%% file: PARTIIIsuperimproved.tex
\documentclass{amsart}

\usepackage{latexsym}
\usepackage{linearb}
\usepackage[LGR, OT1]{fontenc}
\usepackage{amsmath,amsthm}
\usepackage{amssymb}
\usepackage{upgreek}
\usepackage[greek,english]{babel}
\usepackage{teubner}
\usepackage{MnSymbol}
\usepackage{graphicx}
\usepackage{marvosym}
\usepackage{eufrak}
\usepackage[margin=1.25in,dvips]{geometry}
\usepackage{color}
\usepackage{hyperref}
\usepackage{comment}
\usepackage{mathrsfs}

\def\f12{\frac 1 2}

\def\de{\delta}

\def\La{\Lambda}

\def\f12{\frac 1 2}

\newcommand{\nabb}{\mbox{$\nabla \mkern-13mu /$\,}}

\newcommand{\slashg}{\mbox{$g \mkern-9mu /$\,}}

\newcommand\lessflat{{{\mbox{$\flat  \mkern-12mu {}^{\_}$}}}}

\makeindex

\newtheorem{definition}{Definition}[subsection]
\newtheorem{remark}{Remark}[subsection]
\newtheorem{bigremark}{Remark}[section]
\newtheorem{lemma}{Lemma}[subsection]
\newtheorem{theorem}{Theorem}[section]
\newtheorem*{theorem*}{Theorem}
\newtheorem*{corollary*}{Corollary}
\newtheorem{proposition}{Proposition}[subsection]
\newtheorem{bigprop}{Proposition}[section]

\newtheorem{corollary}{Corollary}[subsection]
\newtheorem{bigcorol}{Corollary}[section]

\title[Decay for solutions of the wave equation on Kerr \emph{III}]{Decay for solutions of the wave equation\\
on Kerr exterior spacetimes \emph{III}:\\ The full subextremal
case $|a| < M$}

\author{Mihalis Dafermos}
\address{Princeton University, Department of Mathematics, Fine Hall, Washington Road, Princeton, NJ 08544, United States, {\tt dafermos@math.princeton.edu}}

\address{University of Cambridge, Department of Pure Mathematics and Mathematical Statistics, Wilberforce Road, Cambridge CB3 0WA, United Kingdom, {\tt M.Dafermos@dpmms.cam.ac.uk} }

\author{Igor Rodnianski}
\address{Princeton University, Department of Mathematics, Fine Hall, Washington Road, Princeton, NJ 08544, United States, {\tt irod@math.princeton.edu}}

\address{Massachusetts Institute of Technology, Department of Mathematics, 77 Massachusetts Avenue, Cambridge, MA 02139, United States, {\tt irod@math.mit.edu}}

\author{Yakov Shlapentokh-Rothman}
\address{Massachusetts Institute of Technology, Department of Mathematics, 77 Massachusetts Avenue, Cambridge, MA 02139, United States, {\tt yakovsr@math.mit.edu}}

\date{December 3, 2014}

\begin{document}

\begin{abstract}
This paper concludes the series begun
 in [M.~Dafermos
and I.~Rodnianski
\emph{Decay for solutions of the wave equation on Kerr exterior spacetimes I-II: the cases $|a| \ll M$ or axisymmetry}, arXiv:1010.5132],
providing the complete proof of definitive boundedness and
decay results for the scalar wave equation on  Kerr backgrounds in the general subextremal $|a|<M$ case without symmetry assumptions. The essential ideas of the proof (together with explicit constructions of the most difficult multiplier currents)
have been announced in our survey
[M.~Dafermos and I.~Rodnianski \emph{The black hole stability problem for linear scalar perturbations}, in Proceedings of the 12th Marcel Grossmann Meeting on General Relativity,
T.~Damour et al (ed.), World Scientific, Singapore, 2011, pp.~132–189, arXiv:1010.5137].
Our proof appeals also to the quantitative mode-stability proven in [Y.~Shlapentokh-Rothman \emph{Quantitative Mode Stability for the Wave Equation on the Kerr Spacetime}, arXiv:1302.6902,
to appear, Ann.~Henri Poincar\'e],
together with a streamlined continuity argument in the parameter $a$,
      appearing here for the first time.
While serving as Part III of a series, this paper repeats all necessary notations so that
it can be read independently of previous work.
\end{abstract}

\maketitle

\tableofcontents
\section{Introduction}

The boundedness and decay properties of solutions to the scalar wave equation
\begin{equation}\label{WAVE}
\Box_{g_{a,M}}\psi = 0
\end{equation}
 on the exterior regions of
 Kerr black hole backgrounds $(\mathcal{M},g_{a,M})$
 have been the subject of considerable recent activity, in view of the intimate relation
 of this problem to the stability of these spacetimes themselves in the context of
 Einstein's theory of general relativity (cf.~\cite{ck}).
 Following definitive results~\cite{kw:lss, 2bachelots, dr1, BlueSof0, dr3, BlueSof, BlueSter, dr5}
 in the Schwarzschild case $a=0$, boundedness
 in the very slowly
 rotating Kerr case $|a|\ll M$  was first proven in our~\cite{dr6},
 and subsequently,  decay results
 have been established for $|a|\ll M$ in~\cite{jnotes}
 and in the first parts of this series~\cite{dr7}, and independently
by Tataru--Tohaneanu~\cite{tattoh} and Andersson--Blue~\cite{anblue}. See also~\cite{luk2}. Our~\cite{dr7} also obtained such decay results in
 the general subextremal case $|a|<M$, under the assumption that $\psi$ is itself
 axisymmetric.
 (Let us mention also the previous non-quantitative
 study~\cite{fksy, fksy2} of fixed azimuthal
 modes on Kerr.)
 The main significance of these restrictive assumptions is that the well-known
 difficulty of
 superradiance is controlled by a small parameter
 (the case $|a|\ll M$)
  or is in fact completely absent (the case of axisymmetric $\psi$).
 The present paper represents the culmination of this programme
by dropping these restrictions, extending the above boundedness and decay results to the \emph{general subextremal case $|a|<M$ without axisymmetry}:
\begin{theorem*}
1.~General
 solutions $\psi$ of $(\ref{WAVE})$ on the exterior of a Kerr black hole background
$(\mathcal{M}, g_{a,M})$
\underline{in the full subextremal range $|a|<M$},
 arising from bounded initial energy on a suitable
Cauchy surface $\Sigma_0$, have
bounded energy flux through a global foliation $\Sigma_\tau$ of the exterior,
bounded energy flux through the event horizon $\mathcal{H}^+$ and null infinity
$\mathcal{I}^+$, and satisfy
a suitable version of ``local integrated energy
decay''.

2.~Similar statements hold for higher order energies involving time-translation invariant derivatives.
This
implies immediately uniform \underline{pointwise} bounds on $\psi$ and \underline{all} translation-invariant derivatives
to arbitrary order, up to and including $\mathcal{H}^+$, in terms of a sufficiently high
order initial energy.
\end{theorem*}

The precise statements of parts 1 and 2 of the above Theorem
will be given    in  Section~\ref{statementsections} as Theorems~\ref{theResult} and~\ref{h.o.s.}.
Note that these are exact analogues of Theorems 1.2 and 1.3 of~\cite{dr7},
dropping the assumption of axisymmetry (alternatively, Theorem 1.1 dropping
the assumption $a_0\ll M$).
The main results of the present paper
have been previously announced in our survey paper~\cite{stabi},
which provided both an outline of the proof
and many details of the crucial arguments, including all high frequency multiplier constructions.
To complete the outline,
one required a quantitative refinement of Whiting's classical
mode stability result~\cite{whiting} and a continuity argument in the
parameter $a$. The former refinement has  very recently been obtained~\cite{shlapRot} and will be indeed used in our proof. As for the latter,
the proof presented here introduces a streamlined continuity
argument which as an added benefit in fact \emph{reproves}
the  theorems of the first parts of the series~\cite{dr7} in the case $|a|\ll M$.
We will only rely on~\cite{dr7} for a detailed discussion
of the background, the setup and several of its elementary
propositions. All necessary notations and results from~\cite{dr7} are reviewed and quoted explicitly, however, so that the present paper can be read independently.
We will also repeat all constructions originally introduced in the survey~\cite{stabi}.

In view of our general ``black box'' decay result~\cite{icmp}, a corollary of the above Theorem is
\begin{corollary*}
Solutions $\psi$ of~(\ref{WAVE}) arising from sufficiently regular and \underline{localised}
initial data (i.e.~whose initial suitably higher-order \underline{weighted}
energy is finite) satisfy uniform \underline{polynomial decay}
of the energy flux through a hyperboloidal foliation $\widetilde{\Sigma}_\tau$ of the exterior region as well as uniform pointwise polynomial
decay bounds.
\end{corollary*}
As in the case of Minkowski space, there is a hierarchy of
polynomial decay bounds that can be  obtained, both for energy fluxes and pointwise,
depending on the quantification of the initial localisation assumed on initial data.
The precise statement we  shall give
(Corollary~\ref{thecorol} of Section~\ref{statementsections}) is motivated
by applications to quasilinear problems; further refinements fail to be robust from
this point of view.
We remark explicitly that the decay estimates of the above Corollary are
indeed sufficient for applications to quasilinear problems
with quadratic non-linearities. See~\cite{luk3, shiwu, shiwu2, Yang}.
We note also that the non-quantitative fixed-azimuthal mode
statements of~\cite{fksy, fksy2} are
of course implied \emph{a fortiori} by the above Corollary.
To obtain from our Theorem
alternative more refined corollaries for compactly supported smooth initial data,
see~\cite{tatar}.

As stated, the above Theorem and its Corollary concern the black hole exterior. Note
that boundedness and polynomial decay statements in the Kerr exterior propagate easily to any fixed-$r$ hypersurface in the interior following~\cite{luk}
using the red-shift effect and stationarity, for $r$ \emph{strictly greater} than
its value on the Cauchy horizon. This insight goes back to~\cite{cbh}. On the other hand, by a result
of~\cite{Sbierski}, uniform non-degenerate energy
boundedness cannot hold up to the Cauchy horizon in view of the blue-shift.
Uniform $L^\infty$ bounds on $\psi$ and its \emph{tangential} derivatives up to and including the Cauchy
horizon have been obtained in the thesis of Franzen~\cite{franzen}.

Besides the Kerr family, there is an additional  class of black hole spacetimes of interest in classical general
relativity: the sub-extremal Kerr--Newman metrics. These now form a $3$-parameter family
of metrics, with parameters $a$, $M$ and $Q$ (the latter representing \emph{charge}),
which, when coupled with a suitable
Maxwell field associated to $Q$,
satisfy the \emph{Einstein--Maxwell} equations. See~\cite{he:lssst}.
(Setting $Q=0$, the Maxwell field vanishes and the family reduces to the Kerr
case.) As shown in the thesis of Civin~\cite{civin}, all the miraculous properties of the Kerr family
that allow for the results proven in the present paper in fact  extend to the Kerr--Newman family
in the full sub-extremal
parameter range $\sqrt{a^2+Q^2}<M$, leading to  a precise analogue
of our Theorem and its Corollary in this more general class.

Though outside of the domain of astrophysically relevant black holes,
it is interesting to consider the problem of boundedness and decay
for scalar waves on the
analogues of the Kerr family when
a non-zero cosmological constant $\Lambda$ is added to the Einstein equations.
These spacetimes are known as Kerr-de Sitter ($\Lambda>0$) and
Kerr-anti de Sitter ($\Lambda<0$).
See~\cite{he:lssst}. It is in fact the
negative case $\Lambda<0$ which presents
more surprising new phenomena from the mathematical point of view and
 has been definitively treated in the work
of Holzegel--Smulevici~\cite{holz-smul, holz-smul2}.
See also~\cite{gannot}.
The $\Lambda>0$ case is from some points of view
easier than $\Lambda=0$, and results
in the Schwarzschild-de Sitter $(a=0)$ and very
slowly rotating Kerr-de Sitter $(|a| \ll M, |a|\ll \Lambda)$ case followed soon after
the analogous results for Schwarzschild and very slowly rotating Kerr had been obtained.
See~\cite{dr4, bh, sbvm, vasy, dyatlov1, dyatlov2, schlue2}.
Let us note however, that
Kerr-de Sitter is still not understood in its full subextremal range, in particular in view
of the absence of an analogue of the mode stability statements~\cite{whiting, shlapRot}. The
best results to date have been obtained by Dyatlov~\cite{dyatlov-last}.\footnote{Let
us remark briefly that besides these extensions to $\Lambda\ne 0$, there
are a host of other related problems one can also consider,
including higher dimensional black holes~\cite{schlue, lauletal},
other hyperbolic equations like  Klein--Gordon  (for which
it is now proven~\cite{shlapRot2} that
there are \emph{exponentially growing} solutions for all $|a|\ne 0$),
Maxwell~\cite{anblue2},
linearised gravity and the nonlinear Einstein vacuum equations themselves
(see~\cite{kostakis2, scattering}).
We refer the reader to the many additional references
in the first part of this series~\cite{dr7}, our survey~\cite{stabi} and our lecture notes~\cite{jnotes}.}

Returning to the classical astrophysical domain,
let us recall finally that the Kerr metrics $g_{a,M}$ represent black hole spacetimes in the full \emph{closed} parameter range
$|a|\le M$; the geometry of the extremal case $|a|=M$, however, exhibits several qualitative differences,
most conspicuously, the degeneration
of the celebrated \emph{red-shift effect} at the horizon.
In view of the recently discovered \emph{Aretakis instability}~\cite{aretakis, aretakis2, aretakisHor, luciet, murataetal},
the precise analogue of the above Theorem \emph{does not}
in fact hold without qualification in the case of extremality $|a|=M$. In particular, in
the extremal case, for generic
solutions, certain
higher order time-translation-invariant derivatives asymptotically blow up along the event horizon.
This generic blow up is completely unrelated to superradiance and holds even for solutions $\psi$ restricted to be \emph{axisymmetric}.
For such axisymmetric $\psi$,
restricted decay results of a definitive nature have been obtained by Aretakis~\cite{aretakisKerr}.
 The fundamental remaining problem for scalar waves
on black hole backgrounds of interest in classical general relativity  is thus the precise understanding of  the \underline{$|a|=M$ case} for general
 \underline{non}-axisymmetric solutions.

\subsection{Overview of the main difficulties}
We begin with an overview of the difficulties of the problem and the basic elements
of the proof. In this discussion, we will assume some
familiarity with the Schwarzschild and Kerr
families of metrics as well as basic aspects of the analysis of wave equations
on Lorentzian manifolds. See our
lecture notes~\cite{jnotes}.

\subsubsection{Review of the very slowly rotating case $|a|\ll M$}
\label{reviewsr}
We have discussed at length in the first parts of this series~\cite{dr7}
the various difficulties connected to showing energy boundedness and ``integrated local
energy  decay''
for solutions of the wave equation $(\ref{WAVE})$ in the very slowly rotating case. We review these briefly.

Already
in the Schwarzschild case $a=0$,  to show boundedness, one must face the difficulty that
at the \emph{event horizon} $\mathcal{H}^+$, the conserved $\partial_t$-energy (associated to stationarity of the metric) degenerates.
To show integrated local energy decay, one must moreover understand both ``low'' and ``high'' frequency
obstructions to dispersion,
in particular, the high-frequency obstructions connected to the presence
of \emph{trapped null geodesics}.

The horizon difficulties were overcome by our introduction of the red-shift vector field~\cite{dr3},
while the difficulties concerning both excluding low frequency obstructions and
quantifying the high frequency phenomena connected to trapped geodesics
were overcome in one go by appeal to the energy identity of Morawetz-type (cf.~\cite{mora2}) multipliers associated to a vectorfield $f(r)\partial_r$,
where $f$ is a carefully chosen function vanishing at the so-called photon sphere $r=3M$, the timelike cylinder
to which all trapped null geodesics asymptote.
In the Schwarzschild context, such Morawetz estimates
were pioneered by~\cite{labasoffer, BlueSof}. The boundedness
of the nonnegative space-time integral given by the energy identity associated to this
multiplier is precisely
the statement of ``integrated local energy decay''. See also~\cite{BlueSter, dr3}.
The degeneracy of any such estimate at trapping is necessary in view of
a general result of Sbierski~\cite{Sbierski} in the spirit of the classical~\cite{Ralston}.
On the other hand, the fact that such a degenerate estimate indeed holds (and the precise nature of the degeneracy) means that the trapping is ``good''; at the level of geodesic flow,
this corresponds to the fact that dynamics is hyperbolic near the trapped set.
This estimate also degenerates at the event horizon when
only the initial conserved $\partial_t$-energy is included on the right hand side.
On the other hand, as shown in~\cite{dr3}, again using the red-shift vector field,
this degeneracy is removed by replacing the latter with the initial
\emph{non-degenerate} energy.

Turning from Schwarzschild to the very slowly rotating Kerr case $|a|\ll M$, the above difficulties are combined with a new
one: \emph{superradiance}. Now, the vector field $\partial_t$ is spacelike in a region
outside the horizon known as the \emph{ergoregion},
hence its energy identity gives no obvious \emph{a priori} control
over the solution.
Moreover, it is clear that the high-frequency obstructions to decay
cannot be  captured from classical physical space
vector field multipliers~\cite{alinhac}.
This can be seen at the level of geodesic flow as the projection of the trapped space to physical
space is no longer a codimension-$1$ hypersurface.

The problem of superradiance was first overcome in our proof of boundedness~\cite{dr6}
mentioned previously, which concerned in fact
the more general setting of the wave equation on backgrounds that are small axisymmetric stationary perturbations of Schwarzschild, a  class including the very
slowly rotating Kerr case $|a|\ll M$.
In this class of spacetimes, one can  analyse solutions with respect to
frequencies $\omega\in \mathbb R$ and $m\in \mathbb Z$ corresponding to the stationary and axisymmetric
Killing fields, and decompose general solutions $\psi$ of $(\ref{WAVE})$ into their superradiant and non-superradiant
part. For the latter part, one can prove boundedness as in Schwarzschild.
For the superradiant part, it turns out that one can explicitly prove \emph{both} boundedness
and integrated local energy decay perturbing a Schwarzschild energy identity
yielding both positive boundary and space-time terms that do not degenerate.
The non-degeneracy of this estimate encodes the fact that the \emph{superradiant
part of $\psi$ is not trapped}.
We shall return to this insight later in our discussion of the general $|a|<M$ case.

Turning to the issue of proving decay for $|a|\ll M$,
the problem of capturing the good properties of trapping
was overcome
 using \emph{frequency-localised}
 generalisations of the Morawetz multipliers applied in Schwarzschild.
 There have been three independent approaches~\cite{jnotes, tattoh, anblue}, which all crucially
rely on the additional hidden symmetries of Kerr that are reflected in the existence
of a Killing tensor and  separability
properties of both geodesic flow and the wave equation itself.
In our approach~\cite{jnotes, dr7},
the frequency localisation uses directly
Carter's separation of the wave equation~\cite{cartersep2},
which introduces, in addition to $\omega$ and $m$ above,
a real frequency parameter $\lambda_{m\ell}(a\omega)$ parameterised
by an additional  parameter $\ell\in \mathbb N_0$ such that, localised to
each frequency triple $(\omega, m, \lambda_{m\ell}(a\omega))$, the wave
equation $(\ref{WAVE})$ reduces to the following second order o.d.e.
\begin{equation}
\label{oderedu}
u''+\omega^2 u =V(a\omega, m, \lambda_{m\ell})u
\end{equation}
where $V$ is  a potential and $'$ denotes differentiation in a rescaled $r^*(r)$ coordinate.
The frequencies $\lambda_{m\ell}$ are themselves eigenvalues of
an associated elliptic equation whose eigenfunctions (known as
oblate spheroidal harmonics) appear in the formula reconstructing $\psi$ from $u$.

Note that
in the Schwarzschild ($a=0$) case, the reduction to $(\ref{oderedu})$ corresponds to the classical separation by
spherical harmonics,  and we have explicitly $\lambda_{m\ell}={\ell(\ell+1)}$ (independent of $\omega$)
and
\[
V_{\rm Schw}(r)=(r-2M)\left(\frac{\lambda_{m\ell}}{r^3}+\frac{2M}{r^4}\right).
\]
For all $\lambda_{m\ell}$, the potential $V_{\rm Schw}$ has a unique simple maximum at
an $r$-value $r_{\rm max}(\lambda_{m\ell})$
such that
\[
r_{\rm max}(\lambda_{m\ell})\to 3M
\]
as $\ell\to \infty$. One sees thus the relation of this potential to the trapping phenomenon.
Indeed, the radial dependence of null geodesics with fixed energy and angular momentum is governed by a potential
which coincides with $V$ above in the high frequency rescaled limit.

In the Kerr case, the eigenvalues $\lambda_{m\ell}(a\omega)$ are no longer explicit expressions, and
the form of $V$ is considerably more complicated.
It was shown, however, in~\cite{jnotes}, that for $|a|\ll M$ and
for frequencies in the ``trapping'' regime
\begin{equation}
\label{trappingfr}
1\ll \omega^2 \sim \lambda_{m\ell},
\end{equation}
the potential $V$ retains
its ``good'' Schwarzschild properties. Specifically,
 the potential $V$ in $(\ref{oderedu})$
can be seen to again have a unique simple maximum in this frequency range, whose $r$-value
$r_{\rm max}$
depends on the frequency
parameters
\[
r_{\rm max}=r_{\rm max}(a\omega, m, \lambda_{m\ell}).
\]
This allows, separately for each $(\omega, m , \lambda_{m\ell}(a\omega))$,
the construction of an analogue of the current $f(r)\partial_r$ vanishing exactly at $r_{\rm max}$, yielding the desired positivity properties.
Unlike the Schwarzschild case, however, there does not exist a unique high frequency
limit of $r_{\rm max}$, consistent with the fact~\cite{alinhac} that these currents cannot be
replaced by a classical vector field defined in physical space. See however~\cite{anblue}.
At the level of
geodesic flow, this precisely reflects the fact
that trapped null geodesics exist for a full range of $r$-values in a neighbourhood of
$r=3M$.\footnote{A posteriori, the good structure of trapping in phase space
for $|a|\ll M$
can be understood more conceptually, using
the structural stability properties of normal hyperbolicity, provided the latter
condition is checked for Schwarzschild; see~\cite{WunschZworski}.
Note however that these stability properties depend on
strong regularity assumptions on the metric, whereas
our original boundedness theorem~\cite{dr6} only requires closeness to Schwarzschild in $C^1$.
Thus one expects the domain of validity of~\cite{dr6} to be strictly
bigger than the class of spacetimes where decay results of the type of our main Theorem hold.}

In the remaining frequency regimes, one can in fact  simply  carry over the previous physical-space
Schwarzschild  constructions (see our argument in~\cite{jnotes}) to the more general $|a|
\ll M$ case,
as, restricted to those ranges, the relevant Schwarzschild estimates \emph{do not} degenerate and thus their positivity properties are manifestly stable to small perturbation.
Alternatively, as in the first parts of this series~\cite{dr7},
one can construct  new currents taylored specifically to these frequency ranges.
The latter approach is more flexible (it has the additional benefit of providing~\cite{dr7}
an independent second proof of the Schwarzschild case) and will
 be more useful for the general subextremal case $|a|<M$.

 Let us note that our proof of integrated decay in the first part of this series~\cite{dr7} in fact gives
 a separate proof of
 the boundedness statement of~\cite{dr6}, when the latter is specialised to Kerr. Here, one
exploits the fact that when $|a|\ll M$,  superradiance is controlled by a
small parameter  and thus
boundary terms of the wrong sign can be absorbed by a small multiple of
the red-shift current added to the conserved $\partial_t$ energy.
One obtains thus boundedness and integrated local energy decay
at the same time, without separating the solution into its
superradiant and non-superradiant parts. We shall see,
however, that
for the general case $|a|<M$,
the original insight of~\cite{dr6} will again be fundamental.

In discussing our frequency analysis for both the problems of boundedness and integrated
local energy decay, we have suppressed an important point: To define frequencies $\omega\in
\mathbb R$,
we are applying the Fourier transform in time.
Since  solutions are not known \emph{a priori} to be sufficiently integrable in time, however,
one must first apply suitable cutoffs in the future, generating error terms which must then be absorbed. For this, some weak \emph{a priori} control of these terms
is essential--and here we have used in both~\cite{dr6, dr7} yet
 again\footnote{Let us note that given
  the boundedness result of~\cite{dr6}, then
 one need not appeal again to closeness to Schwarzschild in the
 argument for integrated local energy decay; see our original proof in the lecture notes~\cite{jnotes}. We  have used it again in~\cite{dr7} so as for~\cite{dr7} to retrieve independently
 our previous boundedness result.} the closeness to Schwarzschild.
The analogue of this procedure for the general $|a|<M$ case is our appeal
to the continuity argument of Section~\ref{continuityintro}. We  defer further discussion of this till then.

\subsubsection{Structure of trapping and its disjointness from superradiance}
\label{highfreqinsec}
In passing to the general subextremal case $|a|<M$,
the first fundamental issues that must be addressed
are the ``high frequency'' ones.

The most obvious question is whether
the structure of trapping retains its ``good'' properties which allow in principle for degenerate
integrated decay statements.
At the same time, since superradiance is no longer governed
by a small parameter, one has to understand what is it which finally quantifies its strength,
or, in the context of the proof, how does one guarantee the control of
boundary terms in space-time energy
identities.

Though geodesic flow and various measures of the strength of superradiance have been thoroughly investigated in physics~\cite{carter, chandrasekhar, Starobinsky},  the
properties that turn out to be essential for our argument do not
appear to have been considered explicitly in the classical  literature.
Indeed, it  is somewhat  of a miracle that the Kerr geometry turns out to be well behaved on both accounts,
\underline{for the entire subextremal range $|a|<M$}, specifically:
\begin{enumerate}
\item[(a)]
\emph{The structure of trapping is as in Schwarzschild.}
\item[(b)]
\emph{Superradiant frequencies are not trapped.}
\end{enumerate}

The embodiment of properties (a) and (b)  we shall need were proven already
in our survey paper~\cite{stabi}
by analysing the potential $V$ in $(\ref{oderedu})$.
Concerning (a),
the ``trapping'' frequency range $(\ref{trappingfr})$,
the potential $V$ was  shown to have a unique simple
maximum $V_{\rm max}$
at a (frequency dependent) $r$-value
$r_{\rm max}$, just as in the slowly rotating case $|a|\ll M$.
(This shows \emph{a fortiori} that the underlying  null geodesic flow near
trapping is
hyperbolic.)
Concerning (b), it was
 shown that for high \emph{superradiant}
frequencies,
\begin{equation}
\label{highsupers}
1\ll \lambda + \omega^2,\qquad \omega\left(\omega-\frac{am}{2M(M+\sqrt{M^2+a^2}}\right)<0
\end{equation}
the maximum of $V$ is always
``quantitatively'' above the energy level $\omega^2$, in the sense
\begin{equation}
\label{quantabove}
V_{\rm max} \ge (1+\epsilon) \omega^2.
\end{equation}
(This in turn shows \emph{a fortiori}  that future directed
null geodesics whose tangent vector
has nonnegative inner product with $\partial_t$ are never future trapped; they will always
cross the  event horizon $\mathcal{H}^+$.\footnote{In the borderline case when the time frequency $\omega$ vanishes, this is intimately related to the fact that there are no trapped null geodesics orthogonal to $\partial_t$. This latter observation turns out to be important in the study of black hole uniqueness (see~\cite{alexakis}).}
 \emph{Note that in contrast to the
$|a|\ll M$ case, the stronger
statement that $\partial_t$ is eventually timelike along  any future trapped null
geodesic is not true; i.e.~it is not true that all future trapped null geodesics must leave the ergoregion.})

As with the $|a|\ll M$ case, it is not statements about geodesic flow that we appeal to,
but rather we use directly the properties of the potential $V$ to construct
appropriate energy currents. More specifically, the above properties
 of the potential are used to construct frequency dependent
multiplier currents  yielding both positive bulk \emph{and positive boundary terms} for all
high frequency ranges.
In the superradiant case, property $(\ref{quantabove})$ can be exploited
to arrange such that the bulk term is in fact non-degenerate;
this can be thought of as the definitive embodiment of (b). Note
that these  high-frequency multiplier constructions all appeared
explicitly in our survey~\cite{stabi}.
We will repeat these constructions here with very minor modifications.
See the outline in Section~\ref{outlinesection} below.

It is interesting to note that property (b) above in fact degenerates in
the extremal limit $|a|\to M$ in the following sense: At the endpoint of the superradiant
frequency range $(\ref{highsupers})$, one loses the $\epsilon$ in $(\ref{quantabove})$.
This is an additional (and separate) phenomenon to the degeneration of the red-shift
and
could have interesting implications for the remaining problem of understanding
non-axisymmetric $\psi$
in the extremal case $|a|=M$. See~\cite{AndGlam} and
the discussion in~\cite{aretakisKerr}.

\subsubsection{Absence of bounded frequency superradiant obstructions}
\label{nolow}
The above still leaves us with the problem of understanding \emph{bounded} (i.e.~$|\omega|\lesssim1$) frequencies.

One must first distinguish the near stationary
case $|\omega|\ll1$. This frequency range is
very sensitive to global aspects of the geometry.
It turns out that here an explicit multiplier construction is possible
which adapts our construction of the first parts of this series~\cite{dr7}.
(Interestingly, the cases of $|a|\le \tilde{a}_0$ and $|a|\ge \tilde{a}_0$ are here
handled differently.) These multiplier constructions appear for the first
time in the present paper.

Turning now to the remaining bounded frequencies,
as explained in our survey~\cite{stabi}, whereas
in the  \emph{non-superradiant}
regime, one can explicitly construct multipliers with both nonnegative
bulk and boundary terms,
for bounded non-superradiant frequencies,
adapting the constructions of~\cite{dr7} from the $|a|\ll M$ case,
there does not appear to be a straightforward such construction
for the superradiant regime, when neither can superradiance
be treated as a small parameter, nor can
one exploit $(\ref{quantabove})$ together with either $\omega$, $m$ or $\lambda_{m\ell}$
as a large parameter.
One can indeed construct
currents with a non-negative \emph{bulk} term, but these generate a
boundary term of the wrong
sign which still must be controlled.

As announced already in~\cite{stabi}, to control the remaining term one
requires a quantitative
extension of Whiting's celebrated mode stability~\cite{whiting}, which in particular
excludes the presence not just of growing modes but also resonances on the real axis.
This was achieved in the recent~\cite{shlapRot}.
Appeal to~\cite{shlapRot} will indeed allow us to control the remaining boundary term.
Again, see the outline in Section~\ref{outlinesection}  below.

\subsubsection{Higher
order estimates}
\label{HOEhere}

To obtain higher-order integrated local energy decay in the slowly rotating case $|a|\ll M$, it was sufficient to commute $(\ref{WAVE})$
with $\partial_t$ (which is Killing) and also with the red-shift vector field (the latter an argument
first applied in~\cite{dr6}), exploiting the fact that the latter, though not Killing, generates
positive terms in appropriate energy estimates
modulo terms which can be controlled by the $\partial_t$-commutation.
To show this fact, one uses in turn that control of a second derivative of $\psi$ in  a timelike direction allows control of all second derivatives
of the solution via elliptic estimates (in view of equation $(\ref{WAVE})$).

For the general case $|a|<M$, one appeals to yet another fundamental fact about
Kerr geometry:
\begin{enumerate}
\item[(c)]
\emph{The span of the stationary $\partial_t$ and axisymmetric $\partial_\phi$
Killing fields is timelike outside the horizon for the full range $|a|<M$.}
\end{enumerate}
Thus,  commuting with $\partial_t$, $\chi\partial_\phi$ (where $\chi$ is a cutoff function with
compact support in $r$) and the red-shift
vector field, one can essentially apply the same argument as before.

\subsubsection{Continuity argument}
\label{continuityintro}
We now return to the issue that we have suppressed at the end of Section~\ref{reviewsr}, namely,
the question of how can one justify in the first place
a   frequency analysis based on real frequencies
$\omega$ defined via the Fourier transform in time.
In the case $|a|\ll M$,  closeness to Schwarzschild
gave a small parameter that could be exploited here.
For the general $|a|<M$ case, however, as explained already in our survey~\cite{stabi},
one must exploit a continuity argument in  $|a|$.

Note first that to
justify the Fourier assumption and thus prove integrated local energy decay,
one sees easily that it is sufficient to assume the
 non-quantitative assumption that the energy through
$\Sigma_\tau$ of the projection $\psi_m$ of $\psi$ to each
azimuthal frequency is finite.
This is the statement that we show by continuity:
For each azimuthal frequency number $m\in\mathbb Z$, we define the
subset
\[
\mathcal{A}_m\subset [0,M)=\{|a|: \psi {\rm\  satisfying\ } (\ref{WAVE})
{\rm\ with\ } g_{a,M}\implies {\rm energy\ of\ }\psi_m{\rm\ remains\ finite}\}
\]
We will show that $\mathcal{A}_m$ is a non-empty open and closed
subset of $[0,M)$, and thus, $\mathcal{A}_m=[0,M)$.

We turn to a brief account of the continuity argument.

The non-emptyness of $\mathcal{A}_m$ follows from
 the general boundedness
result for black hole spacetimes without ergoregions proven
in~\cite{jnotes}, specialised to the Schwarzschild case $a=0$.

For openness, one shows that  if
$\mathring{a}\in \mathcal{A}_m$, then  $|a-\mathring{a}|<\epsilon$ satisfies $a\in \mathcal{A}_m$
for sufficiently small $\epsilon$.
One exploits here $\epsilon$ as a small parameter.
The issues associated to openness already appeared in the small $|a|\ll M$ case; see~\cite{dr6}  and~\cite{dr7}.
 The fact that we have fixed the azimuthal mode $m$ makes
 the argument here technically easier to implement.
For this, the fundamental insight is that \underline{for $m$ fixed}
\begin{enumerate}
\item[1.]
trapping occurs \emph{outside the ergoregion}.
\item[2.]
using the energy identity for a vector field of the form $\partial_t+\alpha(r)\partial_\phi$,
one can obtain boundedness modulo lower order terms supported only
in the ergoregion.
\end{enumerate}
To exploit the above,
we first construct  from a fixed-$m$
solution $\psi_m$ to $(\ref{WAVE})$ on $g_{a,M}$ and for each
$\tau\ge 0$, a solution
$\Psi$ of the inhomogeneous wave equation $\Box_g\Psi_m=F_m$
on an interpolating metric $g$ which coincides with $g_{a,M}$ in the region between
$\Sigma_0$ and $\Sigma_{\tau-\delta_0}$ and coincides
with $g_{\mathring{a},M}$ in the region in the future of
$\Sigma_{\tau}$ and to which the integrability properties apply
(since $\mathring{a}\in\mathcal{A}_m$).
Applying our estimates and
using 1.~and 2., we may now absorb (for sufficiently small $\epsilon$)
the error terms arising from the inhomogeneity
to obtain an integrated decay statement for $\psi_m$.
We note that the fixed-$m$ currents used for 1.~and 2.~may find additional applications.

Closedness is easy given the estimates shown and the smooth dependence
of the Kerr family on the parameter $a$.

\subsubsection{Non-degenerate boundedness from integrated local energy decay}
\label{boundednessintro}

The frequency analysis on which our           proof of integrated
local energy decay is based does not directly ``see''
the energy flux on fixed time hypersurfaces $\Sigma_\tau$, only the
energy fluxes on the horizon $\mathcal{H}^+$ and future null infinity $\mathcal{I}^+$.
Thus, it remains to show boundedness of the energy (and higher-order energies)
through $\Sigma_\tau$.

In the slowly rotating case $|a|\ll M$, it is clear that
given integrated local energy decay,
boundedness of the energy flux through a spacelike foliation
easily follows \emph{a posteriori}\footnote{Of course, in our original proof~\cite{dr7},
we proved
those two statements together as we used the boundedness in our version of the continuity
argument. In the new continuity argument presented here, this is not necessary.} by   revisiting the physical space energy  identity of a globally
timelike vector field which coincides with $T$ where the latter is timelike, noting
that, if $|a|$ is sufficiently small $T$ is timelike near trapping.

The above argument again uses in an essential way the disjointness
of the ergoregion and the set--associated to trapping--on which integrated local energy
decay estimate degenerates. As we have remarked earlier, these
sets intersect when $|a|\sim M$--it is only in phase space where
superradiance can be understood as disjoint of trapping.

One approach to boundedness could be to try to exploit again property (b) from
Section~\ref{highfreqinsec}.
It is technically easier to simply exploit the
physical space fact (c) of Section~\ref{HOEhere}, namely that
 Killing fields $\partial_t$ and $\partial_\phi$ together span
a timelike subspace outside the horizon.
Specifically, in
a small neighbourhood of any $r$-value there exists a  combination
of $\partial_t$ and $\partial_\phi$ which is timelike and Killing. We use our
frequency analysis to partition a solution $\psi$
of the wave equation into finitely many pieces
$\tilde\psi_i$, each of which satisfies an analogue of
integrated local energy decay degenerating only in a small neighbourhood of some
$r_i$. Applying the energy estimates corresponding to a suitable $i$-dependent
combination of $\partial_t$ and $\partial_\phi$ to each
$\tilde\psi_i$, and summing, one obtains the desired non-degenerate
uniform boundedness of the energy flux through $\Sigma_\tau$.

\subsection{Outline of the paper}
\label{outlinesection}
We end this introduction with an outline of the structure of the paper.

In Section~\ref{notats}, we will review the set-up and various notations from the first parts of the series~\cite{dr7},
including the ambient manifold, the form of the Kerr family of metrics and useful vector
fields, hypersurfaces and formulas.
This will allow us to give precise formulations of the main
theorems in Section~\ref{statementsections}.
(The reader may wish to refer to this outline again when reading Section~\ref{theLogic},
which will describe the logical flow of the proofs of the various statements.)

Section~\ref{prelimn} contains various preliminaries, including
a review of the propositions
from~\cite{dr7} capturing the redshift effect, an estimate for large $r$, Hardy inequalities and finally, various statements concerning the span of the Killing fields $\partial_t$ and
$\partial_\phi$.

Our frequency localisation based on Carter's separation will be reviewed in
Section~\ref{cartersupersection}. The natural setting for this will be the
class of sufficiently integrable outgoing
functions $\Psi:\mathcal{R}\to \mathbb R$,
a useful notion which we shall define in Section~\ref{suffInt}.
The resulting coefficients $u$ and their corresponding
radial o.d.e.~(cf.~$(\ref{oderedu})$ above) are
obtained in Section~\ref{separationSubsection}
and the ``outgoing'' boundary conditions in Section~\ref{outgoingcondsec}.

The next three sections, Sections 6, 7, and 8, concern the study of
the o.d.e.~$(\ref{oderedu})$ and the proof of uniform estimates
in the frequency parameters $\omega$, $m$, and $\Lambda$.

In Section~\ref{Vpropsec}, we will give salient properties of the potential   $V$
of $(\ref{oderedu})$ which
embody (a) and (b) of Section~\ref{highfreqinsec}.
Versions of the
lemmas of Sections~\ref{Vtrsec} and~\ref{Vsrsec}
together with proofs have in fact already been given in our survey paper~\cite{stabi};
we repeat these here for completeness.
The lemma of
Section~\ref{Vnewsec}, reflecting the properties of trapping for fixed $m$,
is new and will be used in the context of the continuity argument
of Section~\ref{continuityargsec} discussed below.

In Section~\ref{sct}, we shall review
our notation for fixed frequency current templates, which, upon selection
of the free functions, will be used to obtain multiplier estimates
for solutions to $(\ref{oderedu})$.

Section~\ref{freqLocEst} is the heart of the paper. Here, with the help of
well-chosen
functions in the current templates of Section~\ref{sct},
we construct suitable currents
for all relevant frequency ranges
yielding positive bulk terms and thus an estimate for solutions
of the radial o.d.e.~$(\ref{oderedu})$ uniform in frequency parameters.
In the trapping regime, the currents
degenerate at $r_{\rm max}$.
All these currents have appeared previously in our survey paper~\cite{stabi}
with the exception of the near-stationary range of Section~\ref{nearstat}.
The boundary terms can also be made positive, with the exception
of a range of bounded frequencies,  which give rise to
an extra horizon boundary term on the right
hand side of the resulting estimate, which must still be absorbed.

In Section~\ref{summation}, we apply the results of
the previous section  to the coefficients $u$ arising from
the setting of Section~\ref{cartersupersection},
summing the resulting frequency localised estimates to obtain
control of a non-negative definite space-time integral. We note Section~\ref{whitinghere},
where the extra horizon term (arising from low superradiant frequencies)
is  bounded      by  appeal to Proposition~\ref{propShlapRot}, a result of~\cite{shlapRot}.
One obtains finally an
integrated local decay statement for ``future integrable'' solutions
of the  wave equation, and a similar statement for the inhomogeneous equation
in Section~\ref{ILEDinhomo}.

 Higher order decay estimates are then
provided in Section~\ref{higher}, using the structure described in Section~\ref{HOEhere}.

In Section~\ref{continuityargsec}, we implement our
 new  continuity argument discussed in Section~\ref{continuityintro} above, which
 will
 allow us to drop the \emph{a priori} assumption of future integrability,
 and extend our results to general
solutions of the Cauchy problem for $(\ref{WAVE})$.
The reduction to fixed azimuthal frequency is accomplished
in Section~\ref{reduxfixed}.
The most difficult part of the argument is openness, handled in Section~\ref{opennesssec}, while closedness
is considered in Section~\ref{closednesssec}.

Section~\ref{precise} will state
the more precise integrated local energy statement which has actually been obtained
in the proof.

Finally, in Section~\ref{boundSuff}, we prove the boundedness statements, following
our discussion in Section~\ref{boundednessintro}.
This will conclude the paper.

\subsection{Acknowledgements}
During the period where this research has been carried out,
MD was supported by the European Research Council and  the
Engineering and Physical Sciences Research Council.
IR acknowledges support through NSF grant DMS-1001500 and DMS-1065710.
YS acknowledges support through NSF grants DMS-0943787 and DMS-1065710.
The authors thank S.~Aretakis, D.~Civin and G.~Moschidis for discussions and comments
on the manuscript.

\section{Review of the setup}
\label{notats}
In this section, we review the setup and certain notations from the first parts of the
series~\cite{dr7}, so that the present paper can be read independently. The reader
wishing for a more leisurely exposition of this material should refer back to~\cite{dr7};
he or she familiar with~\cite{dr7} can skip to Section~\ref{statementsections}.

\subsection{Ambient manifold and coordinate systems}
The first task is to define an ambient manifold-with-boundary on which the
Kerr family in its subextremal range defines a smooth two-parameter family of
metrics.  The differential structure of
the smooth manifold is defined by what we shall call \emph{fixed coordinates},
while the Kerr metric itself will be defined with the help of auxilliary coordinates
depending on the parameters. We review this here:

\subsubsection{Fixed coordinates $(y^*,t^*,\theta^*,\phi^*)$}
We define first the manifold-with-boundary
\begin{equation}
\label{manwb}
\mathcal{R}=\mathbb R^+\times\mathbb R\times  \mathbb S^ 2.
\end{equation}
Fixed coordinates are just the standard $y^*\in\mathbb R^+$, $t^*\in\mathbb R$ and a choice of standard spherical coordinates $(\theta^*,\phi^*)\in\mathbb S^2$.
Associated to this ambient differentiable structure are the
\emph{event horizon} $\mathcal{H}^+\doteq\partial\mathcal{R}=\{y^*=0\}$,
 the vector fields $T=\partial_{t^*}$, $\Phi=\partial_{\phi^*}$ and
the one-parameter group of transormations $\varphi_\tau$ generated by $T$.

\subsubsection{Kerr-star coordinates $(r,t^*,\theta^*,\phi^*)$}
\label{kerrstardefsec}
We define a new coordinate system which depends on parameters $|a|<M$.

For each choice $|a|<M$, we
first set $r_\pm=M\pm \sqrt{M^2-a^2}$ and then define
a new coordinate $r$ which
is related smoothly to $y^*$, depends smoothly on the parameters and such that,
for fixed parameters, we have $r=r_+(a,M)$ on $\mathcal{H}^+$.\footnote{The precise relation
to fixed coordinates as defined in~\cite{dr7} is as follows:
Let
$\mathcal{P}=\{(x_1,x_2):0\le |x_1|<x_2\}$ denote the parameter space of all admissible
subextremal $(a,M)$.
We chose a smooth map $r:\mathcal{P}\times(0,\infty) \to (x_2+\sqrt{x_2^2-x_1^2}, \infty)$
such that $r|_{\{(x_1,x_2)\}\times (0,\infty)}$
is a diffeomorphism $(0,\infty)\to  (x_2+\sqrt{x_2^2-x_1^2}, \infty)$ which moreover
restricts to the identity map restricted
to $\{(x_1,x_2)\}\times(3x_2,\infty)$. Note that with this definition, then for $r\ge3M$, $r(y^*)$
is independent of $a$.}
Associated to these coordinates is the vector field $Z^*$, defined
to be the smooth extension of the
Kerr-star coordinate vector field $\partial_{r}$ to $\mathcal{R}$.

We will sometimes replace $r$ by a rescaled version, $r^*$,
defined only in the interior of $\mathcal{R}$, by
\begin{equation}
\label{r*def}
\frac{dr^*}{dr}=\frac{r^2+a^2}{\Delta}, \qquad r^*(3M)=0,
\end{equation}
where $\Delta= (r-r_+)(r-r_-)$.
Here we note that $\Delta$ vanishes to first order on $\mathcal{H}^+$,
and the coordinate range $r>r_+$ corresponds to the range $r^*>-\infty$.

\subsubsection{Boyer-Lindquist coordinates $(r,t, \theta,\phi)$}
We define a final coordinate system, again depending on a choice
of fixed parameters $|a|<M$, by further transforming Kerr star coordinates,
by defining
\[
t(t^*,r)= t^* -  \bar t(r), \qquad
\phi(\phi^*,r)= \phi^*- \bar \phi(r) \mod 2\pi, \qquad
\theta=\theta^*
\]
where $\bar{t}$ is a smooth function (see~\cite{dr7} for details)
chosen to satisfy
\begin{equation}
\label{tnearh}
\bar t(r)=r^*(r)-r-r^*(9M/4)+9M/4, \qquad {\rm for}\qquad r_+\le r\le 15M/8,
\end{equation}
\begin{equation}
\label{tnonearh}
\bar t(r)=0 \qquad {\rm for}\qquad r\ge 9M/4,
\end{equation}
\begin{equation}
\label{forspacel}
\frac{d(r^*-\bar{t})}{dr}>0, \qquad 2-\left(1-\frac{2Mr}\rho^2\right)\frac{d(r^*-\bar{t})}{dr}>0.
\end{equation}
Associated to these coordinates is the vector field $Z$ defined
to be (the  extension to  ${\rm int}(\mathcal{R})$ of) the
Boyer-Lindquist coordinate vector
field $\partial_{r}$.\index{vector fields! $Z$ (the Boyer-Lindquist $\partial_r$ coordinate
vector field)}\footnote{Recall that
this vector field is significant as
it will define the directional derivative that does not degenerate in the integrated
decay estimate due to trapping. Note that in Boyer--Lindquist coordinates
the fixed vector fields $T$ and $\Phi$ correspond to the coordinate vector fields
$\partial_t$ and $\partial_\phi$.}

\subsection{The Kerr metric and its properties}
Given these coordinate systems, we may now define the Kerr metric
as a smooth $2$-parameter family on $\mathcal{R}$.

\subsubsection{Explicit form of the metric}
For fixed parameters $|a|<M$, in addition to $\Delta$ above, let us  first set
$\rho^2=r^2+a^2\cos^2\theta$.
The Kerr metric is then defined with respect to Boyer-Lindquist coordinates by
\begin{align}
\label{eleme}
g_{a,M}=
-\frac{\Delta}{\rho^2}\left(dt-a\sin^2\theta d\phi\right)^2
+\frac{\rho^2}{\Delta}dr^2+\rho^2d\theta^2 +\frac{\sin^2\theta}{\rho^2}
\left(a\,dt-(r^2+a^2)d\phi\right)^2.
\end{align}
Though a priori this is only well defined on ${\rm int}(\mathcal{R})$, by
transforming the above into regular coordinates (see~\cite{dr7}), one sees that
the metric $(\ref{eleme})$ extends uniquely to the boundary so that for each
$|a|<M$,
indeed $(\mathcal{R},g_{a,M})$ defines a smooth Lorentzian manifold-with-boundary,
and
such that moreover the metric smoothly
depends on the parameters $a$, $M$.\footnote{The latter can be understood
in the sense that
\[
g:\mathcal{P}\times \mathcal{R} \to T^*\mathcal{R}\otimes T^*\mathcal{R}
\]
is a smooth map.}
These metrics are Ricci flat (i.e.~they satisfy Einstein's vacuum equations).

\subsubsection{Killing fields}
\label{Killingfieldsec}
We note that the fixed vector fields $T$ and $\Phi$ on $\mathcal{R}$
defined in Section~\ref{kerrstardefsec}
are Killing for $g_{a,M}$ for all parameter values $|a|<M$.

For each given $|a|<M$,
the span of $T$ and $\Phi$ yields a timelike subspace of $T_p\mathcal{R}$ for all $p\in
{\rm int}(\mathcal{R})$ (in particular,
$T$ is a timelike vector when $\Phi = 0$).
The event horizon $\mathcal{H}^+=\partial\mathcal{R}$ is
also a \emph{Killing horizon}: the Killing field given by the linear combination
\[
K=T+  \upomega_+\Phi,
\]
where $\upomega_+ \doteq \frac{a}{2Mr_+}$ is the ``anuglar velocity'' of the event horizon, is null and normal to $\mathcal{H}^+$; thus, $\mathcal{H}^+$ is in particular
a null hypersurface.
Note that along $\mathcal{H}^+$ we have
\begin{equation}
\label{eutuxws}
\nabla_K K=
\kappa \, K,\qquad
\kappa= \frac{r_+-r_-}{2(r_+^2+a^2)}>0.
\end{equation}
The quantity $\kappa$ is known as the \emph{surface gravity}.\index{surface gravity}
The positivity $(\ref{eutuxws})$ is what determines the red-shift property, essential
for our estimates (see Section~\ref{Nmult}).
We note that $\kappa$ in fact vanishes in the extremal case $|a|=M$;
this gives rise to the Aretakis instability~\cite{aretakisHor}.

We recall moreover  that the vector $K$ restricted to $\mathcal{H}^+$ coincides with the
smooth extension of the coordinate vector field $\partial_{r^*}$
of the
$(r^*,t,\theta, \phi)$ coordinate system.

\subsubsection{The photon sphere and trapping parameters}
\label{weritsdef}
It is well known that in the Schwarzschild case $a=0$, all future-trapped
null geodesics
asymptote to the timelike hypersurface $r=3M$.

In the statement of Theorem 1.1  of~\cite{dr7}, we defined $s_\pm(a_0,M)$ such that
for all $|a|\le a_0$, then
$r_+<3M-s_-(a_0,M)$ and
all future trapped null geodesics enter the region $3M-s_-(a_0,M)< r<  3M-s_+(a_0,M)$.
We have shown in Section 10.4 of~\cite{stabi} the  existence of such parameters again,
for the full subextremal range $|a|<M$. We will repeat this proof in Section~\ref{trappingParam}.
We note that in the extremal limit $a_0\to M$, $3M-s_-\to r_+(M,M)$.

Given the above parameters,
let $\eta_{[3M-s^-,3M+s^+]}(r)$ denote the indicator function,
and let us define, for each $a_0<M$, the function
\begin{equation}
\label{degenerationfunc}
\zeta(r)= (1-3M/r)^2(1- \eta_{[3M-s^-,3M+s^+]}(r)).
\end{equation}
This function will encode physical space degeneration of the ``integrated
local energy decay'' estimate of
Theorem~\ref{theResult}. The presence of the $(1-3M/r)^2$ factor  ensures
uniformity of the estimate as $a_0\to 0$ so as to retrieve our original Schwarzschild
result~\cite{dr3}.

Finally, since it is derivatives with respect to the vector field $Z$ which do not
degenerate at trapping, but
it is the vector field $Z^*$ which extends to the horizon,
it will be convenient to define
a hybrid vector field that has both good properties.
For this
let us define, for each $|a|<M$, a cutoff funtion $\chi(r)$ such that
$\chi=1$ for $r\ge 3M-s^{-}$ and $\chi=0$ for $r\le (r_++3M-s^-)/2$,
We  define then a new vector field $\tilde{Z}^*=\chi Z+(1-\chi)Z^*$.
This will be the vector field which appears in the statement of  Theorem~\ref{theResult}.

\subsubsection{The ergoregion}
The region $\mathcal{S} \subset \mathcal{R}$ where $T$ is spacelike is known as the \emph{ergoregion}; more explicitly, it is exactly the subset of $\mathcal{R}$ defined
by
\begin{equation}
\label{ergoregionS}
\mathcal{S}=\{ \Delta - a^2\sin^2\theta < 0\}.
\end{equation}
The boundary $\partial\mathcal{S}$ is called the \emph{ergosphere}.

\subsubsection{The $\Sigma_\tau$ hypersurfaces, and the regions
$\mathcal{R}_{(0,\tau)}$, $\mathcal{H}^+_{(0,\tau)}$}\label{hypersurfaceDef}
We have arranged the definition of Kerr-star coordinates in Section~\ref{kerrstardefsec}
so that the hypersurfaces $t^*=c$ are spacelike (see the conditions $(\ref{forspacel})$)
with respect to the metric $g_{a,M}$, for all values of parameters $|a|<M$.

In the region $r\le 15M/8$, we have in fact
\[
g(\nabla t^*, \nabla t^*)= -1-\frac{2Mr}{\rho^2}.
\]

We will define
\[
\Sigma_\tau=\{t^*=\tau\},
\]
\[
\mathcal{R}_{(0,\tau)}=\cup_{0\le \tau^*\le\tau} \Sigma_{\tau^*}
\]
and
\[
\mathcal{R}_0=\cup \mathcal{R}_{(0,\tau)}.
\]
Note that
$\Sigma_0$ is a past Cauchy hypersurface for the regions
$\mathcal{R}_{(0,\tau)}$, $\mathcal{R}_0$.
Let us also define
\[
\mathcal{H}^+_{(0,\tau)} = \mathcal{R}_{(0,\tau)}\cap\mathcal{H}^+,
\qquad
\mathcal{H}^+_0 =  \mathcal{R}_{0}\cap\mathcal{H}^+.
\]

\subsubsection{Angular derivatives and the volume form}\label{usefulcomps}
For future reference, let us introduce here the notation $\slashg$, $\nabb$
 to denote the induced metric and covariant
derivative from $g_{a,M}$~(\ref{eleme})
on the $\mathbb S^2$ factors of $\mathcal{R}$ in the
product $(\ref{manwb})$.

We record finally from~\cite{dr7} some useful properties of the volume form $dV$
of the metric $g_{a,M}$:
With respect to Boyer-Lindquist coordinates, we have
\[
dV= v(r,\theta) \, dt\, dr\, dV_{\slashg} \qquad \text{\rm with\ } v\sim 1
\]
whereas
using the alternative $r^*$ coordinate,
\[
dV =v(r^*,\theta) \,  dt\, dr^*\, dV_{\slashg} \qquad \text{\rm with\ } v\sim \Delta/r^2.
\]
With respect to Kerr-star coordinates, we have
\[
dV =v(r,\theta^*) \,  dt^*\, dr\, dV_{\slashg} \qquad \text{\rm with\ } v\sim 1.
\]

Let $\gamma$ denote the standard unit metric on the sphere
in $(\theta,\phi)$ coordinates.
We have that $\slashg \sim r^2\gamma$, and thus we may
replace $dV_{\slashg}$ in the above using
\[
dV_{\slashg} = v(r,\theta)\, r^2\sin\theta\, d\theta \,d\phi
 \qquad \text{\rm with\ } v\sim
1.
\]
Finally, we note that
\begin{equation}\label{easyCoArea}
dV \sim d\tau\, dV_{\Sigma_{\tau}}.
\end{equation}

For $a_0<M$ and
$|a|\le a_0$, note that
the implicit constants in the above are uniformly bounded, depending only on $a_0$ and $M$.

\subsection{Multiplier currents and the general energy identity}
\label{multsandcomts}
We shall repeat our standard notation for vector field multiplier current identities
associated to ``multiplier'' vector fields $V$
which will be applied to $\psi$  as well as to $\Xi\psi$ for various
commutation vector fields $\Xi$.
See~\cite{dr7} for more details and~\cite{book2} for a systematic discussion.
See~\cite{muchT} for an early application of non-trivial
energy currents to the problem of decay for the wave equation on Minkowski space.
\subsubsection{Currents}
Given a general Lorentzian manifold $(\mathcal{M},g)$,
let $\Psi$
be a sufficiently regular function.
We define
\[
{\bf T}_{\mu\nu}[\Psi]
\doteq \partial_\mu\Psi\partial_\nu\Psi -\frac12 g_{\mu\nu}g^{\alpha\beta}
\partial_\alpha\Psi \partial_\beta\Psi.
\]
Given a sufficiently regular vector field $V_\mu$ and  function $w$ on $\mathcal{M}$,
we will define the currents
\[
{\bf J}^V_\mu[\Psi] = {\bf T}_{\mu\nu}[\Psi] V^\nu, \qquad
{\bf J}^{V,w}_\mu[\Psi] =
{\bf J}_\mu^V[\Psi]+\frac18w\partial_\mu (\Psi^2)-\frac18(\partial_\mu w)\Psi^2,
\]
\[
{\bf K}^V[\Psi] ={\bf T}_{\mu\nu}[\Psi]\nabla^\mu V^\nu, \qquad
{\bf K}^{V,w}[\Psi] = {\bf K}
^V[\Psi] -\frac18\Box_gw (\Psi^2) +\frac14w \nabla^\alpha\Psi\nabla_\alpha\Psi,
\]
\[
\mathcal{E}^V[\Psi] = -(\Box_g\Psi)V^\nu \Psi_{,v}, \qquad
\mathcal{E}^{V,w}[\Psi]= \mathcal{E}^V(\Psi)-\frac14w\Psi \Box_g\Psi.
\]
\begin{remark}Note that even if one is only interested in the study of solutions $\psi$
to the homogeneous~(\ref{WAVE}),
inhomogeneous terms will arise
from applying cutoffs to $\psi$ and also from applying
commutation
vector fields (like vector field $Y$ from Section~\ref{Nmult} below)
which do \underline{not} commute with $\Box_g$.
\end{remark}

\subsubsection{The divergence identity}
The divergence identity between two homologous spacelike
hypersurfaces $S^-$, $S^+$, bounding a region $\mathcal{B}$,
with $S^+$ in the future of $S^-$, yields
\begin{equation}
\label{ingeneralform}
\int_{S^+}{\bf J}^V_\mu [\Psi]n^\mu_{S^+}+
\int_{\mathcal{B}} ({\bf K}^V[\Psi] + \mathcal{E}^V[\Psi])
=\int_{S^-}{\bf J}^V_\mu [\Psi]n^\mu_{S^-},
\end{equation}
where $n_{\Sigma_i}$ denotes the future directed timelike unit normal.
The induced volume forms are to be understood.
A similar identity holds for the ${\bf J}^{V,w}_\mu$ currents, etc.

We shall typically apply  $(\ref{ingeneralform})$ for the Kerr metric
$g_{a,M}$ in the case
where
$S^-=\Sigma_0$ and
$S^+=\Sigma_\tau \cup \mathcal{H}^+_{(0,\tau)}$
and $\Psi$ is compactly supported in $\mathcal{R}_{(0,\tau)}$
to obtain
\begin{equation}
\label{ingeneralform2}
\int_{\Sigma_\tau}{\bf J}^V_\mu [\Psi]n^\mu_{\Sigma_\tau}+
\int_{\mathcal{H}^+_{(0,\tau)}}{\bf J}^V_\mu [\Psi]n^\mu_{\mathcal{H}^+}+
\int_{\mathcal{R}_{(0,\tau)}} ({\bf K}^V[\Psi] + \mathcal{E}^V[\Psi])
=\int_{\Sigma_0}{\bf J}^V_\mu [\Psi]n^\mu_{\Sigma_0}.
\end{equation}
Let us note that the compactness of the support justifies the absence of an additional
boundary term even though $S^\pm$ are not homologous.
Since $\mathcal{H}^+$ is null, its induced normal form
is coupled to the choice of $n^\mu_{\mathcal{H}^+}$. In writing the above,
we shall assume such
a choice has been made such that the formula indeed holds.

\subsubsection{Superradiance in Kerr}\label{supInKerr}
As already mentioned in the introduction, the presence of the ergoregion $\mathcal{S}$
is one of the fundamental difficulties associated with the passage from Schwarzschild to a rotating Kerr spacetime. One particular consequence is that for $a\ne 0$,
the conserved $\mathbf{J}^T_{\mu}[\psi]$ energy flux for a solution to (\ref{WAVE}) may be negative on the horizon $\mathcal{H}^+$. Hence, applying~(\ref{ingeneralform2}), the energy on $\Sigma_{\tau}$ can be larger than the energy on $\Sigma_0$; this phenomenon is known as \emph{superradiance}.\footnote{In this context, it is in fact more appropriate
to refer to the energy flux to null infinity $\mathcal{I}^+$.}

An explicit computation in~(\ref{ingeneralform2}) shows the $\mathbf{J}^T_{\mu}[\psi]$ energy flux along $\mathcal{H}^+(0,\infty)$ is given by
\[\int_{\mathcal{H}^+(0,\infty)}\text{Re}\left(T\psi\overline{\left(T\psi + \upomega_+\Phi\psi\right)}\right),
\]
where $\upomega_+$ was defined in Section~\ref{Killingfieldsec}.
In particular, if one formally considers a (complex-valued) solution of the form
\[
\psi(t^*,r,\theta,\phi^*) = e^{-i\omega t^*}e^{im\phi^*}\psi_0\left(r,\theta\right),
\]
then the sign of the $\mathbf{J}^T_{\mu}[\psi]$ flux on the horizon is determined by the sign of
\[
\omega\left(\omega - \upomega_+m\right).
\]
Thus, we say that the parameters $\omega$ and $m$ are superradiant if
\begin{equation}\label{superradiantParam}
\omega\left(\omega - \upomega_+m\right) < 0.
\end{equation}
Observe that in the case $a \geq 0$, the condition~(\ref{superradiantParam}) is equivalent to
\begin{equation}\label{superradiantParamPosa}
m\omega \in \left(0,\frac{am^2}{2Mr_+}\right]
\end{equation}
We will return to a discussion of the significance of this frequency range
in Section~\ref{Vsrsec}.

\section{The main theorems}
\label{statementsections}
With the  notations of Section~\ref{notats} we may now give precise
statements of the results.

\subsection{Boundedness and integrated local energy decay}
\label{boundedandiled}

Recall the notations  of Section~\ref{notats}, in particular the hypersurfaces
$\Sigma_\tau$, the region $\mathcal{R}_0$, the
vector fields $T$, $\tilde{Z}^*$ and the degeneration function $\zeta$
defined in $(\ref{degenerationfunc})$.
Let $n^\mu_{\Sigma_\tau}$, $n^\mu_{\mathcal{H}^+}$ denote the corresponding normals. The vector field  $N$ below can be taken\footnote{We can  alternatively take
$N$ to be the vector field of Proposition~\ref{specialises..};
this is the vector field we shall use in the proof.
For the statement of Theorem~\ref{theResult}, the
only important feature of $N$ is that it is
$\phi_\tau$-invariant, strictly timelike and asymptotic to $T$ for large $r$.
Whereas we could have used everywhere $n_{\Sigma_\tau}$ in the statement,
we prefer to keep the distinct roles of $n_{\Sigma_\tau}$ and $N$
as this will be important when we replace $\Sigma_\tau$ with
hyperboloidal hypersurfaces $\widetilde{\Sigma}_\tau$ in Section~\ref{corollariessec}.}
 to be $n_{\Sigma_\tau}$, thought of now as a smooth vectorfield on $\mathcal{R}$.

The main theorem of the present paper is
\begin{theorem}
\label{theResult}
Fix $M>0$, $0\le a_0<M$ and $\delta>0$.
There exists a constant $C=C(a_0,M,\delta)$ such that
for all $|a|\le a_0$,
 and all
sufficiently regular solutions $\psi$ of the wave equation $\Box_{g_{a,M}}\psi = 0$ on $\mathcal{R}_0$, the following estimates hold:
\begin{equation}
\label{protasn1b}
\int_{\mathcal{R}_0}\Big(r^{-1}\zeta |\nabb\psi|^2+r^{-1-\delta}\zeta (T\psi)^2+r^{-1-\delta}(\tilde Z^*\psi)^2+ r^{-3-\delta} (\psi-\psi_\infty)^2\Big) \le C
\int_{\Sigma_0} {\bf J}_\mu^N[\psi]n^\mu_{\Sigma_0},
\end{equation}
\begin{equation}
\label{fluxho...}
\int_{\mathcal{H}^+_0}\left( {\bf J}^N_\mu[\psi]n^\mu_{\mathcal{H}^+}
+(\psi-\psi_{\infty})^2\right)
\le C \int_{\Sigma_0} {\bf J}_\mu^N[\psi]n^\mu_{\Sigma_0},
\end{equation}
\begin{equation}\label{fluxNullInf}
\int_{\mathcal{I}^+}{\bf J}^T_\mu[\psi]n^\mu_{\mathcal{I}^+}
\le C \int_{\Sigma_0} {\bf J}_\mu^N[\psi]n^\mu_{\Sigma_0},
\end{equation}
\begin{equation}
\label{bndts1b}
 \int_{\Sigma_\tau} {\bf J}_\mu^N[\psi]
n^\mu_{\Sigma_\tau}\le C \int_{\Sigma_0} {\bf J}_\mu^N[\psi]n^\mu_{\Sigma_0}, \qquad
\forall\tau\ge 0,
\end{equation}
where $4\pi\psi_\infty^2= \lim_{r'\to\infty}\int_{\Sigma_0\cap\{r=r'\}} r^{-2}\psi^2$.
\end{theorem}

Estimate $(\ref{protasn1b})$ is an {\bf integrated local energy decay} statement
degenerating at trapping. The full statement obtained in the proof is  more precise but
cannot be expressed in physical space; see Proposition~\ref{preciseILED} of Section~\ref{precise}.

Estimate $(\ref{fluxho...})$ is the {\bf boundedness of the energy
flux through the event horizon} $\mathcal{H}^+_0$ (as measured by a local observer),
while estimate $(\ref{fluxNullInf})$
is the {\bf boundedness of the energy flux to null infinity} $\mathcal{I}^+$. (The latter
will be explained in Section~\ref{toayplus}.)
These two estimates are obtained concurrently
with $(\ref{protasn1b})$.

Estimate $(\ref{bndts1b})$ is the statement of {\bf uniform energy boundedness
through the foliation} $\Sigma_\tau$.
We note that the proof
of this statement, which is obtained a posteriori, requires the more precise version
 of $(\ref{protasn1b})$ given in Proposition~\ref{preciseILED}.
 Note that
\begin{equation}
\label{sobol1}
\int_{\Sigma_\tau} {\bf J}_\mu^N[\psi]n^\mu_{\Sigma_\tau}\sim
\| \psi\|^2_{\mathring{H}^1(\Sigma_\tau)} +\|n_{\Sigma_\tau}\psi\|^2_{L^2(\Sigma_\tau)}
\sim \int_{\theta,\phi^*} \int_{r_+}^\infty
( |\partial_{t^*}\psi|^2 + |\partial_{r}\psi|^2 +|\nabb\psi|^2_{\slashg})\, dr\, dV_{\slashg}
\end{equation}
with respect to coordinates $(t^*, r, \theta, \phi^*)$,
where here $f\left(\psi\right) \sim g\left(\psi\right)$ means there exist constants
$c$ and $C$ not depending on $\psi$ such that $cg\left(\psi\right) \leq f\left(\psi\right) \leq Cg\left(\psi\right)$.
Thus, $(\ref{bndts1b})$ gives uniform geometric $\mathring{H}^1$ bounds
on the solution.

The reader familiar with Penrose-diagrammatic notation may find the following
useful
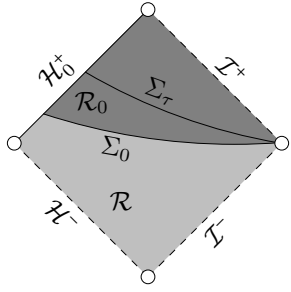
\[
\input{partiii.pstex_t}
\]

\subsection{The higher order statement}
For various applications, it is essential to have a higher-order analogue of
the above. This is given by
\begin{theorem}
\label{h.o.s.}
Let $M$, $a_0$, $a$ be as in Theorem~\ref{theResult}.
Then, for all $\delta>0$ and all integers $j\ge 1$,
there exists a constant $C=C(a_0, M, \delta,j)$
such that
the following inequalities hold
for all sufficiently regular solutions $\psi$
to the wave equation $\Box_{g_{a,M}}\psi = 0$ on $\mathcal{R}_0$
\begin{align}
\nonumber
\label{protasn1}
\int_{\mathcal{R}_0} &
r^{-1-\delta}\zeta \sum_{1\le i_1+i_2+i_3\le j}
|\nabb^{i_1}T^{i_2}(\tilde Z^*)^{i_3}\psi|^2\\
\nonumber
&+r^{-1-\delta}\sum_{1\le i_1+i_2+i_3\le j-1}
\left(|\nabb^{i_1}T^{i_2}(\tilde Z^*)^{i_3+1}\psi|^2+|\nabb^{i_1}T^{i_2}(Z^*)^{i_3}\psi|^2\right)\\
&\le C\int_{\Sigma_0} \sum_{0\le i \le j-1}
{\bf J}^N_\mu[N^{i}\psi]n^\mu_{\Sigma_0},
\end{align}
\begin{equation}
\label{bndts2}
\int_{\mathcal{H}^+_0} \sum_{0\le i \le j-1}{\bf J}^N_\mu[N^{i}\psi]n^\mu_{\mathcal{H}^+}
\le  C\int_{\Sigma_0} {\sum_{0\le i \le j-1}
{\bf J}^N_\mu[N^{i}\psi]n^\mu_{\Sigma_0}},
\end{equation}
\begin{equation}\label{fluxNullInfHigher}
\int_{\mathcal{I}^+}\sum_{0\le i \le j-1}{\bf J}^T_\mu[N^{i}\psi]n^\mu_{\mathcal{I}^+}
\le C \int_{\Sigma_0} \sum_{0\le i \le j-1}{\bf J}_\mu^N[N^i\psi]n^\mu_{\Sigma_0},
\end{equation}
\begin{equation}
\label{bndts1}
\int_{\Sigma_\tau } \sum_{0\le i \le j-1}{\bf J}^N_\mu[N^{i}\psi]n^\mu_{\Sigma_\tau}
\le  C\int_{\Sigma_0} {\sum_{0\le i \le j-1}
{\bf J}^N_\mu[N^{i}\psi]n^\mu_{\Sigma_0}}, \qquad
\forall\tau\ge 0.
\end{equation}
\end{theorem}

Let us note that by an elliptic estimate, we have
\begin{equation}
\label{inasyflatcas}
\int_{\Sigma_\tau } \sum_{0\le i \le j-1}{\bf J}^N_\mu[N^{i}\psi]n^\mu_{\Sigma_\tau}
\sim \sum_{1\le i\le j} \|\psi\|^2_{\mathring{H}^i(\Sigma_\tau)} +\|n_{\Sigma_\tau}\psi\|^2_{\mathring{H}^{i-1}(\Sigma_\tau)}.
\end{equation}
Thus, as with $(\ref{bndts1b})$ before,
we may reexpress statement $(\ref{bndts1})$ of the above theorem as the statement of the uniform
boundedness of geometric Sobolev norms. Note that {\bf uniform pointwise
bounds} on $|\psi|$ and its derivatives
to  \emph{arbitrary} order $|\nabb^{i_1}T^{i_2}(\tilde Z^*)^{i_3}\psi|$
in $\mathcal{R}_0$ follow
as an immediate consequence of the above Theorems in view of the Sobolev
inequality applied on each $\Sigma_\tau$.

The above theorems also
imply pointwise decay statements and decay for the energy flux
through suitable hypersurfaces. We turn to this now.

\subsection{Corollaries}
\label{corollariessec}

Let us note first that by a reduction
proven as Proposition~4.6.1 of~\cite{dr7},  Theorems~\ref{theResult} and~\ref{h.o.s.}
hold where $\Sigma_0$ is replaced by an
 arbitrary ``admissible'' hypersurface $\widetilde{\Sigma}_0$
(see Section 4.4 of~\cite{dr7} for this notion), $\Sigma_\tau$ is replaced by
$\widetilde{\Sigma}_\tau\doteq {\varphi_\tau}(\widetilde\Sigma_0)$,
$n_{\Sigma_\tau}$ is replaced by $n_{\widetilde{\Sigma}_\tau}$,
 $\mathcal{R}_0$, $\mathcal{H}^+_0$ are  redefined as $D^+(\widetilde\Sigma_0)$,
$D^+(\widetilde\Sigma_0)
\cap\mathcal{H}^+$, respectively, and $N$ is kept as is.
This notion includes both asymptotically flat hypersurfaces terminating
at spatial infinity (a special case of admissible hypersurfaces of the first kind) and asymptotically hyperboloidal hypersurfaces terminating at
null infinity (a special case of admissible hypersurface of the second kind).
The latter case
is depicted below
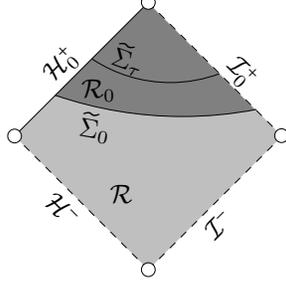
\[
\input{partiii2.pstex_t}
\]
Note, however, that in the latter case, $(\ref{inasyflatcas})$ (with the above substitutions)
will never hold.
It is for this reason that we prefer to state Theorems~\ref{theResult} and~\ref{h.o.s.}
in the form given.

As a consequence of this more general statement, the above theorems allow us to apply
our ``black box'' result of~\cite{icmp} (see~\cite{schlue} and~\cite{moschidis} for detailed treatments).
We obtain
\begin{bigcorol}
\label{thecorol}
Let $a_0$, $M$, $a$, $\delta$ be as in Theorems~\ref{theResult}--\ref{h.o.s.},
and let $R>r_+$.
Let $\widetilde{\Sigma}_0$ be an asymptotically hyperboloidal hypersurface
terminating at null infinity, and denote
$\widetilde{\Sigma}_\tau= \varphi_\tau(\widetilde{\Sigma}_0)$.
Then
for sufficiently regular solutions of the wave equation, we have the following
estimates for the energy flux
\[
\int_{\widetilde{\Sigma}_\tau} {\bf J}^N_\mu [\psi] n^\mu_{\widetilde{\Sigma}_\tau} \le
C(a_0,M) E\tau^{-2}
\]
\[
\int_{\widetilde{\Sigma}_\tau\cap \{r\le R\}} {\bf J}^N_\mu [N\psi] n^\mu_{\widetilde{\Sigma}_\tau}
\le
C(a_0, M,\delta, R) E\tau^{-4+2\delta}
\]
and the following pointwise estimates
\[
\sup_{\widetilde{\Sigma}_\tau}
r
|\psi-\psi_\infty|\le C(a_0, M) \sqrt{E}\, \tau^{-1/2},
\]
\begin{equation}
\label{pdecay1}
\sup_{\widetilde{\Sigma}_\tau\cap\{r\le R\}} |\psi-\psi_{\infty}|\le C(a_0, M, \delta, R) \sqrt{E} \tau^{-3/2+\delta},
\end{equation}
\begin{equation}
\label{pdecay2}
\sup_{\widetilde\Sigma_\tau\cap\{r\le R\}}|n_{\widetilde\Sigma}\psi|+
|\nabla_{\widetilde\Sigma}\psi|
\le C(a_0, M, \delta, R)E \tau^{-2+\delta},
\end{equation}
where in each inequality, $E$ denotes an
appropriate higher order weighted energy on $\widetilde{\Sigma}_0$ (or alternatively
on an asymptotically flat $\Sigma_0$ in the past of $\widetilde{\Sigma}_0$).
\end{bigcorol}

From the point of view of nonlinear applications,
the main significance of the powers on the right hand side
of $(\ref{pdecay1})$ and $(\ref{pdecay2})$ is that they are
integrable in time.

\subsection{The logic of the proof}\label{theLogic}
Now that we have given precise formulations of
the main theorems, we will
give a brief account of the logic of the proof, highlighting
where each statement is proven.

The reader may wish to refer back to the outline of Section~\ref{outlinesection}.
Recall that
Section~\ref{prelimn} concerns various preliminary propositions,
including a reduction (in Section~\ref{WPosed}) to considering $\psi$ arising
from smooth compactly supported data on $\Sigma_0$,
whereas Section~\ref{cartersupersection}  defines a class of functions
for which Carter's separation to the radial o.d.e.~(\ref{e3iswsntouu}) and  appropriate boundary conditions can be justified \emph{a priori}.
Sections~\ref{Vpropsec},~\ref{sct} and~\ref{freqLocEst}, on the other hand, are logically independent of the rest of the paper; they are concerned with the study of classical solutions $u$ to
the o.d.e.~(\ref{e3iswsntouu}) assumed to satisfy appropriate boundary conditions.
The culmination is
Theorem~\ref{phaseSpaceILED} which establishes
estimates on $u$ independent of the frequency parameters in the potential.

The logic of the proof of Theorem~\ref{theResult}
can be properly thought to commence in Section~\ref{summation}.
We define a class of solutions to~(\ref{WAVE}) which we call ``future-integrable''
and which allows us to apply Carter's separation of
Section~\ref{cartersupersection} to a suitably defined function, with the help of
a cutoff. We then apply Theorem~\ref{phaseSpaceILED}
to the resulting $u$. Summing via Plancherel, and using in addition  the preliminary
propositions of Section~\ref{prelimn} and the refined mode stability of~\cite{shlapRot}, we
establish in Proposition~\ref{closedILED} the integrated energy decay statement~(\ref{protasn1b}),
 the horizon energy flux bound~(\ref{fluxho...}) and
the null infinity flux bound~(\ref{fluxNullInf})
\emph{for this class of future-integrable solutions to~(\ref{WAVE}).}

In Proposition~\ref{h.o.s.suff} we will upgrade these to the higher order statements~(\ref{protasn1}),~(\ref{fluxNullInfHigher}) and~(\ref{bndts1}) of Theorem~\ref{h.o.s.}, \emph{again for the class of future-integrable solutions}.

Next, in Proposition~\ref{allIntegrable} we will use a continuity argument to show that \underline{all} solutions to~(\ref{WAVE}) arising from smooth compactly supported data (according to the
reduction of Section~\ref{WPosed}) are future-integrable. We thus unconditionally obtain the statements~(\ref{protasn1b}),~(\ref{fluxho...}),~(\ref{fluxNullInf}),~(\ref{protasn1}),~(\ref{bndts2})
and~(\ref{fluxNullInfHigher}).

Finally, in Proposition~\ref{sufficientIntBound} we unconditionally establish the statements~(\ref{bndts1b}) and~(\ref{bndts1}). This will complete
the proof of Theorems~\ref{theResult} and~\ref{h.o.s.}.

\section{Preliminaries}
\label{prelimn}

\subsection{Well posedness, regularity and smooth dependence}
\label{WPosed}
Let us note that the wave equation $(\ref{WAVE})$
is  well posed in $\mathcal{R}_0$
with initial data $(\uppsi, \uppsi')$ defined on $\Sigma_0$ in $H^j_{\rm loc}(\Sigma_0)\times H^{j-1}_{\rm loc}(\Sigma_0)$
(cf.~Proposition 4.5.1 of~\cite{dr7}).
Moreover, if the initial data are smooth and of compact support on $\Sigma_0$,
then $\psi$ will be smooth, and of compact support on all $\Sigma_\tau$ for $\tau\ge 0$.

In the proof of our theorems, by standard density arguments (applied to $\psi-\psi_{\infty}$),
we may thus assume that $\psi$ indeed arises from such data and thus is
smooth and of compact support for fixed $\Sigma_\tau$ for all $\tau\ge 0$.

Lastly, we observe that the solution $\psi$  to $(\ref{WAVE})$ depends smoothly on $a$, e.g.
\begin{lemma}\label{lemmaSmoothDependence}Let $\left|a_{\infty}\right| < M$, $\{a_k\}_{k=1}^{\infty}$ satisfy $a_k \to a_{\infty}$, $\Box_{g_{a_k,M}}\psi_k = 0$, $\psi_k|_{\Sigma_0} = \psi_{\infty}|_{\Sigma_0}$ and $n_{\Sigma_0}\psi_k|_{\Sigma_0} = n_{\Sigma_0}\psi_{\infty}|_{\Sigma_0}$. Then, for every $j \geq 1$ and $\tau \geq 0$,
\[\lim_{k\to\infty}\int_{\Sigma_{\tau}}\sum_{1\leq i_1 + i_2 + i_3 \leq j}\left|\nabb^{i_1}T^{i_2}\left(\tilde Z^*\right)^{i_3}\psi_k\right|^2 = \int_{\Sigma_{\tau}}\sum_{1\leq i_1 + i_2 + i_3 \leq j}\left|\nabb^{i_1}T^{i_2}\left(\tilde Z^*\right)^{i_3}\psi_{\infty}\right|^2.\]
\end{lemma}
(We shall appeal to the above lemma at the end of Section~\ref{closednesssec}
in the context    of   the closedness part of our continuity argument.)

\subsection{The sign of $a$}\label{signOfa}
For given $a$, $M$, given a solution $\psi$ of $\Box_{g_{M,a}}\psi=0$, then,
defining $\tilde\psi(y^*, t^*,\theta^*, \phi^*)=\psi(y^*,t^*,\theta^*,2\pi-\phi^*)$,
we have that $\tilde\psi$ satisfies
$\Box_{g_{M,-a}}\tilde\psi=0$. Moreover, the estimates of Theorems~\ref{theResult}
and~\ref{h.o.s.} for $\tilde\psi$ with quantities defined
respect to the metric $g_{M,-a}$
are equivalent to the analogous estimates for $\psi$ with respect
to the metric $g_{M,a}$.
Thus,
 it suffices to prove our Theorems for $a \geq 0$.
This reduction is of no conceptual significance, but
it slightly simplifies the notation for discussing the superradiant frequency range,
which then can be given by $(\ref{superradiantParamPosa})$.
For notational convenience we will indeed
use the reduction to $a \geq 0$ in Sections~\ref{Vpropsec}--\ref{freqLocEst} the context of describing
the
properties of the potential $V$ in various frequency regimes
and defining the frequency dependent
multiplier currents.
The reader can assume that $a\ge 0$ globally in this paper, but it is strictly
speaking only necessary for those statements which
refer explicitly to frequency-dependent functions in the separation.

\subsection{Hardy inequalities}
As in the previous parts~\cite{dr7} of this series,
at various points we shall refer to Hardy inequalities.
In view of our comments concerning the volume form (see Section~\ref{usefulcomps}),
the reader can easily derive these
from the one-dimensional inequalities
\begin{equation}\label{hardy1}
\int_{0}^2  x^{-1} |\log x|^{-2} f^2(x) \le C\int_{0}^{2} \left(\frac{df}{dx}\right)^2(x)dx+
C\int_1^2f^2(x) dx,
\end{equation}
\begin{equation}\label{hardy2}
\int_1^\infty  f^2(x)\le C\int_1^\infty x^2\left(\frac{df}{dx}\right)^2 (x)dx,
\end{equation}
where the latter holds for functions $f$ of compact support.

\subsection{Generic constants in inequalities and fixed parameters}
\label{genconstsec}
Let us recall our conventions from~\cite{dr7} regarding constants
depending on the Kerr geometry.

As in the statement of Theorem~\ref{theResult},
all propositions in this paper providing estimates will
explicitly refer to two fixed parameters $a_0<M$ delineating
the range of Kerr parameters allowed.
In the context of inequalities, we shall denote by $B$ potentially large positive constants,
whereas we shall denote by $b$
potentially small positive constants, \underline{both depending only on $M$ and $a_0$}.
 \emph{This dependence
is always to be understood.}
We record the resulting algebra of constants:
\[
b+b=b,  \, B+B=B, \, B\cdot B=B, \, B^{-1}=b, \ldots
\]
We note that these constants will often blow up $B\to\infty$, $b^{-1}\to\infty$ in
the extremal limit $a_0\to M$.

Our constructions will depend on various additional parameters,
for instance, the parameters $\omega_{\text{high}}$, $E$, etc., which
are free in the statements of Propositions~\ref{odeEst1}, etc.,     but are chosen
by the end of the proof of Theorem~\ref{phaseSpaceILED}, in whose
statement they appear as fixed parameters.

When a parameter is required to be ``sufficiently large'' or ``sufficiently small''
without further clarification, this always means
that there exists a constant \underline{depending on $a_0$ and $M$} such that
the parameter can be taken to be an arbitrary value bigger than that constant.
If a parameter is required to be ``sufficiently large'' \emph{given another parameter},
this means that there again exists such a constant depending on $a_0$ and $M$
\underline{and} the other parameter.

Until a parameter has been fixed, e.g.~the parameter $\omega_{\text{high}}$, we shall use
the notation $B(\omega_{\text{high}})$, etc., in the context of inequalities, to denote constants depending
on $\omega_{\text{high}}$ \emph{in addition to $M$
and $a_0$}.
For a parameter, say $c$ which is an explicit function of other parameter(s),
say $\omega_{\text{high}}$, together with $M$ and $a_0$, we will write $c(\omega_{\text{high}})$.
Again, the dependence on $M$ and $a_0$ is \underline{always} to be understood.

The final choices of all initially free
parameters used in the present paper will be made to depend
only on $M$ and $a_0$. Once such choices are made, $B(\omega_{\text{high}})$ is replaced by $B$, following
our conventions.

\subsection{The red-shift}
\label{Nmult}
Understanding the red-shift is an essential part of the dynamics.
Definitive constructions have been given in Section~7  of~\cite{jnotes}.
These depend only on the positivity of the surface gravity $\kappa$,
recalled in Section~\ref{Killingfieldsec}.

\subsubsection{The vectorfield $N$}
Let us recall briefly from~\cite{dr7} the construction of a vector field $N$
capturing the red-shift effect.
\begin{proposition}
\label{specialises..}
Let $|a|\le a_0<M$, $g_{a, M}$ be the Kerr metric and $\mathcal{R}$, etc., be as before.
There exist positive constants $b$ and $B$,
parameters $r_1(a,M)>r_{\rm red}(a,M)>r_+$ and
a $\varphi_\tau$-invariant timelike vector field $N=N(a,M)$
on $\mathcal{R}$, normalised so that $N - K$ is future oriented, traverse to $\mathcal{H}^+$,
and null with $g(N,K)=-2$, such that
\begin{enumerate}
\item
\label{fir}
${\bf K}^N[\Psi] \ge b\, {\bf J}^N_\mu [\Psi]  N^\mu$ for $r\le r_{\rm red}$
\item
\label{item2}
$-{\bf K}^N[\Psi] \le B\, {\bf J}^N_\mu [\Psi] N^\mu$, for $r\ge r_{\rm red}$
\item
\label{lastitem}
$T= N$ for $r\ge r_1$,
\end{enumerate}
where the currents are defined with respect to $g_{a,M}$.
\end{proposition}
Note the implicit $a_0$ and $M$ dependence of constants $b$ and $B$ as described
in Section~\ref{genconstsec} above.
This proposition would fail in the case $a_0=M$. See~\cite{aretakisKerr, Sbierski}.

\subsubsection{The red-shift estimate}
The above leads  immediately to the following estimate (see~\cite{dr7})
\begin{proposition}
\label{ftrs}
Let $g=g_{a,M}$ for $|a|\le a_0<M$,
and let $r_{\rm red}$ be as in the above Proposition.
Then the following is true.
For all $r_+\le \tilde{r}\le r_{\rm red}$ and $\tilde\delta>0$,
there exists
a positive constant $B(\tilde{r},\tilde\delta)$,\index{variable parameters!
$\tilde{r}$ (used in application of red-shift estimates)}\index{variable parameters!
$\tilde{\delta}$ (used in application of red-shift estimates)}
such that for all functions $\Psi$ on $\mathcal{R}_0$, then
\begin{align*}
\int_{\mathcal{R}_{(0,\tau)}\cap\{r\le \tilde{r}\}}  &
({\bf J}_\mu^N[\Psi]N^\mu +|\log(|r-r_+|)^{-2}||r-r_+|^{-1}\Psi^2)+\int_{\mathcal{H}^+_{(0,\tau)}}
{\bf J}^N_\mu[\Psi]n_{\mathcal{H}^+}^\mu
+\int_{\Sigma_\tau\cap\{r\le \tilde{r}\}}{\bf J}^N_\mu[\Psi] n^\mu \\
&\le B(\tilde{r},\tilde\delta)\int_{\Sigma_0} {\bf J}^N_\mu[\Psi] n^\mu + B(\tilde{r},\tilde\delta)
\int_{\mathcal{R}_{(0,\tau)}\cap \{\tilde{r}\le r\le \tilde{r}+\tilde{\delta}
\}}
({\bf J}_\mu^N[\Psi]N^\mu
+\Psi^2)-\mathcal{E}^N\left[\Psi\right].
\end{align*}
\end{proposition}
Again, recall that the additional
dependence of $B$ on $M$ and $a_0$ is now implicit according
to our conventions. Note that the proof of this estimate uses the
Hardy inequality~$(\ref{hardy1})$, so as to include the useful zeroth order
term on the left hand side.
We note that the
same estimate holds with the above zeroth order terms removed from both the right
and the left hand sides.

\subsubsection{Red-shift commutation and the vector field $Y$}
\label{rscsec}
We specialise Theorem 7.2 of~\cite{jnotes}  to the Kerr case.

\begin{proposition}
\label{commuprop}
Let $g=g_{a,M}$,
let $K$ be the vector field of Section~\ref{Killingfieldsec}, let
$Y=N-K$, and let $E_1$, $E_2$ be $\varphi_\tau$-invariant vector fields
such that
$\{K, Y, E_1, E_2\}$ form a local null frame
on $\mathcal{H}^+$.
Then
for all $k\ge 0$ and multi-indices ${\bf m}=(m_1,m_2,m_3,m_4)$,
\[
\Box_g(Y^k\Psi)=\kappa_k Y^{k+1}\Psi + \sum_{|{\bf m}|\le k+1, m_4\le k} c_{{\bf m}}
E_1^{m_1}E_2^{m_2}L^{m_3}Y^{m_4}\Psi +Y^k(\Box_g\Psi)
\]
where $\kappa_k>0$ and the $c_{\bf m}$ are smooth $\varphi_\tau$-invariant functions.
\end{proposition}
The above proposition, which is another manifestation of the red-shift effect, effectively
allows us not only to apply a transversal vector field to the horizon as a multiplier,
but also as a commutation vector field. This is fundamental for retrieving higher order statements
as in Theorem~\ref{h.o.s.}.

\subsection{An estimate for large $r$}\label{largeR}
We will also need the following estimate.
\begin{proposition}
\label{lrp}
Fix $M > 0$ and $a_0 < M$. For each
$\delta>0$, there exist positive  values
$2M < \tilde{R} <R_{\rm large}$, and positive constants $B(\delta)$
such that if $\left|a\right| \leq a_0$,
$\psi$ denotes a solution of $(\ref{WAVE})$ and $\psi_{\infty}=0$, then
for all $\tau\ge 0$
\begin{align*}
\int_{\mathcal{R}_{(0,\tau)}\cap\{r\ge R_{\rm large}\}}& r^{-1}( r^{-\delta}|\partial_r\psi|^2
 +r^{-\delta}|\partial_t\psi|^2+|\nabb \psi|^2_{\slashg}+
r^{-2-\delta}\psi^2)\\
\le& B(\delta)\int_{\Sigma_0} {\bf J}^{N}_\mu[\psi] n^\mu_{\Sigma_0}+
 B(\delta)\int_{\Sigma_\tau} {\bf J}^{N}_\mu[\psi] n^\mu_{\Sigma_\tau} + B(\delta)\int_{\mathcal{R}_{(0,\tau)}\cap\{ r \ge \tilde{R}\}}r^{1+\delta}\left|F\right|^2
 \\
&+ B(\delta)
\int_{\mathcal{R}_{(0,\tau)}\cap\{\tilde{R} \le r\le R_{\rm large}\}}
( |\partial_r\psi|^2  +|\partial_t\psi|^2+|\nabb \psi|^2_{\slashg}
+\psi^2).
\end{align*}
\end{proposition}
The homogeneous case is treated in~\cite{dr7}, and the inhomogeneous case follows by the same argument.

\subsection{A timelike vector field}
We have noted in Section~\ref{Killingfieldsec} that the span of $T$ and $\Phi$ is
a null subspace on the horizon $\mathcal{H}^+$
and a  timelike subspace on $\mathcal{R}\setminus\mathcal{H}^+$.
These statements are in particular implied by the following  Lemma, which
will be important later.

\begin{lemma}\label{timelikeVector}
The vector field
\[T + \frac{2Mar}{\left(r^2+a^2\right)^2}\Phi\]
is timelike  in $\mathcal{R}\setminus\mathcal{H}^+$
and  null on $\mathcal{H}^+$.
\end{lemma}
\begin{proof}
On the horizon $\mathcal{H}^+$, the vector field reduces to
\[T + \frac{a}{2Mr_+}\Phi,\]
which can immediately be seen to be its standard null generator.

Off the horizon, computing $g\left(T + \frac{2Mar}{\left(r^2+a^2\right)^2}\Phi,T + \frac{2Mar}{\left(r^2+a^2\right)^2}\Phi\right)$ in Boyer-Lindquist coordinates reduces the assertion to checking that
\[-\Delta + \sin^2\theta\left(a^2 - \frac{4M^2r^2a^2}{(r^2+a^2)^2} - \frac{4M^2r^2a^4\sin^2\theta\Delta}{(r^2+a^2)^4}\right) < 0.\]
It suffices to consider the case where the quantity in parentheses is positive. In this case, it is sufficient to check that
\begin{equation}\label{check}
-\Delta + a^2 - \frac{4M^2r^2a^2}{(r^2+a^2)^2}  < 0.
\end{equation}

Observe that
\[(r^2+a^2)^2 = (\Delta + 2Mr)^2 = \Delta^2 + 4Mr\Delta + 4M^2r^2.\]
Multiplying through by $(r^2+a^2)^2$ then reduces~(\ref{check}) to
\[
-\Delta^3 - 4Mr\Delta^2 - 4M^2r^2\Delta + a^2\Delta^2 + 4Mra^2\Delta < 0 \Leftrightarrow
-\Delta^2 - (4Mr-a^2)\Delta - 4Mr(Mr-a^2) < 0.
\]
Now it suffices to observe the inequalities $r > M>a$.
\end{proof}

We also note the following, which actually is a trivial consequence
of positivity of surface gravity~$(\ref{eutuxws})$.

\begin{lemma}\label{horizonTimelike}There exists $\epsilon_0  > 0$ such that the vector field
\[T + \frac{a}{2Mr_+}\Phi\]
is timelike for $r \in (r_+,r_++\epsilon_0)$.
\end{lemma}
\begin{proof}A computation gives
\[g\left(T+\frac{a}{2Mr_+}\Phi,T+\frac{a}{2Mr_+}\Phi\right) = \rho^{-2}\left(-\Delta + \sin^2\theta\left(a^2 - \frac{2a^2r}{r_+} + \frac{a^2\left(r^2+a^2\right)^2}{4M^2r_+^2} - \frac{a^4\sin^2\theta\Delta}{4M^2r_+^2}\right)\right).\]
Consider the function
\[F(r) := -\Delta + a^2 - \frac{2a^2r}{r_+} + \frac{a^2\left(r^2+a^2\right)^2}{4M^2r_+^2}.\]
The lemma follows noting that $F(r_+) = 0$ and
\begin{align}
\frac{dF}{dr}(r_+) &= -(r_+-r_-) - \frac{2a^2}{r_+} + \frac{2a^2}{M}
 \\ \nonumber &= -2\sqrt{M^2-a^2} - \frac{2a^2}{M + \sqrt{M^2-a^2}} + \frac{2a^2}{M}
 \\ \nonumber &= \frac{1}{M(M + \sqrt{M^2-a^2})}\left(-2M(M^2 - a^2) -2\sqrt{M^2-a^2}(M^2-a^2)\right)
 \\ \nonumber &< 0.
\end{align}
\end{proof}

\section{The sufficiently integrable outgoing class and Carter's separation}
\label{cartersupersection}
In this section we will define a suitable class of functions
$\Psi$ for which
one can apply Carter's separation, and for which moreover,
the resulting radial ordinary differential equation for $u$ will
satisfy  appropriate ``outgoing'' boundary conditions.
We shall define this  class in Section~\ref{suffInt} below, and then   review
the separation in Section~\ref{karteri}.

\subsection{The class of sufficiently integrable outgoing functions}
\label{suffInt}

We define here a class of functions $\Psi$ for which frequency analysis
 is well defined.
To give some motivation for the class, we give a brief preview of how
separation will be applied.

As described in the introduction,
the frequency analysis necessary for our proof of integrated local energy
decay requires taking a Fourier transform in $t$.
However, a priori, solutions $\psi$ to the wave equation $(\ref{WAVE})$
could even grow exponentially in time. We shall at first instance, however,
restrict to smooth solutions of the wave equation (arising from compactly supported
data) which are \emph{assumed}
to be $L^2$ in time in the future of $\Sigma_0$.\footnote{Recall that in~\cite{dr7} the Fourier transform was only applied to Schwartz functions in $t$. The added flexibility gained by working with square integrable functions in $t$ will be crucial for the continuity argument (see Section~\ref{continuityargsec}).} We shall first appeal
to our estimates with $\Psi = \xi\psi$
where $\psi$ is a solution which is known to be future
integrable, and $\xi(\tau)$ is a suitable
cutoff so that $\xi=1$ for $\tau \ge 1$ and $\xi=0$ for $\tau\le 0$.
See Proposition~\ref{closedILED}.
Note that $\Psi$ satisfies an inhomogeneous equation
\begin{equation}
\label{newinh}
\Box_{g_{a,M}}\Psi=F
\end{equation}
with compactly supported $F$, and $\Psi=0$ to the past of $\Sigma_0$.

In the context of
the openness argument, we shall apply estimates to $\Psi = \xi\psi$ with two different choices of $\psi$ and the cutoff $\xi(\tau)$. First, we will take $\psi$ to solve the wave equation~(\ref{WAVE}), and the cutoff $\xi(\tau)$ will be taken to vanish for $\tau \leq 0$ and $\tau \geq \tau_0$. Next, we will take $\psi$ to be be a solution of $\Box_g\psi=0$ where $g$ is an ``interpolating metric'' between $g_{a,M}$ and a $g_{\mathring{a},M}$,
with interpolation region between hypersurfaces $\Sigma_{\tau-\delta_0}$
and $\Sigma_{\tau}$, and $\xi(\tau)$
will be take to equal $1$ for $\tau \ge 1$ and equal $0$ for $\tau \le 0$.
This will make $\Psi$ an $L^2$ in time solution of the inhomogeneous
$(\ref{newinh})$, where again $F$ is compactly supported in spacetime
and $\Psi=0$ in the past of $\Sigma_0$. See Proposition~\ref{cutoffILED} and Section~\ref{interpolatingSection}.

In the closedness argument, we shall be able to work with solutions $\psi$ which are \emph{a priori} square integrable in time. Thus, we shall apply our estimates to $\Psi = \xi\psi$ where $\xi(\tau)$ is an appropriate cutoff such that $\xi = 1$ for $\tau \geq 1$ and $\xi = 0$ for $\tau \leq 0$. Then, $\Psi$ will satisfy an inhomogeneous equation~(\ref{newinh}) with a compactly supported right hand side, and $\Psi = 0$ to the past of $\Sigma_0$.

Finally, in the context of the boundedness argument, $\psi$ has already been proven
to be $L^2$ in time (both to the future, and, after suitable extension, to the past),
and the argument is applied to $\Psi=\tilde\chi\psi$ where $\tilde\chi(r)$ is a cutoff in $r$
away from the horizon and null infinity. See Section~\ref{boundSuff}.

In all these cases, we note that we apply frequency analysis to
 $\Psi$ which satisfies
(a) $\Psi\left(r\right)$ is square integrable in $t$ for each $r \in [r_+,\infty)$.
and (b) $\Psi$ is supported away from the past event horizon $\mathcal{H}^-$
and ``past null infinity'' $\mathcal{I}^-$ (cf.~the Penrose diagram
of Section~\ref{boundedandiled}), with
$\Box_{g_{a,M}}\Psi$ vanishing for sufficiently large $r$.
It is these properties that motivate the following definitions.

\subsubsection{Sufficiently integrable}
We first introduce the most basic integrability assumptions that will allow us
to apply the separation and make sense of the radial o.d.e.~$(\ref{oderedu})$ discussed already
in Section~\ref{reviewsr}.

\begin{definition}\label{sufficient}
Let $a_0<M$, $|a|\le a_0$ and let $g=g_{a,M}$.
We say that a smooth function
$\Psi:\mathcal{R}\to \mathbb R$ is ``sufficiently integrable'' if
for every $j \geq 1$ and $A > r_+$, we have
            \[\sup_{r \in [r_+,A]}\int_{-\infty}^{\infty}\int_{\mathbb{S}^2}\sum_{0 \leq i_1 + i_2 + i_3 \leq j}\left|\nabb^{i_1}T^{i_2}\left(Z^*\right)^{i_3}\Psi\right|^2\sin\theta\, dt\, d\theta\, d\phi < \infty,\]

            \[\sup_{r \in [r_+,A]}\int_{-\infty}^{\infty}\int_{\mathbb{S}^2}\sum_{0 \leq i_1 + i_2 + i_3 \leq j}\left|\nabb^{i_1}T^{i_2}\left(Z^*\right)^{i_3}\Box_g\Psi\right|^2\sin\theta\, dt\, d\theta\, d\phi < \infty.\]
\end{definition}
\begin{remark}Observe that each fixed-$r$ integral is unchanged under the change of variables $t \mapsto t^*$ and $\phi \mapsto \phi^*$.
\end{remark}

\subsubsection{The ``outgoing'' condition}
\label{outgoingcondsec}
We next introduce an assumption that will imply that solutions of the radial ODE $(\ref{oderedu})$ have
``outgoing'' boundary conditions.

\begin{definition}
\label{sufficient2}
Let $a_0<M$ and $|a|\le a_0$.
We shall say a smooth function $\Psi$ is ``outgoing'' if there exists an $\epsilon>0$
such that
$\Psi$ vanishes in $\Sigma_\tau \cap \{r\le r_++\epsilon\}$ and
$\Sigma_\tau \cap\{r\ge \epsilon^{-1}\}$ for all $\tau \le -\epsilon^{-1}$, and $\Box_{g_{a,M}}\Psi$ vanishes for sufficiently large $r$.
\end{definition}

We shall see the significance of each of these assumptions individually in
Sections~\ref{karteri} and~\ref{boundarySection}  below.
From Section~\ref{summation} onward, we shall always work in the class
defined by assuming \emph{both}
Definitions~\ref{sufficient} and~\ref{sufficient2},
 i.e.~$\Psi$ will always be a sufficiently integrable outgoing function.

\begin{remark}Of course, one could significantly weaken the assumptions in Definitions~\ref{sufficient} and~\ref{sufficient2};
however, this class of functions is very easy to work with, and weakening the assumptions would not simplify the proof of Theorem~\ref{theResult}.
\end{remark}

\subsection{Review of Carter's separation}
\label{karteri}

As we have already discussed in our summary of the first parts of the series
in~Section~\ref{reviewsr},  we shall
view Carter's separation of the wave equation as a
convenient geometric framework
for  frequency-localising energy estimates, closely
tied to the Kerr geometry. In the present section,
we shall review the relevant formalism from~\cite{dr7}.

\subsubsection{The oblate spheroidal harmonics}
\label{oblat}
Let $\nu \in \mathbb{R}$. We begin by recalling the collection
\[
\{S_{m\ell}(\nu,\cos \theta)e^{im\phi}\}_{m\ell}
\]
of the eigenfunctions of the
self-adjoint operator
\[
P(\nu)\, f= -\frac 1{\sin\theta} \frac{\partial}{\partial\theta} \left (\sin\theta \frac{\partial}{\partial\theta}f\right)
-\frac{\partial^2 f}{\partial\phi^2}\frac{1}{\sin^2\theta}
- \nu^2 \cos^2\theta f
\]
on $L^2(\sin\theta\, d\theta\, d\phi)$.
These form a complete orthonormal basis on
$L^2(\sin\theta\, d\theta\, d\phi)$. The eigenfunctions are parametrised by a set of real discrete eigenvalues $\lambda^{(\nu)}_{m\ell}$
\begin{equation}
\label{eigenvals}
P(\nu)\, S_{m\ell}(\nu,\cos \theta)e^{im\phi}=\lambda_{m\ell}(\nu) S_{m\ell}(\nu,\cos \theta)e^{im\phi},
\end{equation}
which have the property that
\begin{equation}\label{eq:lam}
\lambda_{m\ell}^{(\nu)}+\nu^2\ge |m|(|m|+1),
\end{equation}
\begin{equation}\label{lamBound}
\lambda_{m\ell}^{(\nu)} + \nu^2 \geq 2\left|m\nu\right|.\footnote{This follows immediately from the variational characterization of $\lambda_{m\ell}$ and the inequality $\frac{m^2}{\sin^2\theta} + \nu^2\sin^2\theta \geq 2\left|m\nu\right|$.}
\end{equation}
For $\nu=0$ the {\it oblate spheroidal harmonics} $S_{m\ell}(\nu,\cos \theta)e^{im\phi}$
reduce to the
standard spherical harmonics $Y_{m\ell}$ with the corresponding
eigenvalues $\lambda^{(0)}_{m\ell}=\ell(\ell+1)$.

\subsubsection{The coefficients $\Psi_{m\ell}^{(a\omega)}$}\label{separationSubsection}
Given parameters $a$, $M$, let
 $\Psi$ be ``sufficiently integrable'' in the sense of Definition~\ref{sufficient}.
We write
\[
\Psi(t,r,\theta,\phi)=\frac{1}{\sqrt{2\pi}}\int_{-\infty}^\infty e^{-i\omega t} \widehat\Psi (\omega,r,\theta,\phi) d\omega,
\]
and then, setting $\nu = a\omega$ for each $\omega\in \mathbb R$, further decompose
\[
\widehat\Psi(\omega, r,\theta,\phi) = \sum_{m\ell} \Psi^{(a\omega)}_{m\ell}(r)
S_{m\ell}(a\omega,\cos \theta)e^{im\phi},
\]
to arrive at
\[
\Psi(t,r,\theta,\phi)=\frac{1}{\sqrt{2\pi}}\int_{-\infty}^\infty\sum_{m\ell} e^{-i\omega t} \Psi^{(a\omega)}_{m\ell}(r)S_{m\ell}(a\omega,\cos\theta)e^{im\phi} d\omega.
\]
The sufficiently integrable assumption implies that for each $r$, the first and third equality may be interpreted in $L^2_tL^2_{\mathbb{S}^2}$, while the second equality may be interpreted in $L^2_{\omega}L^2_{\mathbb{S}^2}$.

Furthermore, if $\Psi$ satisfies Definition~\ref{sufficient}, then so do
$\partial_t\Psi$, $\partial_\phi\Psi$
and, by the well-known properties of the Fourier
transform, we
have
\[
\partial_t\Psi(t,r,\theta,\phi)=\frac{-i}{\sqrt{2\pi}}\int_{-\infty}^\infty\sum_{m\ell} \omega e^{-i\omega t} \Psi^{(a\omega)}_{m\ell}(r)S_{m\ell}(a\omega,\cos\theta)e^{im\phi} d\omega.
\]
\[
\partial_{\phi}\Psi(t,r,\theta,\phi)=\frac{i}{\sqrt{2\pi}}\int_{-\infty}^\infty\sum_{m\ell}m e^{-i\omega t} \Psi^{(a\omega)}_{m\ell}(r)S_{m\ell}(a\omega,\cos\theta)e^{im\phi} d\omega.
\]
As above, for each $r$ these equalities may be interpreted in $L^2_tL^2_{\mathbb{S}^2}$.

Let us take the opportunity to observe the following consequences of Plancherel's formula and the orthonormality of the $S_{ml}(a\omega,\cos\theta)e^{-im\phi}$:
\begin{align*}
&\int_0^{2\pi}\int_0^\pi\int_{-\infty}^{\infty}
 |\Psi|^2(t,r,\theta,\varphi) \sin\theta\, d\varphi\, d\theta\, dt= \int_{-\infty}^\infty \sum_{m\ell}
|\Psi^{(a\omega)}_{m\ell}(r)|^2\, d\omega,\\
&\int_0^{2\pi}\int_0^\pi\int_{-\infty}^{\infty} \Psi\cdot\overline{\Upsilon} \sin\theta d\varphi\, d\theta\, dt=
\int_{-\infty}^\infty\sum_{m\ell}  \Psi^{(a\omega)}_{m\ell}\cdot\bar \Upsilon^{(a\omega)}_{m\ell}
d\omega,\\
&\int_0^{2\pi}\int_0^\pi\int_{-\infty}^{\infty}
|\partial_r\Psi|^2(t,r,\theta,\varphi) \sin\theta\, d\varphi\, d\theta\, dt = \int_{-\infty}^\infty \sum_{m\ell}
\left|\frac{d}{dr}\Psi^{(a\omega)}_{m\ell}(r)\right|^2
d\omega,\\
&\int_0^{2\pi}\int_0^\pi\int_{-\infty}^{\infty}
|\partial_t\Psi|^2(t,r,\theta,\varphi) \sin\theta\, d\varphi\, d\theta\, dt = \int_{-\infty}^\infty \sum_{m\ell}
\omega^2 |\Psi^{(a\omega)}_{m\ell}(r)|^2
d\omega.
\end{align*}

Finally, we note that a straightforward integration by parts, Plancherel and the orthonormality of $S_{m\ell}(a\omega, \cos\theta)e^{im\phi}$ imply
\begin{align*}
\int_0^{2\pi}\int_0^\pi\int_{-\infty}^{\infty}
&\left[\left|\frac{\partial \Psi}{\partial\theta}\right|^2 + \left|\frac{\partial \Psi}{\partial\phi}\sin^{-1}\theta\right|^2\right](t,r,\theta,\varphi) \sin\theta\, d\varphi\, d\theta\, dt  \\ &=\int_{-\infty}^\infty \sum_{m,\ell}
\lambda_{m\ell}^{(a\omega)}|\Psi^{(a\omega)}_{m\ell}(r)|^2
d\omega
-a^2\int_0^{2\pi}\int_0^\pi\int_{-\infty}^{\infty}\cos^2\theta\left|\partial_t\Psi\right|^2 \sin\theta\, d\varphi\, d\theta\, dt.
\end{align*}

\subsubsection{The radial ordinary differential equation and the potential $V$}\label{theRadialODEPot}

If  $\Psi$ is sufficiently integrable in the sense of Definition~\ref{sufficient},
define
\begin{equation}
\label{homogdefeq}
F=\Box_g\Psi.
\end{equation}
for $g=g_{a,M}$ a Kerr metric with $|a|<M$.

The sufficiently integrable assumption implies that we may  define  the coefficients
$\Psi_{m\ell}^{(a\omega)}(r)$, $\left(\rho^2F\right)_{m\ell}^{(a\omega)}(r)$ as above.

Carter's formal separation~\cite{cartersep2} of the wave operator yields:
\begin{proposition}
Let $a_0<M$, $|a|\le a_0$,
 $\Psi$ be sufficiently integrable, and let $F$ be defined by $(\ref{homogdefeq})$. Then
\begin{align}
\label{CartersODE}
\Delta \frac{d}{dr} \left (\Delta \frac{d\Psi_{m\ell}^{(a\omega)}}{dr}\right)& + \left (a^2m^2 + (r^2+a^2)^2\omega^2-4Mra\omega m - \Delta (\lambda_{m\ell}+a^2\omega^2) \right) \Psi_{m\ell}^{(a\omega)}=\Delta\,
\left(\rho^2F\right)_{m\ell}^{(a\omega)}.
\end{align}
Note that the sufficiently integrable assumption allows us to interpret this equality for each $r$ in $L^2_{\omega}l^2_{m\ell}$.
\end{proposition}

Using the definition $(\ref{r*def})$ of $r^*$ and setting
\begin{equation}\label{uDef}
u^{(a\omega)}_{m\ell}(r)=(r^2+a^2)^{1/2}
 \Psi^{(a\omega)}_{m\ell} (r),
\end{equation}
\begin{equation}\label{hDef}
 H^{(a\omega)}_{m\ell}(r)=\frac{\Delta \left(\rho^2F\right)^{(a\omega)}_{m\ell}(r)}{(r^2+a^2)^{3/2}},
\end{equation}
we obtain
\begin{equation}
\label{e3iswsntouu}
\frac{d^2}{(dr^*)^2}u^{(a\omega)}_{m\ell}+(\omega^2 - V^{(a\omega)}_{m\ell }(r))u =
H^{(a\omega)}_{m\ell}
\end{equation}
where
\begin{equation}
\label{defofV}
V^{(a\omega)}_{m \ell}(r)= \frac{4Mram\omega-a^2m^2+\Delta (\lambda_{m\ell}+\omega^2a^2)}{(r^2+a^2)^2}
+\frac{\Delta(3r^2-4Mr+a^2)}{(r^2+a^2)^3}
-\frac{3\Delta^2 r^2}{(r^2+a^2)^4}.
\end{equation}
In the Schwarzschild case:
\begin{equation}
\label{at0freq}
V^{(0\omega)}_{m\ell}(r) = (r-2M)\left(\frac{\ell\left(\ell+1\right)}{r^3}+\frac{2M}{r^4}\right),
\end{equation}
\begin{equation}
\label{at0freq'}
\left(\frac{dV}{dr^*}\right)^{(0\omega)}_{m\ell}(r)= \frac{r-2M}{r}\left(\frac{2\ell\left(\ell+1\right)(3M-r)}{r^4}+\frac{2M(8M-3r)}{r^5}\right).
\end{equation}
Again, for each $r$,~(\ref{e3iswsntouu}) is to be interpreted in $L^2_{\omega}l^2_{m\ell}$.

\subsubsection{Notational conventions}
Following well-established convention,
in what follows, as in~\cite{dr7}, we shall suppress the dependence of
$u$, $H$ and $V$ on $a\omega$, $m$, $\ell$ in our notation.
We will also use the notation
\begin{equation}
\label{primenotation}
'=\frac{d}{dr^*}.
\end{equation}

\underline{We repeat the following warning from~\cite{dr7}:} Since for fixed $g_{a,M}$,
$r$ is a smoothly invertible function of $r^*$, we
shall often refer to $r^*$-ranges by their corresponding
$r$-ranges (in particular, given an $r$-parameter such as $R$,
we shall very often use the notation $R^*=r^*(R)$ without further
comment), and we shall express
functions appearing in most estimates as functions  of $r$.
The derivative $'$ always is to denote $(\ref{primenotation})$!

\subsection{Boundary conditions}\label{boundarySection}
In this section, we shall establish boundary conditions for the radial ODE $(\ref{e3iswsntouu})$. We will require Definitions~\ref{sufficient} and~\ref{sufficient2}.

\begin{lemma}\label{boundHorizon}
Let $a_0<M$, $|a|\le a_0$, and
$\Psi$ be sufficiently integrable and outgoing. Define $u^{(a\omega)}_{m\ell}(r)$ by~(\ref{uDef}). Then
\begin{equation}
\int_{-\infty}^{\infty}\sum_{m\ell}\left|u'(r)+i\left(\omega-\frac {a m}{2Mr_+}\right)u(r)\right|^2\, d\omega
\end{equation}
is a smooth function of $r$ which vanishes at $r = r_+$.
\end{lemma}
\begin{proof}The assumptions on $\Psi$ imply that
\[\int_{-\infty}^{\infty}\int_{\mathbb{S}^2}\sum_{0 \leq i_1 + i_2 + i_3 \leq j}\left|\nabb^{i_1}T^{i_2}\left(Z^*\right)^{i_3}\Psi\right|^2\sin\theta\, dt\, d\theta\, d\phi\]
is a smooth function of $r$. Combining this with the fact that $\partial_{r^*} = T + \frac{a}{2Mr_+}\Phi$ on $\mathcal{H}^+$, we conclude that
\[\int_{-\infty}^{\infty}\int_{\mathbb{S}^2}\left|\partial_{r^*}\left((r^2+a^2)^{1/2}\Psi\right) - \left(T + \frac{a}{2Mr_+}\right)\left((r^2+a^2)^{1/2}\Psi\right)\right|^2\sin\theta\, dt\, d\theta\, d\phi  = O(r-r_+) \Leftrightarrow \]

\[\int_{-\infty}^{\infty}\sum_{m\ell}\left|u'(r)+i\left(\omega-\frac {a m}{2Mr_+}\right)u(r)\right|^2\, d\omega = O\left(r-r_+\right)\text{ as }r \to r_+.\]
\end{proof}

\begin{lemma}\label{boundInfinity}
Let $a_0<M$, $|a|\le a_0$, and let $\Psi$ be
sufficiently integrable and outgoing. Define $u^{(a\omega)}_{m\ell}(r)$ by~(\ref{uDef}). Then, there exists a sequence $\{r_n\}_{n=1}^{\infty}$ such that $r_n \to \infty$ as $n\to \infty$ and
\begin{equation}
\lim_{n\to\infty}\left|\left(u^{(a\omega)}_{m\ell}\right)'(r_n)-i\omega u^{(a\omega)}_{m\ell}(r_n)\right| \to 0
\end{equation}
for almost every $\omega$.
\end{lemma}
\begin{proof}The ``$r^p$-estimate'' from~\cite{icmp} with $p = 1$, and Definitions~\ref{sufficient} and~\ref{sufficient2} immediately imply that for $R$ sufficiently large,
\[\int_{-\infty}^{\infty}\int_{r\geq R}\int_0^{2\pi}\int_0^{\pi}\left|\left(\partial_t+\partial_{r^*}\right)\left(\left(r^2+a^2\right)^{1/2}\Psi\right)\right|^2\sin\theta\, dt\, dr\, d\theta\, d\phi < \infty.\]
An application of Plancherel and a standard pigeonhole argument imply that there exists a dyadic subsequence $\{r_n\}_{n=1}^{\infty}$ such that
\[\lim_{n\to\infty}\int_{-\infty}^{\infty}\left|\left(u^{(a\omega)}_{m\ell}\right)'(r_n)-i\omega u^{(a\omega)}_{m\ell}(r_n)\right|^2\ d\omega \to 0.\]
Finally, we recall the standard fact that convergence in $L^2$ implies almost everywhere pointwise convergence along a subsequence.
\end{proof}

As noted in~\cite{stabi} we may formally write these boundary conditions as
\begin{align}
&u'+i\left(\omega-\frac {a m}{2Mr_+}\right)u=0,\qquad r=r_+,\label{eq:b-}\\
&u'-i\omega u =0,\qquad\qquad\qquad\quad r=\infty\label{eq:b+}.
\end{align}

\subsection{On the almost everywhere regularity of $u^{(a\omega)}_{m\ell}$}
The most natural setting
for our frequency analysis is the class of functions of $r$ with values in $L^2_{\omega}l^2_{m\ell}$
referred to already after equality $(\ref{CartersODE})$.
However,
for convenience, in Sections~\ref{sct} and~\ref{freqLocEst},
we shall study \emph{classical} solutions $u$ to the o.d.e.~(\ref{e3iswsntouu}).
The following lemma will be used in Section~\ref{summation} to justify the reduction to classical solutions.
\begin{lemma}\label{aeRegular}
Let $a_0<M$, $|a|\le a_0$,
let $\Psi$ be sufficiently integrable and outgoing, and define $u^{(a\omega)}_{m\ell}(r)$ by~(\ref{uDef}). Then, for almost every $\omega$, for all $m$ and $\ell$, $H$ is smooth and $u$ is a smooth solution to the o.d.e.~(\ref{e3iswsntouu})  satisfying the boundary conditions~(\ref{eq:b+}) and~(\ref{eq:b-}).
\end{lemma}
\begin{proof}Pick any labeling of the eigenvalues $\lambda_{m\ell}$ such that $\lambda_{m\ell}$ is a measurable function $\omega$. Then, using the fact that a countable union of measure zero sets is measure zero, it clearly suffices to prove the lemma for each fixed value of $m$ and $\ell$.

For any $j \geq 1$ and $R > r_+$, the fundamental theorem of calculus implies
\begin{align}\label{fromFundCalc}
\sum_{0 \leq i \leq j}\int_{-\infty}^{\infty}\sup_{r \in [r_+,R]}\left|\left(Z^*\right)^iu\right|^2\, d\omega \leq& \sum_{0 \leq i \leq j}\int_{-\infty}^{\infty}\sum_{m\ell}\sup_{r \in [r_+,R]}\left|\left(Z^*\right)^iu\right|^2\, d\omega
\\ \nonumber \leq& \sum_{0 \leq i \leq j}\int_{-\infty}^{\infty}\sum_{m\ell}\left|\left(Z^*\right)^iu\right|^2\Big|_{r = r_+}\, d\omega  +
\sum_{0 \leq i \leq j+1}\int_{-\infty}^{\infty}\sum_{m\ell}\int_{r_+}^R\left|\left(Z^*\right)^iu\right|^2\, d\omega\, dr.
\end{align}
Next, Plancherel (see the explicit formulas in Section~\ref{separationSubsection}), Sobolev inequalities on $\mathbb{S}^2$ and the sufficiently integrable assumption imply that~(\ref{fromFundCalc}) is less than
\begin{align*}
B\sum_{0 \leq i \leq j+1}&\int_{-\infty}^{\infty}\int_{r_+}^R\int_{\mathbb{S}^2}\left|\left(Z^*\right)^i\Psi\right|^2\sin\theta\, dt^* dr\, d\theta\, d\phi^* +
B\sum_{0 \leq i \leq j}\int_{-\infty}^{\infty}\int_{\mathbb{S}^2}\left|\left(Z^*\right)^i\Psi\right|^2\Big|_{r = r_+}\sin\theta\, dt^* d\theta\, d\phi^*
\\ \nonumber \le& B\sum_{0 \leq i \leq j+1,\, k=0,1,2}\int_{-\infty}^{\infty}\int_{r_+}^R\int_{\mathbb{S}^2}\left|\nabb^k\left(Z^*\right)^i\Psi\right|^2\sin\theta\, dt^*\, dr\, d\theta\, d\phi^*
\\ \nonumber&+
B\sum_{0 \leq i \leq j,\, k=0,1,2}\int_{\mathbb{S}^2}\int_{-\infty}^{\infty}\left|\nabb^k\left(Z^*\right)^i\Psi\right|^2\Big|_{r = r_+}\sin\theta\, dt\, d\theta\, d\phi^* < \infty.
\end{align*}

Thus, we conclude that for each $n \in \mathbb{Z}_+$ and $j \geq 0$, $\sum_{0 \leq i \leq j}\sup_{r \in [r_+,r_++n]}\left|\left(Z^*\right)^iu\right|^2$ is an $L^2$ function of $\omega$. Consequently, we may find a set $U^{(j)}_n \subset \mathbb{R}$ such that $\left|\left(U^{(j)}_n\right)^c\right| = 0$ and $\omega \in U^{(j)}_n$ implies that $u^{(a\omega)}_{m\ell}(r)$ is $C^j$ on the interval $(r_+,r_++n)$. Observe that
\[\left|\left(\cap_{j,n=1}^{\infty}U^{(j)}_n\right)^c\right| = \left|\cup_{j,n=1}^{\infty}\left(U^{(j)}_n\right)^c\right| \leq \sum_{j,n=1}^{\infty}\left|\left(U^{(j)}_n\right)^c\right| = 0.\]
Thus, we have a set $U \doteq \cap_{j,n=1}^{\infty}U^{(j)}_n$ such that the complement of $U$ has measure $0$, and $\omega \in U$ implies that $u^{(a\omega)}_{m\ell}$ is a smooth function of $r$. Of course, the same procedure may be carried out for $H^{(a\omega)}_{m\ell}$. We conclude that for almost every $\omega$, $u$ and $H$ are smooth functions of $r$, and hence $u$ is a classical solution of the radial o.d.e.~(\ref{e3iswsntouu}).

Next, we turn to the boundary condition~(\ref{eq:b+}). For every $\omega$ such that $u$ is a classical solution of the radial o.d.e.~(\ref{e3iswsntouu}), an asymptotic analysis of the o.d.e.~(\ref{e3iswsntouu}) as $r^*\to\infty$ implies that we can find constants $A_{out}$ and $A_{in}$ such that
\[u^{(a\omega)}_{m\ell} = A_{out}e^{i\omega r^*} + A_{in}e^{-i\omega r^*} + O\left(r^{-1}\right)\text{ as }r^*\to\infty,\]
where $O\left(r^{-1}\right)$ is preserved upon differentiation. Lemma~\ref{boundInfinity} implies that we must have $A_{in} = 0$, and hence that the boundary condition~(\ref{eq:b+}) holds.

Similarly, an asymptotic analysis of the o.d.e.~(\ref{e3iswsntouu}) as $r^*\to-\infty$ implies that we can find constants $C_{out}$ and $C_{in}$ such that
\[u^{(a\omega)}_{m\ell} = C_{out}e^{-i\left(\omega -\upomega_+m\right)r^*} + C_{in}e^{i\left(\omega-\upomega_+m\right) r^*} + O\left(\left|r^*\right|^{-1}\right)\text{ as }r^*\to-\infty.\]
Lemma~\ref{boundHorizon} implies that we must have $C_{in} = 0$, and hence that the boundary condition~(\ref{eq:b-}) holds.

\end{proof}

\section{Properties of the potential $V$}
\label{Vpropsec}

In this section, we prove certain fundamental properties of the potential $V$
appearing in $(\ref{e3iswsntouu})$,
defined by the expression~$(\ref{defofV})$.
In particular, we shall prove high-frequency regime  properties which will be
essential for the coercivity of the currents of Section~\ref{freqLocEst} in the high
frequency ranges. Sections~\ref{decomppotsec}--\ref{Vsrsec} below follow closely
Section 11.1 of our survey~\cite{stabi}.
Section~\ref{Vnewsec}, relevant
for the fixed-$m$ case which will be
used in our continuity argument of Section~\ref{continuityargsec}, is new.
Finally, we record explicitly in Section~\ref{ASIDEsec} the relation of the properties
of $V$ proven here to properties of
geodesic flow on Kerr.

\begin{bigremark}
Recall from the outline in Section~\ref{outlinesection} and the discussion of
Section~\ref{theLogic} that the
present section, together with Sections~\ref{sct} and~\ref{freqLocEst},
can be understood to form an independent logical unit of this paper
which culminates in Theorem~\ref{phaseSpaceILED} giving
frequency independent estimates for classical solutions
$u$ of the o.d.e.~$(\ref{e3iswsntouu})$ satisfying
the boundary conditions~$(\ref{eq:b-})$ and~$(\ref{eq:b+})$.
Note that for convenience, this analysis will use the reduction to $a\ge 0$
discussed in Section~\ref{signOfa}.
We shall return to the study of $(\ref{WAVE})$ in Section~\ref{summation}.
\end{bigremark}

\subsection{Admissible frequencies}
Recall that the set of eigenvalues $\{\lambda_{m\ell}(a\omega)\}$
defined by $(\ref{eigenvals})$
are not known explicitly in closed form.
As is clear from $(\ref{defofV})$, the potential depends on
$\lambda_{m\ell}(a\omega)$ only through the
quantity
\begin{equation}
\label{LAMBDAfreq}
\Lambda=\lambda_{m\ell}(a\omega)+a^2\omega^2,
\end{equation}
which according to \eqref{eq:lam} and \eqref{lamBound} obeys
\begin{equation}
\label{needsanumber}
\Lambda\ge |m|(|m|+1),
\end{equation}
\begin{equation}\label{useful}
\Lambda \ge 2\left|am\omega\right|.
\end{equation}
It turns out that the results of this section  depend \emph{only} on the constraints
$(\ref{needsanumber})$ and (\ref{useful}), not on the precise values of
the set $\{\lambda_{m\ell}(a\omega)\}$
In what follows, we may thus consider $\omega\in \mathbb R$, $m\in \mathbb Z$,
$\Lambda\in \mathbb R$ to be
\emph{independent} parameters\footnote{In fact, taking $m$ to be integer-valued is
of no significance in this analysis, but we will continue to write $m\in\mathbb Z$
to avoid confusion.} constrained only by $(\ref{needsanumber})$ and
$(\ref{useful})$.
This motivates
\begin{definition}
We call a frequency triple $(\omega, m, \Lambda)$ admissible
if $\omega\in \mathbb R$, $m\in \mathbb Z$, $\Lambda\in \mathbb R$,
where $\Lambda\ge |m|(|m|+1)$ and $\Lambda \geq 2\left|am\omega\right|$.
\end{definition}

\subsection{Decomposition of the potential}
\label{decomppotsec}
Given Kerr parameters $0\le a<M$,
and an admissible frequency triple $(\omega, m, \Lambda)$,
we may now \emph{define}
 the potential as
\begin{equation}
\label{NOWADEF}
V(\omega,m,\Lambda)=V_0(\omega, m,\Lambda)+V_1,
\end{equation}
where
\begin{eqnarray}
\label{NOWADEF2}
V_0&=&\frac{4Mram\omega -a^2m^2+\Delta\Lambda}{(r^2+a^2)^2},\\
\nonumber
V_1&=&\frac{\Delta(3r^2-4Mr+a^2)}{(r^2+a^2)^3} -\frac{3\Delta^2 r^2}{(r^2+a^2)^4}.
\end{eqnarray}
Note that $V_0$ dominates for high frequencies since $V_1$ does not contain
any frequency parameters $m$, $\omega$, $\Lambda$. Note also the
nonnegativity property:
\[
V_1=\frac {\Delta}{(r^2+a^2)^4}\left[a^2\Delta+2Mr(r^2-a^2)\right]\ge 0.
\]

\subsection{The critical points of $V_0$ and the structure of trapping}
\label{Vtrsec}
To understand the nature of trapping, one must first identify the critical points of
$V_0$. This is provided by the following Lemma. (This appeared
as~Lemma 11.1.1  of~\cite{stabi}; we repeat its statement and proof here.)

\begin{lemma}\label{lem:1}
Let $M>0$, $a_0<M$ and $0\le a\le a_0$.
Then for all admissible frequency triples $(\omega, m, \Lambda)$
with $\Lambda>0$,
the potential function $V_0$ defined by $(\ref{NOWADEF2})$
as a function $V_0:(r_+,\infty)\to \mathbb R$
is either (a) strictly decreasing, (b) has a unique critical value $r^0_{\rm max}$ which
is a global maximum, or (c) has exactly two critical values $r^0_{\rm min}<r^0_{\rm max}$ which
are a local minimum and maximum respectively.
The value $r^0_{\rm max}$ is bounded  independently of the frequency
parameters
\[
r^0_{\rm max}\le B.
\]
\end{lemma}
\begin{proof}
We have
\begin{align*}
\frac {d}{dr} V_0&=4maM\omega \left(\frac 1{(r^2+a^2)^2} - \frac {4r^2}{(r^2+a^2)^3}\right)
+\frac {4ra^2m^2}{(r^2+a^2)^3} + \frac \Lambda{(r^2+a^2)^2} \left(2(r-M) -\frac {4r\Delta}{r^2+a^2}\right)\\
&=\frac 1{(r^2+a^2)^3} \left(4maM\omega(-3r^2+a^2) + 4ra^2m^2-2\Lambda(r^3+a^2r-3Mr^2+Ma^2)\right),
\end{align*}
and thus,
\begin{eqnarray*}
\frac d{dr} \left((r^2+a^2)^3 \frac d{dr} V_0\right)&=&-24Mam\omega r + 4a^2m^2-2\Lambda(3r^2-6Mr+a^2)\\ &=&-6\Lambda\left(r^2-2Mr +4Mr\sigma+\frac{a^2}3-\frac 23 a^2 \frac{m^2}\Lambda\right),
\end{eqnarray*}
where we have set
\[
\sigma=\frac {am\omega}\Lambda.
\]

It follows that any  critical points of the function $(r^2+a^2)^3 \frac d{dr} V_0$
must be roots  of the quadratic
\[
r^2-2Mr(1-2\sigma)+\frac {a^2}3\left(1-\frac {2m^2}{\Lambda}\right)
\]
which we may denote as
\[
r_{1,2}=M(1-2\sigma)\pm\sqrt{M^2(1-2\sigma)^2-\frac {a^2}3\left(1-\frac {2m^2}\Lambda\right)}.
\]

Recalling that $r_+>M$, then if $m\omega\geq 0$ (and thus $\sigma\geq 0$), it follows that $Re(r_2)<M$ and thus the
only possible  critical point on the interval $(r_+,\infty)$ would be
\[
r_1=M(1-2\sigma)+\sqrt{M^2(1-2\sigma)^2-\frac {a^2}3\left(1-\frac {2m^2}\Lambda\right)}.
\]
Noting that since $\Lambda>0$, we have
\[
\lim_{r\to\infty} (r^2+a^2)^3 \frac d{dr} V_0= -\infty,
\]
it follows that $\frac d{dr} V_0$ either
(a*) vanishes nowhere, (b*) vanishes  at a unique point to be denoted $r^0_{\rm max}$,
or (c*) vanishes at two points, denoted
$r^0_{\rm min}<r^0_{\rm max}$, where
\[
\frac{ d^2}{dr^2} V_0(r^0_{\rm min})\ge 0,\qquad
\frac{ d^2}{dr^2} V_0(r^0_{\rm max})\le 0.
\]

In case (a*), it follows that $V_0$ is strictly decreasing (case (a) of the lemma).
In case (b*), it follows that either $r_{\rm max}$ is an inflection point
and $V_0$ is again strictly decreasing (corresponding again to
case (a) of the statement of the lemma),
or $r_{\rm max}$ is a global maximum
(case (b) of the statement of the lemma).
In case (c*), it is moreover easy to see that these inequalities
are in fact strict, and thus $r^0_{\rm min}$ and $r^0_{\rm max}$ correspond
to the unique minumum and maximum of $V_0$ on $(r_+,\infty)$
(corresponding to case (c) of the statement of the lemma).

If $m\omega<0$ (and thus $\sigma<0$), then let us reexpress the root  $r_2$ by
\[
r_2=
M(1-2\sigma)\left(1-\sqrt{1-\frac {a^2(1-\frac {2m^2}\Lambda)}{3M^2(1-2\sigma)^2}}\right).
\]
Since $a <M$ and $\sigma<0$, we have
$$
\frac {a^2(1-\frac {2m^2}\Lambda)}{3M^2(1-2\sigma)^2}<\frac 13.
$$
Noting for $0\le x<\frac 13$ the inequality $
\sqrt{1-x}\ge 1-\frac{2x}3$,
it follows that
\[
Re(r_2)<\frac{2M(1-2\sigma)a^2(1-\frac{2m^2}\Lambda)}{9M^2(1-2\sigma)^2}=
\frac {2a^2(1-\frac{2m^2}\Lambda)}{9M(1-2\sigma)}<\frac {2M}9<r_+.
\]
This now implies that $r_1$ is the only possible zero of  $\frac d{dr} \left [(r^2+a^2)^3 \frac d{dr} V_0\right]$
on the interval $[r_+,\infty)$ and the previous argument applies.

The last statement of the lemma easily follows from observing that for all $\Lambda > 0$
$$
(r^2+a^2)^3 \frac d{dr} V_0 = (6\La M -12 Mam\omega)r^2-2\Lambda r^3 + O\left(am\omega,m^2\right)r\text{ as }r\to\infty,
$$
and we have $\Lambda \geq \left|m\right|\left(\left|m\right|+1\right)$, $\Lambda\geq
2a|m\omega|$.
\end{proof}

The next statement effectively establishes that even if $r^0_{\rm min}$
exists, it can only be `trapped'
for the value $\omega=\upomega_+m$. (Again this appeared as Lemma 11.1.2 of~\cite{stabi}.
We repeat its statement and proof here.)
\begin{lemma}\label{lem:2}
Let $M>0$, $a_0<M$ and $0\le a \le a_0$.
For all  admissible frequency triples $(\omega, m, \Lambda)$
we have
\begin{equation}
\label{anineq}
\omega^2\ge V(r_+)
\end{equation}
with equality achieved if and only if $\omega=\upomega_+m$.
In particular, in the notations of the previous lemma, this implies that
\[
\omega^2>V_0(r^0_{\rm min}).
\]
\end{lemma}
\begin{proof}
We simply compute
\[
\omega^2-V(r_+)=\omega^2-\frac {4Mr_+am\omega-a^2m^2}{(r_+^2+a^2)^2}=
\frac {(2Mr_+\omega-am)^2}{4M^2r_+^2}.
\]
\end{proof}

Note that the case of
 equality
in $(\ref{anineq})$
occurs precisely at the threshold of the superradiance condition~(\ref{superradiantParamPosa}):
\[
\omega=\upomega_+m=\frac {am}{2Mr_+}.
\]

\subsection{Superradiant frequencies are not trapped}
\label{Vsrsec}
We now turn specifically to the superradiant frequencies, which under the assumption
$a\ge 0$ are defined by $(\ref{superradiantParamPosa})$. We will show
that these are in fact \underline{not trapped}, in the sense that,
for such frequencies, the maximum of $V$ is always
(quantitatively) above the energy level $\omega^2$.

First, let us show that for a range of frequency parameters including the superradiant regime,
$V_0$ can only have a critical point at a maximum, that is
the point $r^0_{\rm min}$ is absent. (This was Lemma 11.1.3 of~\cite{stabi} augmented
by Remark 11.1.)
\begin{lemma}\label{lem:3}
Let $M>0$, $a_0<M$ and $0\le a\le a_0$.
Then for all
 admissible frequency triples $(\omega, m, \Lambda)$ satisfying in addition
\[
m\omega\le \frac {am^2}{2Mr_+},
\]
we have
\begin{equation}\label{increaseHorizon}
\frac {d}{dr} V(r_+)\ge \frac {d}{dr} V_0(r_+)\ge b\Lambda \ge 0.
\end{equation}
Recall that Lemma~\ref{lem:1} showed that if $r^0_{\rm min}$ exists, we either have
$r^0_{\rm min}<r^0_{\rm max}$ or $\frac{dV_0}{dr} \leq 0$ on $(r_+,\infty)$. Thus~(\ref{increaseHorizon}) implies that $r^0_{\rm min}$ does not exist and the potential $V_0$ has its unique critical point at $r^0_{\rm max}$.

Moreover, for all $\alpha>0$ sufficiently small\footnote{Recall
our conventions from Section~\ref{genconstsec} on the meaning of this term.
This smallness constraint indeed degenerates as $a_0\to M$.},
the same statement holds under the weaker assumption
\begin{equation}\label{withAlpha}
m\omega \leq \frac{am^2}{2Mr_+} + \alpha\Lambda.
\end{equation}
\end{lemma}

\begin{proof}
We begin with the first statement of the lemma. Note
\begin{eqnarray*}
\frac d{dr} V_0(r_+)&=&\frac{4maM\omega}{(r_+^2+a^2)^3}  \left (-3r_+^2+a^2\right)
+\frac {4r_+a^2m^2}{(r_+^2+a^2)^3} + \frac {2(r_+-M)\Lambda}{(r_+^2+a^2)^2} \\ &=&
\frac 1{(r_+^2+a^2)^3}\left(4maM\omega (-3r_+^2+a^2) + 4r_+a^2m^2+2(r_+^2+a^2)(r_+-M)\Lambda\right).
\end{eqnarray*}
For frequency parameters satisfying $m\omega<0$, the conclusion of the lemma is now
obvious, since
$-3r_+^2+a^2<0$. Otherwise,
using the condition
\[
0\le m\omega\le \frac {am^2}{2Mr_+}
\]
we obtain
\begin{align*}
(r_+^2+a^2)^3\frac d{dr} V_0(r_+)&\ge \left (\frac {2a^2m^2}{r_+} (-3r_+^2+a^2) +4r_+a^2m^2+2(r_+^2+a^2)(r_+-M)\Lambda\right)\\ &= \left(\frac {2a^2m^2}{r_+} (-r_+^2+a^2) +2(r_+^2+a^2)(r_+-M)\Lambda\right)\\ &=2(r_+-M)\left(\Lambda (r_+^2+a^2)-2a^2m^2\right)\\&=
4(r_+-M)\left(\Lambda Mr_+-a^2m^2\right).
\end{align*}
The inequalities $\Lambda\ge m^2$ and $r_+>M>a$ imply that
$\frac d{dr} V_0(r_+)\geq b\Lambda$.
We finish the proof of the first statement by  recalling that $V=V_0+V_1$ and observing
the identity
\[
\frac {d}{dr} V_1(r_+)=\frac {4Mr_+(r_+-M)(r_+^2-a^2)}{(r_+^2+a^2)^4}>0.
\]

It is clear that the final assertion of the lemma concerning
the weaker assumption~(\ref{withAlpha})
follows immediately now from the first.
\end{proof}

Recall the superradiant condition $(\ref{superradiantParamPosa})$.
The
statement that superradiant frequencies are not trapped now follows from the following
Lemma
(again, cf.~Lemma 11.1.4 of~\cite{stabi})

\begin{lemma}\label{lem:4}
Let $M>0$, $a_0<M$ and $0\le a\le a_0$.
For all $\alpha\ge0$ sufficiently small, then for all
 admissible frequency triples $(\omega, m, \Lambda)$
satisfying in addition
\[
0 < m\omega\le \frac {am^2}{2Mr_+} + \alpha\Lambda,
\]
the potential $V_0$ satisfies
\[
b\Lambda \leq V_0(r^0_{max}) - \omega^2.
\]
\end{lemma}
\begin{proof}Again, it suffices to prove the lemma with $\alpha = 0$. Let $\epsilon > 0$ be a fixed sufficiently small constant.

We first consider the case when $m\left(\frac{am}{2Mr_+} - \omega\right) \leq \epsilon\left|m\right|\sqrt{\Lambda}$. In this case we have
\[\omega^2 - V_0(r_+) = \left(\omega - \frac{am}{2Mr_+}\right)^2 \leq \epsilon^2\Lambda.\]
Combining this with Lemma~\ref{lem:3} easily shows
\[V_0(r_++\delta)- \omega^2 \geq b\Lambda\]
for some sufficiently small $\delta > 0$ and even smaller $\epsilon$.

Next, we consider the case when $\omega^2 \leq \epsilon \Lambda$. Then we clearly have
\[V_0(r) - \omega^2 \geq \frac{\Lambda}{r^2} + O\left(\frac{\Lambda}{r^3}\right) - \epsilon\Lambda\text{ as }r\to\infty.\]
Therefore, if we let $\tilde{r}$ be sufficiently large, and then let $\epsilon$ be sufficiently small, we can arrange for
\[V_0(\tilde{r}) - \omega^2 \geq b\Lambda.\]

Finally, we consider the case where $m\left(\frac{am}{2Mr_+}-\omega\right) > \epsilon\left|m\right|\sqrt{\Lambda}$ and $\omega^2 > \epsilon \Lambda$. In this case, $r_0 := \frac{am}{2M\omega}$ will satisfy $r_0 \in \left[r_++\delta,R\right]$ for some $\delta > 0$ and $R < \infty$. Letting $\Delta_{r_0}$ denote $r_0^2 - 2Mr_0 + a^2$, we then compute

\begin{eqnarray*}
\omega^2-V_0(r_0)&=&\omega^2-\frac {4Mr_0am\omega-a^2m^2+\Delta_{r_0}\Lambda}{(r_0^2+a^2)^2}\\ &=&
\frac 1{(r_0^2+a^2)^2}\left[(r_0^2+a^2)^2\omega^2-4Mr_0am\omega+a^2m^2-\Delta_{r_0}\Lambda\right]\\ &=&
\frac 1{(r_0^2+a^2)^2}\left[4M^2r_0^2\omega^2-4Mr_0am\omega+a^2m^2+\omega^2\left((r_0^2+a^2)^2-4Mr_0^2\right)-\Delta_{r_0}\Lambda\right]\\&=&\frac {\omega^2(r_0^2-2Mr_0+a^2)(r_0^2+2Mr_0+a^2)-\Delta_{r_0}\Lambda}{(r_0^2+a^2)^2}\\
&=&\frac{ \Delta_{r_0}}{(r_0^2+a^2)^2}\left(\frac {a^2 m^2}{4M^2} \left(1+\frac{2M}{r_0}+\frac{a^2}{r_0^2}\right)-\Lambda\right).
\end{eqnarray*}

We now recall that $a<M<r_0$ and that $\Lambda\ge \left|m\right|(\left|m\right|+1)$ to conclude that
\[V_0\left(r_0\right) -\omega^2 \geq b\frac{\Delta_{r_0}}{\left(r_0^2+a^2\right)^2}\Lambda \geq b\Lambda.\]
In the last inequality we have used that $r_0$ is bounded away from $r_+$ and $\infty$ independently of the frequency parameters.

\end{proof}

\subsection{Trapping for fixed-azimuthal mode solutions}
\label{Vnewsec}

The final result of this section shows in the case of a fixed azimuthal frequency $m$, large $\Lambda$ and $\omega^2 \sim \Lambda$, $r^0_{\rm max}$ occurs outside the ergoregion.
\begin{lemma}\label{aziTrap}
Let $M>0$, $a_0<M$ and $|a|\le a_0$.
Recall that we previously defined $\sigma = \frac{am\omega}{\Lambda}$. There exists a small constant $c>0$ such that $|\sigma| \leq c$, $m^2 \leq c\Lambda$ and $c^{-1} \leq \Lambda$ imply that $r^0_{\rm max} > (1+\sqrt{2})M$.
\end{lemma}
\begin{proof}
A previous computation showed
\[\left(r^2+a^2\right)^3\frac{dV_0}{dr} = 4maM\omega\left(-3r^2+ a^2\right) + 4ra^2m^2 -2\Lambda\left(r^3-3Mr^2 + a^2r + a^2M\right).\]
Since $r^0_{\rm max}$ is the final critical point of $V_0$, we have that $r \geq r^0_{\rm max}$ implies $\frac{dV_0}{dr}\left(r\right) \leq 0$. Hence, the lemma will follow if we can check that $\frac{dV_0}{dr}\left(r = \left(1+\sqrt{2}\right)M\right) > 0$:

\[\Lambda^{-1}\left(r^2+a^2\right)^3\frac{dV_0}{dr}\Big|_{r = \left(1+\sqrt{2}\right)M} =O\left(c\right) - 2\left(M^3\left(1+\sqrt{2}\right)^3 - 3M^3\left(1+\sqrt{2}\right)^2 + a^2M\left(1+\sqrt{2}\right) + a^2M\right).\]
Since we have
\[\left(1+\sqrt{2}\right)^2 = 3 + 2\sqrt{2},\]
\[\left(1+\sqrt{2}\right)^3 = 7 + 5\sqrt{2},\]
we obtain
\begin{align*}
\Lambda^{-1}\left(r^2+a^2\right)^3\frac{dV_0}{dr}\Big|_{r = \left(1+\sqrt{2}\right)M} &=
O\left(c\right) - 2\left(7M^3 + 5\sqrt{2}M^3 - 9M^3 -6\sqrt{2}M^3 + a^2M + \sqrt{2}a^2M + a^2M\right) \\
&=O\left(c\right) - 2\left(2M\left(a^2-M^2\right) + \sqrt{2}M\left(a^2 - M^2\right)\right).
\end{align*}
This is positive for sufficiently small $c>0$.
\end{proof}
\begin{remark}The importance of the value $r = \left(1+\sqrt{2}\right)M$ comes from the fact that this is the unique location of trapping for axisymmetric solutions to the wave equation on an extreme Kerr background, see~\cite{aretakisKerr}.
\end{remark}
\begin{remark}Note that in the case $a = 0$, one may drop the assumptions $|\sigma| \leq c$ and $\left|m\right|^2 \leq c\Lambda$ and the $O(c)$'s which occur in the proof.
\end{remark}
\begin{remark}\label{ergoTrap}Of course, the Killing vector field $T$ satisfies
\[g\left(T,T\right) = -\left(\frac{r^2 -2Mr + a^2\cos^2\theta}{r^2+a^2\cos^2\theta}\right),\]
which is manifestly negative for $r  \geq \left(1+\sqrt{2}\right)M > 2M$.
\end{remark}

\subsection{Aside: relation with null geodesic flow}
\label{ASIDEsec}

We note that the potential $V_0$ is intimately related to the potential which appears for the radial dependence of solutions of the geodesic equation, i.e. let $\gamma(s) = \left(t\left(s\right),r\left(s\right),\theta\left(s\right),\phi\left(s\right)\right)$ be a null geodesic.

The conserved quantities associated to stationarity and axisymmetry are
\[E \doteq g\left(\dot\gamma,T\right) = -\left(1-\frac{2Mr}{\rho^2}\right)\dot t - \frac{2Mra\sin^2\theta}{\rho^2}\dot\phi,\]
\[L \doteq -g\left(\dot\gamma,\Phi\right) = \frac{2Mra\sin^2\theta}{\rho^2}\dot t - \sin^2\theta\frac{(r^2+a^2)^2 - a^2\sin^2\theta\Delta}{\rho^2}\dot\phi.\]
Carter's hidden conserved quantity is
\[Q \doteq \rho^4\left(\dot\theta\right)^2 + \frac{L^2}{\sin^2\theta} - a^2E^2\cos^2\theta.\footnote{Instead of $Q$ one often finds the Carter constant defined as $K := \rho^4\left(\dot\theta\right)^2 + \frac{(L-aE\sin^2\theta)^2}{\sin^2\theta}$, but $Q$ will relate more naturally to our conventions for the wave equation.}\]

Geodesic motion then reduces to the following system (see~\cite{carter})

\begin{equation*}
\rho^2\dot t = a\left(Ea\sin^2\theta-L\right) + \frac{\left(r^2+a^2\right)\left(La-\left(r^2+a^2\right)E\right)}{\Delta},
\end{equation*}
\begin{equation*}
\rho^2\dot\phi = \frac{Ea\sin^2\theta-L}{\sin^2\theta} + \frac{a\left(La-(r^2+a^2)E\right)}{\Delta},
\end{equation*}
\begin{equation*}
\rho^4\left(\dot\theta\right)^2 = Q + a^2E^2 - 2aEL - \frac{\left(L-aE\sin^2\theta\right)^2}{\sin^2\theta},
\end{equation*}
\begin{equation}\label{rEqn}
\rho^4\left(\dot r\right)^2 = \left((r^2+a^2)E-aL\right)^2 -\Delta\left(Q + a^2E^2 - 2aEL\right).
\end{equation}

Note that the right hand side of (\ref{rEqn}) be re-arranged to
\begin{align}\label{arrange}
\left(r^2+a^2\right)^2E^2 - 4MarEL + a^2L^2 - \Delta\left(Q + a^2E^2\right).
\end{align}

Under the correspondence $E \mapsto \omega$, $L \mapsto m$ and $Q \mapsto \lambda_{ml}$, (\ref{arrange}) is exactly equal to $(r^2+a^2)^2\left(\omega^2-V_0\right)$. Hence, we can write $\dot r$'s equation as
\begin{equation*}
\frac{\rho^4}{(r^2+a^2)^2}\left(\dot r\right)^2 = E^2 - V_0\left(E,L,Q,r\right).
\end{equation*}

As a corollary of Lemmas~\ref{lem:1},~\ref{lem:2},~\ref{lem:3},~\ref{lem:4} and~\ref{aziTrap}, one has
that
(a) null geodesic flow is hyperbolic in a neighborhood of the set of future trapped
null geodesics
(b) null geodesics $\gamma$ whose future tangent $\dot\gamma$
has $g_{a,M}(\dot\gamma,T)\ge 0$ are \emph{not} future trapped; they intersect $\mathcal{H}^+$
(c) trapped null geodesics orthogonal to $\partial_\phi$ lie
outside of the ergoregion.
We shall not however make direct use of any of these facts at the level of
geodesic flow.

\section{The separated current templates}
\label{sct}
Before turning to our estimates we recall the separated current templates of~\cite{dr7} and~\cite{stabi}.

\subsection{The frequency-localised virial currents ${\bf J}^{X,w}$}
First, we  define the frequency-localised analogue of the virial currents ${\bf J}^{X,w}$
where $X$ is in the direction of $\partial_{r^*}$, and $w$ is a suitable function.

Fix Kerr parameters $M>0$ and $|a|<M$ and frequency parameters
$\omega\in\mathbb R$, $m\in\mathbb Z$, and $\Lambda\in \mathbb R$.
Let $f(r^*)$, $h(r^*)$ and $y(r^*)$ be arbitrary sufficiently regular functions.\footnote{In general, $f$ will be bounded and $C^2$,  $h$ will be bounded, $C^1$ and piecewise $C^2$ and $y$ will be bounded, $C^0$ and piecewise $C^1$.}
With the notation $(\ref{primenotation})$, let us
define\footnote{For better or for worse, we follow here the notation we instituted in the first
parts of this series~\cite{dr7}. As this notation proved somewhat
unpopular,  we suggest that readers who dislike archaic Greek simply substitute
${\text Q}^y$, ${\text Q}^h$ for both $\text {\fontencoding{LGR}\selectfont \koppa}^y$ and $\text {\fontencoding{LGR}\selectfont \Coppa}^h$,
as we shall consistently use
functions named $f$, $h$ and $y$, according to whether  we mean
${\text Q}^f$, $\text {\fontencoding{LGR}\selectfont \Coppa}^h$ or
$\text {\fontencoding{LGR}\selectfont \koppa}^y$.
Note that in our survey~\cite{stabi}, we used the notation
$Q^f_0={\text Q}^f$, $Q^h_1=\text {\fontencoding{LGR}\selectfont \Coppa}^h$,
$Q^y_2=\text {\fontencoding{LGR}\selectfont \koppa}^y$.}
 the currents
 \begin{eqnarray*}
{\text{Q}}^f[u]&=&
f \left ( |u'|^2 + (\omega^2-V ) |u|^2\right) + f'
{\rm Re}\left(u'\bar u\right)-
\frac 12f'' |u|^2,\\
\text {\fontencoding{LGR}\selectfont \Coppa}^h[u] &=& h {\rm Re} (u'\bar u)-\frac12 h' |u|^2,\\
\text {\fontencoding{LGR}\selectfont \koppa}^y[u] &=& y\left(|u'|^2+(\omega^2-V)|u|^2\right),
\end{eqnarray*}
associated to
the choice of an arbitrary
 smooth function $u(r^*)$.\footnote{Recall that $\text{Q}^f$ is itself the combination
$\text {\fontencoding{LGR}\selectfont \Coppa}^{h}[u]+
\text {\fontencoding{LGR}\selectfont \koppa}^{y}[u]$, with $y=f$ and $h=f'$, but sufficiently important
to deserve its own name!}

For $u$ satisfying~(\ref{e3iswsntouu}), we compute:
\begin{equation}\label{eq:Qfor}
(\text{Q}^f[u])'= 2 f' |u'|^2  - f V' |u|^2 + {\rm Re}(2 f \bar{H} u'+f' \bar{H} u)-\frac 12 f''' |u|^2,
\end{equation}
\begin{equation}\label{eq:Q1for}
(\text {\fontencoding{LGR}\selectfont \Coppa}^h[u])'=h\left(|u'|^2 +(V-\omega^2)|u|^2\right) -\frac12 h'' |u|^2 +h\, {\rm Re} (u\bar H),
\end{equation}
\begin{equation}\label{eq:Q2for}
(\text {\fontencoding{LGR}\selectfont \koppa}^y[u])'= y ' \left(|u'|^2+(\omega^2-V)|u|^2\right) -yV'|u|^2+ 2y\,{\rm Re} (u'\bar H).
\end{equation}

The virial currents we shall employ will be various combinations  of
$\text{Q}$,
$\text {\fontencoding{LGR}\selectfont \Coppa}$, $\text {\fontencoding{LGR}\selectfont \koppa}$
with
suitably selected functions $f$, $h$, $y$.
Note that the choice of these functions may
depend on $a$, $\omega$, $m$, $\Lambda$, but, again, we temporarily
suppress this from the notation.

\subsection{The frequency-localised conserved energy currents}
As in our survey~\cite{stabi}, we shall need, in addition to the above,
a frequency-localised analogue of the conserved energy current ${\bf J}^T$.
Whereas in~\cite{stabi}, we introduced also a frequency-localised version of
the red-shift current ${\bf J}^N$, here we shall use in its place
a frequency-localised version of the (again conserved) current ${\bf J}^K$.

Again, fix Kerr parameters $M>0$ and $|a|<M$ and frequency parameters
$\omega\in\mathbb R$, $m\in\mathbb Z$, and $\Lambda\in \mathbb R$.
The ``frequency-localised'' versions of ${\bf J}^T$ and ${\bf J}^K$ are then  defined as follows:
\begin{eqnarray*}
\text{Q}^T[u]&=&\omega\, {\text Im}  (u'\overline {u}),\\
\text{Q}^K[u]&=& \left(\omega - \upomega_+m\right)\, {\text Im} (u'\overline {u}),
\end{eqnarray*}
where $\upomega_+ = \frac{a}{2Mr_+}$ is the ``angular velocity'' of the event horizon.
\index{energy currents! fixed-frequency current templates! $\text{Q}^T[u]$}
\index{energy currents! fixed-frequency current templates!  $\text {Q}^K[u]$}
For $u$ satisfying $(\ref{e3iswsntouu})$, we have
\begin{equation}\label{eq:Q3for}
\left(\text{Q}^T[u]\right)' = \omega\, {\text Im}  (H\overline {u}),
\end{equation}
\begin{equation}\label{eq:Q4for}
\left(\text{Q}^K[u]\right)' = \left(\omega - \upomega_+m\right)\, {\text Im}  (H\overline {u}).
\end{equation}

\section{The frequency localised multiplier estimates}\label{freqLocEst}
In the present
section, using the current templates of Section~\ref{sct}, we
will estimate smooth solutions $u$ to the radial o.d.e.~(\ref{e3iswsntouu}) with a general smooth right hand side $H$ and which satisfy
the boundary conditions~(\ref{eq:b-}) and~(\ref{eq:b+}).
The point is to obtain estimates which are uniform in the frequency
parameters $(\omega, m, \Lambda)$. In view of future applications, we will write the result as an independent theorem.
We apply this theorem several times in the present paper (in slightly different contexts)
in Sections~\ref{summation},~\ref{continuityargsec},~\ref{precise} and~\ref{boundSuff}.
We remark that the theorem can in principle be applied in future applications independently of the specific setup
of Section~\ref{cartersupersection}.

Before stating the theorem, given $|a|\le a_0<M$,
set $R_- \doteq r_+ + \frac{1}{2}\left(r_{\rm red}-r_+\right)$ where $r_{\rm red}$ is the constant from Proposition~\ref{specialises..} and set $R_+ \doteq 2R_{\rm large}$, where $R_{\rm large}$ is the constant from Proposition~\ref{lrp}. These values will be referred to below.
 The precise statement of the main result of this section is

\begin{theorem}\label{phaseSpaceILED}
Given  $0\le a_0<M$,
there exist positive
parameters $\omega_{\rm high}$, $\omega_{\rm low}$, $\epsilon_{\rm width}$, $E$ and $R_{\infty}^*$,
such that the following is true.

Let $0\le a\le a_0$ and let $(\omega, m, \Lambda)$ be an admissible frequency triple.

Then there exist
functions $f$, $h$, $y$, $\hat y$, $\tilde y$, $\chi_1$ and $\chi_2$, and
a value $r_{\rm trap}$,
depending on the parameters $a_0$, $M$, $a$
and the frequency triple
$(\omega, m, \Lambda)$ but
satisfying the uniform bounds
\[
|r_{\rm trap}-r_+|^{-1}+|r_{\rm trap}|+
\left|f\right| + \Delta^{-1}r^2\left|f'\right| + \left|h\right| + \left|y\right| + \left|\tilde y\right| + \left|\hat y\right| + \left|\chi_1\right| + \left|\chi_2\right|  \leq B,
\]
\[f + y = 1,\ f' = 0,\ h = 0,\ \left|\tilde y\right| \leq B\exp\left(-br^*\right),\ \hat y = 0,\ \chi_1 = 0,\ \chi_2 = 1\text{ for }r^* \geq R^*_{\infty},
\]
such that, for all smooth solutions $u$ to the radial o.d.e.~(\ref{e3iswsntouu}) with
right hand side $H$, satisfying moreover the boundary conditions~(\ref{eq:b-}) and~(\ref{eq:b+}), we have,
\begin{align}\label{fromPhaseSpace2}
b&\int_{R^*_-}^{R_+^*}\left[\left|u'\right|^2 + \left(\left(1-r^{-1}r_{\rm trap}\right)^2\left(\omega^2 + \Lambda\right) + 1\right)\left|u\right|^2\right]\, dr^*
\\ \nonumber &\le \int_{-\infty}^{\infty}H \cdot  (f, h, y, \chi) \cdot (u, u')\, dr^* + 1_{\{\omega_{\rm low} \leq \left|\omega\right| \leq \omega_{\rm high}\}\cap \{\Lambda \leq \epsilon_{\rm width}^{-1}\omega_{\rm high}^2\}}\left|u\left(-\infty\right)\right|^2.
\end{align}
The symbol $1_{\{\omega_{\rm low} \leq \left|\omega\right| \leq \omega_{\rm high}\}\cap \{\Lambda \leq \epsilon_{\rm width}^{-1}\omega_{\rm high}^2\}}$ denotes the indicator function for the set $\{\omega_{\rm low} \leq \left|\omega\right| \leq \omega_{\rm high}\}\cap \{\Lambda \leq \epsilon_{\rm width}^{-1}\omega_{\rm high}^2\}$, and
\begin{align}
H \cdot  (f, h, y, \chi) \cdot (u, u') \doteq &-2f\text{Re}\left(u'\overline{H}\right) -f'\text{Re}\left(u\overline{H}\right) + h\text{Re}\left(u\overline{H}\right) -E\chi_2\omega\text{Im}\left(H\overline{u}\right)
\\ \nonumber &\,\,- E\chi_1\left(\omega-\upomega_+m\right)\text{Im}\left(H\overline{u}\right) -2y\text{Re}\left(u'\overline{H}\right)-2\tilde y\text{Re}\left(u'\overline{H}\right) -2\hat y\text{Re}\left(u'\overline{H}\right).
\end{align}
\end{theorem}
Before discussing the proof of the theorem, we give a few remarks pertaining to the application of Theorem~\ref{freqLocEst} in Section~\ref{summation} in the context of $u$ arising from Carter's separation applied to a solution $\Psi$ of the inhomogeneous wave equation.
\begin{bigremark}\label{Traprem}
For frequencies
in the trapping regime, $r_{\rm trap}$ will denote the unique trapped value of $r$ associated to the triple $\left(\omega,m,\Lambda\right)$. Otherwise, $r_{\rm trap}$ will be set to $0$. This will capture the degeneration due to trapping.
\end{bigremark}

\begin{bigremark}\label{R2}
The specific behaviour of the functions $f$, $h$, $y$, $\hat y$, $\tilde y$, $\chi_1$ and $\chi_2$ in the region $r^* \geq R_{\infty}^*$ will be useful in Section~\ref{summation} when we sum~(\ref{fromPhaseSpace2}) to produce a physical space estimate.
\end{bigremark}

\begin{bigremark}\label{thisIsPhaseILED}
If we consider the right hand side of the estimate~(\ref{fromPhaseSpace2}) as ``data'', a direct application of Plancherel (see the explicit formulas in Section~\ref{separationSubsection}) shows that~(\ref{fromPhaseSpace2}) is the phase space versions of integrated local energy decay.
\end{bigremark}

\begin{bigremark}Let us draw particular attention to the term $1_{\{\omega_{\rm low} \leq \left|\omega\right| \leq \omega_{\rm high}\}\cap \{\Lambda \leq \epsilon_{\rm width}^{-1}\omega_{\rm high}^2\}}\left|u\left(-\infty\right)\right|^2$ on the right hand side of the estimate~(\ref{fromPhaseSpace2}). This term must
initially be put on the right hand side of the corresponding integrated energy decay statement (cf.~Remark~\ref{thisIsPhaseILED}). Eventually, this term will be dealt with
in Section~\ref{whitinghere} using the quantitative refinement~\cite{shlapRot} of mode stability.
\end{bigremark}

The proof proper of Theorem~\ref{phaseSpaceILED} will be given in Section~\ref{putting}.
It will be based on a series of propositions proven
in Sections~\ref{largeSuper}--\ref{whit} below,
where $(\ref{fromPhaseSpace2})$ is successively
obtained for various ranges of admissible frequency triples.
These frequency ranges, however,
are determined by parameters which must be suitably optimised so
as for our constructions to be possible.
We begin
thus with  a discussion of these ranges and an overview of the constructions.

\subsection{The frequency ranges}
\label{freerange}
Let $a_0<M$.
Fix a parameter $\alpha$ (depending only on $a_0$, $M$)
 satisfying the statement of Lemma~\ref{lem:3}.
 For each $0\le a\le a_0$,
and all $\omega_{\rm high}>0$, $\epsilon_{\rm width}>0$, we define the frequency ranges
$\mathcal{G}_{\mbox{$\flat$}}(\omega_{\rm high})$,
$\mathcal{G}_{\lessflat}(\omega_{\rm high}, \epsilon_{\rm width})$,
$\mathcal{G}_{\mbox{$\natural$}}(\omega_{\rm high}, \epsilon_{\rm width})$,
 $\mathcal{G}_{\mbox{$\sharp$}}(\omega_{\rm high},\epsilon_{\rm width})$,
 $\mathcal{G}^{\mbox{$\sharp$}}(\omega_{\rm high},\epsilon_{\rm width})$
 by
\begin{itemize}
\item
$\mathcal{G}^{\mbox{$\sharp$}}=\{(\omega, m, \Lambda){\rm\ admissible\ }$ :
$\Lambda \ge (\frac{a}{2Mr_+} + \alpha)^{-2}\omega_{\rm high}^2$, $m\omega\in (0,\frac {am^2}{2Mr_+}+\alpha\Lambda]\}$,
\item
$\mathcal{G}_{\mbox{$\sharp$}}=\{(\omega, m, \Lambda){\rm\ admissible\ }$ :
$|\omega|\ge \omega_{\rm high}$, $\Lambda<\epsilon_{\rm width} \omega^2, m\omega\not\in (0,\frac {am^2}{2Mr_+}+\alpha\Lambda]\}$,
\item
$\mathcal{G}_{\mbox{$\lessflat$}}=\{(\omega, m, \Lambda){\rm\ admissible\ }$ :
$\Lambda \geq \epsilon_{\rm width}^{-1}\omega_{\rm high}^2$, $\epsilon_{\rm width} \Lambda > \omega^2, m\omega\not\in (0,\frac {am^2}{2Mr_+} + \alpha\Lambda]\}$,
\item
$\mathcal{G}_{\mbox{$\natural$}}=\{(\omega, m, \Lambda){\rm\ admissible\ }$ :
$|\omega|\ge \omega_{\rm high}$, $\epsilon_{\rm width} \Lambda \le \omega^2\le \epsilon_{\rm width}^{-1} \Lambda, m\omega\not\in (0,\frac {am^2}{2Mr_+}+\alpha\Lambda]\}$,
\item
$\mathcal{G}_{\mbox{$\flat$}}=\{(\omega, m, \Lambda){\rm\ admissible\ }$ :
$|\omega| < \omega_{\rm high}$, $\Lambda < \epsilon_{\rm width}^{-1}\omega_{\rm high}^2\}$.
\end{itemize}
The parameters $\omega_{\rm high}$ and $\epsilon_{\rm width}$ will be fixed in the course of the proof of Theorem~\ref{freqLocEst}, see Section~\ref{putting}.

We see easily that
\begin{lemma}
\label{everythingCovered}
With the above notation, for all $0\le a\le a_0$
if $(\omega, m, \Lambda)$ is admissible,
then,
for all choices of parameters $\omega_{\rm high}$, $\epsilon_{\rm width}$,
$(\omega, m, \Lambda)$ lies in exactly  one of the
frequency ranges $\mathcal{G}^{\mbox{$\sharp$}}$,
$\mathcal{G}_{\mbox{$\sharp$}}$, $\mathcal{G}_{\mbox{$\lessflat$}}$,
$\mathcal{G}_{\mbox{$\natural$}}$, or $\mathcal{G}_{\mbox{$\flat$}}$.
\end{lemma}
\begin{proof}
To see this, observe that
\[|\omega| \geq \omega_{\rm high} \text{ and }m\omega \in \left(0,\frac{am^2}{2Mr_+}+\alpha\Lambda\right] \Rightarrow
\Lambda \ge \left(\frac{a}{2Mr_+} + \alpha\right)^{-2}\omega_{\rm high}^2.\]
\end{proof}

Our constructions of currents will vary according to the frequency range of the triple $(\omega, m,\Lambda)$. We now give an overview of these constructions.
\subsection{Overview}
For each admissible triple $(\omega, m, \Lambda)$,
we would like to find a current $\text{Q}$ consisting of various combinations of
$\text{Q}^f$,
$\text {\fontencoding{LGR}\selectfont \koppa}^y$, $\text {\fontencoding{LGR}\selectfont \Coppa}^h$, $\text{Q}^T$ and $\text{Q}^K$ satisfying the bulk coercivity property
\begin{equation}
\label{yaxvw}
\int_{-\infty}^{\infty}  \text{Q}'[u] \ge b\int_{R^*_-}^{R^*_+}  \left(|u'|^2 +\left(1-r^{-1}r_{\rm trap}\right)^2(\Lambda +\omega^2)|u|^2 +|u|^2\right)  - \int_{-\infty}^{\infty}
H \cdot  (f, h, y, \chi) \cdot (u, u'),
\end{equation}
and, ideally, the boundary positivity property
\begin{equation}\label{boundaryOK}
\text{Q}\left(\infty\right) - \text{Q}\left(-\infty\right) \leq 0.
\end{equation}

The terms Q, $r_{\rm trap}$, $H$, $f$, $h$, $y$ should
all be understood to depend on $\omega$, $m$, and $\Lambda$,
here omitted for brevity, and the integrals are with respect to $r^*$. One restricts the domain of integration on
the first term to $[R_-^*,R_+^*]$
on the right hand side because one expects this virial current
not to control things at the horizon and infinity.

The most difficult aspect of establishing~(\ref{yaxvw}) is the need to understand trapping. In order to do this this we will heavily rely on the analysis of the potential $V_0$ carried out in Section~\ref{Vpropsec}. For frequencies for which trapping is relevant, $r_{\rm trap}$ will denote the unique value of $r$, associated with the
frequency triple, where the estimate must degenerate. For frequencies where trapping is not relevant, $r_{\rm trap} = 0$.

The fundamental obstruction to achieving~(\ref{boundaryOK}), on the other hand, is superradiance (see Section~\ref{supInKerr}). For non-superradiant frequencies, i.e.~frequencies which satisfy $\omega\left(\omega - \upomega_+m\right) \geq 0$, one may easily\footnote{For the moment we are suppressing the fact that this estimate may be insufficiently strong if $0 \leq \omega\left(\omega-\upomega_+m\right) \ll \left(\omega-\upomega_+m\right)^2$. See Section~\ref{boundaryStillaProblem}.} control these fluxes via a sufficiently large multiple of the conserved Q$^T$ current:
\begin{equation}\label{microEnergyEst}
\int_{-\infty}^{\infty}\text{Im}\left(H\overline{u}\right) = \int_{-\infty}^{\infty}\left(\text{Q}^T\right)' = \text{Q}^T(\infty) - \text{Q}^T(-\infty) = \omega^2\left|u(\infty)\right|^2 + \omega\left(\omega-\upomega_+m\right)\left|u(-\infty)\right|^2.
\end{equation}
However, for superradiant frequencies, where $\omega\left(\omega-\upomega_+m\right) < 0$, no conserved current gives a coercive estimate for the boundary terms and it is thus no longer clear how to arrange for~(\ref{boundaryOK}).  As it turns out, see Section~\ref{largeSuper} below, one of the miracles of the Kerr geometry is that trapping and superradiance are disjoint;
exploiting this, one may indeed  establish~(\ref{boundaryOK}) for sufficiently large frequencies
with the help of~(\ref{yaxvw}) and a large positive parameter. Unfortunately, for bounded superradiant frequencies, one does not have a large parameter at hand. We will not be able to carry out such a scheme, and we will not in fact establish~(\ref{boundaryOK}); see
Section~\ref{theseAreBounded}.

We now turn to a more detailed discussion of the difficulties in each frequency range.
The reader may wish to refer to this when reading
Sections~\ref{largeSuper}--\ref{whit} below.

\subsubsection{The $\mathcal{G}^{\mbox{$\sharp$}}$ range}\label{largeSuper000}

This is the large frequency superradiant regime. Lemma~\ref{lem:4} shows that these frequencies are not trapped. Thus, it is not difficult to establish~(\ref{yaxvw}) via the combination of a
$\text{Q}^f$ and
$\text {\fontencoding{LGR}\selectfont \Coppa}^h$
 current with a monotonically increasing $f$ which switches signs at the unique maximum of the potential and a positive function $h$ which peaks near the maximum of the potential.

As for the boundary terms, despite the lack of a coercive conserved current, we will appeal to the aforementioned miracle that superradiant frequencies are not trapped to find a large parameter which will still allow us to achieve~(\ref{boundaryOK}). Briefly put, Lemma~\ref{lem:4} shows that we have a quantitatively large ``classically forbidden region'',~and from this one expects to derive an estimate for $u$ near $r_{\max}$ which comes with a large parameter.

\subsubsection{The $\mathcal{G}_{\mbox{$\sharp$}}$ range}\label{lanosudiscu}
This is a non-superradiant regime where the time frequency $\omega$ is large and dominates the other parameters. It is easy to see that a $\text {\fontencoding{LGR}\selectfont \koppa}^y$ current with an appropriate choice of $y$ will establish~(\ref{yaxvw}).

Of course, the boundary terms may be easily controlled with~(\ref{microEnergyEst}).

\subsubsection{The $\mathcal{G}_{\mbox{$\lessflat$}}$ range}\label{boundaryStillaProblem}
This is a non-superradiant regime where the angular frequency $\Lambda$ is large and dominates the other parameters. One may easily show that the conclusions of Lemma~\ref{lem:4} still hold, and, as in Section~\ref{largeSuper}, it is not difficult to establish~(\ref{yaxvw}).

Turning to the boundary terms, note that $\omega\left(\omega-\upomega_+m\right)$ and $\left(\omega-\upomega_+m\right)^2$ are not necessarily comparable in this regime. Thus, even though the flux $\text{Q}^T[u]|_{r=\infty}$ may be controlled with~(\ref{microEnergyEst}), the estimate~(\ref{microEnergyEst}) does not provide sufficient control of the
flux $\text{Q}^K[u]|_{r=r_+}$. Fortunately, we may apply the same argument as in the Section~\ref{largeSuper} to control the horizon flux.

\subsubsection{The $\mathcal{G}_{\mbox{$\natural$}}$ range}\label{spartaheretoo}
This is a non-superradiant regime where the angular frequency $\Lambda$ and the time frequency $\omega$ are large and comparable. \underline{This is the regime of trapping} and hence the only frequency range where $r_{\rm trap}$ will be non-zero. The estimate~(\ref{yaxvw}) is achieved via a Q$^f$ current with a monotonically increasing function $f$ which switches sign at the unique maximum of the potential. The construction the function $f$ will heavily depend on the critical point analysis of $V_0$ carried out in Section~\ref{Vtrsec}.

The estimate~(\ref{boundaryOK}) is easily achieved via~(\ref{microEnergyEst}).

\subsubsection{The $\mathcal{G}_{\mbox{$\flat$}}$ range}\label{theseAreBounded}
This is a bounded frequency regime. It turns out to be useful to further split this frequency regime into the following four sub-regimes.
\begin{enumerate}
    \item $\left|\omega\right| \leq \omega_{\rm low}$, $0\le a < \tilde a_0$ and $m \neq 0$.
    \item $\left|\omega\right| \leq \omega_{\rm low}$ and $m =0$.
    \item $\left|\omega\right| \leq \omega_{\rm low}$, $m \neq 0$ and $a \geq \tilde a_0$.
    \item $|\omega|\ge \omega_{\rm low}$.
\end{enumerate}
Here $\omega_{\rm low}$ is the small parameter mentioned in Theorem~\ref{freqLocEst} and $\tilde a_0$ is a small parameter to be fixed in the course of the proof.

For the estimate~(\ref{yaxvw}) we will exploit two types of estimates. If $|\omega|\ge \omega_{\rm low}$ or  $\left|\omega\right| \leq \omega_{\rm low}$, $m \neq 0$ and $a \geq \tilde a_0$, then we will either have $\omega^2 \sim 1$ or $\left(\omega - \upomega_+m\right)^2 \sim 1$. In this case we will employ $\text {\fontencoding{LGR}\selectfont \koppa}^y$ currents with exponential multipliers $y \doteq \exp\left(\int \upsilon\right)$ and appropriate functions $\upsilon$. If $\left|\omega\right| \leq \omega_{\rm low}$ and $\omega_{\rm low}$ is sufficiently small, then in regions with $1 \lesssim V$ we will have $1 \lesssim V - \omega^2$. We will apply $\text {\fontencoding{LGR}\selectfont \Coppa}^h$ currents to exploit this positivity of $V - \omega^2$.

As in Section~\ref{largeSuper}, the fundamental difficulty is a lack of control of the boundary terms for superradiant frequencies. It turns out that when $\omega^2 \ll 1$, i.e. $\left|\omega\right| \leq \omega_{\rm low}$ for $\omega_{\rm low}$ sufficiently small, then $\omega$ arises naturally as a small parameter and we will again be able to achieve~(\ref{boundaryOK}). However, for bounded frequencies with $|\omega| \geq \omega_{\rm low}$ there is no large or small parameter to exploit. Instead, for this frequency range we will only be
able to establish the weaker
\begin{equation}
\label{asthenes}
\text{Q}(\infty) - \text{Q}(-\infty) \leq B\left|u\left(-\infty\right)\right|^2.
\end{equation}
This is the origin of the term $1_{\{\omega_{\rm low} \leq \left|\omega\right| \leq \omega_{\rm high}\}\cap \{\Lambda \leq \epsilon_{\rm width}^{-1}\omega_{\rm high}^2\}}\left|u\left(-\infty\right)\right|^2$ on the right hand side of the estimate~(\ref{fromPhaseSpace2}).

We now turn to the detailed constructions of the currents for each frequency regime.

\subsection{The $\mathcal{G}^{\mbox{$\sharp$}}$ range}\label{largeSuper}
As discussed in Section~\ref{largeSuper000}, this defines a large frequency superradiant regime,
and by the results of Section~\ref{Vsrsec}, frequencies in this regime can
be viewed as non-trapped.

Once we have made our final choice of the parameter $\omega_{\rm high}$, then
for $(\omega, m, \Lambda)\in  \mathcal{G}^{\mbox{$\sharp$}}(\omega_{\rm high})$,
we will set the functions $y$, $\hat y$ and $\tilde y$ together with
the parameter $r_{\rm trap}$ from the statement of Theorem~\ref{phaseSpaceILED}
to be $0$. The desired coercivity in this range
and remaining functions $f$, $h$, $\chi_1$ and $\chi_2$ are given by the following:
\begin{proposition}\label{odeEst1}
Let $a_0<M$. Then, for all
$E\ge 2$,
 for all
$\omega_{\rm high}$ sufficiently big depending on $E$, for all
$R_{\infty}$ sufficiently big, for all $0\le a\le a_0$,
$(\omega, m, \Lambda)\in  \mathcal{G}^{\mbox{$\sharp$}}(\omega_{\rm high})$,
there exist functions $f$, $h$, $\chi_1$ and $\chi_2$
satisfying the uniform bounds
\[
\left|f\right| + \Delta^{-1}r^2\left|f'\right| + \left|h\right| + \left|\chi_1\right| + \left|\chi_2\right| \leq B\left(\omega_{\rm high}\right),\]
\[f = 1,\ h = 0,\ \chi_1 = 0\text{ and }\chi_2 = 1\text{ for }r\geq R_{\infty},
\]
such that, for all
smooth solutions $u$ to the radial o.d.e.~(\ref{e3iswsntouu}) with
right hand side $H$, satisfying moreover the boundary conditions~(\ref{eq:b-}) and~(\ref{eq:b+}) we have the estimate
\begin{align}\label{sharpEst}
b&\int_{R^*_-}^{R^*_+}\left (|u'|^2+(\omega^2+\La) |u|^2\right)
\\ &\leq \int_{-\infty}^\infty \left (-2f \,{\text{Re}} (u'\overline{H}) - (f'+h)\, {\text{Re}} (u\overline{H})+E\chi_2\omega {\text{Im}} (H\overline u) + E\chi_1\left(\omega-\upomega_+m\right){\text{Im}}(H\overline u)\right ).
\end{align}
\end{proposition}

\begin{proof}
As $\mathcal{G}^{\mbox{$\sharp$}}$ is a superradiant regime with $\Lambda > 0$,  the
conclusions of both Lemma~\ref{lem:3} and Lemma~\ref{lem:4} apply. In particular, the potential
$V_0$ has a unique $r^0_{\rm max}$ which is a maximum,
and satisfies
\begin{equation}\label{reallyNotTrapped}
 V_0(r^0_{\rm max}) - \omega^2 \geq c\Lambda,
\end{equation}
for some positive constant $c$ depending only on $a_0$ and $M$.

We shall first need to establish the following lemma, which
 shows that the full potential $V$ behaves similarly in the range $\mathcal{G}^{\mbox{$\sharp$}}\left(\omega_{\rm high}\right)$ for sufficiently large $\omega_{\rm high}$.

\begin{lemma}\label{theOtherLemma}There exists a $\delta > 0$ depending only on $a_0$ and $M$ such that for sufficiently large $\omega_{\rm high}$ and $(\omega,m,\Lambda) \in \mathcal{G}^{\mbox{$\sharp$}}\left(\omega_{\rm high}\right)$, then $V$ has a unique critical point $r_{\rm max}$ and satisfies
\[V(r)-\omega^2\ge b\Lambda,\qquad \forall r\in (r_{\max}-\de,r_{\max}+\de),\]
\begin{equation}\label{eq:Vdeg}
-(r-r_{\max}) \frac d{dr} V(r)\ge b \Lambda \frac {(r-r_{\max})^2}{r^{4}}, \qquad \forall r\in [r_+,\infty),
\end{equation}
\[\left|r_{\rm max} - r^0_{max}\right| \leq B\Lambda^{-1}.\]
\end{lemma}
\begin{proof}
Let us first refine
our estimates on $V_0$ for frequencies $(\omega, m, \Lambda)\in
 \mathcal{G}^{\mbox{$\sharp$}}\left(\omega_{\rm high}\right)$.

Using the fact that $\left|\frac{dV_0}{dr}\right| \leq B\Lambda$,~(\ref{reallyNotTrapped}) implies that we may find a $\delta_1 > 0$ depending only on $a_0$ and $M$ such that
\[
V_0(r)-\omega^2\ge \frac{c}{2} \Lambda,\qquad \forall\, r\in [r^0_{max}-\delta_1,r^0_{max}+\delta_1].
\]

Lemma~\ref{lem:1} implies that
$r^0_{\rm max}$ is bounded from above independently of the frequency parameters:
\[
r^0_{\rm max} \le B.
\]
Furthermore, Lemma~\ref{lem:3} and the bound $\left|\frac{d^2V_0}{dr^2}\right| \leq B\Lambda$ implies $r^0_{\rm max}$ is also bounded away from $r_+$ independently of the frequency parameters:
\[
r^0_{\rm max} - r_+ \ge b.
\]

Lemma~\ref{lem:3} also implies that the full potential
$V=V_0+V_1$ satisfies
\begin{equation}\label{increaseOnHorizon}
\frac {d}{dr} V(r_+)\ge \frac {d}{dr} V_0(r_+)\ge \hat c \Lambda
\end{equation}
for a positive constant $\hat c$ depending only on $a_0$ and $M$.

Recall now that the proof of Lemma~\ref{lem:1} showed that the function
\[
\frac {d}{dr} \left ((r^2+a^2)^3\frac {d}{dr} V_0(r)\right)
\]
is either non-positive on $[r_+,\infty)$ or there exists a unique point
$r_+ \leq r_1 < r^0_{\rm max}$ such that $\frac {d}{dr} \left ((r^2+a^2)^3\frac {d}{dr} V_0(r)\right)$ is positive on $[r_+,r_1)$ and negative on $(r_1,\infty)$.

We first consider the case where the point $r_1$ exists. Then,
\[
 \frac {d}{dr} V_0(r)\ge \hat c \frac{(r_+^2+a^2)^3}{(r_1^2+a^2)^3}\Lambda\qquad \forall\, r\in [r_+,r_1].
\]
Next, recall from the proof of Lemma~\ref{lem:1} that
\[\frac{d}{dr}\left((r^2+a^2)^3\frac{dV_0}{dr}\right) = -6\Lambda\left(r^2-2Mr +4Mr\sigma+\frac{a^2}3-\frac 23 a^2 \frac{m^2}\Lambda\right),\]
and furthermore, by definition, $\frac{d}{dr}\left((r^2+a^2)^3\frac{dV_0}{dr}\right)$ is negative on $(r_1,\infty)$. Thus, we can choose a value $r'_1\in (r_1,r_{max}^0)$ such that
\begin{equation}\label{aProp1}
\frac {d}{dr} V_0(r)\ge \frac {\hat c}{2}\frac{(r_+^2+a^2)^3}{(r_1^2+a^2)^3}\Lambda,\qquad \forall r\in [r_+,r_1'],
\end{equation}
and
\begin{equation}\label{aProp2}
\frac {d}{dr} \left ((r^2+a^2)^3\frac {d}{dr} V_0(r)\right)\le -\tilde c \Lambda r^2,\qquad \forall r\in [r_1',\infty),
\end{equation}
for a positive constant $\tilde c$ independent of the frequency parameters.

In the case where $r_1$ does not exists, the same argument \emph{mutatis mutandis} will produce a value $r_1'$ with the properties~(\ref{aProp1}) and~(\ref{aProp2}).

Now, we simply observe that the potential $V_1$ satisfies the bounds
\[\left|V_1\right|\le B r^{-3},\qquad
\left|\frac {d}{dr} V_1(r)\right|\le B r^{-4},\qquad \left|\frac {d}{dr} \left ((r^2+a^2)^3\frac {d}{dr} V_1(r)\right)\right|\le B r.\]

For $\omega_{\rm high}$ sufficiently large (and hence large $\Lambda$),
 it immediately follows that, for
 $(\omega, m, \Lambda)\in
 \mathcal{G}^{\mbox{$\sharp$}}\left(\omega_{\rm high}\right)$,
 the full potential $V=V_0+V_1$ cannot  have any critical points on $[r_+,r_1']$ and has a unique maximum
$r_{\max}\in [r_1',\infty)$ which satisfies $\left|r_{\max} - r^0_{max}\right| \leq B\Lambda^{-1}$.

The proof concludes by  applying the fact that $\left|\frac{dV}{dr}\right| \leq B\Lambda$.
\end{proof}

We now proceed to the construction of a suitable current for the regime
$\mathcal{G}^{\mbox{$\sharp$}}$.
The current will be of the form:
$$
\text{Q}=\text{Q}^f + \text {\fontencoding{LGR}\selectfont \Coppa}^h - E\chi_1\text{Q}^K - E\chi_2\text{Q}^T,
$$
for appropriate functions $f$, $h$, $\chi_1$ and $\chi_2$ and large constant $E$.

It is simpler to describe this procedure in three stages.

{\bf Stage 1.}~We first apply current $\text{Q}^f$ where $f$
is a function chosen such that
\begin{equation}
\label{choices1}
f=-1 {\rm\ at\ } r=r_+, \qquad f=0 {\rm\ at\ } r=r_{\max},\qquad
f=1 {\rm\ when\ } r^* \geq R^*_{\infty},
\end{equation}
\begin{equation}
\label{choices2}
f'(r^*) > 0 {\rm\ for\ all\ } r \leq R_1, \qquad f'(r^*) \geq 0 {\rm\ for\ all\ } r>r_+, \qquad \left|f\right| + \Delta^{-1}r^2\left|f'\right| \leq B.
\end{equation}

Application of $(\ref{eq:Qfor})$ yields then
\begin{align}
\label{metastra}
\nonumber
\int_{-\infty}^\infty \left (2f'|u'|^2-fV' |u|^2 -\frac 12 f{'''}|u|^2\right )&=
\left (|u'|^2+(\omega-\upomega_+m)^2 |u|^2\right)_{r=r_+} + \left (|u'|^2+\omega^2 |u|^2\right)_{r=\infty}
\\ &\qquad- \int_{-\infty}^\infty \left (2 f \,{\text{Re}} (u'\overline{H}) + f'\, {\text{Re}} (u\overline{H})\right).
\end{align}
Let us moreover require that $f$ above has been chosen so that
in addition to $(\ref{choices1})$, $(\ref{choices2})$, the following coercivity
property holds
\begin{equation}
\label{choices3}
-fV'-\frac 12 f'''\ge \La \frac {\Delta (r-r_{\max})^2}{r^7}, {\rm\ for\ all\ } r>r_+.
\end{equation}
 Since $f$ vanishes at $r=r_{\max}$ and $V'$ obeys the
property \eqref{eq:Vdeg}, we can easily arrange such that
in addition to $(\ref{choices1})$, $(\ref{choices2})$ and $(\ref{choices3})$,
we have
\begin{equation}
\label{choices4}
fV'\ge b \La \frac {\Delta (r-r_{\max})^2}{r^7}.
\end{equation}
It remains to impose
\begin{equation}
\label{choices5}
f'''(r)<0  {\rm\ in\ a\ small\ neighbourhood\ of\ }r_{\max}, \qquad
|f'''(r)|\le B \Delta r^{-5}.
\end{equation}

Note that the reader may easily construct a function $f$ satisfying the conditions (\ref{choices1}), (\ref{choices2}), (\ref{choices3}), (\ref{choices4}) and (\ref{choices5}). With the above choice of $f$,
the left hand side of $(\ref{metastra})$ is now
non-negative, but still degenerate at $r=r_{\max}$. As discussed
in Section~\ref{Vsrsec}, the bound
$V(r_{\max})-\omega^2\ge b\La$ indicates this regime is non-trapped and thus
the degeneracy
may be removed with the help of the current $\text {\fontencoding{LGR}\selectfont \Coppa}^h$.
The more serious problem is a lack of
control of the boundary terms on the right
hand side, due to the  superradiant condition. However, as we shall see below,
we will be able to overcome this by exploiting the largeness of the potential in the region $(r_{\max}-\de,r_{\max}+\de)$.

{\bf Stage 2.}~We now add a $\text {\fontencoding{LGR}\selectfont \Coppa}^h$ current with a function $h \doteq A\tilde h$ such that
\begin{equation}\label{choices6}
h \geq 0,\qquad \left|\tilde h\right| \leq B,
\end{equation}
\begin{equation}\label{choices7}
\text{supp}\left(h\right) \subset [r_{\max}-\de,r_{\max} + \de],\qquad \tilde h = 1\text{ for } r \in [r_{\max}-\de/2,r_{\max}+\de/2]
\end{equation}
and $A$ is a constant to be determined.

We obtain
\begin{align}
\label{metatastra2}
\nonumber
\int_{-\infty}^\infty &\left ((2f'+Ah) |u'|^2+\left (A\tilde h(V-\omega^2)-fV'\right) |u|^2 -\frac 12 (f{'''}+Ah'')|u|^2\right )\\&=
\left(|u'|^2+(\omega-\upomega_+m)^2 |u|^2\right)_{r=r_+} + \left (|u'|^2+\omega^2 |u|^2\right)_{r=\infty}
- \int_{-\infty}^\infty \left(2 f \,{\text{Re}} (u'\overline{H}) + (f'+h)\, {\text{Re}} (u\overline{H})\right).
\end{align}
Note that as long as $A \leq \tilde\epsilon\omega_{\rm high}^2$ for a sufficiently small constant $\tilde \epsilon$ only depending on $a_0$ and $M$, the integrand of the left hand side of $(\ref{metatastra2})$
will be positive. Moreover, this integrand has the property
that it satisfies
$$
\ge bA (|u'|^2 + \Lambda |u|^2),\quad \forall\, r\in [r_{\max}-\frac{\delta}2,r_{\max}+\frac\delta 2].
$$

{\bf Stage 3.}~We now let $\chi_1(r)$ be a smooth function such that
\begin{equation}\label{somerequirements}
\chi_1 = 1\text{ for } r\in [r_+,r_{\max}-\frac{\de}{2}],\qquad \chi_1 = 0\text{ for }r \in [r_{\max}+\frac{\delta}{2},\infty),\qquad \left|\chi_1\right| \leq B.
\end{equation}
Since $E \geq 2$, we have
\begin{align*}
\left (|u'|^2+(\omega-\upomega_+m)^2 |u|^2\right)_{r=r_+} &\leq E\int_{-\infty}^{\infty}(\chi_1\text{Q}^K)' \\&=
E\int_{r_{\max}-\frac{\de}{2}}^{r_{\max}+\frac{\de}{2}}\chi_1'\left(\omega-\upomega_+m\right)\text{Im}\left(u'\overline{u}\right) + E\int_{-\infty}^{\infty}\chi_1\left(\omega-\upomega_+m\right)\text{Im}\left(H\overline{u}\right).
\end{align*}
Now, we require that $\omega_{\rm high}$ be sufficiently large so as to satisfy $E\delta^{-1} \ll (1/2)\tilde\epsilon\omega_{\rm high}^2$, and then we set $A \doteq (1/2)\tilde\epsilon\omega_{\rm high}^2$. This choice of $A$ will both maintain the coercivity of the left hand side of~(\ref{metatastra2}) and yield
\begin{eqnarray*}
E\left|\int_{r_{\max}-\frac{\de}{2}}^{r_{\max}+\frac{\de}{2}}\chi_1'\left(\omega-\upomega_+m\right)\text{Im}\left(u'\overline{u}\right)\right| &\leq& E\delta^{-1}\int_{r_{\max}-\frac{\de}{2}}^{r_{\max}+\frac{\de}{2}}\left(\left|u'\right|^2 + \left(\omega^2+m^2\right)\left|u\right|^2\right) \\
&\ll& A\int_{r_{\max}-\frac{\de}{2}}^{r_{\max}+\frac{\de}{2}}\left(\left|u'\right|^2 + \Lambda\left|u\right|^2\right).
\end{eqnarray*}
We can, of course, carry out an analogous construction with a cutoff $\chi_2$, satisfying
\begin{equation}\label{somerequirements2}
\chi_2 = 1\text{ for }r \in [r_{\max} + \frac{\de}{2},\infty),\qquad \chi_2 = 0\text{ for }r\in [r_+,r_{\max}-\frac{\de}{2}],\qquad \left|\chi_2\right| \leq B,
\end{equation}
and the current $\text{Q}^T$. Then, adding the currents $-E\chi_1\text{Q}^K-E\chi_2\text{Q}^T$ will give us the necessary control of the boundary terms.

Observing that the left hand side of the resulting estimate is coercive (with weights which degenerate however as $r^* \to \pm\infty$), restricting the domain of integration of the left hand side then yields~(\ref{sharpEst}).
\end{proof}

\subsection{The $\mathcal{G}_{\mbox{$\sharp$}}$ range}
\label{sub:sharp}
As discussed in Section~\ref{lanosudiscu},
$\mathcal{G}_{\mbox{$\sharp$}}$ defines a large frequency regime (whose definition still depends on
parameters $\omega_{\rm high}$ and $\epsilon_{\rm width}$, yet to be fixed)
where time frequencies will dominate angular frequencies.
The regime is manifestly non-superradiant, and, for suitable
choice of parameters, non-trapped.

Once we have made our final choice of parameters $\omega_{\rm high}$ and
$\epsilon_{\rm width}$, then for
 $(\omega, m, \Lambda)\in \mathcal{G}_{\mbox{$\sharp$}}(\omega_{\rm high},\epsilon_{\rm width})$,
we will set the functions $f$, $h$, $\hat y$, $\tilde y$ and $\chi_1$
appearing in Theorem~\ref{phaseSpaceILED} together with
the parameter $r_{\rm trap}$ to be $0$. The remaining function $y$ and
the desired coercivity property are given by

\begin{proposition}\label{odeEst2}Let $a_0<M$. Then, for all $\omega_{\rm high}$, $\epsilon_{\rm width}^{-1}$, $R_{\infty}$ sufficiently big, for all
$E\ge 2$,
$0\le a\le a_0$,
$(\omega, m, \Lambda)\in  \mathcal{G}_{\mbox{$\sharp$}}(\omega_{\rm high},\epsilon_{\rm width})$, there exists a  function $y$
satisfying the uniform bounds
\[
\left|y\right| \leq B,
\]
\begin{equation}\label{someyStuff}
 y = 1\text{ for }r^* \geq R^*_{\infty},
\end{equation}
such that,
for all
smooth solutions $u$ to the radial o.d.e.~(\ref{e3iswsntouu}) with
right hand side $H$, satisfying moreover the boundary conditions~(\ref{eq:b-}) and~(\ref{eq:b+}) we have the estimate
\begin{align*}
b\int_{R^*_-}^{R^*_+}\left( |u'|^2+{(\omega^2+\La)} |u|^2\right)\le
\int_{-\infty}^\infty \left (-2 y \,{\text{Re}} (u'\overline{H}) +E\omega {\text{Im}} (H \overline u)\right ).
\end{align*}
\end{proposition}

\begin{proof}
The construction of our currents will exploit the fact
that the range $\mathcal{G}_{\mbox{$\sharp$}}$ defines a large frequency
regime in which $\Lambda \ll \omega^2 $ (and thus also $m^2 \ll \omega^2$).
To handle the boundary terms, we will use that this regime is moreover manifestly non-superradiant,
and thus addition of a sufficiently large multiple of the $\text{Q}^T$ current provides
positive terms at $r=r_+$ and $r=\infty$.

We turn to the details. First of all, it is easy to see that the admissibility inequalities $\Lambda \geq 2a|m\omega|$ and $\Lambda \geq \left|m\right|(\left|m\right|+1)$ imply that there exists a constant $R_{\rm dec}^* \geq 2R^*_+$ only depending on $a_0$ and $M$ such that
\begin{equation}\label{decrease}
V' < 0\text{ for }r^* \geq R_{\rm dec}^*.
\end{equation}

Define a current given by the following expression:
\[
\text{Q}=\text {\fontencoding{LGR}\selectfont \koppa}^y - E \text{Q}^T.
\]
We require that
\begin{equation}\label{choices8}
\left|y'\right| \leq B,\qquad y' \geq 0\text{ for }r \in [r_+,\infty),\qquad y' > 0\text{ for }r^* \in
[R^*_-,R^*_1],
\end{equation}
\begin{equation}\label{choices9}
\frac{1}{2} \leq y \leq 1\text{ for }r^* \in (-\infty,R^*_+],\qquad y(-\infty) = 1/2,\qquad y = 1\text{ for }r^* \geq R^*_{\rm dec}.
\end{equation}
Such a $y$ is trivial to construct.

We obtain from~(\ref{eq:Q2for}) and~(\ref{eq:Q3for}) the identity
\begin{align}
\label{whatweobtainh}
\nonumber
&\int_{-\infty}^\infty \left (y' |u'|^2+\left (y'(\omega^2-V)-yV'\right) |u|^2 \right )\\&\qquad
\nonumber
-\left (\frac{1}{2}|u'|^2+\left(\frac{1}{2}(\omega-\upomega_+m)^2-E\omega(\omega-\upomega_+m)\right) |u|^2\right)_{r=r_+} - \left (|u'|^2+(1-E)\omega^2 |u|^2\right)_{r=\infty}
\\ &= \int_{-\infty}^\infty \left (-2 y \,{\text{Re}} (u'\overline{H})+E\omega {\text{Im}} (\overline H u) \right).
\end{align}
Next, we observe the bound
\begin{equation}\label{someBoundS}
|V|\le B\left(\epsilon^{-1}\Lambda + \epsilon\omega^2 + \omega_{\rm high}^{-2})\omega^2\right),\qquad |V'|\le \frac{B\Delta}{r^5}
\left(\epsilon^{-1}\Lambda + \epsilon\omega^2 + \omega_{\rm high}^{-2})\omega^2\right)
\end{equation}
where $\epsilon > 0$ is arbitrary.

Now, we fix a sufficiently small $\epsilon > 0$, require that $\epsilon_{\rm width}$ is sufficiently small depending on $\epsilon$, and combine the inequality $\omega^2 > \epsilon_{\rm width}^{-1}\Lambda$ with the inequalities~(\ref{someBoundS}) and~(\ref{decrease}). We conclude the integrand on the left hand side of $(\ref{whatweobtainh})$ is non-negative and bounds
from above the expression
\[
b \int_{R_-^*}^{R_+^*}\left(|u'|^2+(\omega^2+\La) |u|^2\right).
\]
The boundary terms are non-negative due to the boundary conditions~(\ref{eq:b-}) and~(\ref{eq:b+}), the non-superradiance condition and the requirement $E \geq 2$. Requiring that
$R^*_{\infty} > R^*_{\rm dec}$ ensures that~(\ref{someyStuff}) is satisfied.
\end{proof}

\subsection{The $\mathcal{G}_{\mbox{$\lessflat$}}$ range}\label{angularDominated}
As described in Section~\ref{boundaryStillaProblem}, this
 is again a large frequency regime (whose definition still depends on parameters
$\omega_{\rm high}$ and $\epsilon_{\rm width}$ yet to be fixed), but where angular frequencies
will now dominate
time frequencies.
The regime is again manifestly non-superradiant, and, for suitable
parameters, non-trapped, but
as we shall see, we will have to handle the horizon boundary term as in the superradiant
regime.

Once we have made our final choice of the parameters $\omega_{\rm high}$
and $\epsilon_{\rm width}$, then
for $(\omega, m, \Lambda)\in \mathcal{G}_{\mbox{$\lessflat$}}(\omega_{\rm high},
\epsilon_{\rm width})$,
we set the functions $y$, $\tilde y$ and $\hat y$ together with the parameter $r_{\rm trap}$
to be $0$. The remaining functions $f$, $h$ and $\chi_1$ and the desired
coercivity properties are given by
\begin{proposition}\label{odeEst3}
Let $a_0<M$. Then, for all $\omega_{\rm high}$, $R_{\infty}$ and $\epsilon_{\rm width}^{-1}$ sufficiently large, for all $E\ge 2$,
$0\le a\le a_0$, $(\omega, m, \Lambda)\in  \mathcal{G}_{\mbox{$\sharp$}}(\omega_{\rm high},
\epsilon_{\rm width})$,
there exist functions $f$, $h$ and $\chi_1$
satisfying the uniform bounds
\[
\left|f\right| + \Delta^{-1}r^2\left|f'\right| + \left|h\right| + \left|\chi_1\right| + \left|\chi_2\right| \leq B\left(\omega_{\rm high},\epsilon_{\rm width}\right),
\]
\[
f = 1,\ h = 0,\ \chi_1 = 0\text{ for }r^* \geq R^*_{\infty},
\]
such that,
for all
smooth solutions $u$ to the radial o.d.e.~(\ref{e3iswsntouu}) with
right hand side $H$, satisfying moreover the boundary conditions~(\ref{eq:b-}) and~(\ref{eq:b+}), we have the estimate
\begin{align*}
b\int_{R^*_-}^{R^*_+}&\left (|u'|^2+(\omega^2+\La) |u|^2\right)
\\ &\leq \int_{-\infty}^\infty \left (-2f \,{\text{Re}} (u'\overline{H}) - (f'+h)\, {\text{Re}} (u\overline{H})+E\omega {\text{Im}} (H\overline u) + \chi_1\left(\omega-\upomega_+m\right){\text{Im}}(H\overline u)\right ).
\end{align*}

\end{proposition}
\begin{proof}
For the construction of our currents, we again shall exploit that $\mathcal{G}_{\mbox{$\lessflat$}}$ defines a large
frequency regime, where now, however, $\omega^2\ll \Lambda$. Since this is a non-superradiant regime, the boundary term of $r^* = \infty$ may be
controlled with the Q$^T$ current; however, we shall handle the boundary term at the horizon as we did for the regime  $\mathcal{G}^{\mbox{$\sharp$}}$. As we explained in Section~\ref{boundaryStillaProblem} this is necessary because the boundary term at the horizon is proportional to $\left(\omega - \upomega_+m\right)^2\left|u\left(-\infty\right)\right|^2$, and the Q$^T$ current would only give an estimate for $\omega\left(\omega-\upomega_+m\right)\left|u\left(-\infty\right)\right|^2$. In the frequency regime under consideration these are \emph{not} necessarily comparable.

Turning to the proof, we begin by arguing that $\epsilon_{\rm width}$ sufficiently small implies $m\omega \leq 0$. Suppose $m\omega > 0$. Then we have
\begin{equation}\label{mBigOmega}
m\omega \geq \frac{am^2}{2Mr_+} + \alpha\Lambda \geq \alpha\epsilon_{\rm width}^{-1}\omega^2 \Rightarrow \left|m\right| \geq \alpha\epsilon_{\rm width}^{-1}\left|\omega\right|.
\end{equation}
On the other hand,
\begin{equation}\label{omegaBigm}
m\omega \geq \frac{am^2}{2Mr_+} + \alpha\Lambda \Rightarrow \left|\omega\right| \geq \alpha\Lambda\left|m\right|^{-1} \geq \alpha\left|m\right|.
\end{equation}
Combining~(\ref{mBigOmega}) and~(\ref{omegaBigm}) implies
\[
\left|\omega\right| \geq \alpha^2\epsilon_{\rm width}^{-1}\left|\omega\right|.
\]
This is a contradiction if we take $\epsilon_{\rm width} < \alpha^2$. Thus, we indeed
have $m\omega < 0$.

From the above inequality, it follows that Lemma~\ref{lem:3} applies, and we may thus
conclude that
the potential $V_0$ is increasing at $r_+$, and hence has only one critical point at  $r=r^0_{\rm max}$ where it attains a maximum. As in the proof of Proposition~\ref{odeEst1}
concerning the regime
 $\mathcal{G}^{\mbox{$\sharp$}}$, we again infer that, for $\omega_{\rm width}$ sufficiently large, in the regime $\mathcal{G}_{\mbox{$\lessflat$}}\left(\omega_{\rm high},\epsilon_{\rm width}\right)$, the
 potential $V$ has a unique non-degenerate
critical  point at $r_{\max}$, where it attains a maximum, and that $r_{\max}$
is uniformly bounded away from $r_+$ and is uniformly bounded from above. Similarly, we also obtain the existence of
 an interval $(r_{\max}-\de,r_{\max}+\de)$, where $\delta$ is independent of frequency
 parameters, such that $V$ satisfies the two relations
\[
V(r)-\omega^2\ge b\Lambda,\qquad \forall r\in (r_{\rm max}-\delta,r_{\rm max}+\delta)
\]
and
\[
(r-r_{\rm max}) \frac d{dr} V(r)\ge b \Lambda \frac {(r-r_{\rm max})^2}{r^4}, \qquad r>r_+.
\]

We may now follow the construction given in Proposition~\ref{odeEst1} for
the range $\mathcal{G}^{\mbox{$\sharp$}}$.
We define first a current
$\text{Q}=\text{Q}^f+\text {\fontencoding{LGR}\selectfont \Coppa}^h$ with the
same choice of
functions $f$ and $h = A\tilde h$ as for $\mathcal{G}^{\mbox{$\sharp$}}$. This
gives the inequality
\begin{align}\label{someEstimateWithA}
\int_{-\infty}^\infty &\left ((2f'+Ah) |u'|^2+\left (Ah(V-\omega^2)-fV'\right) |u|^2 -\frac 12 (f{'''}+Ah'')|u|^2\right )\\&\nonumber=
\left (|u'|^2+(\omega-\upomega_+m)^2 |u|^2\right)_{r=r_+} + \left (|u'|^2+\omega^2 |u|^2\right)_{r=\infty}
\\  \nonumber &\qquad- \int_{-\infty}^\infty \left (2 f \,{\text{Re}} (u'\overline{H}) + (f'+Ah)\, {\text{Re}} (u\overline{H})\right),
\end{align}
where the integrand on the left hand side
is positive definite.
As in the $\mathcal{G}^{\mbox{$\sharp$}}$ regime, we may gain a large parameter in the region $(r_{\max}-\de,r_{\max}+\de)$ by observing that there exists a small constant $\tilde\delta$ only depending on $a_0$ and $M$ so that, as long as $A \leq  \tilde\delta\epsilon_{\rm width}^{-1}\omega_{\rm high}^2$, the left-hand side of~(\ref{someEstimateWithA}) will give a coercive estimate. We fix such an $A$. Finally, using a $\chi_1\left(\omega-\upomega_+m\right)\text{Q}^K$ current we may handle the boundary term at the horizon, \emph{mutatis mutandis}, as we did for the $\mathcal{G}^{\mbox{$\sharp$}}$ regime. We obtain
\begin{align*}
b\int_{-\infty}^\infty &\left ((2f'+Ah) |u'|^2+\left (Ah(V-\omega^2)-fV'\right) |u|^2 -\frac 12 (f{'''}+Ah'')|u|^2\right )\\&\leq \left (|u'|^2+\omega^2 |u|^2\right)_{r=\infty}
- \int_{-\infty}^\infty \left (2 f \,{\text{Re}} (u'\overline{H}) + (f'+Ah)\, {\text{Re}} (u\overline{H}) + \chi_1\left(\omega-\upomega_+m\right)\text{Im}\left(H\overline u\right)\right).
\end{align*}
Finally, for any $E \geq 2$, the boundary term at infinity is controlled easily with a Q$^T$ current:
\begin{align*}
b\int_{-\infty}^\infty &\left ((2f'+Ah) |u'|^2+\left (Ah(V-\omega^2)-fV'\right) |u|^2 -\frac 12 (f{'''}+Ah'')|u|^2\right )\\&\leq- \int_{-\infty}^\infty \left (2 f \,{\text{Re}} (u'\overline{H}) + (f'+Ah)\, {\text{Re}} (u\overline{H}) + \chi_1\left(\omega-\upomega_+m\right)\text{Im}\left(H\overline u\right) + E\omega\text{Im}\left(H\overline u\right)\right).
\end{align*}
Restricting the domain of integration of the left hand side of our estimate then finishes the proof.
\end{proof}

\subsection{The $\mathcal{G}_{\mbox{$\natural$}}$ range}
\label{thisissparta}
This range is manifestly non-superradiant. By the results of
Section~\ref{Vsrsec},    it will follow that,
after suitable such choices of $\omega_{\rm high}$ and $\epsilon_{\rm width}$,
this will be the only range which can
contain trapping phenomena; thus, it is only in this range
for which we will define
a non-zero parameter $r_{\rm trap}$.

After the final choices of parameters $\omega_{\rm high}$ and $\epsilon_{\rm width}$
have been made, then for
$(\omega, m, \Lambda)\in\mathcal{G}_{\mbox{$\natural$}}(\omega_{\rm high}, \epsilon_{\rm width})$,
we set the functions  $h$, $\tilde y$, $\hat y$ and $\chi_1$ appearing
in the statement of Theorem~\ref{phaseSpaceILED} to be identically $0$.
The remaining functions $f$ and $\hat{y}$, the parameter $r_{\rm trap}$,
and the desired coercivity properties are given by the following:

\begin{proposition}\label{odeEst4}
Let $a_0<M$. Then, for all
$\epsilon_{\rm width} > 0$, for all $\omega_{\rm high}$, $R_{\infty}$ and $E$
sufficiently big depending on $\epsilon_{\rm width}$, and for all
$0\le a\le a_0$, $(\omega, m, \Lambda)\in  \mathcal{G}_{\mbox{$\natural$}}(\omega_{\rm high},
\epsilon_{\rm width})$,
 there exist  functions $f$ and $\hat y$ and a value $r_{\rm trap}$
 satisfying the uniform bounds
\[
r_{\rm trap} =0 \qquad {\rm\ or\ } \qquad 0<b< r_{\rm trap}-r_+<B,
\]
\[
 \left|f\right| + \Delta^{-1}r^2\left|f'\right| + \left|y\right| \leq B\left(\epsilon_{\rm width}\right),\]
\[f = 1,\ \hat y = 0\text{ for }r^* \geq R^*_{\infty},
\]
such that,
 for all
smooth solutions $u$ to the radial o.d.e.~(\ref{e3iswsntouu}) with
right hand side $H$, satisfying moreover the boundary conditions~(\ref{eq:b-}) and~(\ref{eq:b+}),
we have the estimate
 \begin{align}\label{thisisspartaEst}
b\left(\epsilon_{\rm width}\right)\int_{R^*_-}^{R^*_+}&\left (|u'|^2+ \left((\omega^2+\La)\left(1-r^{-1}r_{\rm trap}\right)^2+1\right)|u|^2\right)\\ \nonumber &\le  \int_{-\infty}^\infty \left (-2 f \,{\text{Re}} (u'\overline{H}) -f'{\text{Re}} (u\overline H)+E\omega {\text{Im}} (H\overline u)\right) + \int_{-\infty}^{\infty} 2\hat y {\text{Re}} (u'\overline H).
\end{align}
\end{proposition}

\begin{proof}
As noted above, this frequency range, where $\omega^2$
is comparable to $\Lambda$, contains the trapping phenomena,
but is non-superradiant.
For  frequencies $(\omega, m, \Lambda)\in \mathcal{G}_{\mbox{$\natural$}}$,
 Lemma~\ref{lem:1} implies that the potential $V_0$
may have at most two critical points. Furthermore, either a maximum $r^0_{\rm max}$ exists or $V_0$ is non-increasing on $[r_+,\infty)$; if the maximum exists, then there may also exist a minimum $r^0_{\rm min}$ which will satisfy $r^0_{\rm min} < r^0_{\rm max}$.

In analogy to Lemma~\ref{theOtherLemma} we first must show that for $\omega_{\rm high}$ sufficiently large, the full potential $V$ enjoys similar properties.

\begin{lemma}For $\epsilon_{\rm width}$ as above, for all $\omega_{\rm high}$ sufficiently large depending on $\epsilon_{\rm width}$ and for $(\omega,m,\Lambda) \in \mathcal{G}_{\mbox{$\natural$}}(\omega_{\rm high},\epsilon_{\rm width})$, there exists an $r_3 \in (r_+,\infty]$ depending on the frequency triple but bounded away from $r_+$,
\[r_3 - r_+ \geq b(\epsilon_{\rm width}),\]
such that for $ r\in [r_+,r_3]$
$$
V(r)\le \omega^2-b\left(\epsilon_{\rm width}\right) \La.
$$
Furthermore, in the case when $r_3 < \infty$, then in fact $r_3 \leq B\left(\omega_{\rm high},\epsilon_{\rm width}\right)$, $r^0_{max}$ exists and the potential $V$ has a unique non-degenerate maximum $r_{\max} \in [r_3,\infty)$,
$\left|r_{\rm max}-r_{\rm max}^0\right|\leq B\left(\epsilon_{\rm width}\right)\Lambda^{-1}$ and $\frac {d^2}{dr^2} V(r_{\max})<-b\left(\epsilon_{\rm width}\right) \La$.
\end{lemma}

\begin{proof}
Since $m\omega\not\in(0,\frac {am^2}{2Mr_+}+\alpha \La]$ and $\epsilon_{\rm width}\La \leq \omega^2 \leq \epsilon_{\rm width}^{-1}\La$ , we find
\[\omega^2-V(r_+)=\omega^2-V_0(r_+)\ge c \Lambda,\]
where $c = c\left(\epsilon_{\rm width}\right)$ only depends on the value of $\epsilon_{\rm width}$. We define $r_0 \in (r_+,\infty]$ to be the largest value with the property that
for all $r\in [r_+,r_0)$
$$
V_0(r)\le V_0(r_+)+\frac {c}2\La.
$$
If $r_0$ is finite then we must have a maximum $r_{\rm max}^0$. Furthermore, $\frac{dV_0}{dr}(r_0) \geq 0$; hence, Lemma~\ref{lem:1} implies that if $r_{\rm min}^0$ exists, then
\[r_{\rm min}^0<r_0\le r_{\rm max}^0.\]
Moreover, Lemma~\ref{lem:1} implies that $r_{\rm max}^0$ is bounded
from above by a constant only depending on $\epsilon_{\rm width}$. On the other hand, since
$\left|\frac {d}{dr} V_0(r)\right|\le B\left(\epsilon_{\rm width}\right) \La r^{-3}$, the value $r_0-r_+$ and thus $r_{max}^0-r_+$
is bounded from below by a constant only depending on $\epsilon_{\rm width}$.

We continue to consider the case where $r_0 < \infty$. Recall from the proof of Lemma~\ref{lem:1} that either $\frac{d}{dr}\left((r^2+a^2)^3\frac{dV_0}{dr}\right)$ is negative on $[r_+,\infty)$ or there exists a unique value $r_1 \in [r_+,r_{\rm max}^0)$ such that $\frac{d}{dr}\left((r^2+a^2)^3\frac{dV_0}{dr}\right)$ is positive on $[r_+,r_1)$ and negative on $(r_1,\infty)$. Moreover, since for frequency triples in $\mathcal{G}_{\mbox{$\natural$}}$,
the parameter $\sigma=am\omega/\Lambda$ is bounded by a constant only depending on $\epsilon_{\rm width}$, the value of $r_1$ is uniformly bounded from above by a constant only depending on $\epsilon_{\rm width}$. We first consider the case where the point $r_1$ exists and further split the analysis into two sub-cases based on the value of $V_0(r_1)$.

If $V_0(r_1)\le V_0(r_+)+\frac {3c}4\La$, then, in view of the fact that $V_0$ has a unique maximum at $r^0_{\rm max}$,
we have that
$$
V_0(r)\le V_0(r_+)+\frac {3c}4\La,\qquad \forall r\in [r_+,\max(r_0,r_1)].
$$
Moreover, using that $\left|\frac{dV_0}{dr}\right| \le B\left(\epsilon_{\rm width}\right)\La r^{-3}$ and that $\Lambda^{-1} \frac {d}{dr} \left ((r^2+a^2)^3 \frac d{dr} V_0(r)\right)$ is a quadratic polynomial, with coefficients bounded by $\epsilon_{\rm width}$,
vanishing at the unique point $r_1$ on the interval $[r_+,\infty)$, we can find a small constant $\de = \de\left(\epsilon_{\rm width}\right) >0$ only depending on $\epsilon_{\rm width}$ such that
$$
V_0(r)\le V_0(r_+)+\frac {3c}5\La,\qquad \forall r\in [r_+,\max(r_0,r_1)+\de]
$$
and
$$
\frac {d}{dr} \left ((r^2+a^2)^3 \frac d{dr} V_0(r)\right)<-c_1\left(\epsilon_{\rm width}\right)\La r^2,\qquad\forall  r\in [\max(r_0,r_1)+\de,\infty),
$$
where the positive constant $c_1$ only depends on $\epsilon_{\rm width}$.

Now we consider the case where $V_0(r_1)\ge V_0(r_+)+\frac {3c}4\La$. Then, once again using the bound $\left|\frac{dV_0}{dr}\right| \leq B\left(\epsilon_{\rm width}\right)\La$, we conclude that $r_1-r_0$ is bounded from below by a small positive constant just depending on $\epsilon_{\rm width}$. Furthermore, since $\frac d{dr} V_0(r_0)\ge 0$,
we can find a value
$r_0'\in [r_0,r_1]$ such that
$$
V_0(r)\le  V_0(r_+)+\frac {3c}4\La,\qquad \forall r\in [r_+,r_0']
$$
and
$$
\frac d{dr} V_0(r)\ge c_2 \La,\qquad \forall r\in [r_0',r_1],
$$
where $c_2 = c_2\left(\epsilon_{\rm width}\right)$ is a positive constant which only depends on $\epsilon_{\rm width}$.

Moreover, after slightly changing $c_2$, the last property can be easily extended to a slightly larger interval
$$
\frac d{dr} V_0(r)\ge c_2 \La,\qquad \forall r\in [r_0',r_1+\de],
$$
so that $\delta$ only depends on the constant $\epsilon_{\rm width}$.
$$
\frac {d}{dr} \left ((r^2+a^2)^3 \frac d{dr} V_0(r)\right)<-c_3\La r^2,\qquad\forall  r\in [r_1+\de,\infty),
$$
for a positive constant $c_3 = c_3\left(\epsilon_{\rm width}\right)$ which only depends on $\epsilon_{\rm width}$.

If $r_1$ does not exists, the above arguments \emph{mutatis mutandis} will produce a value $r_0'$ only depending on the value $\epsilon_{\rm width}$ such that
$$
V_0(r)\le  V_0(r_+)+\frac {3c}4\La,\qquad \forall r\in [r_+,r_0'],
$$
$$
\frac {d}{dr} \left ((r^2+a^2)^3 \frac d{dr} V_0(r)\right)<-c_4\left(\epsilon_{\rm width}\right)\La r^2,\qquad\forall  r\in [r_0',\infty),
$$
for a positive constant $c_4 = c_4\left(\epsilon_{\rm width}\right)$ only depending on $\epsilon_{\rm width}$.

Finally, in both cases $r_0 < \infty$ and $r_0 = \infty$ we may therefore claim the existence of a value $r_3$ (possibly infinite), bounded away from $r_+$ by a constant only depending on $\epsilon_{\rm width}$, such that
$$
V_0(r)\le  V_0(r_+)+\frac {3c}4\La\le \omega^2-\frac c4 \La,\qquad \forall r\in [r_+,r_3]
$$
and, such that for any $r\in [r_3,\infty)$, either
$$
\frac d{dr} V_0(r)\ge b\left(\epsilon_{\rm width}\right) \La
$$
or
$$
\frac {d}{dr} \left ((r^2+a^2)^3 \frac d{dr} V_0(r)\right)<-b\left(\epsilon_{\rm width}\right)\La r^2.
$$
We note that if $r_3$ is finite, then it is bounded from above by a constant only depending on $\epsilon_{\rm width}$.

Now, just as we argued in the frequency range $\mathcal{G}^{\mbox{$\sharp$}}$,
adding the bounded potential $V_1$, and requiring that $\omega_{\rm high}$ is sufficiently large finishes the proof.
\end{proof}

Before constructing our current, it will be useful to recall that, as observed in Section~\ref{sub:sharp}, the inequalities $\Lambda \geq \left|m\right|\left(\left|m\right| + 1\right)$ and $\Lambda \geq 2a\left|m\omega\right|$ imply that there exists a constant $R_{\rm dec}^* \geq 2R^*_+$ only depending on $a_0$ and $M$ such that
\begin{equation}\label{decrease2}
V' < 0\text{ for }r^* \geq R_{\rm dec}^*.
\end{equation}

We now construct our current, first under the assumption that $r^*_3 < R^*_{\rm dec}$. Given $E$ sufficiently large depending on $\epsilon_{\rm width}$, we shall use a combination
\[
\text{Q}=\text{Q}^f-\text {\fontencoding{LGR}\selectfont \koppa}^y-E \text{Q}^T
\]
of the currents $\text{Q}^f, \text {\fontencoding{LGR}\selectfont \koppa}^y$ and $Q^T$
where $f$, $y$  are chosen as described below.

The current $\text{Q}^f$ is applied with a function $f$ such that
\begin{equation}\label{choices10}
\left|f\right| + \Delta^{-1}r^2\left|f'\right| \leq B\left(\epsilon_{\rm width}\right),\qquad f(r_+) = 0,\qquad f' > 0\text{ for }r \in [r_3,R_{\infty}],
\end{equation}
\begin{equation}\label{choices11}
f\text{ switches from negative to positive at }r = r_{\max},\qquad f = 1\text{ for }r^* \geq R^*_{\rm dec},
\end{equation}
\begin{equation}\label{choices12}
-fV' - \frac 12 f'''(r) > b(\epsilon_{\rm width})\La \frac {\Delta (r-r_{\max})^2}{r^7},\qquad \forall r\in [r_3,\infty).
\end{equation}
In view of the properties of $V$ proven above, such a function can easily be constructed.

The second current will be
$\text {\fontencoding{LGR}\selectfont \koppa}^{\hat y}$, with
\begin{equation}\label{choices13}
\hat y = 0\text{ for }r \geq r_3,\qquad \hat y' > 0\text{ for }r \leq r_3,\qquad \left|\hat y\right| + \left|\hat y'\right| \leq B\left(\epsilon_{\rm width}\right).
\end{equation}
Such a $\hat{y}$ is now trivial to construct.

Finally, we subtract the multiple $E\text{Q}^T$ of the current $\text{Q}^T$.
We obtain:
\begin{align*}
\int_{-\infty}^{r_3} &\left(\hat y'\left (|u'|^2+(\omega^2-V)|u|^2\right)-\hat yV'|u|^2\right)+
\int_{-\infty}^\infty \left (2f' |u'|^2-(fV'+\frac 12 f''') |u|^2 \right )\\+
&\left(-f|u'|^2+(\frac{1}{2}E-f)\omega^2 |u|^2\right)_{r=\infty}+\left(\frac{1}{2}E\omega(\omega-\upomega_+m)-2\hat y(\omega-\upomega_+m)^2\right) |u|^2|_{r=r_+}
\\ & =- \int_{-\infty}^\infty \left (2 f \,{\text{Re}} (u'\overline{H})+f' {\text{Re}} (u\overline{H})-E\omega {\text{Im}} (H\overline{u})  \right)+ \int_{-\infty}^{r_3}
2\hat y \,{\text{Re}} (u'\overline{H}) .
\end{align*}
By the described properties of the potential $V$, the expression $-(fV'+\frac 12 f''')$ is positive on the interval $[r_3,\infty)$.
On the interval
$(r_+,r_3]$, we need to choose a function $\hat y$ so that in addition to~(\ref{choices13}) we have
\begin{equation}\label{somanychoices}
\hat y'  (\omega^2-V)-\hat y  V'-(fV'+\frac 12 f''')\ge 0.
\end{equation}
Since for these values of $r$
$$
(\omega^2-V)\ge b\left(\epsilon_{\rm width}\right) \La,\qquad |V'|\le B\left(\epsilon_{\rm width}\right)\Lambda \frac {\Delta}{r^2}, \qquad |f|+|f'''|\le B\left(\epsilon_{\rm width}\right)\frac {\Delta}{r^2},
$$
it suffices to fulfill the inequality
\begin{equation}\label{aSplendidIneq}
\frac {d}{dr} \hat y\ge -\hat y C + C,
\end{equation}
provided that $C$ is sufficiently large only depending on $\epsilon_{\rm width}$.
The function
$$
\hat y= 1-e^{C(r_3-r)}
$$
satisfies all the above criteria. Note that the constant $C$ only depends on $\epsilon_{\rm width}$. Finally, for all $E$ such that $C \ll E$, the non-superradiant condition $m\omega\not\in (0,m\upomega_+]$ and the boundary condition
$u'=i\omega u$ at $r=\infty$ ensure that both boundary terms at $r=r_+$ and $r=\infty$ are positive. After restricting the domain of integration of the left hand side of our estimate, we have obtained $(\ref{thisisspartaEst})$, defining
\[
r_{\rm trap} = r_{\rm max}.
\]

In the case $\infty\ge r_3 \geq R^*_{\rm dec}$ we construct our current as follows. As above we will have
\begin{equation}\label{aNiceCurrent}
\text{Q}=\text{Q}^f+\text {\fontencoding{LGR}\selectfont \koppa}^{\hat y}-E \text{Q}^T.
\end{equation}
We define
\[\hat y = 1-e^{\hat C\left(R_{\rm dec} + 2 - r\right)}\text{ for }r \leq R_{\rm dec} + 2,\]
\[\hat y = 0\text{ for }r\geq R_{\rm dec} + 2.\]
Note we shall satisfy~(\ref{aSplendidIneq}) with $C$ replaced by $\hat C$. Thus, arguing just as in the case when $r_3 < R_{\rm dec}$, for a sufficiently large $\hat C$ we will have
\[
\int_{-\infty}^{R_{\rm dec} + 1} \left(\hat y'\left (|u'|^2+(\omega^2-V)|u|^2\right)-\hat yV'|u|^2\right) \geq b\left(\epsilon_{\rm width}\right)\int_{R^*_-}^{R_{\rm dec} + 1}\left(\left|u'\right|^2 + \omega^2\left|u\right|^2\right).
\]
Next, we let $f$ be any smooth function such that
\begin{equation}\label{thechoicesneverend}
f' \geq 0,\qquad f = 0\text{ for }r\in [r_+,R_{\rm dec}], \qquad f = 1\text{ for }[R_{\rm dec}+1,\infty),\qquad \left|f\right| +|f'| + \left|f'''\right| \leq B.
\end{equation}
Such an $f$ is trivial to construct.

Requiring $\omega_{\rm high}$ to be sufficiently large depending on $\epsilon_{\rm width}$, we shall have
\[\int_{R_{\rm dec}}^{\infty} \left(\hat y'\left (|u'|^2+(\omega^2-V)|u|^2\right)-\hat yV'|u|^2\right)+
\int_{R_{\rm dec}}^\infty \left (2f' |u'|^2-(fV'+\frac 12 f''') |u|^2 \right ) \]
\[\geq \int_{R_{\rm dec}^*}^{R_{\rm dec}^*+1}\left(b\omega_{\rm high}^2 - \frac 12f'''\right)\left|u\right|^2 \geq 0.\]
Thus, the bulk term of the estimate corresponding to $\text{Q}$ is positive. Just as in the case $r_3 < R_{\rm dec}$, requiring that $E$ is
large enough depending on $\epsilon_{\rm width}$ will guarantee that the boundary terms are controlled. Finally, we require that $R^*_{\infty} \geq R_{\rm dec}^* + 1$.
This gives again $(\ref{thisisspartaEst})$
defining $r_{\rm trap}=0$.
 \end{proof}

\subsection{The $\mathcal{G}_{\mbox{$\flat$}}$ range} \label{whit}
This range again depends on $\omega_{\rm high}$, and $\epsilon_{\rm width}$.
As opposed to the Propositions concerning the other ranges which restrict the choices of
one or both these
parameters, in the range $\mathcal{G}_{\mbox{$\flat$}}(\omega_{\rm high}, \epsilon_{\rm width})$,
estimates can be obtained for \underline{arbitrary} $\omega_{\rm high}>0$ and
$\epsilon_{\rm width}>0$,
but the relevant constants will degenerate as $\omega_{\rm high}\to\infty$, $\epsilon_{\rm width}\to 0$.

We shall split the frequency range $\mathcal{G}_{\mbox{$\flat$}}$ into four subcases, considering each separately. We will see the above degeneration in the last of the cases.
We note that our decomposition will not however distinguish between
superradiant and non-superradiant frequencies. It should be clear to the reader how the constructions could be simplified if restricted to the non-superradiant case.

The split will rely on the introduction of a further small
parameter $\tilde a_0$. This parameter is for now free--we choose it in Section~\ref{putting}.

\subsubsection{The subrange $\left|\omega\right| \leq \omega_{\rm low}$, $0\le a < \tilde a_0$ and $m \neq 0$}
Given the final choice of parameters,
$\omega_{\rm high}$, $\epsilon_{\rm width}$ and $\omega_{\rm low}$, then
for  $(\omega, m, \Lambda)\in \mathcal{G}_{\mbox{$\flat$}}(\omega_{\rm high}, \epsilon_{\rm width})$
such that $|\omega|\le \omega_{\rm low}$ and $a < \tilde a_0$, we will set  the functions $f$ and $\tilde y$ together
with the parameter $r_{\rm trap}$ to be $0$.
The remaining functions $y$, $\hat y$, $h$, $\chi_1$, $\chi_2$ and the desired coercivity properties
are given by the following

\begin{proposition}\label{odeEst5}
Let $a_0<M$. Then, for all $\omega_{\rm high}>0$, $\epsilon_{\rm width}>0$,
for all $\omega_{\rm low} > 0$, $\tilde{a}_0 > 0$ sufficiently small depending on $\omega_{\rm high}$ and $\epsilon_{\rm width}$, for all $R_{\infty}$ sufficiently large, for all $E \geq 2$,
$0\le a\le a_0$, and for all
$(\omega, m, \Lambda)\in  \mathcal{G}_{\mbox{$\flat$}}(\omega_{\rm high}, \epsilon_{\rm width})$
such that $\left|\omega\right| \leq \omega_{\rm low}$ and $0\le a < \tilde a_0$, there exist
functions $y$, $\hat y$, $\chi_1$, $\chi_2$ and $h$, satisfying the uniform bounds
\[
\left|y\right| + \left|\hat y\right| + \left|h\right| + \left|\chi_2\right| \leq B,
\]
\[
\chi_2 = 1,\ \chi_1 = 0,\\ y = 1,\ \hat y = 0,\ h = 0\text{ for }r^* \geq R^*_{\infty},
\]
such that,
for all
smooth solutions $u$ to the radial o.d.e.~(\ref{e3iswsntouu}) with
right hand side $H$, satisfying moreover the boundary conditions~(\ref{eq:b-}) and~(\ref{eq:b+}),
we have the estimate
\begin{align}
b\int_{R^*_-}^{R^*_+}\left(\left|u'\right|^2 + \left|u\right|^2\right) &\leq
\int_{-\infty}^\infty \left (2(y+\hat y) \,{\text{Re}} (u'\overline{H})+ h\,{\text{Re}}(u\overline H)+E\omega \chi_1{\text{Im}} (H\overline u)+\chi_2\left(\omega-\upomega_+m\right) {\text{Im}} (H\overline u)\right).
\end{align}
\end{proposition}
\begin{proof}The construction of our current is inspired by the treatment of similar frequency regimes in~\cite{aretakisKerr} and~\cite{holz-smul}.

The following three properties are easily verified:
\begin{enumerate}
    \item For every $-\infty < \alpha < \beta < \infty$, if we require $\tilde a$ and $\omega_{\rm low}$ sufficiently small, both depending on $\alpha$ and $\beta$, then $r \in [\alpha,\beta] \Rightarrow V - \omega^2 > 0$.
    \item For sufficiently large $r^*$, independent of the frequency parameters, we have $V' < 0$.
    \item For sufficiently small $\tilde a_0$ and sufficiently negative $r^*$, independent of the frequency parameters, we have $V' > 0$.
\end{enumerate}

Let's introduce the set of relevant constants.
\begin{enumerate}
    \item Requiring that $\tilde a_0$ is sufficiently small, let $R^*_1 < R^*_-$ be a fixed negative constant chosen so that $r^* \leq R^*_1$ implies that $V' > 0$ and $\left(r^*\left(V - V|_{r=r_+}\right)\right)' > 0$.
    \item Let $R^*_2 > R_+^*$ be a fixed positive constant chosen so that $r^* \geq R^*_2$ implies $V' < 0$ and $\left(r^*V\right)' < 0$.
    \item Let $\epsilon > 0$ be a sufficiently small positive constant to be fixed later.
    \item Let $p = p\left(\epsilon\right) > 0$ be a sufficiently small positive constant depending on $\epsilon$.
\end{enumerate}

We now construct our current Q in a step by step fashion. Choose a function $h$ satisfying
\begin{equation}
\label{ccchoices}
h=1 {\rm\ for\ } r^* \in [R^*_1,R^*_2], \qquad h=0 {\rm\ for\ } r^* \in (-\infty,e^{p^{-1}}R^*_1],\qquad h \geq 0,
\end{equation}
\begin{equation}
\label{ccchoices2}
h=0 {\rm\ for\ } r^* \in [e^{p^{-1}}R^*_2,\infty), \quad \left|h''\right| \leq \frac{Bp}{\left|r^*\right|^2} {\rm\ when\ } r^* \in [e^{p^{-1}}R_1^*,R^*_1] \cup [R_2^*,e^{p^{-1}}R_2^*].
\end{equation}
Note that one may easily construct a function $h$ satisfying (\ref{ccchoices}) and (\ref{ccchoices2}).

We then apply a $\text {\fontencoding{LGR}\selectfont \Coppa}^h$ current:
\begin{align}\label{axisym1}
\int_{-\infty}^\infty &\left (h |u'|^2+\left(h(V-\omega^2)-\frac 12 h''\right)|u|^2\right ) = - \int_{-\infty}^\infty h\, {\text{Re}} (u\overline{H}).
\end{align}
The integrand of the left hand side of the estimate~(\ref{axisym1}) will cease to be non-negative for $r^* \in [e^{p^{-1}}R_1^*,R^*_1] \cup [R_2^*,e^{p^{-1}}R_2^*]$. We will produce a non-negative integrand by adding in $\text {\fontencoding{LGR}\selectfont \koppa}^y$ and $\text {\fontencoding{LGR}\selectfont \koppa}^{\hat y}$ currents.

Define a function $y$ by
\begin{equation}
\label{ccchoices3}
y = 0{\rm\ for\ } r^* \in (-\infty,R^*_2-1), \qquad y = \frac{r^*-R_2^*+1}{2}{\rm\ for\ }r^* \in [R^*_2-1,R^*_2),
\end{equation}
\begin{equation}
\label{ccchoices4}
y = \epsilon\left(\frac{1}{r^*V} - \frac{1}{R_2^*V|_{r^* = R^*_2}}\right) + \frac{1}{2}{\rm\ for\ }r^* \in [R^*_2,e^{p^{-1}}R^*_2],
\end{equation}
\begin{equation}\label{cccoices5}
y = \epsilon\left(\frac{1}{e^{p^{-1}}R_2^*V|_{r^* = e^{p^{-1}}R_2^*}} - \frac{1}{R_2^*V|_{r^* = R_2^*}}\right) + \frac{1}{2}{\rm\ for\ }r^* \in [e^{p^{-1}}R^*_2,\infty).
\end{equation}
Note that we have chosen $R^*_2$ so that we will have $y' \geq 0$. Of course, we also have $y \geq 0$.

Now we add in a $\text {\fontencoding{LGR}\selectfont \koppa}^y$ current to~(\ref{axisym1}) and obtain
\begin{align}\label{axisym2}
\int_{-\infty}^\infty &\left (\left(h + y'\right)|u'|^2+\left(y'\omega^2 + h(V-\omega^2) - (yV)' - \frac{1}{2}h''\right)|u|^2\right )=
\\ \nonumber &y(\infty)\left (|u'|^2+\omega^2 |u|^2\right)_{r=\infty}- \int_{-\infty}^\infty \left (2y \,{\text{Re}} (u'\overline{H}) + h\, {\text{Re}} (u\overline{H})\right).
\end{align}

We will now show that if we require $\omega_{\rm low}$ and $\tilde a_0$ to be sufficiently small depending on appropriate choices of $\epsilon$ and $p$, the integrand of the left hand side of~(\ref{axisym2}) is non-negative in the region $r^* \in [R^*_+,\infty)$. Since $h$, $y' \geq 0$ it suffices to show that the term $h(V-\omega^2) - (yV)' - \frac{1}{2}h''$ is non-negative.

The function $y$ vanishes and $h = 1$ in the region $r^* \in [R^*_+,R^*_2-1)$, and thus we have
\[r^* \in [R^*_+,R^*_2-1) \Rightarrow h(V-\omega^2) - (yV)' - \frac{1}{2}h'' = V - \omega^2.\]
If $\omega_{\rm low}$ and $\tilde a_0$ are sufficiently small, then $V - \omega^2$ will be positive in this region.

Next, we have
\[r^* \in [R^*_2-1,R^*_2) \Rightarrow h(V-\omega^2) - (yV)' - \frac{1}{2}h'' = V - \omega^2 - \frac{1}{2}V - yV'.\]
Recall that we chose $R^*_2$ so that $V' < 0$ in this region. Since $y > 0$, we then get
\[ r^* \in [R^*_2-1,R^*_2) \Rightarrow V - \omega^2 - \frac{1}{2}V - yV' \geq \frac{1}{2}V - \omega^2.\]
Now, it is clear that if $\omega_{\rm low}$ and $\tilde a_0$ are sufficiently small, then $\frac{1}{2}V - \omega^2$ will be positive in this region.

Next we consider the region $r^* \in [R^*_2,e^{p^{-1}}R^*_2)$. As usual, we start by noting that if we require $\omega_{\rm low}$ and $\tilde a_0$ to be sufficiently small depending on $p$, then $V - \omega^2 > 0$ in the region $r^* \in [R^*_2,e^{p^{-1}}R^*_2)$. Hence, we will have
\begin{equation}\label{someInequality}
r^* \in [R^*_2,e^{p^{-1}}R^*_2) \Rightarrow h(V - \omega^2) - (yV)' - \frac{1}{2}h'' \geq \frac{\epsilon - Bp}{(r^*)^2} - \left(\frac{1}{2} - \frac{\epsilon}{R^*_2V|_{r^* = R^*_2}}\right)V'.
\end{equation}
Again, we recall that $V' < 0$ for $r^* > R^*_5$. Furthermore, as long as we require $\epsilon$ to be sufficiently small, we will have $\frac{1}{2} - \frac{\epsilon}{R^*_2V|_{r^* = R^*_2}} > 0$. Finally, we may choose $p$ small enough depending on $\epsilon$ so that the first term on the right hand side of~(\ref{someInequality}) is also positive.

In the region $r^* \in [e^{p^{-1}}R^*_2,\infty)$ we have that $h = 0$ and $y$ is constant. Since $V' < 0$ in this region, we have
\[ h(V - \omega^2) - (yV)' - \frac{1}{2}h'' = -y(\infty)V' > 0.\]

Thus as long as $\epsilon > 0$ is sufficiently small, $p$ is sufficiently small depending on $\epsilon$ and $\omega_{\rm low}$ and $\tilde a_0$ are sufficiently small depending on $p$, the integrand of the left hand side of~(\ref{axisym2}) is non-negative for $r^* \geq R^*_2$; however, it is still not non-negative for $r^* < R^*_1$. To remedy this we will employ a $\text {\fontencoding{LGR}\selectfont \koppa}^{\hat y}$ current with a function $\hat y$ whose properties as $r^* \to -\infty$ will mimic the properties of $y$ as $r^* \to \infty$. The key point which allows us to carry out an analogous construction is that $V' > 0$ for $r^*$ sufficiently close to $-\infty$.

We define
\begin{equation}
\label{ccchoices6}
\hat y = 0{\rm\ for\ } r^* \in (R^*_1+1,\infty), \qquad \hat y = \frac{r^*-R_1^*-1}{2}{\rm\ for\ }r^* \in [R^*_1,R^*_1+1),
\end{equation}
\begin{equation}
\label{ccchoices7}
\hat y = \epsilon\left(\frac{1}{r^*\tilde V} - \frac{1}{R_1^*\tilde V|_{r^* = R^*_1}}\right) - \frac{1}{2}{\rm\ for\ }r^* \in [e^{p^{-1}}R^*_1,R^*_1),
\end{equation}
\begin{equation}\label{ccchoices8}
\hat y = \epsilon\left(\frac{1}{e^{p^{-1}}R_1^*\tilde V|_{r^* = e^{p^{-1}}R_1^*}} - \frac{1}{R_1^*\tilde V|_{r^* = R_1^*}}\right) - \frac{1}{2}{\rm\ for\ }r^* \in (-\infty,e^{p^{-1}}R^*_1].
\end{equation}
Here $\tilde V \doteq V - V|_{r=r_+}$. Note that we have chosen $R^*_1$ so that $\hat y' \geq 0$. Of course, we also have $\hat y \leq 0$.

Now we add a $\text {\fontencoding{LGR}\selectfont \koppa}^{\hat y}$ current to~(\ref{axisym2}). We obtain
\begin{align}\label{axisym3}
\int_{-\infty}^\infty &\left (\left(h + y' + \hat y '\right)|u'|^2+\left(y'\omega^2 + \hat y'\left(\omega-\upomega_+m\right)^2 + h(V-\omega^2) - (yV)' - (\hat y \tilde V)'- \frac{1}{2}h''\right)|u|^2\right )=
\\ \nonumber &y(\infty)\left (|u'|^2+\omega^2 |u|^2\right)_{r=\infty} + \left|\hat y(-\infty)\right|\left(|u'|^2 + (\omega-\upomega_+m)^2|u|^2\right)_{r = r_+} - \int_{-\infty}^\infty \left ((y+\hat y) \,{\text{Re}} (u'\overline{H}) + h\, {\text{Re}} (u\overline{H})\right).
\end{align}
Now, keeping in mind that $V' > 0$ for sufficiently negative $r^*$ and repeating the argument, \emph{mutatis mutandis}, which showed that $r^* \geq R^*_2 \Rightarrow h(V-\omega^2) - (yV)' - \frac{1}{2}h'' \geq 0$ we obtain that
\[r^* \leq R^*_1 \Rightarrow h(V-\omega^2) - (\hat y\tilde V)' - \frac{1}{2}h'' \geq 0.\]
We conclude that the integrand of the left hand side of~(\ref{axisym3}) is non-negative and is greater than
\[b\int_{R^*_-}^{R^*_+}\left(|u'|^2 + |u|^2\right).\]
We may now fix the constants $\epsilon$ and $p$.

It remains to absorb the boundary terms on the right hand side of~(\ref{axisym3}). We start with the term at $r = \infty$. Let $\chi_2$ be a function which is identically $1$ for $r^* \geq R_+^*$ and identically $0$ for $r^* \leq R^*_-$. Requiring that $E \geq 2$, we obtain
\begin{align}\label{getThatBoundaryTerm00}
y(\infty)\left(|u'|^2 + \omega^2|u|^2\right)_{r = \infty} &\leq Ey(\infty)\int_{-\infty}^{\infty}\left(\chi_2\text{Q}^T\right)'
\\ \nonumber &\leq B(\omega_{\rm high},\epsilon_{\rm width})\omega_{\rm low}\int_{R^*_-}^{R^*_+}\left(|u'|^2 + |u|^2\right) + Ey(\infty)\omega\int_{-\infty}^{\infty}\chi_2\text{Im}(H\overline u).
\end{align}

Taking $\omega_{\rm low}$ sufficiently small, we may add this in to our previous estimate and obtain
\begin{align}\label{axisym4}
b\int_{R^*_-}^{R^*_+}&\left(|u'|^2 + |u|^2\right) \leq
\\ \nonumber &B\left(|u'|^2 + (\omega-\upomega_+m)^2|u|^2\right)_{r = r_+} - \int_{-\infty}^\infty \left ((y+\hat y) \,{\text{Re}} (u'\overline{H}) + h\, {\text{Re}} (u\overline{H})+ Ey(\infty)\omega\chi_2\text{Im}(H\overline{u})\right).
\end{align}

Let $\chi_1$ be a function which is identically $1$ for $r^* \in (-\infty,R^*_1)$ and identically $0$ for $r^* \geq R^*_2$. We obtain
\begin{align}\label{getThatBoundaryTerm}
B\left(|u'|^2 + (\omega-\upomega_+m)^2|u|^2\right)_{r = r_+} &= B\int_{-\infty}^{\infty}\left(\chi_1\text{Q}^K\right)'
\\ \nonumber &\leq B(\omega_{\rm high},\epsilon_{\rm width})\left(\omega_{\rm low} + \tilde a_0\right)\int_{R^*_-}^{R^*_+}\left(|u'|^2 + |u|^2\right) + B\left(\omega-\upomega_+m\right)\int_{-\infty}^{\infty}\chi_1\text{Im}(H\overline u).
\end{align}
Thus, it is clear that if require that $\omega_{\rm low}$ and $\tilde a_0$ are sufficiently small, depending on $\omega_{\rm high}$ and $\epsilon_{\rm width}$, we may multiply $\chi_1$ by a bounded constant, add in $\chi_1\text{Q}^K$ to our current and obtain
\begin{align}\label{axisym40}
b\int_{R^*_-}^{R^*_+}&\left(|u'|^2 + |u|^2\right) \leq - \int_{-\infty}^\infty \left ((y+\hat y) \,{\text{Re}} (u'\overline{H}) + h\, {\text{Re}} (u\overline{H})+ Ey(\infty)\omega\chi_2\text{Im}(H\overline{u}) + (\omega-\upomega_+m)\chi_1\text{Im}(H\overline u)\right).
\end{align}

Finally, we may rescale all of the multipliers so that $y(\infty) = 1$. We obtain
\begin{align}\label{axisym5}
b\int_{R^*_-}^{R^*_+} &\left (|u'|^2+|u|^2\right) \leq - \int_{-\infty}^\infty \left ((y+\hat y) \,{\text{Re}} (u'\overline{H}) + h\, {\text{Re}} (u\overline{H})+ E\chi_2\text{Im}(H\overline{u})+ (\omega-\upomega_+m)\chi_1\text{Im}(H\overline u)\right).
\end{align}

Of course, $R^*_{\infty}$ is simply required to be larger than $e^{p^{-1}}R^*_2$.
\end{proof}
\begin{remark}The above proof does not use the assumption $m \neq 0$. We only include $m \neq 0$ in the definition of the frequency range so that the set of frequencies covered by Proposition~\ref{odeEst5} is disjoint from the set of frequencies covered by Proposition~\ref{odeEst5b}.
\end{remark}

\subsubsection{The subrange $\left|\omega\right| \leq \omega_{\rm low}$ and $m =0$}

Given the final choice of parameters,
$\omega_{\rm high}$, $\epsilon_{\rm width}$ and $\omega_{\rm low}$, then
for  $(\omega, m, \Lambda)\in \mathcal{G}_{\mbox{$\flat$}}(\omega_{\rm hig}, \epsilon_{\rm width})$
such that $|\omega|\le \omega_{\rm low}$ and $m=0$, we will set the functions $f$, $\tilde y$ and $\chi_1$ together
with the parameter $r_{\rm trap}$ to be $0$.
The remaining functions $y$, $\hat y$, $h$ and the desired coercivity properties
are given by the following

\begin{proposition}\label{odeEst5b}
Let $a_0<M$. Then, for all $\omega_{\rm high}>0$, $\epsilon_{\rm width}>0$,
for all $\omega_{\rm low} > 0$ sufficiently small depending on $\omega_{\rm high}$ and $\epsilon_{\rm width}$, for all $R_{\infty}$ sufficiently large, for all $E \geq 2$, $0\le a\le a_0$, and for all
$(\omega, m, \Lambda)\in  \mathcal{G}_{\mbox{$\flat$}}(\omega_{\rm high}, \epsilon_{\rm width})$
such that $\left|\omega\right| \leq \omega_{\rm low}$ and $m =0$, there exists
functions $y$, $\hat y$ and $h$, satisfying the uniform bounds
\[
\left|y\right| + \left|\hat y\right| + \left|h\right| \leq B,
\]
\[
y = 1,\ h = 0,\ \hat y= 0\text{ for }r^* \geq R^*_{\infty},
\]
such that,
for all
smooth solutions $u$ to the radial o.d.e.~(\ref{e3iswsntouu}) with
right hand side $H$, satisfying moreover the boundary conditions~(\ref{eq:b-}) and~(\ref{eq:b+}),
we have the estimate
\begin{align}
b\int_{R^*_-}^{R^*_+}\left(\left|u'\right|^2 + \left|u\right|^2\right) &\leq-\int_{-\infty}^\infty \left (2(y+\hat y) \,{\text{Re}} (u'\overline{H})+ h\,{\text{Re}}(u\overline H)+E\omega {\text{Im}} (H\overline u)\right).
\end{align}
\end{proposition}
\begin{proof}Observe that the properties of the potential $V$ used in the proof of Proposition~\ref{odeEst5} also hold here:
\begin{enumerate}
    \item For every $-\infty < \alpha < \beta < \infty$, if we require $\omega_{\rm low}$ sufficiently small, both depending on $\alpha$ and $\beta$, then $r \in [\alpha,\beta] \Rightarrow V - \omega^2 > 0$.
    \item For sufficiently large $r^*$, independent of the frequency parameters, we have $V' < 0$.
    \item For sufficiently negative $r^*$, independent of the frequency parameters, we have $V' > 0$.
\end{enumerate}
Using these observations, one may repeat, \emph{mutatis mutandis}, the current construction from the proof of Proposition~\ref{odeEst5}. In fact, the situation is strictly better here; since this proposition concerns a non-superradiant regime, we may set $\chi_1 = 0$. One obtains
\begin{align}\label{axisym6}
b\int_{R^*_-}^{R^*_+} &\left (|u'|^2+|u|^2\right) \leq - \int_{-\infty}^\infty \left (2(y+\hat y) \,{\text{Re}} (u'\overline{H}) + h\, {\text{Re}} (u\overline{H})+ E\text{Im}(H\overline{u})\right).
\end{align}

\end{proof}

\subsubsection{The subrange $\left|\omega\right| \leq \omega_{\rm low}$, $m \neq 0$ and $a \geq \tilde a_0$ (the near stationary subcase)}
\label{nearstat}
Although these frequencies are near-stationary,
we will here be able to effectively exploit the non-vanishing of $a$ and the bound $\left|m\right|\ge 1$.

For $(\omega, m, \Lambda)\in \mathcal{G}_{\mbox{$\flat$}}$, $m\ne0$, and $a\ge \tilde a_0$,
we set the functions $f$ and $\hat y$ together with the parameter $r_{\rm trap}$ to $0$.
The remaining functions $\tilde{y}$, $y$, $h$, $\chi_1$ and $\chi_2$ and
the desired coercivity properties are given by the following:

\begin{proposition}\label{odeEst6}Let $a_0<M$. Then, for all $\omega_{\rm high}>0$, $\epsilon_{\rm width}>0$, $E\geq 2$, for all $\omega_{\rm low} > 0$ sufficiently small depending on $\tilde a_0$ and $E$, for all $R_{\infty}$ sufficiently large depending on $\tilde a_0$, $0\le a\le a_0$, and for all
$(\omega, m, \Lambda)\in  \mathcal{G}_{\mbox{$\flat$}}(\omega_{\rm high}, \epsilon_{\rm width})$
such that $\left|\omega\right| \leq \omega_{\rm low}$, $m\neq 0$ and $a \geq \tilde a_0$, there exist
functions $\tilde y$, $y$, $h$, $\chi_1$ and $\chi_2$, satisfying the uniform bounds
\[\left|\tilde y\right| + \left|y\right| + \left|h\right| + \left|\chi_1\right| + \left|\chi_2\right| \leq B\left(\tilde a_0\right),\]
\[\left|\tilde y\right| \leq B\exp\left(-br\right),\ y = 1,\ h = 0\text{ for }r^* \geq R^*_{\infty},\]
such that,
for all
smooth solutions $u$ to the radial o.d.e.~(\ref{e3iswsntouu}) with
right hand side $H$, satisfying moreover the boundary conditions~(\ref{eq:b-}) and~(\ref{eq:b+}),
we have the estimate
\begin{align}\label{someEstimate}
b\left(\tilde a_0\right)&\int_{R_-^*}^{R_+^*} \left (|u'|^2+|u|^2 \right ) \\ &\nonumber\le
\int_{-\infty}^\infty \left (-2\tilde y \,{\text{Re}} (u'\overline{H}) - E\chi_2\omega\,{\text{Im}}(H \overline u) - 2\chi_1\left(\omega-\upomega_+m\right)\,{\text{Im}}( H \overline u)- h\,{\text{Re}}(u\overline H)- 2y\,{\text{Re}}(u'\overline H)\right).
\end{align}
\end{proposition}
\begin{proof}
Our current will be of the form
\[\text{Q} = \text {\fontencoding{LGR}\selectfont \koppa}^{\tilde y} + \text {\fontencoding{LGR}\selectfont \Coppa}^{h} + \text {\fontencoding{LGR}\selectfont \koppa}^{y} - \chi_1\text{Q}^K - E\chi_2\text{Q}^T,\]
for suitable functions $\tilde{y}$, $h$, $y$, $\chi_1$ and $\chi_2$.

As we did for the frequency range $\mathcal{G}^{\mbox{$\sharp$}}$, for the purposes of exposition we shall construct the current in a step by step fashion.
The first important observation is that the assumptions $m \neq 0$ and $a \geq \tilde a_0$ imply that $\omega_0^2 := \left(\omega - \upomega_+m\right)^2 \geq b\left(\tilde a_0\right)$ as long as $\omega_{\rm low} \ll \tilde a_0$. The second important observation is that $V = \frac{\Lambda}{r^2} + O(r^{-3})\text{ as }r\to\infty$ and, since $m \neq 0$, $\Lambda \geq 2$. This implies that for any $1 \ll \alpha \ll \beta < \infty$, then requiring that
$\omega_{\rm low}$ is small enough, depending on $\alpha$ and $\beta$, we have
\begin{equation}\label{potPos}
r \in [\alpha,\beta]\Rightarrow V - \omega^2 \geq \frac{b}{r^2}
\end{equation}
in this frequency range.
We shall exploit this positivity via the use of a $\text {\fontencoding{LGR}\selectfont \Coppa}^h$ current.

Let us now introduce the set of relevant constants. Let $p > 0$, $R_1 < R_2 < R_3 < e^{p^{-1}}R_3$, $C > 0$ and $c > 0$ be constants such that
\begin{enumerate}
    \item $C = C(\tilde a_0)$ is sufficiently large.
    \item $c$ is sufficiently small.
    \item $R_1$ is sufficiently large.
    \item $c\omega_0^{-2}R_1\exp\left(BCR_1\right) \ll R_2$.
    \item $R_2 \ll R_3$.
    \item $p \ll R_3^{-3}$.
    \item $\omega_{\rm low}^2 \ll cE^{-1}\omega_0^{2}\exp\left(-p^{-1}\right)\exp\left(-BCR_2\right)R_3^{-4}$ and, requiring $\omega_{\rm low}$ sufficiently small,  $r \in [R_1,e^{p^{-1}}R_3]\Rightarrow V - \omega^2 \geq br^{-2}$.
\end{enumerate}

We write
\[\omega^2 - V =: \omega_0^2 - \tilde{V}\]
where $\tilde{V}(r_+) = 0$. Let $\upsilon(r)$ be a positive function such that
\begin{equation}\label{evenmorechoices}
\upsilon = \Delta\text{ near }r_+,\qquad \upsilon = 1\text{ when }r^* \geq R^*_{\infty},\qquad \left|\upsilon\right| \leq B.
\end{equation}
Then we define
\[\tilde y(r^*) := -\exp\left(-C\int_{-\infty}^{r^*}\upsilon dr^*\right),\]
and consider the current $\text {\fontencoding{LGR}\selectfont \koppa}^{\tilde y}$.  Note that $\tilde y\left(-\infty\right) = -1$ and $\tilde y\left(\infty\right) = 0$. We obtain
\begin{align}\label{smallOmegaFirst}
\int_{-\infty}^\infty &\left (\tilde y' |u'|^2+\left (\tilde y'\omega_0^2 - \left(\tilde y\tilde{V}\right)'\right) |u|^2 \right )=
\left (|u'|^2+(\omega-\upomega_+m)^2 |u|^2\right)_{r=r_+} - \int_{-\infty}^\infty \left (2\tilde y \,{\text{Re}} (u'\overline{H})\right).
\end{align}
\begin{remark}
Note that unlike every other microlocal current we have considered, this $\tilde y$ cannot be taken independent of the frequency parameters when $r^* \geq R^*_{\infty}$ (since this is the only regime which employs a $\text {\fontencoding{LGR}\selectfont \koppa}^{\tilde y}$ current where the seed function $\tilde y$ is negative for large $r$). Nevertheless, the exponential decay of $\tilde y$ as $r \to \infty$ will allow us to handle this when we re-sum (see Section~\ref{summation}).
\end{remark}

We now turn to the $\left(\tilde y\tilde{V}\right)'|u|^2$ term on the left hand side of~(\ref{smallOmegaFirst}) which threatens to destroy our estimate:
\begin{align*}
\left|\int_{-\infty}^{\infty}\left(\tilde y\tilde{V}\right)'\left|u\right|^2\right| = 2\left|\int_{-\infty}^{\infty}2\tilde y\tilde{V}\text{Re}\left(u'\overline u\right)\right| \leq \epsilon \int_{-\infty}^{\infty}\tilde y'\left|u'\right|^2 + B\epsilon^{-1}\int_{-\infty}^{\infty}\tilde y'\omega_0^2\frac{\tilde y^2\tilde{V}^2}{\left(\tilde y'\right)^2\omega_0^2}\left|u\right|^2 =
\end{align*}
\begin{align*}
\epsilon \int_{-\infty}^{\infty}\tilde y'\left|u'\right|^2 + B\epsilon^{-1}\int_{-\infty}^{\infty}\tilde y'\omega_0^2\frac{\tilde{V}^2}{C^2\upsilon^2\omega_0^2}\left|u\right|^2 \leq \epsilon \int_{-\infty}^{\infty}\tilde y'\left|u'\right|^2 + B\epsilon^{-1}C^{-2}\omega_0^{-2}\int_{-\infty}^{\infty}\tilde y'\omega_0^2\left|u\right|^2.
\end{align*}
Hence, taking $\epsilon$ sufficiently small and then $C = C(\tilde a_0)$ sufficiently large gives us the estimate
\begin{align}\label{smallOmegaSecond}
b\int_{-\infty}^\infty &\left(\tilde y' |u'|^2+\tilde y'\omega_0^2 |u|^2 \right)\leq
\left (|u'|^2+(\omega-\upomega_+m)^2 |u|^2\right)_{r=r_+} - \int_{-\infty}^\infty  2\tilde y \,{\text{Re}} (u'\overline{H}).
\end{align}

As in the frequency regime $\mathcal{G}^{\mbox{$\sharp$}}$, we need to find a large parameter in order to handle the boundary term $\left (|u'|^2+(\omega-\upomega_+m)^2 |u|^2\right)_{r=r_+}$. We employ a $\text {\fontencoding{LGR}\selectfont \Coppa}^h$ current where
\begin{equation}\label{evenmorechoices2}
h = 0\text{ for }r \in [r_+,R_1],\qquad h'' = c\tilde y'\omega_0^2\text{ for }r \in [R_1,R_2],\qquad h'' = 0\text{ for }r \in [R_2,R_3],
\end{equation}
\begin{equation}\label{evenmorechoices3}
\left|h'\right| \leq \frac{BR_3p}{r}\text{ for }r \in [R_3,e^{p^{-1}}R_3],\qquad \left|h''\right| \leq \frac{BR_3p}{r^2}\text{ for }r\in[R_3,e^{p^{-1}}R_3],
\end{equation}
\begin{equation}\label{evenmorechoices4}
 h = 0\text{ for }r \in [e^{p^{-1}}R_3,\infty).
\end{equation}
Note that one may easily construct an $h$ satisfying~(\ref{evenmorechoices2}),~(\ref{evenmorechoices3}) and~(\ref{evenmorechoices4}).

In order to help orient the reader for the estimates below, let us briefly describe the rationale behind the construction of $h$. First of all, the $\text {\fontencoding{LGR}\selectfont \Coppa}^h$ current gives a good estimate when $h$ is positive, $h\left(V-\omega^2\right)$ is positive and if the error terms from the $-\frac{1}{2}h''$ term can be controlled. Since $V-\omega^2$ is only positive for large enough $r$, we set $h$ to be $0$ for $r \leq R_1$. In order for $h$ to become non-zero, it is necessary for $-\frac{1}{2}h'' < 0$. Thus, the definition of $h$ on $[R_1,R_2]$ is motivated by the desire to increase $h$ as fast as possible while still being able to absorb the error term $-\frac{1}{2}h''$ with the estimate~(\ref{smallOmegaSecond}). This successfully produces a positive $h$, but we still need to find a large parameter. In the region $[R_2,R_3]$ we achieve this by setting $h'' = 0$, and then letting $h$ grow linearly. Since we have taken $\omega^2$ small enough so that $V-\omega^2$ is positive on $[R_2,R_3]$, by taking $R_3$ very large we can arrange for $h$ to be as large as we wish. The crucial estimates for absorbtion of the boundary term $\left|u(-\infty)\right|^2$ are
\[r \in [R_2,R_3] \Rightarrow bc\omega_0^2\exp\left(-BCR_1\right)\left(r - R_1\right) \leq h \leq BR_3,\]
\[r \in [R_2,R_3] \Rightarrow h^{-1} \leq B\left(V - \omega^2\right)h,\]
see the estimates~(\ref{smallOmegaHorizon}),~(\ref{smallOmegaHorizonError}) and~(\ref{smallOmegaSixth}). Now that we have succeeded in finding a large parameter to absorb the boundary term, we need to take $h$ back down to $0$. Keeping in mind that $pR_3 \ll 1$, the choice of $h$ on $[R_3,e^{p^{-1}}R_3]$ is motivated by the desire to take $h$ down to $0$ in a such a way that the error term $-\frac{1}{2}h''$ is as small as possible. See estimates~(\ref{smallOmegaIntermediate}),~(\ref{yProperty}) and~(\ref{smallOmegaFifth}) for the details of how these error terms are dealt with.

We now turn to the specifics. Applying the current $\text {\fontencoding{LGR}\selectfont \Coppa}^h$ gives
\begin{align}\label{smallOmegaThird}
\int_{-\infty}^\infty &\left (\left(b\tilde y'+h\right) |u'|^2+\left(b\tilde y'\omega_0^2 + h\left(V-\omega^2\right) - \frac{1}{2}h''\right)|u|^2 \right ) \\ &\nonumber \leq
\left (|u'|^2+(\omega-\upomega_+m)^2 |u|^2\right)_{r=r_+} - \int_{-\infty}^\infty \left (2\tilde y \,{\text{Re}} (u'\overline{H}) + h\,{\text{Re}}(u\overline H)\right).
\end{align}
Recall that we explicitly required that $\omega_{\rm low}$ be sufficiently small so that in particular $V-\omega^2$ is positive on $[R_1,e^{p^{-1}}R_3]$. Given this, the only negative terms on the left hand side of this estimate come from the $-\frac{1}{2}h''$ term on the intervals $[R_1,R_2]$ and $[R_3,e^{p^{-1}}R_3]$. By construction of $h$ and the assumption that $c\ll 1$, the negative terms on $[R_1,R_2]$ can be controlled by the $by'\omega_0^2\left|u\right|^2$ term. Therefore, we have
\begin{align}\label{smallOmegaFourth}
b\int_{-\infty}^\infty &\left (\left(\tilde y'+h\right) |u'|^2+\left(\tilde y'\omega_0^2 + h\left(V-\omega^2\right)\right)|u|^2 \right ) \\ &\nonumber\le
\left (|u'|^2+(\omega-\upomega_+m)^2 |u|^2\right)_{r=r_+} + BpR_3\int_{R_3}^{e^{p^{-1}}R_3}\left|u\right|^2r^{-2} -\int_{-\infty}^\infty \left (2\tilde y \,{\text{Re}} (u'\overline{H}) + h\,{\text{Re}}(u\overline H)\right).
\end{align}

The left hand side now is sufficiently strong to absorb the boundary term on the right hand side in a similar fashion as in the $\mathcal{G}^{\mbox{$\sharp$}}$ regime, i.e.~by
an application of the current $\chi Q^K$ for a suitable cutoff $\chi$. However, we still need to address the term $BpR_3\int_{R_3}^{e^{p^{-1}}R_3}\left|u\right|^2r^{-2}$. For this we use a $\text {\fontencoding{LGR}\selectfont \koppa}^y$ current with a function $y$ which is determined by
\begin{equation}\label{evenmorechoices5}
y = 0\text{ for }r \in [r_+,R_2],\qquad y' = \frac{h}{2}\text{ for }r \in [R_2,R_3],
\end{equation}
\begin{equation}\label{evenmorechoices6}
y = \frac{r-R_3}{R_3^2} + \frac{1}{2}\int_{R_2}^{R_3}h\text{ for }r \in [R_3,e^{p^{-1}}R_3],
\end{equation}
\begin{equation}\label{evenmorechoices7}
y = \frac{e^{p^{-1}}R_3 - R_3}{R_3^2} + \frac{1}{2}\int_{R_2}^{R_3}h\text{ for }r \in [e^{p^{-1}}R_3,\infty).
\end{equation}

We obtain
\begin{align}\label{smallOmegaIntermediate}
\int_{R_2}^{\infty}&\left(y'\left|u'\right|^2 + \left(y'\omega^2 - \left(yV\right)'\right)\left|u\right|^2\right)  \\ \nonumber &=
\left(\frac{e^{p^{-1}}R_3 - R_3}{R_3^2} + \frac{1}{2}\int_{R_2}^{R_3}h\right)\left (|u'|^2+\omega^2 |u|^2\right)_{r=\infty} - \int_{R_2}^{\infty}2y\,{\text{Re}}(u'\overline H).
\end{align}
Observe that $r \in [R_3,e^{p^{-1}}R_3]$ implies
\begin{align}\label{yProperty}
-\frac{d}{dr}\left(yV\right) &= -\left(R_3^{-2}\left(\frac{\Lambda}{r^2} + O\left(r^{-3}\right)\right) + \left(\frac{r-R_3}{R_3^2}+\frac{1}{2}\int_{R_2}^{R_3}h\right)\left(\frac{-2\Lambda}{r^3} + O\left(r^{-4}\right)\right)\right)\\
\nonumber & \geq R_3^{-2}\frac{\Lambda}{r^2} + R_3^{-1}O\left(r^{-3}\right) + \left(\frac{1}{2}\int_{R_2}^{R_3}h\right)\frac{2\Lambda}{r^3} \geq bR_3^{-2}r^{-2}.
\end{align}
We have used that
\[\int_{R_2}^{R_3}h \geq bc\omega_0^2\exp\left(-BCR_1\right)\left(R_3^2 - R_2^2 - R_1\right),\]
and that $R_3$ has been chosen to dominate $R_2 + c\omega_0^{-2}\exp\left(BCR_1\right)$. Of course, $\frac{d}{dr^*}\left(yV\right) = \left(1 + O(r^{-1})\right)\frac{d}{dr}\left(yV\right)$. We conclude that
\[-\left(yV\right)' \geq bR_3^{-2}r^{-2}.\]

Next, keeping in mind that $p \ll R_3^{-3}$ and that $R_1$ sufficiently large implies that $V' < 0$ for $r \in [R_1,\infty)$, we may add~(\ref{smallOmegaIntermediate}) to~(\ref{smallOmegaFourth}) to obtain
\begin{align}\label{smallOmegaFifth}
b\int_{-\infty}^\infty &\left (\left(\tilde y'+h\right) |u'|^2+\left(b\tilde y'\omega_0^2 + h\left(V-\omega^2\right)\right)|u|^2 \right )\\ &\nonumber\leq
\left (|u'|^2+(\omega-\upomega_+m)^2 |u|^2\right)_{r=r_+} + \left(\frac{e^{p^{-1}}R_3 - R_3}{R_3^2} + \int_{R_2}^{R_3}h\right)\left (|u'|^2+\omega^2 |u|^2\right)_{r=\infty}
\\ &\nonumber\qquad - \int_{-\infty}^\infty \left (2\tilde y \,{\text{Re}} (u'\overline{H}) + h\,{\text{Re}}(u\overline H)+ 2y\,{\text{Re}}(u'\overline H)\right).
\end{align}

Lastly, it remains to absorb the boundary terms on the right hand side. We start with the horizon term. Let $\chi_1$ be a smooth function such that
\begin{equation}\label{evenmorechoices8}
\chi_1 = 1\text{ for }r \in [r_+,R_2],\qquad \chi_1 = 0\text{ for }r \in [R_3,\infty),
\end{equation}
\begin{equation}\label{evenmorechoices9}
\left|\chi_1'\right| \leq B\left(R_3-R_2\right)^{-1},\qquad \left|\chi_1\right| \leq B.
\end{equation}
We have
\begin{align}\label{smallOmegaHorizon}
\left (|u'|^2+(\omega-\upomega_+m)^2 |u|^2\right)_{r=r_+} &= 2\int_{-\infty}^{\infty}\left(\chi_1\text{Q}^K\right)'  \\ \nonumber &= \left(\omega-\upomega_+m\right)\int_{R_2}^{R_3}\chi_1'\text{Im}\left(u'\overline u\right) + \left(\omega-\upomega_+m\right)\int_{-\infty}^{\infty}\chi_1\text{Im}\left(H\overline u\right).
\end{align}

Now, recall that $r \in [R_2,R_3]$ implies that $h \geq bc\omega_0^2\exp\left(-BCR_1\right)\left(r-R_1\right)$ which in turn implies $h\left(V-\omega^2\right) \geq bc\omega_0^2\exp\left(-BCR_1\right)r^{-1} - BR_1R_2^{-1}r^{-1} - \omega_{\rm low}^2R_3 \geq bc\omega_0^2\exp\left(-BCR_1\right)r^{-1}$. Thus, keeping in mind that $R_1R_2^{-1} \ll h\left(V-\omega^2\right)$ on $[R_2,R_3]$, we conclude that
\[r \in [R_2,R_3]\Rightarrow h^{-1} \leq Bh\left(V-\omega^2\right).\]
Thus,
\begin{align}\label{smallOmegaHorizonError}
&\left|\int_{R_2}^{R_3}\chi_1'\text{Im}\left(u'\overline u\right)\right| \leq B\left(R_3-R_2\right)^{-1}\int_{R_2}^{R_3}\left|u'\right|\left|u\right| \\ \nonumber &\qquad\leq B\left(R_3-R_2\right)^{-1}\int_{R_2}^{R_3}\left(h\left|u'\right|^2 + h^{-1}\left|u\right|^2\right) \leq  B\left(R_3-R_2\right)^{-1}\int_{R_2}^{R_3}\left(h\left|u'\right|^2 + h\left(V-\omega^2\right)\left|u\right|^2\right).
\end{align}

Hence, we may combine~(\ref{smallOmegaHorizon}),~(\ref{smallOmegaFifth}) and~(\ref{smallOmegaHorizonError}) to obtain
\begin{align}\label{smallOmegaSixth}
b\int_{-\infty}^\infty &\left (\left(\tilde y'+h\right) |u'|^2+\left(\tilde y'\omega_0^2 + h\left(V-\omega^2\right)\right)|u|^2 \right) \\ &\nonumber
\le \left(\frac{e^{p^{-1}}R_3 - R_3}{R_3^2} + \int_{R_2}^{R_3}h\right)\left (|u'|^2+\omega^2 |u|^2\right)_{r=\infty}
\\ &\qquad -\nonumber \int_{-\infty}^\infty \left (2\tilde y \,{\text{Re}} (u'\overline{H}) + 2\chi_1\left(\omega-\upomega_+m\right)\,{\text{Im}}(H\overline u)+ h\,{\text{Re}}(u\overline H)+ 2y\,{\text{Re}}(u'\overline H)\right).
\end{align}

Now we shall handle the boundary term at $\infty$. Let $\chi_2$ be a smooth function such that
\begin{equation}\label{evenmorechoices10}
\chi_2 = 1\text{ for }r \in [R_2,\infty),\qquad \chi_2 = 0 \text{ for }r \in [r_+,R_1],\qquad \left|\chi_2\right| \leq B.
\end{equation}
Then, we have
\begin{align}\label{smallOmegaInfinity}
&\frac{E}{2}\left(\frac{e^{p^{-1}}R_3 - R_3}{R_3^2} + \int_{R_2}^{R_3}h\right)\left [|u'|^2+\omega^2 |u|^2\right]_{r=\infty}= E\left(\frac{e^{p^{-1}}R_3 - R_3}{R_3^2} + \int_{R_2}^{R_3}h\right)\int_{-\infty}^{\infty}\left(\chi_2\text{Q}^T\right)' \\ \nonumber &= E\left(\frac{e^{p^{-1}}R_3 - R_3}{R_3^2} + \int_{R_2}^{R_3}h\right)\omega\int_{R_1}^{R_2}\chi_2'\text{Im}\left(u'\overline u\right)+ E\left(\frac{e^{p^{-1}}R_3 - R_3}{R_3^2} + \int_{R_2}^{R_3}h\right)\omega\int_{-\infty}^{\infty}\chi_2\text{Im}\left(H\overline u\right).
\end{align}

We have
\begin{align}\label{smallOmegaHorizonError2}
&E\left(\frac{e^{p^{-1}}R_3 - R_3}{R_3^2} + \int_{R_2}^{R_3}h\right)\left|\omega\int_{R_1}^{R_2}\chi_2'\text{Im}\left(u'\overline u\right)\right|  \\ \nonumber &\qquad\le BE\omega_0^{-2}\exp\left(BCR_2\right)\left(\frac{e^{p^{-1}}R_3 - R_3}{R_3^2} + \int_{R_2}^{R_3}h\right)\omega_{\rm low}\int_{R_1}^{R_2}\left(\tilde y'\left|u'\right|^2 + \omega_0^2\tilde y'\left|u\right|^2\right).
\end{align}

Thus, using that $E \geq 2$, we may combine~(\ref{smallOmegaInfinity}) and~(\ref{smallOmegaSixth}) and obtain
\begin{align}\label{smallOmegaSeventh}
b&\int_{-\infty}^\infty \left (\left(\tilde y'+h\right) |u'|^2+\left(\tilde y'\omega_0^2 + h\left(V-\omega^2\right)\right)|u|^2 \right ) \\ &\nonumber\le
\int_{-\infty}^\infty \left (-2\tilde y \,{\text{Re}} (u'\overline{H}) - Ey(\infty)\chi_2\omega\,{\text{Im}}(H \overline u) - 2\chi_1\left(\omega-\upomega_+m\right)\,{\text{Im}}(H\overline u)- h\,{\text{Re}}(u\overline H)- 2y\,{\text{Re}}(u'\overline H)\right).
\end{align}
At this point, it is clear that we may rescale the functions $\tilde{y}$, $h$  and $y$ by an $\tilde a_0$ dependent constant so that $y$ is identically $1$ for $r^* \geq R^*_{\infty}$.
\end{proof}

\subsubsection{The subrange $|\omega|\ge \omega_{\rm low}$ (the non-stationary subcase)}\label{nonStat}
We turn finally to our last frequency range.
It is only this range which gives rise to the term
$1_{\{\omega_{\rm low} \leq \left|\omega\right| \leq \omega_{\rm high}\}\cap \{\Lambda \leq \epsilon_{\rm width}^{-1}\omega_{\rm high}^2\}}\left|u(-\infty)\right|^2$ on the right hand side of $(\ref{fromPhaseSpace2})$ in the statement of Theorem~\ref{phaseSpaceILED}.

Let $(\omega, m,\Lambda)\in \mathcal{G}_{\mbox{$\flat$}}$ where $|\omega|\ge \omega_{\rm low}$.
When the
final choices of $\omega_{\rm high}$ $\epsilon_{\rm width}$, $\omega_{\rm low}$ have been made,
we will set the functions $f$, $h$, $\tilde y$, $\hat y$ and $\chi_1$ together with the trapping
parameter $r_{\rm trap}$ to be $0$.
The remaining function $y$ and desired coercivity property is given by

\begin{proposition}\label{odeEst7}Let $a_0<M$. Then, for all $\omega_{\rm high}>0$, $\epsilon_{\rm width}>0$, $\omega_{\rm low} > 0$, $E\geq 2$, for all $R_{\infty}$ sufficiently large depending on $\omega_{\rm high}$, $\omega_{\rm low}$, $0\le a\le a_0$, and for all
$(\omega, m, \Lambda)\in  \mathcal{G}_{\mbox{$\flat$}}(\omega_{\rm high}, \epsilon_{\rm width})$
such that $|\omega|\ge \omega_{\rm low}$, there exists a function $y$ satisfying the uniform bounds
\[
\left|y\right| \leq B,
\]
\[
y = 1\text{ for }r^* \geq R^*_{\infty},
\]
such that,
for all
smooth solutions $u$ to the radial o.d.e.~(\ref{e3iswsntouu}) with
right hand side $H$, satisfying moreover the boundary conditions~(\ref{eq:b-}) and~(\ref{eq:b+}), we have
\begin{align*}
b\left(\omega_{\rm low}, \omega_{\rm high}, \epsilon_{\rm width}\right)\int_{R_-^*}^{R_+^*} &\left(|u'|^2+|u|^2 \right)
\\ &\le B\left (\left|\omega\left(\omega-\upomega_+m\right)\right||u|^2\right)_{r=r_+}- \int_{-\infty}^\infty \left (2 y \,{\text{Re}} (u'\overline{H})-E\omega {\text{Im}} (\overline H u) \right).
\end{align*}
\end{proposition}
\begin{proof}

As in the previous section we will treat the superradiant and non-superradiant frequencies concurrently. However, as previously discussed, it is only
for the sake of the superradiant frequencies for which we include the first
term on the right hand side of the estimates of the proposition.

Our current will be of the form
\[\text{Q} = \text {\fontencoding{LGR}\selectfont \koppa}^y - E\text{Q}^T\]
for
\[y(r^*) := \exp\left(-C\int_{r^*}^{\infty}\chi_{R_{\infty}^*}r^{-2}dr\right),\]
where $C = C(\omega_{\rm low},\omega_{\rm high},\epsilon_{\rm width})$ is a sufficiently large constant. The function $\chi_{R_{\infty}^*}$ is a smooth function which is identically $1$ on $[r_+,R_{\infty}-1)$ and identically $0$ on $[R_{\infty},\infty)$. Note that $y|_{r^* \geq R^*_{\infty}} = 1$ and $y\left(-\infty\right) = 0$. Applying the current $\text{Q}$ gives
\begin{align}\label{nonStatFirst}
\int_{-\infty}^{\infty} &\left (y' |u'|^2+\left (y'\omega^2-\left(yV\right)'\right) |u|^2 \right) - \left (|u'|^2+(1-E)\omega^2 |u|^2\right)_{r=\infty}
\\ &=\left(\omega(\omega-\upomega_+m) |u|^2\right)_{r=r_+}- \int_{-\infty}^\infty \left (2 y \,{\text{Re}} (u'\overline{H})+E\omega {\text{Im}} (u\overline H) \right).
\end{align}

Let $R_{\infty}$ be sufficiently large and $R_6$ be chosen such that $1 \ll R_6 \ll R_{\infty}-1$. Then, let $\chi_2$ be a smooth function such that
\begin{equation}\label{evenevenmorechoices}
\chi_2 = 1\text{ for }r\in [r_+,R_6],\qquad \chi_2 = 0\text{ for }R \in [R_{\infty}-1,\infty),
\end{equation}
\begin{equation}\label{evenevenmorechoices2}
\left|\chi_2'\right| \leq B\left(R_{\infty}-R_6\right)^{-1},\qquad \left|\chi_2\right| \leq B.
\end{equation}
Then set $V_{\leq} := \chi_2V$ and $V_{\geq} = \left(1-\chi_2\right)V$. Of course, we have $V = V_{\leq} + V_{\geq}$. For $V_{\leq}$ we have
\begin{align}\label{absorbV1}
\left|\int_{-\infty}^{\infty}\left(yV_{\leq}\right)'\left|u\right|^2\right| &= 2\left|\int_{-\infty}^{\infty}yV_{\leq}\text{Re}\left(u'\overline u\right)\right| \\
\nonumber
&\leq \epsilon \int_{-\infty}^{\infty}y'\left|u'\right|^2 + B\epsilon^{-1}\int_{-\infty}^{\infty}y'\omega^2\frac{y^2V_{\leq}^2}{\left(y'\right)^2\omega^2}\left|u\right|^2\\ \nonumber &=\epsilon \int_{-\infty}^{\infty}y'\left|u'\right|^2 + B\left(\omega_{\rm high},\epsilon_{\rm width}\right)\epsilon^{-1}\int_{-\infty}^{\infty}y'\omega^2\frac{V_{\leq}^2}{C^2r^{-4}\omega^2}\left|u\right|^2\\
&\nonumber \leq \epsilon \int_{-\infty}^{\infty}y'\left|u'\right|^2 + B\left(\omega_{\rm high},\epsilon_{\rm width}\right)\epsilon^{-1}C^{-2}\omega^{-2}\int_{-\infty}^{\infty}y'\omega^2\left|u\right|^2,
\end{align}
while  for $V_{\geq}$ we have
\begin{align}\label{absorbV2}
-\int_{-\infty}^{\infty}\left(yV_{\geq}\right)'\left|u\right|^2 &= \int_{-\infty}^{\infty}\left(-y'V_{\geq} - yV_{\geq}'\right)\left|u\right|^2 \\ \nonumber
&\ge \int_{-\infty}^{\infty}\left(-B\left(\omega_{\rm high},\epsilon_{\rm width}\right)R_6^{-2}\omega^{-2}\left(y'\omega^2\right)\right.\\
\nonumber &\qquad\qquad\left.
 -B\left(\omega_{\rm high},\epsilon_{\rm width}\right)\left(R_{\infty}-R_6\right)^{-1}1_{\text{supp}\left(\chi_2'\right)}yV + by\left(1-\chi_2\right)r^{-3}\right)\left|u\right|^2  \\ \nonumber & \ge-B\left(\omega_{\rm high},\epsilon_{\rm width}\right)\text{max}\left(R_{\infty}^{-2}\omega^{-2},\left(R_{\infty}-R_6\right)^{-1}\right)\int_{-\infty}^{\infty}\left(y'\omega^2\right)\left|u\right|^2.
\end{align}

It is now clear that choosing $C$, $R_{\infty}$ and $R_{\infty} - R_6$ sufficiently large depending on $\omega_{\rm low}$ or $\omega_{\rm low}$ and $\omega_{\rm high}$ and $\epsilon_{\rm width}$ and combining~(\ref{nonStatFirst}),~(\ref{absorbV1}) and~(\ref{absorbV2}) will imply the proposition.

\end{proof}

\subsection{Putting everything together}\label{putting}
In this section we will combine the propositions of the above sections to prove Theorem~\ref{phaseSpaceILED}.

First of all, keeping Lemma~\ref{everythingCovered} in mind, we observe that for any choice of $\omega_{\rm high}$ and $\epsilon_{\rm width}$, every admissible frequency triple $\left(\omega,m,\Lambda\right)$ lies in exactly one of the frequency ranges:
$\mathcal{G}_{\mbox{$\flat$}}(\omega_{\rm high}, \epsilon_{\rm width})$,
$\mathcal{G}_{\lessflat}(\omega_{\rm high}, \epsilon_{\rm width})$,
$\mathcal{G}_{\mbox{$\natural$}}(\omega_{\rm high}, \epsilon_{\rm width})$,
$\mathcal{G}_{\mbox{$\sharp$}}(\omega_{\rm high}, \epsilon_{\rm width})$,
$\mathcal{G}^{\mbox{$\sharp$}}(\omega_{\rm high}, \epsilon_{\rm width})$.
Thus, it only remains to choose the constants
$\epsilon_{\rm width}$, $E$,
$\omega_{\rm high}$, $\omega_{\rm low}$, $\tilde a_0$ and $R_{\infty}$ in the correct order so that it is possible to apply simultaneously all of the above propositions.

The first constant we fix is a sufficiently small $\epsilon_{\rm width}$, consistent with the requirements of Propositions~\ref{odeEst2} and~\ref{odeEst3}. Then,
depending on the choice of $\epsilon_{\rm width}$, for all large enough $\omega_{\rm high}$, $R_{\infty}$ and $E$ we may apply Propositions~\ref{odeEst2},~\ref{odeEst3} and,
in addition, Proposition~\ref{odeEst4}, corresponding to  the frequency regimes $\mathcal{G}_{\mbox{$\sharp$}}$, $\mathcal{G}_{\lessflat}$ and $\mathcal{G}_{\mbox{$\natural$}}$.
Now we fix the choice of $E$ consistent with the above requirement. Then, depending on this choice of $E$, for all large enough $\omega_{\rm high}$ and $R_{\infty}$ we may apply,
in addition to the above Propositions, also
Proposition~\ref{odeEst1}, corresponding to the frequency regime $\mathcal{G}^{\mbox{$\sharp$}}$. Finally, we fix the constant $\omega_{\rm high}$ consistent also with this requirement.

Since $\epsilon_{\rm width}$ and $\omega_{\rm high}$ are both fixed,
the frequency ranges $\mathcal{G}_{\mbox{$\sharp$}}$, $\mathcal{G}_{\lessflat}$, $\mathcal{G}_{\mbox{$\natural$}}$,  $\mathcal{G}^{\mbox{$\sharp$}}$,  $\mathcal{G}_{\mbox{$\flat$}}$
are now determined.

We still must determine the four subranges of $\mathcal{G}_{\mbox{$\flat$}}$ which
depend on additional parameters $\omega_{\rm low}$ and $\tilde a_0$,
and make our final choice of $R_{\infty}$.

We choose first $\tilde a_0$ and $\omega_{\rm low}$ sufficiently small
so that for $R_\infty$ sufficiently large, we can apply Propositions~\ref{odeEst5} and~\ref{odeEst5b}. We then fix our choice of $\tilde a_0$.
Then chose sufficiently small $\omega_{\rm low}$ depending on $\tilde a_0$,
and note that for sufficiently large $R_{\infty}$ depending on $\tilde a_0$
we may apply (in addition to all previous Propositions) also Proposition~\ref{odeEst6}.
Finally, choose $R_{\infty}$ so that we may apply, in addition to all the previous Propositions,
also  Proposition~\ref{odeEst7}.

With these choices,
all frequency ranges are determined so as to indeed simultaneously satisfy the assumptions of  Propositions \ref{odeEst1}--\ref{odeEst7}.
We now for each  frequency range
define the functions $f$, $y$, $r_{\rm trap}$, etc., as given in
the corresponding Proposition or else set them to $0$ (as explained before each statement).
The statement of Propositions~\ref{odeEst1}--\ref{odeEst7} then give Theorem~\ref{phaseSpaceILED} for
 frequencies $(\omega, m, \Lambda)$ in the corresponding range.
Since these ranges cover all admissible frequencies, the proof is complete.

\subsection{Trapping parameters}\label{trappingParam}
We finally define the trapping parameters $s_{\pm}$ which appear in the definition of the degeneration function $\zeta$ (see~(\ref{degenerationfunc})) which
in turn appears in the statement of Theorem~\ref{theResult}.

\begin{definition}Let $0\le a_0 < M$ and let $\omega_{\rm high} = \omega_{\rm high}(a_0,M)$ and $\epsilon_{\rm width} = \epsilon_{\rm width}(a_0,M)$ be the parameters
from Theorem~\ref{phaseSpaceILED}. We define the trapping parameters $s_{\pm}$ by
\begin{equation}\label{sminus}
s_-\left(a_0,M\right) \doteq 3M - \inf_{0\le a \leq a_0,(\omega,m,\Lambda) \in \mathcal{G}_{\mbox{$\natural$}}, r_{\rm trap}\ne 0}r_{\rm trap}\left(\omega,m,\Lambda\right) - \varepsilon\left(a_0\right),
\end{equation}
\begin{equation}\label{splus}
s_+\left(a_0,M\right) \doteq \sup_{0\le a \leq a_0,(\omega,m,\Lambda) \in \mathcal{G}_{\mbox{$\natural$}}r_{\rm trap}\ne 0}r_{\rm trap}\left(\omega,m,\Lambda\right) - 3M + \varepsilon\left(a_0\right).
\end{equation}
where
$\varepsilon\left(a_0\right)$ is
a fixed choice of continuous function such  that $\varepsilon(0)=0$ and
$\varepsilon\left(a_0\right) > 0$ for $a_0>0$,
and such that $s_\pm$ satisfy for all $0\le a\le a_0$ the relations
\[
r_+(a,M)<3M-s_-(a_0,M)
<3M + s_+\left(a_0,M\right) < \infty.
\]

\end{definition}

The proof of Proposition~\ref{odeEst4} shows that $\varepsilon(a_0)$ can be chosen
ensuring that $s_\pm$, $\varepsilon$ enjoy the  properties claimed in the above
definition.

\begin{remark}
Let us observe that we then necessarily have
    \[
    \lim_{a_0\to 0}s_{\pm}\left(a_0,M\right) = 0, \qquad
    \lim_{a_0\to M} s_-(a_0,M) =2M=3M-r_+(M,M)
    \]
From the latter,
it follows that  we must also have $\lim_{a_0\to M}\varepsilon\left(a_0\right) = 0$.
\end{remark}

Recall the definition of the physical space degeneration function
$\zeta$ (see~(\ref{degenerationfunc})) which in particular required the definition of the points $s_{\pm}$. It follows from our definition of
$r_{\rm trap}$ in Theorem~\ref{phaseSpaceILED} that
we now have  for all admissible $(\omega, m, \Lambda)$
then $r_{\rm trap}=0$ or
\[
3M-s_-+\varepsilon(a_0) \le r_{\rm trap} \le 3M+s_+-\varepsilon(a_0)
\]
It follows that  for all admissible $(\omega, m, \Lambda)$, we have the uniform bound
\begin{equation}
\label{insteadofthis}
b\zeta \le (1-r^{-1}r_{\rm trap})^2.
\end{equation}
In particular, the statement of Theorem~\ref{phaseSpaceILED} holds with $(1-r^{-1}r_{\rm trap})^2$ replaced by $\zeta$.
It is this weaker statement that we will in fact apply in the following section.

\section{Summing and integrated local energy decay for future-integrable solutions}\label{summation}

In this section, we will combine the estimates of Sections~\ref{Nmult},~\ref{largeR}
and the o.d.e.~analysis of Section~\ref{freqLocEst} to prove
integrated local energy decay for solutions of the wave equation sufficiently integrable
towards the future. We begin by defining this class and stating the  main proposition.

\subsection{Future-integrable solutions of the wave equation}\label{futIntSec}
Let $a_0<M$,
$|a| \le a_0$ and let $\psi$ be as in the reduction of Section~\ref{WPosed}, i.e.,
a  solution of the wave equation $(\ref{WAVE})$ on $\mathcal{R}_0$
arising from smooth  compactly supported data at $\Sigma_0$.
Let $\xi\left(\tau\right)$ be smooth function which is $0$ in the past of $\Sigma_0$ and
identically $1$ in the future of $\Sigma_1$. Then we define
\[
\psi_{\text{\Rightscissors}} \doteq \xi\psi.
\]
We have
\begin{equation}
\label{inhogo}
\Box_g\psi_{\text{\Rightscissors}} = F \doteq 2\nabla^{\mu}\xi\nabla_{\mu}\psi + \left(\Box_g\xi\right)\psi.
\end{equation}
\begin{definition}
Let $|a|<M$ and
let $\psi$ be a solution of $(\ref{WAVE})$ as in the reduction of Section~\ref{WPosed}.
We shall say that $\psi$ is future-integrable if
$\psi_{\text{\Rightscissors}}$ satisfies Definition~\ref{sufficient}.
\end{definition}

Note that $\psi_{\text{\Rightscissors}}$ by its construction
will
then automatically
satisfy Definitions~\ref{sufficient2}.

Recall the degeneration function $\zeta$ defined by $(\ref{degenerationfunc})$ in Section~\ref{weritsdef}, and~\ref{easyCoArea}. The main result of this section
is
\begin{proposition}\label{closedILED}Let $a_0<M$, $|a|\le a_0$, and let
$\psi$ be a future
integrable solution of  $(\ref{WAVE})$.
Then, for every $\delta > 0$
\[
\int_{\mathcal{H}^+_0}\mathbf{J}^N_{\mu}[\psi]n^{\mu}_{\mathcal{H}^+}+ \int_{\mathcal{I}^+}\mathbf{J}^T_{\mu}[\psi]n^{\mu}_{\mathcal{I}^+}+
\int_0^{\infty}\int_{\Sigma_{\tau}}\left(\left|\tilde{Z}^*\psi\right|^2r^{-1-\delta} + r^{-3-\delta}\left|\psi\right|^2 +
\zeta \left|T\psi\right|^2r^{-1-\delta} + \zeta\left|\nabb\psi\right|^2r^{-1}\right)\, d\tau\]
\[ \leq
B\left(\delta\right)\int_{\Sigma_0}\mathbf{J}^N_{\mu}[\psi]n^{\mu}_{\Sigma_0}.
\]
\end{proposition}

The proof of this proposition will be carried out in Sections~\ref{cutoffSec}--\ref{whitinghere} below.
In view of the reduction of Section~\ref{signOfa}, we may
assume in this proof that   $a\ge 0$, in order
to appeal to the the results of Section~\ref{freqLocEst} as stated.

\subsection{Finite in time energy estimate}\label{cutoffSec}
Defining $\psi_{\text{\Rightscissors}}$ as above, by Section~\ref{karteri},
we may apply Carter's separation to the inhomogeneous equation
$(\ref{inhogo})$ to define the function $u^{(a\omega)}_{m\ell}$.
Lemma~\ref{aeRegular} implies that for almost every $\omega$, then for
all $m$, $\ell$. the function $H^{(a\omega)}_{m\ell}$ is smooth and $u^{(a\omega)}_{m\ell}$ is a smooth solution to the radial o.d.e.~(\ref{e3iswsntouu}) satisfying the boundary conditions~(\ref{eq:b+}) and~(\ref{eq:b-}). For each such $\omega$ we may apply the estimates of Section~\ref{freqLocEst} to the admissible
triples $(\omega, m, \Lambda= \Lambda_{m\ell}(a\omega))$
and conclude that Theorem~\ref{phaseSpaceILED} holds.

Thus, integrating $(\ref{fromPhaseSpace2})$
in $\omega$, summing in $m$ and $\ell$ and applying~(\ref{insteadofthis}) yields the estimate
\begin{align}\label{fromPhaseSpace}
b&\int_{-\infty}^{\infty}\sum_{m\ell}\int_{R_-^*}^{R_+^*}\left(\left|u'\right|^2 + \left(\zeta\left(\omega^2 + \Lambda_{m\ell}\right) + 1\right)\left|u\right|^2\right)\, dr^*\, d\omega
\\ \nonumber &\le \int_{-\infty}^{\infty}\sum_{m\ell}\int_{-\infty}^{\infty}H \cdot  (f, h, y, \chi) \cdot (u, u')\,dr^*\, d\omega + \int_{\omega_{\rm low} \leq \left|\omega\right| \leq \omega_{\rm high}}\sum_{\{m\ell : \Lambda \leq \epsilon_{\rm width}^{-1}\omega_{\rm high}^2\}}\left|u\left(-\infty\right)\right|^2\, d\omega.
\end{align}

An application of Plancherel to~$(\ref{fromPhaseSpace})$ (see the explicit formulas in Section~\ref{separationSubsection} and the discussion of the volume form in Section~\ref{usefulcomps}) yields
\begin{align}\label{conseqSecFreqLocEst}
&b\int_0^{\infty}\int_{\Sigma_{\tau} \cap [R_-,R_+]}\left(\left|\partial_{r^*}\psi_{\text{\Rightscissors}}\right|^2 + \left|\psi_{\text{\Rightscissors}}\right|^2 +
\zeta \left|T\psi_{\text{\Rightscissors}}\right|^2 + \zeta\left|\nabb\psi_{\text{\Rightscissors}}\right|^2\right)\, d\tau  \\ \nonumber &\,\,\le \int_{-\infty}^{\infty}\sum_{m\ell}\left(\int_{-\infty}^{\infty}H \cdot  (f, h, y, \chi) \cdot (u, u')\right)\, d\omega+ \int_{\omega_{\rm low} \leq \left|\omega\right| \leq \omega_{\rm high}}\sum_{\{m\ell : \Lambda \leq \epsilon_{\rm width}^{-1}\omega_{\rm high}^2\}}\left|u\left(-\infty\right)\right|^2\, d\omega.
\end{align}

Recall that the last term on the right hand side of both of these estimates arises from Section~\ref{nonStat}. We further remark that this term would be controlled by the physical space quantity $\int_{\mathcal{H}^+(0,\infty)}\mathbf{J}^N_{\mu}[\psi]n^{\mu}_{\mathcal{H}^+}$, if we had
control for the latter--in general, we do not, however! In Section~\ref{whitinghere} we shall exploit the localisation of the integral in $\omega$ to control this using the quantitative
mode stability result~\cite{shlapRot}.

The first thing we observe is that $\partial\psi_{\text{\Rightscissors}}$ only differ from $\partial\psi$ when $\tau \in [0,1]$. However, in this region, the energy can simply be controlled by a finite in time energy inequality and a Hardy inequality.

For the $\left|\psi_{\text{\Rightscissors}}\right|^2$ term, we observe
\begin{align*}
\int_0^1\int_{\Sigma_{\tau}\cap [R_-, R_+]}\left|\psi\right|^2 &\leq B\int_0^1\int_{\Sigma_{\tau}}\frac{\left|\psi\right|^2}{r^2} \leq B\int_0^1\int_{\Sigma_{\tau}}\mathbf{J}^N_{\mu}[\psi]n^{\mu}_{\Sigma_{\tau}} \leq B\int_{\Sigma_0}\mathbf{J}^N_{\mu}[\psi]n^{\mu}_{\Sigma_{\tau}},
\end{align*}
where we have used a Hardy inequality and a finite in time energy estimate.

We conclude
\begin{align}\label{addFiniteEnergyEst}
&b\int_0^{\infty}\int_{\Sigma_{\tau}\cap \left[R_-,R_+\right]}\left(\left|\partial_{r^*}\psi\right|^2 + \left|\psi\right|^2 +
\zeta \left|T\psi\right|^2 + \zeta\left|\nabb\psi\right|^2\right)\, d\tau  \\ \nonumber &\le \int_{-\infty}^{\infty}\sum_{m\ell}\left(\int_{-\infty}^{\infty}H \cdot  (f, h, y, \chi) \cdot (u, u')\right)\, d\omega
\\ \nonumber &\qquad+\int_{\omega_{\rm low} \leq \left|\omega\right| \leq \omega_{\rm high}}\sum_{\{m\ell : \Lambda \leq \epsilon_{\rm width}^{-1}\omega_{\rm high}^2\}}\left|u\left(-\infty\right)\right|^2\, d\omega + \int_{\Sigma_0}\mathbf{J}^N_{\mu}[\psi]n^{\mu}_{\Sigma_0}.
\end{align}

\subsection{Adding in the red-shift}\label{addInTheRedShift}
Next, we add a small constant (depending only on $M$)
times the red-shift estimate of Proposition~\ref{ftrs} to~(\ref{addFiniteEnergyEst}).
Recalling that $R_- = r_+ + \frac{1}{2}\left(r_{\rm red} - r_+\right)$, we thus obtain
\begin{align}\label{redshiftAdded}
&b\int_0^{\infty}\int_{\Sigma_{\tau}\cap \left[r_+,R_+\right]}\left[\left|\tilde{Z}^*\psi_{\text{\Rightscissors}}\right|^2 + \left|\psi_{\text{\Rightscissors}}\right|^2 +
\zeta \left|T\psi_{\text{\Rightscissors}}\right|^2 + \zeta\left|\nabb\psi_{\text{\Rightscissors}}\right|^2\right]\, d\tau +b\int_{\mathcal{H}^+(0,\infty)}\mathbf{J}^N_{\mu}[\psi]n^{\mu}_{\mathcal{H}^+}  \\ \nonumber &\le \int_{-\infty}^{\infty}\sum_{m\ell}\left(\int_{-\infty}^{\infty}H \cdot  (f, h, y, \chi) \cdot (u, u')\right)\, d\omega
\\ \nonumber &\qquad + \int_{\Sigma_0}\mathbf{J}^N_{\mu}[\psi]n^{\mu}_{\Sigma_0} + \int_{\omega_{\rm low} \leq \left|\omega\right| \leq \omega_{\rm high}}\sum_{\{m\ell : \Lambda \leq \epsilon_{\rm width}^{-1}\omega_{\rm high}^2\}}\left|u\left(-\infty\right)\right|^2\, d\omega.
\end{align}

\subsection{Adding in the large $r$ current}
Next, we would like to add in a small constant times the large $r$ estimate of Section~\ref{largeR}. However, we must be careful because that estimate produces an error proportional to $\int_{\Sigma_{\tau}}\mathbf{J}^N_{\mu}[\psi]n^{\mu}_{\Sigma_{\tau}}$, and we do not yet have a uniform energy boundedness statement.

We surmount this difficulty as follows. Since $\psi$ is future-integrable, we know that $\int_{\Sigma_{\tau} \cap [r_+,R_+]}\mathbf{J}^N_{\mu}[\psi]n^{\mu}_{\Sigma_{\tau}}$, as a function of $\tau$, is in $L^1_{\tau}[0,\infty)$. A pigeonhole argument implies that we may find a constant $C\left(\psi\right)$ and a dyadic sequence $\{\tau_n\}_{n=1}^{\infty}$ with $\lim_{n\to\infty}\tau_n = \infty$ such that
\begin{equation}\label{dyadicEnergyDecay}
\int_{\Sigma_{\tau_n}\cap[r_+,R_+]}\mathbf{J}^N_{\mu}[\psi]n^{\mu}_{\Sigma_{\tau_n}} \leq \frac{C}{\tau_n}.
\end{equation}
Note that $R_+ = 2R_{\rm large} \geq 4M$, so that $T$ is timelike in the region $r \geq R_+$. Then, a $\mathbf{J}_{\mu}^T$ energy estimate implies
\begin{align}\label{conseqDyadicEnergyDecay}
\int_{\Sigma_{\tau_n}}\mathbf{J}^N_{\mu}[\psi]n^{\mu}_{\Sigma} &\leq \int_{\Sigma_{\tau_n}\cap [r_+,R_+]}\mathbf{J}^N_{\mu}[\psi]n^{\mu}_{\Sigma} + B\int_{\Sigma_{\tau_n}\cap [R_+,\infty)}\mathbf{J}^T_{\mu}[\psi]n^{\mu}_{\Sigma}
\\ \nonumber &\leq B\frac{C}{\tau_n} + B\int_{\mathcal{H}^+(0,\tau_n)}\mathbf{J}^N_{\mu}[\psi]n^{\mu}  + B\int_{\Sigma_0}\mathbf{J}^N_{\mu}[\psi]n^{\mu}.
\end{align}
Now, combine Proposition~\ref{lrp} with~(\ref{conseqDyadicEnergyDecay}), multiply the result by a sufficiently small constant, and then add the result to~(\ref{redshiftAdded}). In particular, the horizon term on the right hand side of~(\ref{conseqDyadicEnergyDecay}) will be multiplied by a small constant, and hence can be absorbed into the left hand side of~(\ref{redshiftAdded}). We obtain
\begin{align}\label{addedLargeR}
b\int_0^{\tau_n}\int_{\Sigma_{\tau}}&\left[|\tilde{Z}^*\psi|^2r^{-1-\delta} + \left|\psi\right|^2r^{-3-\delta} +
\zeta \left|T\psi\right|^2r^{-1-\delta} + \zeta\left|\nabb\psi\right|^2r^{-1}\right]\, d\tau  \\ \nonumber \le&
\int_{-\infty}^{\infty}\sum_{m\ell}\left(\int_{-\infty}^{\infty}H \cdot  (f, h, y, \chi) \cdot (u, u')\right)\, d\omega + B\left(\delta\right)\int_{\Sigma_0}\mathbf{J}^N_{\mu}[\psi]n^{\mu}_{\Sigma_0} \\ \nonumber &+ \int_{\omega_{\rm low} \leq \left|\omega\right| \leq \omega_{\rm high}}\sum_{\{m\ell : \Lambda \leq \epsilon_{\rm width}^{-1}\omega_{\rm high}^2\}}\left|u\left(-\infty\right)\right|^2\, d\omega + \frac{C}{\tau_n}.
\end{align}
Taking $n \to \infty$ gives
\begin{align}\label{addedLargeR2}
\int_0^{\infty}\int_{\Sigma_{\tau}}&\left[\left|\tilde{Z}^*\psi\right|^2r^{-1-\delta} + \left|\psi\right|^2r^{-3-\delta} +
\zeta \left|T\psi\right|^2r^{-1-\delta} + \zeta\left|\nabb\psi\right|^2r^{-1}\right]\, d\tau  \\ \nonumber \le& \int_{-\infty}^{\infty}\sum_{m\ell}\left(\int_{-\infty}^{\infty}H \cdot  (f, h, y, \chi) \cdot (u, u')\right)\, d\omega
\\\nonumber&+
 B\left(\delta\right)\int_{\Sigma_0}\mathbf{J}^N_{\mu}[\psi]n^{\mu}_{\Sigma_0} + \int_{\omega_{\rm low} \leq \left|\omega\right| \leq \omega_{\rm high}}\sum_{\{m\ell : \Lambda \leq \epsilon_{\rm width}^{-1}\omega_{\rm high}^2\}}\left|u\left(-\infty\right)\right|^2\, d\omega.
\end{align}

\subsection{Boundedness of the energy flux to $\mathcal{I}^+$}
\label{toayplus}

The estimates of the previous section give in addition an estimate for
the energy flux to null infinity
\begin{equation}
\label{tonullinfinity}
\int_{\mathcal{I}^+}\mathbf{J}^T_{\mu}[\psi]n^{\mu}_{\mathcal{I}^+},
\end{equation}
which is well-defined by an easy limiting operation (see~\cite{dr6})
which we omit here.

To bound $(\ref{tonullinfinity})$, we only need the easily proven property
of its definition, that
applying $\mathbf{J}^T$ energy estimates outside the ergoregion,
$(\ref{tonullinfinity})$ can be seen to satisfy
\begin{equation}\label{toBoundNullInf}
\int_{\mathcal{I}^+}\mathbf{J}^T_{\mu}[\psi]n^{\mu}_{\mathcal{I}^+} \leq \limsup_{n\to\infty}\int_{\Sigma_{s_n}}\mathbf{J}^N_{\mu}[\psi]n^{\mu}_{\Sigma_{s_n}},
\end{equation}
where $\{s_n\}$ is any sequence with $\lim_{n\to\infty}s_n = \infty$.

Now, taking the limit as $n\to\infty$ in the inequality~(\ref{conseqDyadicEnergyDecay}) and then applying~(\ref{redshiftAdded}) gives
\begin{align}\label{almostNullInfBound}
\limsup_{n\to\infty}\int_{\Sigma_{\tau_n}}\mathbf{J}^N_{\mu}[\psi]n^{\mu}_{\Sigma} &\leq B\int_{\mathcal{H}^+(0,\infty)}\mathbf{J}^N_{\mu}[\psi]n^{\mu}  + B\int_{\Sigma_0}\mathbf{J}^N_{\mu}[\psi]n^{\mu}
\\ \nonumber &\leq B\int_{\Sigma_0}\mathbf{J}^N_{\mu}[\psi]n^{\mu}_{\Sigma_0}.
\end{align}
Then~(\ref{toBoundNullInf}) implies
\begin{align}\label{boundNullInf}
\int_{\mathcal{I}^+}\mathbf{J}^T_{\mu}[\psi]n^{\mu}_{\mathcal{I}^+} \leq
\limsup_{n\to\infty}\int_{\Sigma_{\tau_n}}\mathbf{J}^N_{\mu}[\psi]n^{\mu}_{\Sigma} \leq B\int_{\Sigma_0}\mathbf{J}^N_{\mu}[\psi]n^{\mu}_{\Sigma_0}.
\end{align}

An alternative approach for controlling $\int_{\mathcal{I}^+}\mathbf{J}^N_{\mu}[\psi]n^{\mu}_{\mathcal{I}^+}$ can be found in the proof of the inequality~(\ref{boundedEnergyInfinity}) where a cut-off $\mathbf{J}^T$ energy estimate is applied and the errors are absorbed with the integrated energy decay statement. Lastly, we note that yet another approach would be to first show that (up to a normalisation constant)
\[\int_{\mathcal{I}^+}\mathbf{J}^T_{\mu}[\psi]n^{\mu}_{\mathcal{I}^+} = \int_{-\infty}^{\infty}\sum_{m\ell}\omega^2\left|u\left(\infty\right)\right|^2,\]
and then use the fact that the estimates of Section~\ref{freqLocEst} give a bound for the right hand side of this equality.

\subsection{Error terms associated to the cutoff}
In this subsection we will study closely the terms $\int_{-\infty}^{\infty}\sum_{m\ell}\left(\int_{-\infty}^{\infty}H \cdot  (f, h, y, \chi) \cdot (u, u')\right)\, d\omega$ on the right hand side of~(\ref{addedLargeR2}). Recall that when $r \geq R_{\infty}$, we have arranged for our multipliers to be independent of the frequency parameters or decay exponentially in $r$ (see Remark~\ref{R2}). In particular, we may split the error terms associated to the cutoff into:
\begin{align}\label{errorsFromCutoff}
&\int_{-\infty}^{\infty}\sum_{m\ell}\left(\int_{-\infty}^{\infty}H \cdot  (f, h, y, \chi) \cdot (u, u')\right)\, d\omega\\
\nonumber&\qquad  = \int_{-\infty}^{\infty}\sum_{m\ell}\left(\int_{-\infty}^{\infty}\chi_{R_{\infty}}H \cdot  (f, h, y, \chi) \cdot (u, u')\right)\, d\omega -\int_{-\infty}^{\infty}\sum_{m\ell}\left(\int_{-\infty}^{\infty}2\left(1-\chi_{R_{\infty}}\right)\text{Re}\left(u'\overline H\right)\right)\, d\omega
\\ \nonumber &\qquad\qquad+\int_{-\infty}^{\infty}\sum_{m\ell}\left(\int_{-\infty}^{\infty}E\omega\left(1-\chi_{R_{\infty}}\right)\text{Im}\left(H\overline u\right)\right)\, d\omega
\\ \nonumber &\qquad\qquad+2\int_{\omega_{\rm low} \leq \left|\omega\right| \leq \omega_{\rm high}}\sum_{\{m\ell : \Lambda \leq \epsilon_{\rm width}^{-1}\omega_{\rm high}^2\}}\left(\int_{-\infty}^{\infty}\left(1-\chi_{R_{\infty}}\right)\tilde y\text{Re}\left(u'\overline H\right)\right)\, d\omega.
\end{align}
Here $\chi_{R_{\infty}}$ is a cutoff which is identically $1$ on $[r_+,R_{\infty}]$ and $0$ on $[R_{\infty}+1,\infty)$.
\subsubsection{The bounded $r$ error terms}\label{boundedRError}
The error terms in the region $[r_+,R_{\infty}]$ are the easiest to deal with. We simply observe that an application of Plancherel (see Sections~\ref{separationSubsection} and~\ref{usefulcomps}), finite in time energy estimates and Hardy inequalities
\begin{align}
&\left|\int_{-\infty}^{\infty}\sum_{m\ell}\left(\int_{-\infty}^{\infty}\chi_{R_{\infty}}H \cdot  (f, h, y, \chi) \cdot (u, u')\right)\, d\omega\right|
\\ \nonumber &\qquad\leq B\int_{-\infty}^{\infty}\sum_{m\ell}\int_{-\infty}^{R_{\infty}+1}\left|H\right|\left(\left|u'\right| + \left|u\right|\right)
\\ \nonumber &\qquad\leq B\epsilon^{-1}\int_0^{\infty}\int_{\Sigma_{\tau}\cap [r_+,R_{\infty}+1)}\left|F\right|^2 + \epsilon\int_0^{\infty}\int_{\Sigma_{\tau}\cap [r_+,R_{\infty}+1)}\left(\left|\partial_{r^*}\psi\right|^2 + \left|\psi\right|^2\right)
\\ \nonumber &\qquad\leq B\epsilon^{-1}\int_0^{2}\int_{\Sigma_{\tau}\cap [r_+,\infty)}\mathbf{J}^N_{\mu}[\psi]n^{\mu}_{\Sigma_{\tau}} + \epsilon\int_0^{\infty}\int_{\Sigma_{\tau}\cap [r_+,R_{\infty}+1)}\left(\left|\partial_{r^*}\psi\right|^2 + \left|\psi\right|^2\right)
\\ \nonumber &\qquad\leq B\epsilon^{-1}\int_{\Sigma_0}\mathbf{J}^N_{\mu}[\psi]n^{\mu}_{\Sigma_0} + \epsilon\int_0^{\infty}\int_{\Sigma_{\tau}\cap [r_+,R_{\infty}+1)}\left(\left|\partial_{r^*}\psi\right|^2 + \left|\psi\right|^2\right).
\end{align}

\subsubsection{Large $r$: $\text {\fontencoding{LGR}\selectfont \koppa}^y$ error terms}\label{largeRkoppa}

For error terms supported in the $r \in [R_{\infty},\infty)$ region we must be careful that lower order terms come with appropriate $r$-weights so that either a Hardy inequality can be applied or they can be absorbed into the left hand side of~(\ref{addedLargeR2}). First of all, an application of Plancherel (see Section~\ref{separationSubsection}) gives
\begin{align}
&\int_{-\infty}^{\infty}\sum_{m\ell}\left(\int_{-\infty}^{\infty}\left(1-\chi_{R_{\infty}}\right)\text{Re}\left(u'\overline H\right)\right)\, d\omega
\\ \nonumber\qquad = &\int_0^{\infty}\int_{-\infty}^{\infty}\int_{\mathbb{S}^2}\left(1-\chi_{R_{\infty}}\right)\text{Re}\left(\partial_{r^*}\left(\left(r^2+a^2\right)^{1/2}\psi_{\text{\Rightscissors}}\right)\overline{\Delta\left(r^2+a^2\right)^{-1/2}F}\right)\ \sin\theta\, dt\, dr^*\, d\theta\, d\phi.
\end{align}
Recall that
\[F = \left(r^2+a^2\right)^{-1}\rho^2\left(2\nabla^{\mu}\xi\nabla_{\mu}\psi + \left(\Box_g\xi\right)\psi\right).\]
For sufficiently large $r$, $\xi$ is just a function of $t$. Hence,
\[F = \left(r^2+a^2\right)^{-1}\rho^2\left(2g^{tt}\dot\xi\partial_t\psi + 2g^{t\phi}\dot\xi\partial_{\phi}\psi + g^{tt}\ddot\xi\psi\right).\]
Thus, (suppressing the $\sin\theta\, dt\, dr^*\, d\theta\, d\phi$)
\begin{align}
&\int_0^{\infty}\int_{-\infty}^{\infty}\int_{\mathbb{S}^2}\left(1-\chi_{R_{\infty}}\right)\text{Re}\left(\partial_{r^*}\left(\left(r^2+a^2\right)^{1/2}\psi_{\text{\Rightscissors}}\right)\overline{\Delta\left(r^2+a^2\right)^{-1/2}F}\right)
\\&  =\nonumber \int_0^{\infty}\int_{-\infty}^{\infty}\int_{\mathbb{S}^2}\left(1-\chi_{R_{\infty}}\right)\text{Re}\left(\partial_{r^*}\left(\left(r^2+a^2\right)^{1/2}\psi_{\text{\Rightscissors}}\right)\overline{\Delta\left(r^2+a^2\right)^{-3/2}\rho^2\left(2g^{tt}\dot\xi\partial_t\psi + 2g^{t\phi}\dot\xi\partial_{\phi}\psi\right)}\right)
\\ \nonumber &\qquad+\int_0^{\infty}\int_{-\infty}^{\infty}\int_{\mathbb{S}^2}\left(1-\chi_{R_{\infty}}\right)\text{Re}\left(\partial_{r^*}\left(\left(r^2+a^2\right)^{1/2}\psi_{\text{\Rightscissors}}\right)\overline{\Delta\left(r^2+a^2\right)^{-3/2}\rho^2g^{tt}\ddot\xi\psi}\right).
\end{align}

First we consider the term with $g^{tt}\dot\xi\partial_t\psi$:
\begin{align}\label{gttTerm}
&\left|\int_0^{\infty}\int_{-\infty}^{\infty}\int_{\mathbb{S}^2}\left(1-\chi_{R_{\infty}}\right)\text{Re}\left(\partial_{r^*}\left(\left(r^2+a^2\right)^{1/2}\psi_{\text{\Rightscissors}}\right)\overline{\Delta\left(r^2+a^2\right)^{-3/2}\rho^2\left(2g^{tt}\dot\xi\partial_t\psi\right)}\right) \right|
\\ \nonumber \le&B\left|\int_0^{\infty}\int_{-\infty}^{\infty}\int_{\mathbb{S}^2}\left(1-\chi_{R_{\infty}}\right)\text{Re}\left(\left(\partial_{r^*}\psi_{\text{\Rightscissors}}\right)\overline{\Delta\left(r^2+a^2\right)^{-1}\rho^2\left(2g^{tt}\dot\xi\partial_t\psi\right)}\right)\right|
\\ \nonumber &+B\left|\int_0^{\infty}\int_{-\infty}^{\infty}\int_{\mathbb{S}^2}\left(1-\chi_{R_{\infty}}\right)\frac{r}{(r^2+a^2)^{1/2}}\text{Re}\left(\left(\psi_{\text{\Rightscissors}}\right)\overline{\Delta\left(r^2+a^2\right)^{-3/2}\rho^2\left(2g^{tt}\dot\xi\partial_t\psi\right)}\right)\right|
\\ \nonumber\leq & B\int_0^1\int_{\Sigma_{\tau}\cap [R_{\infty},\infty)}\left(\mathbf{J}^N_{\mu}[\psi]n^{\mu}_{\Sigma_{\tau}} + \frac{\left|\psi\right|^2}{r^2}\right)
 \nonumber\leq  B\int_{\Sigma_0}\mathbf{J}^N_{\mu}[\psi]n^{\mu}_{\Sigma_{\tau}}.
\end{align}

Keeping in mind that $g^{t\phi} = O\left(r^{-3}\right)$, the term with $g^{t\phi}\dot\xi\partial_{\phi}\psi$ can be treated like~(\ref{gttTerm}). Finally, recalling that $\partial_{r^*}\xi = 0$ for $r \geq R_{\infty}$, we have
\begin{align}\label{lowerOrder}
&\left|\int_0^{\infty}\int_{-\infty}^{\infty}\int_{\mathbb{S}^2}\left(1-\chi_{R_{\infty}}\right)\text{Re}\left(\partial_{r^*}\left(\left(r^2+a^2\right)^{1/2}\psi_{\text{\Rightscissors}}\right)\overline{\Delta\left(r^2+a^2\right)^{-3/2}\rho^2g^{tt}\ddot\xi\psi}\right)\right|
\\ \nonumber &= \left|\int_0^{\infty}\int_{-\infty}^{\infty}\int_{\mathbb{S}^2}\left(1-\chi_{R_{\infty}}\right)\Delta\left(r^2+a^2\right)^{-2}\rho^2g^{tt}\ddot\xi\xi\text{Re}\left(\partial_{r^*}\left(\left(r^2+a^2\right)^{1/2}\psi\right)\overline{\left(r^2+a^2\right)^{1/2}\psi}\right)\right|
\\ \nonumber &=
\frac{1}{2}\left|\int_0^{\infty}\int_{-\infty}^{\infty}\int_{\mathbb{S}^2}\partial_{r^*}\left(\left(1-\chi_{R_{\infty}}\right)\Delta\left(r^2+a^2\right)^{-2}\rho^2g^{tt}\ddot\xi\xi\right)\left(r^2+a^2\right)\left|\psi\right|^2\right|
\\ \nonumber &\leq B\int_0^1\int_{\Sigma_{\tau}\cap [R_{\infty},\infty)}\frac{\left|\psi\right|^2}{r^2}
\leq B\int_{\Sigma_0}\mathbf{J}^N_{\mu}[\psi]n^{\mu}_{\Sigma_0}.
\end{align}

Combining everything implies
\[
\left|\int_{-\infty}^{\infty}\sum_{m\ell}\left(\int_{-\infty}^{\infty}2\left(1-\chi_{R_{\infty}}\right)\text{Re}\left(u'\overline H\right)\right)\, d\omega\right| \leq B\int_{\Sigma_0}\mathbf{J}^N_{\mu}[\psi]n^{\mu}_{\Sigma_0}.
\]

\subsubsection{Large $r$: $\text{Q}^T$ error terms}\label{largeRQT}

An application of Plancherel (see Section~\ref{separationSubsection}) gives
\begin{align}
&\left|\int_{-\infty}^{\infty}\sum_{m\ell}\left(\int_{-\infty}^{\infty}\omega\left(1-\chi_{R_{\infty}}\right)\text{Im}\left(H\overline u\right)\right)\, d\omega\right|\\  \nonumber &=\left|\int_0^{\infty}\int_{-\infty}^{\infty}\int_{\mathbb{S}^2}\Delta\left(1-\chi_{R_{\infty}}\right)\text{Re}\left(\partial_t\psi_{\text{\Rightscissors}}\overline{F}\right)\ \sin\theta\, dt\, dr^*\, d\theta\, d\phi\right|.
\end{align}
We have
\[F = \left(r^2+a^2\right)^{-1}\rho^2\left(2g^{tt}\dot\xi\partial_t\psi +
 2g^{t\phi}\dot\xi\partial_{\phi}\psi + g^{tt}\ddot\xi\psi\right),\qquad
\partial_t\psi_{\text{\Rightscissors}} = \dot\xi\psi + \xi\partial_t\psi.\]
To ease the notation, let us introduce
\[W\left(r,\theta\right) := \Delta\left(r^2+a^2\right)^{-1}\left(1-\chi_{R_{\infty}}\right)\rho^2.\]
Then, suppressing the $\sin\theta\, dt\, dr^*\, d\theta\, d\phi$, we have
\begin{align}
\left|\int_0^{\infty}\right.&\left.\int_{-\infty}^{\infty}\int_{\mathbb{S}^2}\Delta\left(1-\chi_{R_{\infty}}\right)\text{Re}\left(\partial_t\psi_{\text{\Rightscissors}}\overline{F}\right)\right| \\
\nonumber
 \le &2\left|\int_0^{\infty}\int_{-\infty}^{\infty}\int_{\mathbb{S}^2}W\left(r,\theta\right)g^{t\phi}\dot\xi\text{Re}\left(\left(\dot\xi\psi + \xi\partial_t\psi\right)\overline{\partial_{\phi}\psi}\right)\right|
\\ \nonumber &+\left|\int_0^{\infty}\int_{-\infty}^{\infty}\int_{\mathbb{S}^2}W\left(r,\theta\right)\text{Re}\left(\left(\dot\xi\psi + \xi\partial_t\psi\right)\overline{\left(2g^{tt}\dot\xi\partial_t\psi + g^{tt}\ddot\xi\psi\right)}\right)\right|.
\end{align}
Several of the above terms are easy to deal with:
\begin{align*}
&\left|\int_0^{\infty}\int_{-\infty}^{\infty}\int_{\mathbb{S}^2}W\left(r,\theta\right)g^{t\phi}\dot\xi\text{Re}\left(\left(\dot\xi\psi\right)\overline{\partial_{\phi}\psi}\right)\right| = \left|\int_0^{\infty}\int_{-\infty}^{\infty}\int_{\mathbb{S}^2}W\left(r,\theta\right)g^{t\phi}\left(\dot\xi\right)^2\partial_{\phi}\left|\psi\right|^2\right| = 0.
\end{align*}
\begin{align*}
&\left|\int_0^{\infty}\int_{-\infty}^{\infty}\int_{\mathbb{S}^2}W\left(r,\theta\right)g^{t\phi}\dot\xi\xi\text{Re}\left(\left(\partial_t\psi\right)\overline{\partial_{\phi}\psi}\right)\right| \leq \int_0^1\int_{\Sigma_{\tau}}\mathbf{J}^N_{\mu}[\psi]n^{\mu}_{\Sigma_{\tau}} \leq \int_{\Sigma_0}\mathbf{J}^N_{\mu}[\psi]n^{\mu}_{\Sigma_0}.
\end{align*}
\begin{align*}
2\left|\int_0^{\infty}\int_{-\infty}^{\infty}\int_{\mathbb{S}^2}W\left(r,\theta\right)g^{tt}\xi\dot\xi\left|\partial_t\psi\right|^2\right| \leq \int_0^1\int_{\Sigma_{\tau}}\mathbf{J}^N_{\mu}[\psi]n^{\mu}_{\Sigma_{\tau}} \leq \int_{\Sigma_0}\mathbf{J}^N_{\mu}[\psi]n^{\mu}_{\Sigma_0}.
\end{align*}
Combining everything yields
\begin{align}\label{partialProgAbsorbQT}
&\left|\int_{-\infty}^{\infty}\sum_{m\ell}\left(\int_{-\infty}^{\infty}\omega\left(1-\chi_{R_{\infty}}\right)\text{Im}\left(H\overline u\right)\right)\, d\omega\right|
\\ \nonumber &\leq\left|\int_0^{\infty}\int_{-\infty}^{\infty}\int_{\mathbb{S}^2}W\left(r,\theta\right)\left(\text{Re}\left(\dot\xi\psi\overline{\left(2g^{tt}\dot\xi\partial_t\psi + g^{tt}\ddot\xi\psi\right)}\right) + \text{Re}\left(\xi\partial_t\psi\overline{g^{tt}\ddot\xi\psi}\right)\right)\right| + \int_{\Sigma_0}\mathbf{J}^N_{\mu}[\psi]n^{\mu}_{\Sigma_0}.
\end{align}

We now focus on the first term on the right hand side:
\begin{align}
&\left|\int_0^{\infty}\int_{-\infty}^{\infty}\int_{\mathbb{S}^2}g^{tt}W\left(r,\theta\right)\left(2\left(\dot\xi\right)^2\text{Re}\left(\psi\overline{\partial_t\psi}\right) + \dot\xi\ddot\xi\left|\psi\right|^2 + \xi\ddot\xi\text{Re}\left(\partial_t\psi\overline{\psi}\right) \right)\right|
\\ \nonumber &=\left|\int_0^{\infty}\int_{-\infty}^{\infty}\int_{\mathbb{S}^2}g^{tt}W\left(r,\theta\right)\left(2\left(\dot\xi\right)^2\text{Re}\left(\psi\overline{\partial_t\psi}\right) - \left(\dot\xi\right)^2\text{Re}\left(\psi\overline{\partial_t\psi}\right) + \xi\ddot\xi\text{Re}\left(\partial_t\psi\overline{\psi}\right) \right)\right|
\\ \nonumber &=\left|\int_0^{\infty}\int_{-\infty}^{\infty}\int_{\mathbb{S}^2}g^{tt}W\left(r,\theta\right)\left(\left(\dot\xi\right)^2\text{Re}\left(\psi\overline{\partial_t\psi}\right)+ \xi\ddot\xi\text{Re}\left(\partial_t\psi\overline{\psi}\right) \right)\right|
\\ \nonumber &=\left|\int_0^{\infty}\int_{-\infty}^{\infty}\int_{\mathbb{S}^2}g^{tt}W\left(r,\theta\right)\left(\left(\dot\xi\right)^2\text{Re}\left(\psi\overline{\partial_t\psi}\right)- \left(\dot\xi\right)^2\text{Re}\left(\partial_t\psi\overline{\psi}\right)-\xi\dot\xi\text{Re}\left(\partial_t^2\psi\overline{\psi}\right) - \xi\dot\xi\left|\partial_t\psi\right|^2 \right)\right|
\\ \nonumber
&=\left|\int_0^{\infty}\int_{-\infty}^{\infty}\int_{\mathbb{S}^2}g^{tt}W\left(r,\theta\right)\left(\xi\dot\xi\text{Re}\left(\partial_t^2\psi\overline{\psi}\right) + \xi\dot\xi\left|\partial_t\psi\right|^2 \right)\right|
\\ \nonumber
&\leq \int_{\Sigma_0}\mathbf{J}^N_{\mu}[\psi]n^{\mu}_{\Sigma_0} +
\left|\int_0^{\infty}\int_{-\infty}^{\infty}\int_{\mathbb{S}^2}g^{tt}W\left(r,\theta\right)\left(\xi\dot\xi\text{Re}\left(\partial_t^2\psi\overline{\psi}\right)\right)\right|.
\end{align}
Instead of additional integration by parts on this last term, we use that $\psi$ solves the wave equation, which we write out here for reference:
\begin{align}\label{waveEqn}
&g^{tt}\partial_t^2\psi = \frac{4Mar}{\rho^2\Delta}\partial_{t,\phi}^2\psi - \frac{\Delta-a^2\sin^2\theta}{\Delta\rho^2\sin^2\theta}\partial_{\phi}^2\psi -
\frac{r^2+a^2}{\Delta\rho^2}\partial_{r^*}\left(\left(r^2+a^2\right)\partial_{r^*}\psi\right) - \frac{1}{\rho^{2}\sin\theta}\partial_{\theta}\left(\sin\theta\partial_{\theta}\psi\right).
\end{align}
Substituting the right hand side of~(\ref{waveEqn}) for $g^{tt}\partial_t^2\psi$ , carrying out a straightforward integration by parts, and applying a finite in time energy inequality shows
\[\left|\int_0^{\infty}\int_{-\infty}^{\infty}\int_{\mathbb{S}^2}g^{tt}W\left(r,\theta\right)\left(\xi\dot\xi\text{Re}\left(\partial_t^2\psi\overline{\psi}\right)\right)\right| \leq B\int_{\Sigma_0}\mathbf{J}^N_{\mu}[\psi]n^{\mu}_{\Sigma_0}.\]

Thus, we have shown
\[
\left|\int_{-\infty}^{\infty}\sum_{m\ell}\left(\int_{-\infty}^{\infty}\omega\left(1-\chi_{R_{\infty}}\right)\text{Im}\left(H\overline u\right)\right)\, d\omega\right| \leq B\int_{\Sigma_0}\mathbf{J}^N_{\mu}[\psi]n^{\mu}_{\Sigma_0}.
\]

\subsubsection{Large $r$: $\text {\fontencoding{LGR}\selectfont \koppa}^{\tilde y}$ error terms}\label{largeRkoppa2}

Since $\left|\tilde{y}\right| \leq \exp\left(-br^*\right)\text{ as }r^*\to \infty$,
we may estimate the term
$\int_{\omega_{\rm low} \leq \left|\omega\right| \leq \omega_{\rm high}}\sum_{\{m\ell : \Lambda \leq \epsilon_{\rm width}^{-1}\omega_{\rm high}^2\}}\left(\int_{-\infty}^{\infty}\left(1-\chi_{R_{\infty}}\right)\tilde y\text{Re}\left(u'\overline H\right)\right)\, d\omega$
exactly as in
Section~\ref{boundedRError}. We obtain
\begin{align*}
&\left|\int_{\omega_{\rm low} \leq \left|\omega\right| \leq \omega_{\rm high}}\sum_{\{m\ell : \Lambda \leq \epsilon_{\rm width}^{-1}\omega_{\rm high}^2\}}\left(\int_{-\infty}^{\infty}\left(1-\chi_{R_{\infty}}\right)\tilde y\text{Re}\left(u'\overline H\right)\right)\, d\omega\right| \\
&\qquad \leq B\epsilon^{-1}\int_{\Sigma_0}\mathbf{J}^N_{\mu}[\psi]n^{\mu}_{\Sigma_0} + \epsilon\int_0^{\infty}\int_{\Sigma_{\tau}\cap [R_{\infty},\infty)}e^{-br^*}\left(\left|\partial_{r^*}\psi\right|^2 + \left|\psi\right|^2\right).
\end{align*}

\subsubsection{Absorbing the error terms}
Combining the results of Sections~\ref{boundedRError},~\ref{largeRkoppa},~\ref{largeRQT} and~\ref{largeRkoppa2} gives
\begin{align}\label{cutoffErrorHandled}
&\int_{-\infty}^{\infty}\sum_{m\ell}\left(\int_{-\infty}^{\infty}H \cdot  (f, h, y, \chi) \cdot (u, u')\right)\, d\omega
\\ \nonumber &\qquad\leq B\int_{\Sigma_0}\mathbf{J}^N_{\mu}[\psi]n^{\mu}_{\Sigma_0} +\epsilon\int_0^{\infty}\int_{\Sigma_{\tau}}r^{-1-\delta}\left(\left|\partial_{r^*}\psi\right|^2 + r^{-2}\left|\psi\right|^2 + \zeta\mathbf{J}^N_{\mu}[\psi]n^{\mu}_{\Sigma_{\tau}}\right).
\end{align}
Taking $\epsilon$ sufficiently small and combining~(\ref{cutoffErrorHandled}) with~(\ref{addedLargeR2}),~(\ref{redshiftAdded}) and~(\ref{boundNullInf}) implies
\begin{align}\label{cutoffsDealtWith}
&b\int_{\mathcal{H}^+(0,\infty)}\mathbf{J}^N_{\mu}[\psi]n^{\mu}_{\mathcal{H}^+} + b\int_{\mathcal{I}^+_0}\mathbf{J}^N_{\mu}[\psi]n^{\mu}_{\mathcal{I}^+}\\
\nonumber &\qquad+b\int_0^{\infty}\int_{\Sigma_{\tau}}\left(\left|\partial_{r^*}\psi\right|^2r^{-1-\delta} + \left|\psi\right|^2r^{-3-\delta} +
\zeta \left|T\psi\right|^2r^{-1-\delta} + \zeta\left|\nabb\psi\right|^2r^{-1}\right)\, d\tau  \\
 \nonumber &\le B\left(\delta\right)\int_{\Sigma_0}\mathbf{J}^N_{\mu}[\psi]n^{\mu}_{\Sigma_0} + \int_{\omega_{\rm low} \leq \left|\omega\right| \leq \omega_{\rm high}}\sum_{\{m\ell : \Lambda \leq \epsilon_{\rm width}^{-1}\omega_{\rm high}^2\}}\left|u\left(-\infty\right)\right|^2\, d\omega.
\end{align}

\subsection{The non-stationary bounded frequency horizon term}
\label{whitinghere}

Finally, we come to the term $\int_{\omega_{\rm low} \leq \left|\omega\right| \leq \omega_{\rm high}}\sum_{\{m\ell : \Lambda \leq \epsilon_{\rm width}^{-1}\omega_{\rm high}^2\}}\left|u\left(-\infty\right)\right|^2\, d\omega$. Since we do not have a small parameter, we cannot hope to absorb this error term into the left hand side of~(\ref{cutoffsDealtWith}); however, this term has already been dealt with in the context of the quantitative mode stability work~\cite{shlapRot}:

\begin{proposition}\label{propShlapRot}Let $\psi$ be a future-integrable solution to~(\ref{WAVE}). Define $u$ by~(\ref{uDef}) with $\Psi = \psi_{\text{\Rightscissors}}$. Then
\[\int_{\omega_{\rm low} \leq \left|\omega\right| \leq \omega_{\rm high}}\sum_{\{m\ell : \Lambda \leq \epsilon_{\rm width}^{-1}\omega_{\rm high}^2\}}\left|u\left(-\infty\right)\right|^2\, d\omega \leq B\int_{\Sigma_0}\mathbf{J}^N_{\mu}[\psi]n^{\mu}_{\Sigma_0}.\]
\end{proposition}
\begin{proof}This follows immediately from Theorem 1.9 of~\cite{shlapRot} if we replace $\Sigma_0$ with a hyperboloidal hypersurface $\tilde\Sigma_0$. For any $1 \ll R$ one can easily find a hyperboloidal hypersurface $\tilde\Sigma_0$ which agrees with $\Sigma_0$ on $\{r \leq R\}$ and which lies to the future of $\Sigma_0$. If we choose $R$ large enough so that $T$ is timelike in the region $\{r \geq R\}$, then a $\mathbf{J}^T_{\mu}$ energy estimate will immediately imply
\begin{align*}
\int_{\omega_{\rm low} \leq \left|\omega\right| \leq \omega_{\rm high}}\sum_{\{m\ell : \Lambda \leq \epsilon_{\rm width}^{-1}\omega_{\rm high}^2\}}\left|u\left(-\infty\right)\right|^2\, d\omega &\leq B\int_{\tilde\Sigma_0}\mathbf{J}^N_{\mu}[\psi]n^{\mu}_{\tilde \Sigma_0}\leq B\int_{\Sigma_0}\mathbf{J}^N_{\mu}[\psi]n^{\mu}_{\Sigma_0}.
\end{align*}
\end{proof}
\begin{remark}We observe that the appeal to~\cite{shlapRot} is not necessary in the case of $a \ll M$ or if $\psi$
is only supported on sufficiently high azimuthal frequencies.
\end{remark}

Combining~(\ref{cutoffsDealtWith}) with Proposition~\ref{propShlapRot} finishes
the proof of Proposition~\ref{closedILED}.

\subsection{An inhomogeneous estimate}\label{ILEDinhomo}
In Sections~\ref{higher} and~\ref{continuityargsec} we will need to consider future-integrable solutions $\Psi$ to the \emph{inhomogeneous} wave equation $\Box_{g_{a,M}}\Psi = F$.

Let us first generalise the definition of future-integrability to apply
to general smooth $\Psi$.
\begin{definition}
With
cutoff $\xi\left(\tau\right)$ as in Section~\ref{futIntSec}, we say that
$\Psi:\mathcal{R}_0\to\mathbb R$ is future-integrable if $\xi\Psi$ satisfies
Definitions~\ref{sufficient} and~\ref{sufficient2}.
\end{definition}

\begin{proposition}\label{closedILEDinhomo}Let $\Psi$ be a future
integrable solution of the inhomogeneous wave equation $\Box_{g_{a,M}}\Psi = F$.
Then, for every $\delta > 0$ and $\epsilon > 0$,
\begin{align}\label{inhomo1}
&\int_{\mathcal{H}^+(0,\infty)}\mathbf{J}^N_{\mu}[\Psi]n^{\mu}_{\mathcal{H}^+}+\int_{\mathcal{I}^+}\mathbf{J}^N_{\mu}[\Psi]n^{\mu}_{\mathcal{I}^+}\\
\nonumber&\qquad+
\int_0^{\infty}\int_{\Sigma_{\tau}}\left(\left|\tilde{Z}^*\Psi\right|^2r^{-1-\delta} + \left|\Psi\right|^2r^{-3-\delta} +
\zeta \left|T\Psi\right|^2r^{-1-\delta} + \zeta\left|\nabb\Psi\right|^2r^{-1}\right)\, d\tau
\\ \nonumber &\leq
B\left(\delta\right)\left(\int_{\Sigma_0}\mathbf{J}^N_{\mu}[\Psi]n^{\mu}_{\Sigma_0} + \int_{\Sigma_0}\left|\Psi\right|^2 + \int_0^{\infty}\int_{\Sigma_{\tau}}\left[\epsilon^{-1}r^{1+\delta}\left|F\right|^2 + \epsilon\left(1-\zeta\right)\left(\left|T\Psi\right|^2+\left|\Phi\Psi\right|^2\right)\right]\right).
\end{align}
If $F$ is supported in the region $\{r \leq R\}$, then one may drop the $\int_{\Sigma_0}\left|\Psi\right|^2$ term:
\begin{align}\label{inhomo2}
{\rm L.H.S.\ of\ }(\ref{inhomo1})\,\, \le B\left(\delta,R\right)\left(\int_{\Sigma_0}\mathbf{J}^N_{\mu}[\Psi]n^{\mu}_{\Sigma_0}+ \int_0^{\infty}\int_{\Sigma_{\tau}}\left[\epsilon^{-1} r^{1+\delta}\left|F\right|^2 + \epsilon\left(1-\zeta\right)\left(\left|T\Psi\right|^2+\left|\Phi\Psi\right|^2\right)\right]
\right).
\end{align}
If $F$ is supported in the region $\{r \geq 3M + s_+\}$, then one may drop the $\left(1-\zeta\right)\left(\left|T\Psi\right|^2+\left|\Phi\Psi\right|^2\right)$ term:
\begin{align}\label{inhomo3}
{\rm L.H.S.\ of\ }(\ref{inhomo1})\,\, \le B\left(\delta \right)\left(\int_{\Sigma_0}\mathbf{J}^N_{\mu}[\Psi]n^{\mu}_{\Sigma_0}+ \int_0^{\infty}\int_{\Sigma_{\tau}}r^{1+\delta}\left|F\right|^2
+\int_{\Sigma_0}|\Psi|^2\right).
\end{align}
If $F$ is supported in the region $\{R\ge r \geq 3M + s_+\}$, then one may drop both:
\begin{align}\label{inhomo4}
{\rm L.H.S.\ of\ }(\ref{inhomo1})\,\, \le B\left(\delta,R\right)\left(\int_{\Sigma_0}\mathbf{J}^N_{\mu}[\Psi]n^{\mu}_{\Sigma_0}+ \int_0^{\infty}\int_{\Sigma_{\tau}}r^{1+\delta}\left|F\right|^2
\right).
\end{align}
\end{proposition}
\begin{proof}
Repeating the proof of Proposition~\ref{closedILED} \emph{mutatis mutandis} yields
\begin{align*}
\int_{\mathcal{H}^+(0,\infty)}\mathbf{J}^N_{\mu}[\Psi]n^{\mu}_{\mathcal{H}^+}&+\int_{\mathcal{I}^+}\mathbf{J}^N_{\mu}[\Psi]n^{\mu}_{\mathcal{I}^+}
\\
&+
\int_0^{\infty}\int_{\Sigma_{\tau}}\left(\left|\tilde{Z}^*\Psi\right|^2r^{-1-\delta} + \left|\Psi\right|^2r^{-3-\delta} +
\zeta \left|T\Psi\right|^2r^{-1-\delta} + \zeta\left|\nabb\Psi\right|^2r^{-1}\right)\, d\tau \\
\le&
B\left(\delta\right)\left(\int_{\Sigma_0}\mathbf{J}^N_{\mu}[\Psi]n^{\mu}_{\Sigma_0} + \int_0^{\infty}\int_{\Sigma_{\tau}}\left[\epsilon^{-1}r^{1+\delta}\left|F\right|^2 + \epsilon\left(1-\zeta\right)\left(\left|T\Psi\right|^2+\left|\Phi\Psi\right|^2\right)\right]\right)
 \\ \nonumber &+ \int_0^1\int_{\Sigma_s}\mathbf{J}^N_{\mu}[\Psi_{\text{\Rightscissors}}]n^{\mu}_{\Sigma_s}\, ds.
\end{align*}
We cannot absorb the $\epsilon\left(1-\zeta\right)\left(\left|T\Psi\right|^2+\left|\Phi\Psi\right|^2\right)$ term into the left hand side because of the degeneration due to trapping. The final term on the right hand side arises due to the fact that $\Psi_{\text{\Rightscissors}}$ and $\Psi$ differ when $\dot\xi \neq 0$; since there are no weights in $r$, we cannot hope to absorb this term into the left hand side. However, an application of the fundamental theorem of calculus and Hardy inequalities easily allows us to finish the proof of~(\ref{inhomo1}).

In the case where $F$ is compactly supported in the region $\{r \leq R\}$, the proof of~(\ref{closedILED}) yields
\begin{align*}
\int_{\mathcal{H}^+(0,\infty)}\mathbf{J}^N_{\mu}[\Psi]n^{\mu}_{\mathcal{H}^+}
&+\int_{\mathcal{I}^+}\mathbf{J}^N_{\mu}[\Psi]n^{\mu}_{\mathcal{I}^+}\\
&\qquad+
\int_0^{\infty}\int_{\Sigma_{\tau}}\left(\left|\tilde{Z}^*\Psi\right|^2r^{-1-\delta} + \left|\Psi\right|^2r^{-3-\delta} +
\zeta \left|T\Psi\right|^2r^{-1-\delta} + \zeta\left|\nabb\Psi\right|^2r^{-1}\right)\, d\tau \\
\le&
B\left(\delta\right)\left(\int_{\Sigma_0}\mathbf{J}^N_{\mu}[\Psi]n^{\mu}_{\Sigma_0} + \int_0^{\infty}\int_{\Sigma_{\tau}}\left[\epsilon^{-1}r^{1+\delta}\left|F\right|^2 + \epsilon\left(1-\zeta\right)\left(\left|T\Psi\right|^2+\left|\Phi\Psi\right|^2\right)\right]\right)
\\ \nonumber &+ \int_0^1\int_{\Sigma_s\cap [r_+,R]}\mathbf{J}^N_{\mu}[\Psi_{\text{\Rightscissors}}]n^{\mu}_{\Sigma_s}\, ds.
\end{align*}
The proof of~(\ref{inhomo2}) follows from an application of Hardy inequalities and a finite in time energy estimate to the last term on the right hand side.

The proof of~(\ref{inhomo3}) and $(\ref{inhomo4})$ follow from the same reasoning as above \emph{mutatis mutandis}.
\end{proof}
\begin{remark}After one has proved Theorem~\ref{theResult} it is possible to revisit the inhomogeneous problem and prove a sharper version of this proposition; however, we shall refrain from a systematic treatment of the inhomogeneous problem.
\end{remark}

\section{The higher order statement for future-integrable solutions}\label{higher}

Section~\ref{summation} has established the integrated decay statement~(\ref{protasn1b}) for the class of future-integrable solutions to the wave equation $(\ref{WAVE})$.
Retaining
this restriction, we will  in this section upgrade this statement to the higher order~(\ref{protasn1}).\begin{bigprop}
\label{h.o.s.suff}
Let $M>0$, $0\le a_0<M$.
Let $|a|\le a_0$ and
let $\psi$ be a future integrable solution of $(\ref{WAVE})$ on $\mathcal{R}_0$.
Then, for all $\delta>0$ and all integers $j\ge 1$, the following bound holds
\begin{align}
\nonumber
&\int_{\mathcal{H}^+(0,\infty)}\sum_{1\le i_1+i_2+i_3\le j}
|\nabb^{i_1}T^{i_2}(\tilde Z^*)^{i_3}\psi|^2 + \int_{\mathcal{I}^+}\sum_{1 \leq i \leq j-1}\mathbf{J}^N_{\mu}[N^i\psi]n^{\mu}_{\mathcal{I}^+} \\
\nonumber &\qquad +\int_{\mathcal{R}_0}
r^{-1-\delta}\zeta\sum_{1\le i_1+i_2+i_3\le j}
|\nabb^{i_1}T^{i_2}(\tilde Z^*)^{i_3}\psi|^2\\
\nonumber
&\qquad\qquad +r^{-1-\delta}\sum_{1\le i_1+i_2+i_3\le j-1}
\left(|\nabb^{i_1}T^{i_2}(\tilde Z^*)^{i_3+1}\psi|^2+|\nabb^{i_1}T^{i_2}(Z^*)^{i_3}\psi|^2\right)\\
\le& B\left(\de,j\right)\int_{\Sigma_0} \sum_{0\le i \le j-1}
\mathbf{J}^N_{\mu}[N^i\psi]n^{\mu}_{\Sigma_0}.
\end{align}
\end{bigprop}

\subsection{Elliptic estimates}
Before turning to the proof of Proposition~\ref{h.o.s.suff},
we
will require a few standard elliptic estimates for solutions of the wave equation $(\ref{WAVE})$.

Throughout this section,
$M>0$, $0\le a_0<M$, $|a|\le a_0$ will be fixed, and
 $\chi$ will be a cutoff which is identically $1$ on $[r_+,R_1]$ and identically $0$ on $[R_1+1,\infty)$ for a sufficiently large constant $R_1$ whose $r$-value
 will in particular lie
  outside the ergoregion $\mathcal{S}$,
  and $Y$ will be the red-shift commutation vector field from Section~\ref{Nmult}.

 Lastly, for the reader's benefit we recall the following pointwise relation which follows immediately from the algebraic properties of the energy-momentum tensor:
\[
\mathbf{J}^N_{\mu}[\Psi]n^{\mu}_{\Sigma_{\tau}} \geq b\left(\left(T\Psi\right)^2 +
(\tilde{Z}^*\Psi)^2+ \left|\nabb\Psi\right|^2\right).
\]

All the lemmas below refer to solutions $\psi$ of the wave equation $(\ref{WAVE})$
as in the reduction of Section~\ref{WPosed}.

\begin{lemma}\label{Nelliptic}
For $\psi$ as above, we have
\begin{align*}
&\int_{\Sigma_{\tau}}\sum_{1\le i_1+i_2+i_3\le 2}
|\nabb^{i_1}T^{i_2}(\tilde Z^*)^{i_3}\psi|^2 \leq B\int_{\Sigma_{\tau}}\left(\mathbf{J}^N_{\mu}[N\psi] + \mathbf{J}^N_{\mu}[\psi]\right)n^{\mu}_{\Sigma_{\tau}}.
\end{align*}
\end{lemma}
\begin{proof}This is standard: Let $\hat{\Sigma}_{\tau}$ be an extension of $\Sigma_{\tau}$ from $r \in [r_+,\infty)$ to $r \in [r_+-\epsilon,\infty)$. By a standard extension lemma, one may extend $\psi$ to a function $\hat\psi$ on $\hat{\Sigma}_{\tau}$ in such a way that $\left\vert\left\vert \Delta_{\hat{\Sigma}_{\tau}}\hat\psi\right\vert\right\vert_{L^2\left(\hat{\Sigma}_{\tau}\right)} \leq B\left\vert\left\vert \Delta_{\Sigma_{\tau}}\psi\right\vert\right\vert_{L^2\left(\Sigma_{\tau}\right)}$. The lemma then follows from a local elliptic estimate.
\end{proof}

\begin{lemma}\label{ellipticEstimates}
For $\psi$ as above, we have
\begin{align*}
&\int_{\Sigma_{\tau}}\sum_{1\le i_1+i_2+i_3\le 2}
|\nabb^{i_1}T^{i_2}(\tilde Z^*)^{i_3}\psi|^2 \le B\int_{\Sigma_{\tau}}\left(\mathbf{J}^N_{\mu}[T\psi]n^{\mu}_{\Sigma_{\tau}} +\mathbf{J}^N_{\mu}[\chi\Phi\psi]n^{\mu}_{\Sigma_{\tau}}+\mathbf{J}^N_{\mu}[Y\psi]n^{\mu}_{\Sigma_{\tau}} +\mathbf{J}^N_{\mu}[\psi]n^{\mu}_{\Sigma_{\tau}}\right).
\end{align*}

\end{lemma}
\begin{proof}This is standard: One uses elliptic estimates on spheres near the horizon and an elliptic estimate on $\Sigma_{\tau} \cap \{r \geq r_0\}$ away from the horizon. The key point is that $T$ and $\Phi$ span a timelike direction away from the horizon, and $Y$, $T$ and $\Phi$ span a timelike direction near the horizon.
\end{proof}

\begin{lemma}\label{horizonSphere}
For $\psi$ as above, we have
\begin{align*}
&\int_{\mathcal{H}^+(0,\infty)}\sum_{1\le i_1+i_2+i_3\le 2}
|\nabb^{i_1}T^{i_2}(\tilde Z^*)^{i_3}\psi|^2
\\ \nonumber &\qquad \le B\int_{\mathcal{H}^+(0,\infty)}\left(\mathbf{J}^N_{\mu}[T\psi]n^{\mu}_{\mathcal{H}^+} +\mathbf{J}^N_{\mu}[\chi\Phi\psi]n^{\mu}_{\mathcal{H}^+}+\mathbf{J}^N_{\mu}[Y\psi]n^{\mu}_{\mathcal{H}^+} +\mathbf{J}^N_{\mu}[\psi]n^{\mu}_{\mathcal{H}^+}\right).
\end{align*}
\end{lemma}
\begin{proof}This follows from elliptic estimates on spheres.
\end{proof}
One can, of course, localise Lemma~\ref{ellipticEstimates}:
\begin{lemma}\label{ellipticEstimatesLocal}
For $\psi$ as above, then for any $R < \infty$, we have
\begin{align*}
&\int_{\Sigma_{\tau}\cap [r_+,R]}\sum_{1\le i_1+i_2+i_3\le 2}
|\nabb^{i_1}T^{i_2}(\tilde Z^*)^{i_3}\psi|^2
\\ \nonumber &\qquad \le B\int_{\Sigma_{\tau}\cap [r_+,R+1]}\left(\mathbf{J}^N_{\mu}[T\psi]n^{\mu}_{\Sigma_{\tau}} +\mathbf{J}^N_{\mu}[\chi\Phi\psi]n^{\mu}_{\Sigma_{\tau}}+\mathbf{J}^N_{\mu}[Y\psi]n^{\mu}_{\Sigma_{\tau}} +\mathbf{J}^N_{\mu}[\psi]n^{\mu}_{\Sigma_{\tau}}\right).
\end{align*}
\end{lemma}

The next four lemmas give control
of the solution without including a
$Y$-commuted energy on the right hand side.
\begin{lemma}\label{ellipticOffHorizon}
For $\psi$ as above, then for any $r_0 > r_+$,
\begin{align*}
&\int_{\Sigma_{\tau} \cap \{r \geq r_0\}}\mathbf{J}^N_{\mu}[Y\psi]n^{\mu}_{\Sigma_{\tau}} \leq
B(r_0)\int_{\Sigma_{\tau}}\left(\mathbf{J}^N_{\mu}[T\psi]n^{\mu}_{\Sigma_{\tau}} + \mathbf{J}^N_{\mu}[\chi\Phi\psi]n^{\mu}_{\Sigma_{\tau}} +\mathbf{J}^N_{\mu}[\psi]n^{\mu}_{\Sigma_{\tau}}\right).
\end{align*}
\end{lemma}
\begin{proof}
This follows from an elliptic estimate away from the horizon using the fact
that the span of $T$ and $\Phi$ is timelike.
The straightforward proof is omitted.
\end{proof}
\begin{lemma}\label{ellipticOffHorizonLocal}
For $\psi$ as above, then for any $r_+ < r_0 < r_1 < \infty$, $\delta > 0$,
\begin{align*}
&\int_{\Sigma_{\tau} \cap \{[r_0,r_1]\}}\mathbf{J}^N_{\mu}[Y\psi]n^{\mu}_{\Sigma_{\tau}} \leq
B(r_0,r_1,\delta)\int_{\Sigma_{\tau}\cap [r_0-\delta,r_1+\delta]}\left(\mathbf{J}^N_{\mu}[T\psi]n^{\mu}_{\Sigma_{\tau}} + \mathbf{J}^N_{\mu}[\chi\Phi\psi]n^{\mu}_{\Sigma_{\tau}} +\mathbf{J}^N_{\mu}[\psi]n^{\mu}_{\Sigma_{\tau}}\right).
\end{align*}
\end{lemma}
\begin{proof}This straightforward proof is omitted.
\end{proof}
\begin{lemma}\label{ellipticOffHorizonDecay}
For $\psi$ as above, then for any, $r_+ < r_0 < \infty$ and $\delta_1,\delta_2 > 0$,
\begin{align*}
&\int_{\Sigma_{\tau}\cap [r_0,\infty)}\sum_{1\le i_1+i_2+i_3\le 2}
r^{-1-\delta_1}|\nabb^{i_1}T^{i_2}(\tilde Z^*)^{i_3}\psi|^2
\\ \nonumber &\qquad\le B(r_0,\delta_2)\int_{\Sigma_{\tau}\cap [r_0-\delta_2,\infty)}r^{-1-\delta_1}\left(\mathbf{J}^N_{\mu}[T\psi]n^{\mu}_{\Sigma_{\tau}} +\mathbf{J}^N_{\mu}[\chi\Phi\psi]n^{\mu}_{\Sigma_{\tau}} +\mathbf{J}^N_{\mu}[\psi]n^{\mu}_{\Sigma_{\tau}}\right).
\end{align*}
\end{lemma}
\begin{proof}This straightforward proof is omitted.
\end{proof}
\begin{lemma}\label{ellipticOffHorizonLargeDecay}
For $\psi$ as above, then for any $2M + 1 \leq  r_0 < \infty$ and $\delta_1,\delta_2 > 0$,\begin{align*}
&\int_{\Sigma_{\tau} \cap \{[r_0,\infty)\}}r^{-1-\delta_1}\mathbf{J}^N_{\mu}[\chi\Phi\psi]n^{\mu}_{\Sigma_{\tau}}\leq B(\delta_2)\int_{\Sigma_{\tau}\cap [r_0-\delta_2,\infty)}\left(r^{-1-\delta_1}\mathbf{J}^N_{\mu}[T\psi]n^{\mu}_{\Sigma_{\tau}} +r^{-1-\delta_1}\mathbf{J}^N_{\mu}[\psi]n^{\mu}_{\Sigma_{\tau}}\right).
\end{align*}
\end{lemma}
\begin{proof}It suffices to remark that the region $[2M+1,\infty)$ lies outside the ergoregion (see~(\ref{ergoregionS})),
and apply elliptic estimates as before.
\end{proof}
The following lemma will be used in conjunction with red-shift estimate of Proposition~\ref{ftrs} and the commutation formula for $Y$ given in Proposition~\ref{commuprop}.
\begin{lemma}\label{thetaNearHorizon}
For $\psi$ as above, then for all $\epsilon > 0$, we may find a $r_0 > r_+$
depending on $\epsilon$ such that
\begin{align*}
&\int_{\Sigma_{\tau}\cap\{r \leq r_0\}}\left|\nabb^2\psi\right|^2 \leq
B\int_{\Sigma_{\tau}\cap [r_+,r_0)}\left(\mathbf{J}^N_{\mu}[T\psi]n^{\mu}_{\Sigma_{\tau}}+ \mathbf{J}^N_{\mu}[\chi\Phi\psi]n^{\mu}_{\Sigma_{\tau}} + \mathbf{J}^N_{\mu}[\psi]n^{\mu}_{\Sigma_{\tau}}\right) + \epsilon\int_{\Sigma_{\tau}}\mathbf{J}^N_{\mu}[Y\psi]n^{\mu}_{\Sigma_{\tau}}.
\end{align*}
\end{lemma}
\begin{proof}Since $Y$ is null on $\mathcal{H}^+$, on $\mathcal{H}^+$ there is no $Y^2$ term in the wave equation. In particular, the second derivative terms in the wave equation which contain a $Y$ derivative may be controlled by $\mathbf{J}^N_{\mu}[T\psi]n^{\mu}_{\Sigma_{\tau}}$ and $\mathbf{J}^N_{\mu}[\Phi\psi]n^{\mu}_{\Sigma_{\tau}}$. Given these observations, the lemma easily follows from elliptic estimates on spheres.
\end{proof}
We will also need some integrated in time estimates:
\begin{lemma}\label{iledEllipticT}Let $\psi$
be as above, and let $R < \infty$. Then
\begin{align*}
&\int_0^{\infty}\int_{\Sigma_{\tau}\cap [3M+1,R)}\mathbf{J}^N_{\mu}[\psi]n^{\mu}_{\Sigma_{\tau}} \leq
B\int_0^{\infty}\int_{\Sigma_{\tau}\cap [3M,R+1)}\left(\left|T\psi\right|^2 + \left|\partial_{r^*}\psi\right|^2+\left|\psi\right|^2\right) + B\int_{\Sigma_0}\mathbf{J}^N_{\mu}[\psi]n^{\mu}_{\Sigma_0}.
\end{align*}
\end{lemma}
\begin{proof}This is standard: One writes the wave equation as
\begin{align*}
&g^{tt}\partial_t^2\psi - \frac{4Mar}{\rho^2\Delta}\partial_{t,\phi}^2\psi = \frac{\Delta-a^2\sin^2\theta}{\Delta\rho^2\sin^2\theta}\partial_{\phi}^2\psi -
\frac{r^2+a^2}{\Delta\rho^2}\partial_{r^*}\left(\left(r^2+a^2\right)\partial_{r^*}\psi\right) - \frac{1}{\rho^{2}\sin\theta}\partial_{\theta}\left(\sin\theta\partial_{\theta}\psi\right),
\end{align*}
multiplies by a cutoff to $r \in [3M,R+1)$, multiplies by $\psi$, integrates by parts,
remembers the comments concerning the volume form in Section~\ref{usefulcomps},
and applies Hardy inequalities when appropriate.
\end{proof}
\begin{lemma}\label{iledEllipticTPhi}
For $\psi$ as above, then for any
 $r_0 > r_+$, $R < \infty$ and $\delta > 0$,
\begin{align*}
&\int_0^{\infty}\int_{\Sigma_{\tau}\cap [r_0+\delta,R-\delta]}\mathbf{J}^N_{\mu}[\psi]n^{\mu}_{\Sigma_{\tau}}
\\ \nonumber &\qquad \le B(r_0, \delta)\int_0^{\infty}\int_{\Sigma_{\tau}\cap [r_0,R])}\left(\left|T\psi\right|^2 +\left|\partial_{r^*}\psi\right|^2+\left|\Phi\psi\right|^2 + \left|\psi\right|^2\right) + B\int_{\Sigma_0}\mathbf{J}^N_{\mu}[\psi]n^{\mu}_{\Sigma_0}.
\end{align*}
\end{lemma}
\begin{proof}This is proven in the same fashion as Lemma~\ref{iledEllipticT}.
\end{proof}

\subsection{Proof of Proposition~\ref{h.o.s.suff}}
Now we will prove Proposition~\ref{h.o.s.suff}.
\begin{proof}
Let $a_0$, $M$, $a$ and $\psi$ be as in the statement of the proposition.
Let us be given moreover $\delta>0$.
We will consider the case $j=2$. The case of $j \geq 3$ follows by induction in a straightforward fashion.

First, we commute the wave equation with $T$ and obtain $\Box_g\left(T\psi\right) = 0$. Since $T\psi$ is future-integrable, the integrated energy decay statement~(\ref{protasn1b}) holds for $T\psi$.

Next, we commute with $\chi\Phi$, where $\chi$ is a cutoff which is identically $1$ on $[r_+,R_1]$ and identically $0$ on $[R_1+1,\infty)$. We obtain $\Box_g\left(\chi\Phi\psi\right) = \left(\Box_g\chi\right)\Phi\psi + 2\nabla^{\mu}\chi\nabla_{\mu}\Phi\psi$. Now, Lemma~\ref{ellipticOffHorizonLargeDecay} implies
\begin{align}\label{errorPhi}
\int_0^{\infty}\int_{\Sigma_{\tau}}\left|\Box_g\left(\chi\Phi\psi\right)\right|^2 &\leq B\int_0^{\infty}\int_{\Sigma_{\tau}\cap [R_1,\infty)}r^{-1-\delta}\left(\mathbf{J}^N_{\mu}[T\psi]n^{\mu}_{\Sigma_{\tau}} + \mathbf{J}^N_{\mu}[\psi]n^{\mu}_{\Sigma_{\tau}}\right)
\\ \nonumber &\leq B\int_{\Sigma_0}\left(\mathbf{J}^N_{\mu}[T\psi]n^{\mu}_{\Sigma_0} + \mathbf{J}^N_{\mu}[\psi]n^{\mu}_{\Sigma_0}\right).
\end{align}

In the last inequality, we used that the integrated energy decay statement holds for $T\psi$. Now, statement~(\ref{inhomo4}) of Proposition~\ref{closedILEDinhomo} implies that the integrated energy decay statement holds for $\chi\Phi\psi$ as long as we add $B\int_{\Sigma_0}\left(\mathbf{J}^N_{\mu}[T\psi]n^{\mu}_{\Sigma_0} + \mathbf{J}^N_{\mu}[\psi]n^{\mu}_{\Sigma_0}\right)$
to the right hand side of the inequality.

Finally, we turn to commutation with $Y$. We recall Proposition~\ref{commuprop} which implies
\begin{equation}\label{goodSign}
\Box_g(Y\Psi)=\kappa_1 Y^2\Psi + \sum_{|{\bf m}|\le 2, m_4\le 1} c_{{\bf m}}
E_1^{m_1}E_2^{m_2}L^{m_3}Y^{m_4}\Psi
\end{equation}
where $\kappa_1>0$ is proportional to the surface gravity. Next, for any $\tilde r \leq r_{\rm red}$, we apply the energy estimate associated to the red-shift vector field $N$, in between the hypersurfaces $\Sigma_0$ and $\Sigma_{\tau}$:
\begin{align}\label{redImplyHigher2}
&\int_{\Sigma_{\tau}}\mathbf{J}^N_{\mu}[Y\psi]n^{\mu}_{\Sigma_{\tau}} + \int_{\mathcal{H}^+(0,\tau)}\mathbf{J}^N_{\mu}[Y\psi]n^{\mu}_{\mathcal{H}^+} + \int_{0}^{\tau}\int_{\Sigma_s\cap\{r \leq \tilde r\}}\mathbf{J}^N_{\mu}[Y\psi]n^{\mu}_{\Sigma_s}\, ds  \\
\nonumber&\qquad \le B\int_{0}^{\tau}\int_{\Sigma_s}\left(1_{r \in \left[\tilde r,\tilde r+\delta\right]}\mathbf{J}^N_{\mu}[Y\psi]n^{\mu}_{\Sigma_{s}} + \mathcal{E}^N[Y\psi]\right)\, ds + \int_{\Sigma_{s_1}}\mathbf{J}^N_{\mu}[Y\psi]n^{\mu}_{\Sigma_{s_1}}.
\end{align}

For any $\epsilon > 0$, we may choose $\tilde r$ close enough to $r_+$, $\delta$ small enough so that $\tilde r + 2\delta < 3M - s^-$ and use~(\ref{goodSign}), Lemma~\ref{ellipticOffHorizonLocal}, Lemma~\ref{thetaNearHorizon}, Lemma~\ref{ellipticEstimatesLocal} and the fact that $N = K+Y$, to show that
\begin{align}\label{errorControl2}
&\int_{\Sigma_s}\mathcal{E}^N[Y\psi]
\\ \nonumber &\,\,\leq \epsilon\int_{\Sigma_s\cap\{r\leq \tilde r\}}\mathbf{J}^N_{\mu}[Y\psi]n^{\mu}_{\Sigma_s} + B\epsilon^{-1}\int_{\Sigma_s\cap [r_+,\tilde r+2\delta]}\left(\mathbf{J}^N_{\mu}[T\psi]n^{\mu}_{\Sigma_s} + \mathbf{J}^N_{\mu}[\chi\Phi\psi]n^{\mu}_{\Sigma_s}+\mathbf{J}^N_{\mu}[\psi]n^{\mu}_{\Sigma_s}\right)
\\ \nonumber &\,\,\leq \epsilon\int_{\Sigma_s\cap\{r\leq \tilde r\}}\mathbf{J}^N_{\mu}[Y\psi]n^{\mu}_{\Sigma_s}+B\epsilon^{-1}\int_{\Sigma_0}\left(\mathbf{J}^N_{\mu}[T\psi]n^{\mu}_{\Sigma_0}+\mathbf{J}^N_{\mu}[\chi\Phi\psi]n^{\mu}_{\Sigma_0} + \mathbf{J}^N_{\mu}[\psi]n^{\mu}_{\Sigma_0}\right).
\end{align}

Combining~(\ref{errorControl2}) and~(\ref{redImplyHigher2}) implies
\begin{align}
&\int_{\Sigma_{\tau}}\mathbf{J}^N_{\mu}[Y\psi]n^{\mu}_{\Sigma_{\tau}} + \int_{\mathcal{H}^+(0,\tau)}\mathbf{J}^N_{\mu}[Y\psi]n^{\mu}_{\mathcal{H}^+}+\int_0^{\tau}\int_{\Sigma_s\cap\{r\leq \tilde r\}}\mathbf{J}^N_{\mu}[Y\psi]n^{\mu}_{\Sigma_s}\, ds  \\ \nonumber &\qquad \le B\int_{\Sigma_0}\left(\mathbf{J}^N_{\mu}[T\psi]n^{\mu}_{\Sigma_0}+\mathbf{J}^N_{\mu}[\chi\Phi\psi]n^{\mu}_{\Sigma_0} + \mathbf{J}^N_{\mu}[\psi]n^{\mu}_{\Sigma_0}\right).
\end{align}

Now, the proof concludes with applications of
Lemmas~\ref{horizonSphere},
\ref{ellipticEstimates}, \ref{ellipticEstimatesLocal}, \ref{iledEllipticTPhi},
\ref{ellipticOffHorizonDecay}, \ref{Nelliptic}, \ref{iledEllipticT} and (for the null infinity
$\mathcal{I}^+$ term) straightforward $\mathbf{J}^T$ energy estimates in a large $r$ region.
\end{proof}

\section{The continuity argument}
\label{continuityargsec}

In this section, we will prove
\begin{bigprop}\label{allIntegrable}
Let $M>0$ and $|a|<M$.
All solutions $\psi$ to the wave equation~(\ref{WAVE})
on $\mathcal{R}_0$ as in the
reduction of Section~\ref{WPosed} (i.e.~arising from smooth, compactly supported initial data on $\Sigma_0$)
are future-integrable.
\end{bigprop}

\subsection{The reduction to fixed azimuthal frequency}
\label{reduxfixed}

We begin with the following easy but important Lemma.
\begin{lemma}\label{azimuthal}It suffices to prove Proposition~\ref{allIntegrable} for solutions
$\psi$ to
$(\ref{WAVE})$ assumed moreover to be
supported on an arbitrary but fixed azimuthal frequency $m$.
\end{lemma}
\begin{proof}Let $\psi$ be a solution to the wave equation arising from smooth, compactly supported initial data, and suppose we have established Proposition~\ref{allIntegrable} for solutions supported on any fixed azimuthal frequency. We may expand $\psi$ into its azimuthal modes: $\psi = \sum_{m \in \mathbb{Z}}\psi_m$. Since each $\psi_m$ is future-integrable,
it follows by Proposition~\ref{closedILED} that
 the integrated energy decay statements~(\ref{protasn1b}) and~(\ref{protasn1}) hold for $\psi_m$. Orthogonality immediately implies that~(\ref{protasn1b}) and~(\ref{protasn1}) hold for $\psi$. Finally, we simply observe that the fundamental theorem of calculus implies that
 \begin{align*}
&\sup_{r \in [r_+,A]}\int_0^{\infty}\int_{\mathbb{S}^2}\sum_{1\le i_1+i_2+i_3\le j}
|\nabb^{i_1}T^{i_2}(\tilde Z^*)^{i_3}\psi|^2\sin\theta\, dt\, d\theta\, d\phi \\
& \,\,\,\leq
B\left(\int_{\mathcal{H}^+(0,\infty)}\sum_{1\le i_1+i_2+i_3\le j}
|\nabb^{i_1}T^{i_2}(\tilde Z^*)^{i_3}\psi|^2 +\int_0^{\infty}\int_{\Sigma_s\cap [r_+,A]}\sum_{1\le i_1+i_2+i_3\le j+1}
|\nabb^{i_1}T^{i_2}(\tilde Z^*)^{i_3}\psi|^2\right).
\end{align*}
\end{proof}

Lemma~\ref{azimuthal} thus implies that Proposition~\ref{allIntegrable} follows
from the following proposition:
\begin{proposition}\label{allIntegrablem}
Let $M>0$, $|a|<M$ and $m\in \mathbb Z$.
Let $\psi$ be a solution to the wave equation as in the reduction of Section~\ref{WPosed}
such that moreover, $\psi$ is supported only on the azimuthal frequency $m$.
Then $\psi$ is sufficiently integrable.
\end{proposition}

The following Lemma will be very useful for the proof of Proposition~\ref{allIntegrablem}.
\begin{lemma}\label{cutoffILED}
Let $M$, $a$, $m$, and $\psi$ be as in the statement of Proposition~\ref{allIntegrablem}.
Then, for every $\tau \geq 0$ and $\delta > 0$,
\begin{align*}
\int_{\mathcal{H}^+(0,\tau)}\mathbf{J}^N_{\mu}[\psi]n^{\mu}_{\mathcal{H}^+} + &\int_0^{\tau}\int_{\Sigma_s}\Big(r^{-1}(1-\eta_{[\left(1+\sqrt{2}\right)M,3M+s^+]})(1-3M/r)^2 \big(|\nabb\psi|^2+r^{-\delta}\left|T\psi\right|^2\big)\\
&\qquad\qquad+r^{-1-\delta}\left|\tilde Z^*\psi\right|^2+ r^{-3-\delta}\left|\psi\right|^2\Big)\\
\nonumber
&\leq B(\de,m)\left(\int_{\Sigma_0}\mathbf{J}^N_{\mu}[\psi]n^{\mu}_{\Sigma_0} + \int_{\Sigma_{\tau}}\mathbf{J}^N_{\mu}[\psi]n^{\mu}_{\Sigma_{\tau}}\right).
\end{align*}

\end{lemma}
\begin{proof}One modifies the cutoff $\xi$ from Section~\ref{cutoffSec}; now let $\xi$ be identically $1$ in between $\Sigma_1$ and $\Sigma_{\tau-1}$ and identically $0$ to the past of $\Sigma_0$ and the future of $\Sigma_{\tau}$. Then, one may easily check that the arguments of Section~\ref{summation} will imply the lemma. Note that we can write $\eta_{[\left(1+\sqrt{2}\right)M,3M+s^+]}$ instead of $\eta_{[3M-s^-,3M+s^+]}$ because Lemma~\ref{aziTrap} tells us that for fixed $m$ and large $\Lambda$ the trapped set is contained in $[\left(1+\sqrt{2}\right)M,\infty)$.
\end{proof}
\begin{remark}Let us emphasise that we do \underline{not} assume that $\psi$ is future-integrable. This is why we must have the term $\int_{\Sigma_{\tau}}\mathbf{J}^N_{\mu}[\psi]n^{\mu}_{\Sigma_{\tau}}$ on the right hand side.
\end{remark}
\begin{remark}
As observed in Remark~\ref{ergoTrap}, we see that (for a fixed azimuthal frequency) trapping and the ergoregion are non-overlapping! This will be extremely useful in what follows.
\end{remark}

We will also need higher order versions of Lemma~\ref{cutoffILED}.
\begin{lemma}\label{cutoffILEDhigher}
Let $M$, $a$, $m$, and $\psi$ be as in the statement of Proposition~\ref{allIntegrablem}.
 Then, for every $\tau \geq 0$, $j \geq 1$ and $\delta > 0$,
\begin{align*}
&\int_{\mathcal{H}^+(0,\tau)}\sum_{1\le i_1+i_2+i_3\le j}
|\nabb^{i_1}T^{i_2}(\tilde Z^*)^{i_3}\psi|^2 \\
&+\int_0^{\tau}\int_{\Sigma_s}
r^{-1-\delta}(1-\eta_{[\left(1+\sqrt{2}\right)M,3M+s^+]})(1-3M/r)^2\sum_{1\le i_1+i_2+i_3\le j}
|\nabb^{i_1}T^{i_2}(\tilde Z^*)^{i_3}\psi|^2\\
\nonumber
&\qquad\qquad\qquad+r^{-1-\delta}\sum_{1\le i_1+i_2+i_3\le j-1}
\left(|\nabb^{i_1}T^{i_2}(\tilde Z^*)^{i_3+1}\psi|^2+|\nabb^{i_1}T^{i_2}(Z^*)^{i_3}\psi|^2\right)\\
&\le B(\de,j,m)\left(\int_{\Sigma_0} \sum_{0\le i\le j-1}\mathbf{J}^N_{\mu}[N^i\psi]n^{\mu}_{\Sigma_0}+\int_{\Sigma_\tau} \sum_{0\le i \le j-1}
\mathbf{J}^N_{\mu}[N^i\psi]n^{\mu}_{\Sigma_{\tau}}\right).
\end{align*}

\end{lemma}
\begin{proof}This follows from repeating the arguments of Section~\ref{higher}.
\end{proof}

We have the following easy corollary.
\begin{corollary}\label{boundEquivSuff}
Let $M$, $a$, $m$, and $\psi$ be as in the statement of Proposition~\ref{allIntegrablem}.
 Then, $\psi$ is future-integrable if
\begin{equation}\label{aHigherBound}
\sup_{\tau \geq 0}\int_{\Sigma_\tau} \sum_{1\le i_1+i_2+i_3\le j}
|\nabb^{i_1}T^{i_2}(\tilde Z^*)^{i_3}\psi|^2 < \infty\quad\forall\ j\geq 1.
\end{equation}
\end{corollary}
\begin{proof}
As in the proof of Lemma~\ref{azimuthal} we need only observe that
\begin{align*}
&\sup_{r \in [r_+,A]}\int_0^{\infty}\int_{\mathbb{S}^2}\sum_{1\le i_1+i_2+i_3\le j}
|\nabb^{i_1}T^{i_2}(\tilde Z^*)^{i_3}\psi|^2\sin\theta\, dt\, d\theta\, d\phi\\
& \leq
B(j)\left(\int_{\mathcal{H}^+(0,\infty)}\sum_{1\le i_1+i_2+i_3\le j}
|\nabb^{i_1}T^{i_2}(\tilde Z^*)^{i_3}\psi|^2 +\int_0^{\infty}\int_{\Sigma_s\cap [r_+,A]}\sum_{1\le i_1+i_2+i_3\le j+1}
|\nabb^{i_1}T^{i_2}(\tilde Z^*)^{i_3}\psi|^2\right).
\end{align*}
\end{proof}

The proof of Proposition~\ref{allIntegrablem} will be a continuity argument in the rotation
parameter $a$ of the black hole. That is, fix $M > 0$, and define for each $m\in\mathbb Z$,
the set
\[\mathcal{A}_m := \{\left|a\right| \in [0,M) : \text{ the statement~(\ref{aHigherBound}) holds for } g = g_{a,M}\}.\]
We shall prove that $\mathcal{A}_m = [0,M)$ by showing that it is non-empty, open and closed.
Proposition~\ref{allIntegrablem} then follows by Corollary~\ref{boundEquivSuff}.

We note first
\begin{proposition}For all $m\in \mathbb Z$, the set $\mathcal{A}_m$ is non-empty.
\end{proposition}
\begin{proof}When $a = 0$, it is well known that~(\ref{aHigherBound}) holds (even without the restriction to a fixed azimuthal frequency). One may find the (relatively short) argument in the lecture notes~\cite{jnotes}. Thus $0\in \mathcal{A}_m$.
\end{proof}

We now turn to openness.
\subsection{Openness}
\label{opennesssec}

In this section, we will prove
\begin{proposition}\label{open}
For all $m\in\mathbb Z$, the set $\mathcal{A}_m$ is open. That is, suppose $\mathring{a} \in \mathcal{A}_m$. Then there exists $\epsilon > 0$ such that $\left|a-\mathring{a}\right| < \epsilon$ implies $a \in \mathcal{A}_m$.
\end{proposition}
The proof proper will be given in Section~\ref{interpolatingSection} below. We begin with some preliminaries.

\subsubsection{Gaining derivatives}
We start with a definition.
\begin{definition}\label{aVectorField}
Let $|a|<M$
and let $\epsilon_0 > 0$ be from Lemma~\ref{horizonTimelike}. Let $\alpha(r)$ be a function such that $V := T + \alpha(r)\Phi$ is a smooth vector field timelike in $\mathcal{R}$ which satisfies
\begin{align*}
V& = T + \frac{a}{2Mr_+}\Phi, \text{ when } r \in [r_+,r_+ + \epsilon_0/2],\\
V& = T + \frac{2Mar}{\left(r^2+a^2\right)^2}\Phi, \text{ when }r \in \left[r_++\epsilon_0,\frac{M\left(7+\sqrt{2}\right)}{4}\right],\\
V& = T, \text{ when } r \geq \frac{M\left(3+\sqrt{2}\right)}{2}.
\end{align*}
\end{definition}
\begin{remark}\label{VkillTrap}
Note that  $2M < \frac{M\left(3+\sqrt{2}\right)}{2} < M\left(1+\sqrt{2}\right)$. In particular, $V$ is Killing in the region where trapping occurs in Lemmas~\ref{cutoffILED} and~\ref{cutoffILEDhigher}.
\end{remark}

The following Lemma can be thought of as a derivative gaining converse to Lemma~\ref{cutoffILED}.
\begin{lemma}\label{prepLemm}
Let $|a|\le a_0<M$, let $m\in\mathbb Z$,
and let
 $\psi$ be a solution the wave equation~(\ref{WAVE})
 as in the reduction of Section~\ref{WPosed} which is furthermore supported on the fixed azimuthal frequency $m$. Then
\begin{align*}
\int_{\Sigma_{\tau}}\mathbf{J}^N_{\mu}[\psi]n^{\mu}_{\Sigma_{\tau}}
&\leq B(m)\left(\int_0^{\tau}\int_{\Sigma_{s} \cap \left\{r \leq \frac{M\left(3+\sqrt{2}\right)}{2}\right\}}\left|\Phi\psi\right|^2\, ds + \int_{\Sigma_0}\mathbf{J}^N_{\mu}[\psi]n^{\mu}_{\Sigma_0}\right)\quad \forall\ \tau \geq 0\\
&\leq B(m)\left(\int_0^{\tau}\int_{\Sigma_{s} \cap \left\{r \leq \frac{M\left(3+\sqrt{2}\right)}{2}\right\}}\left|\psi\right|^2\, ds + \int_{\Sigma_0}\mathbf{J}^N_{\mu}[\psi]n^{\mu}_{\Sigma_0}\right)\quad \forall\ \tau \geq 0.
\end{align*}
\end{lemma}
\begin{proof}We apply the energy identity associated to the vector field $V$:
\begin{align}\label{vEstimate}
\int_{\Sigma_{\tau}}\mathbf{J}^V_{\mu}[\psi]n^{\mu}_{\Sigma_{\tau}} \leq B\int_0^{\tau}\int_{\Sigma_s}\left|\mathbf{K}^V[\psi]\right|\, ds + \int_{\Sigma_0}\mathbf{J}^V_{\mu}[\psi]n^{\mu}_{\Sigma_0}.
\end{align}
In view of the fact that $T$ and $\Phi$ are Killing vector fields, we have
\[\mathbf{K}^V[\psi] = \mathbf{K}^{(\alpha\Phi)}[\psi] = 2\mathbf{T}(\nabla\alpha,\Phi)[\psi] = 2\frac{\Delta}{\rho^2}\frac{d\alpha}{dr}\mathbf{T}(Z,\Phi)[\psi] = 2\frac{\Delta}{\rho^2}\frac{d\alpha}{dr}\text{Re}(Z\psi\overline{\Phi\psi}). \]
Recall that $\frac{d\alpha}{dr}$ is supported away from the horizon, so that $Z$ is a regular vector field when the expression above is non-zero. We may conclude that
\begin{align}\label{spacetimeControl}
\left|\mathbf{K}^V[\psi]\right| \leq B1_{\text{supp}(\frac{d\alpha}{dr})}\left(\epsilon|\partial_r\psi|^2 + \epsilon^{-1}|\Phi\psi|^2\right).
\end{align}
Lemma~\ref{cutoffILED} implies
\begin{align}\label{coupleILED}
\int_0^{\tau}\int_{\Sigma_s}1_{\text{supp}(\frac{d\alpha}{dr})}|\partial_r\psi|^2\, ds \leq B(m)\left(\int_{\Sigma_{\tau}}\mathbf{J}_{\mu}^N[\psi]n^{\mu}_{\Sigma_{\tau}} + \int_{\Sigma_0}\mathbf{J}_{\mu}^N[\psi]n^{\mu}_{\Sigma_0}\right).
\end{align}
Combining $B(m)\epsilon$ times estimate~(\ref{coupleILED}) with estimates~(\ref{vEstimate}) and~(\ref{spacetimeControl}) implies
\begin{align}\label{almostThere}
\int_{\Sigma_{\tau}}\mathbf{J}^V_{\mu}[\psi]n^{\mu}_{\Sigma_{\tau}} \leq  B(m)\left(\epsilon^{-1}\int_0^{\tau}\int_{\Sigma_{s} \cap \left\{r \leq \frac{M\left(3+\sqrt{2}\right)}{2}\right\}}\left|\Phi\psi\right|^2\, ds + \int_{\Sigma_0}\mathbf{J}^N_{\mu}[\psi]n^{\mu}_{\Sigma_{1}} + \epsilon\int_{\Sigma_{\tau}}\mathbf{J}^N_{\mu}[\psi]n^{\mu}_{\Sigma_{\tau}}\right).
\end{align}
In order to finish the Lemma we apply the standard red-shift argument (see the lecture notes~\cite{jnotes}). Set
\[A := \sup_{0 \leq s\leq \tau}\epsilon^{-1}\int_0^s\int_{\Sigma_{s'} \cap \left\{r \leq \frac{M\left(3+\sqrt{2}\right)}{2}\right\}}\left|\Phi\psi\right|^2\, ds' + \int_{\Sigma_0}\mathbf{J}^N_{\mu}[\psi]n^{\mu}_{\Sigma_0} + \epsilon\int_{\Sigma_s}\mathbf{J}^N_{\mu}[\psi]n^{\mu}_{\Sigma_s}.\]
For every $0 \leq s_1 < s_2 \leq \tau$ and $\tilde r$ sufficiently close to $r_+$, the red-shift estimate~(\ref{ftrs}) implies
\begin{align}\label{redImply}
&\int_{\Sigma_{s_2}}\mathbf{J}^N_{\mu}[\psi]n^{\mu}_{\Sigma_{s_2}} + \int_{s_1}^{s_2}\int_{\Sigma_s\cap\{r \leq \tilde r\}}\mathbf{J}^N_{\mu}[\psi]n^{\mu}_{\Sigma_s}\, ds  \\ \nonumber&\qquad \le
B(m)\int_{s_1}^{s_2}\int_{\Sigma_s\cap\{\tilde r \leq r \leq \tilde r + 1\}}\mathbf{J}^N_{\mu}[\psi]n^{\mu}_{\Sigma_{s}}\, ds + \int_{\Sigma_{s_1}}\mathbf{J}^N_{\mu}[\psi]n^{\mu}_{\Sigma_{s_1}}.
\end{align}
Now, we observe that in the region $\{\tilde r \leq r\}$ the quantities $\mathbf{J}^N_{\mu}[\psi]n^{\mu}_{\Sigma_{s}}$ and $\mathbf{J}^V_{\mu}[\psi]n^{\mu}_{\Sigma_s}$ are comparable. Thus, adding $\int_{s_1}^{s_2}\int_{\Sigma_s}\mathbf{J}^V_{\mu}[\psi]n^{\mu}_{\Sigma_s}\, ds$ to both sides of~(\ref{redImply}) implies
\begin{align}\label{redImply2}
\int_{\Sigma_{s_2}}\mathbf{J}^N_{\mu}[\psi]n^{\mu}_{\Sigma_{s_2}} + b(m)\int_{s_1}^{s_2}\int_{\Sigma_s}\mathbf{J}^N_{\mu}[\psi]n^{\mu}_{\Sigma_s}\, ds \leq B(m)\int_{s_1}^{s_2}\int_{\Sigma_s}\mathbf{J}^V_{\mu}[\psi]n^{\mu}_{\Sigma_{s}}\, ds + \int_{\Sigma_{s_1}}\mathbf{J}^N_{\mu}[\psi]n^{\mu}_{\Sigma_{s_1}}.
\end{align}
Now, estimate~(\ref{almostThere}) (with $\tau$ on the left hand side replaced by $s$) implies
\begin{align}\label{redImply3}
\int_{\Sigma_{s_2}}\mathbf{J}^N_{\mu}[\psi]n^{\mu}_{\Sigma_{s_2}} + b(m)\int_{s_1}^{s_2}\int_{\Sigma_s}\mathbf{J}^N_{\mu}[\psi]n^{\mu}_{\Sigma_s}\, ds \leq B(m)A(s_2-s_1) + \int_{\Sigma_{s_1}}\mathbf{J}^N_{\mu}[\psi]n^{\mu}_{\Sigma_{s_1}}.
\end{align}
Let
\[f(s) := \int_{\Sigma_s}\mathbf{J}^N_{\mu}[\psi]n^{\mu}_{\Sigma_s}.\]
We may rewrite equation~(\ref{redImply3}) as
\[f(s_2) + b\int_{s_1}^{s_2}f(s)\, ds \leq B(m)A(s_2-s_1) + f(s_1)\qquad\text{ for every }0 \leq s_1 < s_2 \leq \tau.\]
An easy argument shows that this implies
\[f(s) \leq B(m)\left(A + f(0)\right).\]
Writing this out gives
\begin{align*}
&\sup_{0\leq s\leq \tau}\int_{\Sigma_s}\mathbf{J}^N_{\mu}[\psi]n^{\mu}_{\Sigma_s} \leq B(m)\left(\epsilon^{-1}\int_0^{\tau}\int_{\Sigma_{s} \cap \left\{r \leq \frac{M\left(3+\sqrt{2}\right)}{2}\right\}}\left|\Phi\psi\right|^2\, ds + \int_{\Sigma_0}\mathbf{J}^N_{\mu}[\psi]n^{\mu}_{\Sigma_{0}} + \epsilon\sup_{0\leq s\leq \tau}\int_{\Sigma_s}\mathbf{J}^N_{\mu}[\psi]n^{\mu}_{\Sigma_s}\right).
\end{align*}
We conclude the proof by taking $\epsilon$ sufficiently small.
\end{proof}
\begin{remark}Observe that the proof does \underline{not} exploit the fact that the ergoregion and trapping are disjoint; indeed, even without the restriction to fixed $m$, we could have proved the first line of the proposition, with a constant $B$ not depending on $m$, simply by exploiting the fact that the $\partial_r$   derivative  does not degenerate in the integrated local energy decay.  Rather, the point is that for fixed $m$, the presence of the ergoregion is only a low-frequency obstruction to boundedness.
\end{remark}
\begin{remark}Note that the proof crucially uses
that we can upgrade a degenerate energy boundedness statement to a non-degenerate energy boundedness statement without a full integrated local energy decay.
\end{remark}

Next, we play Lemmas~\ref{prepLemm} and Lemma~\ref{cutoffILEDhigher} off each other. We end up being able to gain an \emph{arbitrary} number of derivatives.
\begin{lemma}\label{higherPrepLemm}
Let $|a|\le a_0<M$, let $m\in\mathbb Z$
and let
 $\psi$ be a solution the wave equation~(\ref{WAVE})
 as in the reduction of Section~\ref{WPosed}, which is furthermore supported on the fixed azimuthal frequency $m$. Then, for every $j \geq 1$,
\begin{align*}
&\int_{\Sigma_{\tau}}\sum_{1\le i_1+i_2+i_3\le j}
|\nabb^{i_1}T^{i_2}(\tilde Z^*)^{i_3}\psi|^2
\\ \nonumber &\le B(j,m)\left(\int_0^{\tau}\int_{\Sigma_{s} \cap \left\{r \leq \frac{M\left(3+\sqrt{2}\right)}{2}\right\}}\left|\psi\right|^2\, ds + \int_{\Sigma_0}\sum_{1\le i_1+i_2+i_3\le j}
|\nabb^{i_1}T^{i_2}(\tilde Z^*)^{i_3}\psi|^2\right)\quad \forall\ \tau \geq 0.
\end{align*}
\end{lemma}
\begin{proof}We first consider the case $j= 2$. We begin by commuting the wave equation with $T$ and applying Lemma~\ref{prepLemm}. We obtain
\begin{align}\label{tCommute}
&\int_{\Sigma_{\tau}}\mathbf{J}^N_{\mu}[T\psi]n^{\mu}_{\Sigma_{\tau}} \leq B(m)\left(\int_0^{\tau}\int_{\Sigma_{s} \cap \left\{r \leq \frac{M\left(3+\sqrt{2}\right)}{2}\right\}}\left|T\psi\right|^2\, ds + \int_{\Sigma_0}\mathbf{J}^N_{\mu}[T\psi]n^{\mu}_{\Sigma_{1}}\right).
\end{align}
Now commute the wave equation with the red-shift commutation vector field $Y$. On the horizon $\mathcal{H}^+$ we will have
\begin{align}\label{yCommute}
\Box_g\left(Y\psi\right) = \kappa_1 Y^2\psi + \sum_{i+j+k \in [0,2], k\leq 1}c_{ijk}T^i\partial_{\theta}^jY^k\psi,
\end{align}
where $\kappa_1 > 0$ is proportional to the surface gravity of $\mathcal{H}^+$.

Next, we apply
Proposition~\ref{ftrs},
 the energy estimate associated to the red-shift multiplier $N$, to $\Psi=Y\psi$. For every $1 \leq s_1 < s_2 \leq \tau$, we obtain
\begin{align}\label{redImplyHigher}
&\int_{\Sigma_{s_2}}\mathbf{J}^N_{\mu}[Y\psi]n^{\mu}_{\Sigma_{s_2}} + \int_{s_1}^{s_2}\int_{\Sigma_s\cap\{r \leq \tilde r\}}\mathbf{J}^N_{\mu}[Y\psi]n^{\mu}_{\Sigma_s}\, ds  \\
\nonumber&\qquad\le B\int_{s_1}^{s_2}\int_{\Sigma_s}\left(1_{r \in \left[\tilde r,\tilde r+1\right]}\mathbf{J}^N_{\mu}[Y\psi]n^{\mu}_{\Sigma_{s}} + \mathcal{E}^N[Y\psi]\right)\, ds + \int_{\Sigma_{s_1}}\mathbf{J}^N_{\mu}[Y\psi]n^{\mu}_{\Sigma_{s_1}}.
\end{align}
For any $\epsilon > 0$ we may choose $\tilde r$ close enough to $r_+$, and use~(\ref{yCommute}), Lemma~\ref{ellipticOffHorizon}, Lemma~\ref{thetaNearHorizon} and the fact that $N = K+Y$, to show that
\begin{align}\label{errorControl}
\int_{\Sigma_s}\mathcal{E}^N[Y\psi] \leq \epsilon\int_{\Sigma_s\cap\{r\leq \tilde r\}}\mathbf{J}^N_{\mu}[Y\psi]n^{\mu}_{\Sigma_s} + B\epsilon^{-1}\int_{\Sigma_s}\left(\mathbf{J}^N_{\mu}[T\psi]n^{\mu}_{\Sigma_s} + \mathbf{J}^N_{\mu}[\psi]n^{\mu}_{\Sigma_s}\right).
\end{align}
Adding $\int_{s_1}^{s_2}\int_{\Sigma_s}\left(\mathbf{J}^N_{\mu}[T\psi]n^{\mu}_{\Sigma_s} + \mathbf{J}^N_{\mu}[\psi]n^{\mu}_{\Sigma_s}\right)\, ds$ to both sides, using Lemma~\ref{ellipticOffHorizon} and using~(\ref{errorControl}) implies
\begin{align}
&\int_{\Sigma_{s_2}}\mathbf{J}^N_{\mu}[Y\psi]n^{\mu}_{\Sigma_{s_2}} + b\int_{s_1}^{s_2}\int_{\Sigma_s}\mathbf{J}^N_{\mu}[Y\psi]n^{\mu}_{\Sigma_s}\, ds  \\
\nonumber&\qquad \le  B\int_{s_1}^{s_2}\int_{\Sigma_s}\left(\mathbf{J}^N_{\mu}[T\psi]n^{\mu}_{\Sigma_{s}} + \mathbf{J}^N_{\mu}[\psi]n^{\mu}_{\Sigma_s}\right)\, ds + \int_{\Sigma_{s_1}}\mathbf{J}^N_{\mu}[Y\psi]n^{\mu}_{\Sigma_{s_1}}.
\end{align}
Now we use~(\ref{tCommute}), Lemma~\ref{prepLemm} and the same argument which occurs at the end of the proof of Lemma~\ref{prepLemm} to conclude
\begin{align}\label{redConsequence}
\int_{\Sigma_{\tau}}\mathbf{J}^N_{\mu}[Y\psi]n^{\mu}_{\Sigma_{\tau}} \leq &B(m)\int_0^{\tau}\int_{\Sigma_{s} \cap \left\{r \leq \frac{M\left(3+\sqrt{2}\right)}{2}\right\}}\left(\left|T\psi\right|^2 + \left|\psi\right|^2\right)\, ds  \\ \nonumber &+B(m)\int_{\Sigma_0}\left(\mathbf{J}^N_{\mu}[T\psi]n^{\mu}_{\Sigma_{1}} + \mathbf{J}^N_{\mu}[Y\psi]n^{\mu}_{\Sigma_{1}} + \mathbf{J}^N_{\mu}[\psi]n^{\mu}_{\Sigma_{1}}\right).
\end{align}
Next, Lemma~\ref{ellipticEstimates} allows us to combine~(\ref{redConsequence}) and~(\ref{tCommute}) to get
\begin{align}\label{veryClose}
&\int_{\Sigma_{\tau}}\sum_{1\le i_1+i_2+i_3\le 2}
|\nabb^{i_1}T^{i_2}(\tilde Z^*)^{i_3}\psi|^2
\\ \nonumber &\le B(m)\left(\int_0^{\tau}\int_{\Sigma_{s} \cap \left\{r \leq \frac{M\left(3+\sqrt{2}\right)}{2}\right\}}\left(\left|T\psi\right|^2 + \left|\psi\right|^2\right)\, ds + B(m)\int_{\Sigma_0}\sum_{1\le i_1+i_2+i_3\le 2}
|\nabb^{i_1}T^{i_2}(\tilde Z^*)^{i_3}\psi|^2\right).
\end{align}
It remains to remove the spacetime integral of $\left|T\psi\right|^2$ from the right hand side; however, we observe the following immediate consequence of Lemmas~\ref{cutoffILED} and~\ref{prepLemm}:
\begin{align}\label{derivativeGain}
\int_0^{\tau}\int_{\Sigma_s\cap\left\{r \leq \frac{M\left(3+\sqrt{2}\right)}{2}\right\}}\left|T\psi\right|^2 &\leq B(m)\left(\int_{\Sigma_{\tau}}\mathbf{J}^N_{\mu}[\psi]n^{\mu}_{\Sigma_{\tau}} + \int_{\Sigma_0}\mathbf{J}^N_{\mu}[\psi]n^{\mu}_{\Sigma_0}\right)\\
\nonumber &\leq B(m)\left(\int_0^{\tau}\int_{\Sigma_s\cap\left\{r \leq \frac{M\left(3+\sqrt{2}\right)}{2}\right\}}\left|\psi\right|^2 + \int_{\Sigma_0}\mathbf{J}^N_{\mu}[\psi]n^{\mu}_{\Sigma_0}\right).
\end{align}
Combining~(\ref{veryClose}) and~(\ref{derivativeGain}) completes the proof for $j = 2$.\footnote{Observe that this final trick relies on the fact that trapping and the ergoregion are
disjoint in physical space when the azimuthal frequency is fixed.}

The case for general $j \geq 2$ follows by induction in a straightforward fashion.
\end{proof}
The following straightforward corollary will be useful in what follows.
\begin{corollary}\label{combineLemmas}
Let $|a|\le a_0<M$, let $m\in\mathbb Z$,
and let
 $\psi$ be a solution the wave equation~(\ref{WAVE})
 as in the reduction of Section~\ref{WPosed} which is furthermore supported on the fixed azimuthal frequency $m$.  Then, for all $\delta>0$, $j \geq 1$,
\begin{align*}
&\sup_{\tau' \leq \tau}\int_{\Sigma_{\tau'}}\sum_{1\le i_1+i_2+i_3\le j}
|\nabb^{i_1}T^{i_2}(\tilde Z^*)^{i_3}\psi|^2  \\
&+\int_0^{\tau}\int_{\Sigma_{\tau}}\Big(\sum_{1\le i_1+i_2+i_3\le j-1}
|\nabb^{i_1}T^{i_2}(\tilde Z^*)^{i_3}\psi|^2 + \sum_{1\le i_1+i_2+i_3\le j-1}
|\nabb^{i_1}T^{i_2}(\tilde Z^*)^{i_3+1}\psi|^2
\\ &\nonumber \qquad\qquad\qquad+r^{-2-\delta}\left|\psi\right|^2 + 1_{\left[r_+,\left(1+\sqrt{2}\right)M\right]}\sum_{1\le i_1+i_2+i_3\le j}
|\nabb^{i_1}T^{i_2}(\tilde Z^*)^{i_3}\psi|^2\Big)\cdot r^{-1-\delta}  \\
&\le B(\de,j,m)\left(\int_0^{\tau}\int_{\Sigma_{s} \cap \left\{r \leq \frac{M\left(3+\sqrt{2}\right)}{2}\right\}}\left|\psi\right|^2\, ds + \int_{\Sigma_0}\sum_{1\le i_1+i_2+i_3\le j}
|\nabb^{i_1}T^{i_2}(\tilde Z^*)^{i_3}\psi|^2\right)
\end{align*}
\end{corollary}
\begin{proof}This is a simple combination of Lemmas~\ref{higherPrepLemm} and~\ref{cutoffILEDhigher}.
\end{proof}

\subsubsection{An interpolating metric and the proof of Proposition~\ref{open}}\label{interpolatingSection}
We now turn to the proof of Proposition~\ref{open}.

\begin{proof}
Recall that we have fixed $M>0$.
Let us suppose $m\in \mathbb Z$ is fixed and $\mathring{a}\in \mathcal{A}_m$.
Let us choose $a_0$ such that $|\mathring{a}|<a_0<M$.
We will find an $\epsilon>0$ with
$|\mathring{a}|+\epsilon \le a_0$
 such that
\begin{equation}
\label{1stcondonep}
|a-\mathring{a}|<\epsilon
\end{equation}
implies $a\in \mathcal{A}_m$.

Let $a$ satisfy $(\ref{1stcondonep})$, for $\epsilon$ to be determined, and
$\psi$ be a solution of the wave equation $(\ref{WAVE})$ on $g_{M,a}$,
as in the reduction of Section~\ref{WPosed},
such that moreover, $\psi$ is supported on the fixed azimuthal frequency $m$.

Recall that Corollary~\ref{boundEquivSuff} implies that solutions $\widetilde\psi$
to $\Box_{g_{\mathring{a},M}}\widetilde\psi = 0$ which are supported on a fixed azimuthal frequency are known to be future-integrable. In order to exploit this ``black box'' knowledge about $g_{\mathring{a},M}$ it is useful to introduce a metric $\tilde{g}_{\tau}$ which interpolates between $g_{\mathring{a},M}$ and $g_{a,M}$. Fortunately, we will \emph{not} need to fine tune $\tilde{g}_{\tau}$.

\begin{definition}\label{interpolatingMetric}Pick $\tau \geq 1$. Recalling that the hypersurfaces $\Sigma_s$ are independent of $a$, let $\chi_{\tau}$ be a cut-off which is $0$ in the past of $\Sigma_{\tau-\delta_0}$ and identically one in the future of $\Sigma_{\tau}$ for some sufficiently small $\delta_0 > 0$. We define the interpolating metric $\tilde{g}_{\tau}$ by
\[\tilde{g}_{\tau} \doteq \chi_{\tau}g_{\mathring{a},M} + \left(1-\chi_{\tau}\right)g_{a,M}.\]
\end{definition}

If $\epsilon$ in $(\ref{1stcondonep})$ is assumed sufficiently small,
then $\tilde{g}_{\tau}$ defines a Lorentzian metric on $\mathcal{R}$.
\begin{remark}\label{stillKilling}Note that it is easy to see that $\Phi$ is a Killing vector field for the metric $\tilde{g}_{\tau}$, and that for all $\tau\ge 0$, $\Sigma_\tau$ is a past Cauchy
hypersurface for $\mathcal{R}_0\setminus \mathcal{R}_{(0,\tau)}$ with respect
to $\tilde{g}_\tau$.
\end{remark}

Corresponding to our interpolating metric, we will need an ``interpolating'' solution to the wave equation.
\begin{definition}\label{interpolatingSolution}
Let $\psi$ be our solution to $\Box_{g_{a,M}}\psi = 0$ defined above.
We define the interpolating solution $\tilde{\psi}_{\tau}$ by solving $\Box_{\tilde{g}_{\tau}}\tilde{\psi}_{\tau} = 0$ with the same initial data as
$\psi$ on $\Sigma_0$.
\end{definition}

Of course, $\tilde{\psi}_{\tau}$ will exactly equal $\psi$ in the past of $\Sigma_{\tau-\delta_0}$, and in the future of $\Sigma_{\tau}$, $\tilde{\psi}_{\tau}$ is a solution to
$\Box_{g_{\mathring{a},M}}\tilde{\psi}_{\tau} = 0$. Furthermore, since $\Phi$ is a Killing vector field for $\tilde{g}_{\tau}$, the interpolating solution $\tilde{\psi}_{\tau}$ will be supported on the same azimuthal frequency $m$ as the original solution $\psi$. Hence, by
the assumption $\mathring{a}\in\mathcal{A}_m$ and Corollary~\ref{boundEquivSuff},
it follows
that $\tilde{\psi}_{\tau}$ is future-integrable with respect to $\mathring{a}$.

We write
\begin{align}\label{errorCont}
\Box_{g_{\mathring{a},M}}\tilde{\psi}_{\tau} = \left(\Box_{g_{\mathring{a},M}} - \Box_{\tilde{g}_{\tau}}\right)\tilde{\psi}_{\tau},
\end{align}
and observe
\begin{align}
 r^{1+\delta}\left|\left(\Box_{g_{\mathring{a},M}} - \Box_{\tilde{g}_{\tau}}\right)\tilde{\psi}_{\tau}\right|^2 \leq B\left(\delta_0^{-1}\right)\left|a-\mathring{a}\right|^2r^{-2}\sum_{1 \leq i_1+i_2+i_3\leq 2}\left|\nabb^{i_1}T^{i_2}(\tilde Z^*)^{i_3}\tilde{\psi}_{\tau}\right|^2.
\end{align}
In this statement, and in what follows, metric defined quantities (such as $\nabb$ and $\mathbf{J}^N_{\mu}[\psi]n^{\mu}$) will refer to $g_{\mathring{a},M}$. Now we apply the $g_{\mathring{a},M}$ integrated local energy estimate to $\tilde{\psi}$.

Keeping in mind that~(\ref{errorCont}) is supported in the past of $\Sigma_{\tau}$,
Proposition~\ref{closedILEDinhomo} implies
\begin{align}\label{ILEDinterpolate}
&\int_0^{\tau-\delta_0}\int_{\Sigma_s\cap \left\{r \leq M\left(1+\sqrt{2}\right)
\right\}}\left(\mathbf{J}^N_{\mu}[\psi]n^{\mu}_{\Sigma_s} + \left|\psi\right|^2\right)\, ds  \\
\nonumber&\qquad\leq B(\delta_0,m)\left|a-\mathring{a}\right|\int_0^{\tau}\int_{\Sigma_s}\sum_{1 \leq i_1+i_2+i_3\leq 2}r^{-2}\left|\nabb^{i_1}T^{i_2}(\tilde Z^*)^{i_3}\tilde{\psi}_{\tau}\right|^2\, ds
\\ \nonumber &\qquad +B\left(\delta_0,m\right)\left|a-\mathring{a}\right|\int_{\tau}^{\infty}\int_{\Sigma_s\cap [r_+,\left(1+\sqrt{2}\right)M]}\left[\left|T\tilde\psi_{\tau}\right|^2 + \left|\tilde\psi_{\tau}\right|^2\right]+ B(m)\int_{\Sigma_0}\left[\mathbf{J}^N_{\mu}[\psi]n^{\mu}_{\Sigma_0} + \left|\psi\right|^2\right].
\end{align}
For $\delta_0$ sufficiently small (and then fixing the value of $\delta_0$), finite in time energy estimates (and an easy domain of dependence argument) imply
\begin{align}\label{moveDelta}
\int_0^{\tau}\int_{\Sigma_s\cap \left\{r \leq \frac{M\left(3+\sqrt{2}\right)}{2}\right\}}\left(\mathbf{J}^N_{\mu}[\psi]n^{\mu}_{\Sigma_s} + \left|\psi\right|^2\right)\, ds \leq
B\int_0^{\tau-\delta_0}\int_{\Sigma_s\cap \left\{r \leq M\left(1+\sqrt{2}\right)\right\}}\left(\mathbf{J}^N_{\mu}[\psi]n^{\mu}_{\Sigma_s} + \left|\psi\right|^2\right)\, ds.
\end{align}
Since $\tilde{\psi}_{\tau}$ is equal to $\psi$ along $\Sigma_{\tau-\delta_0}$, finite in time energy estimates for $\Box_{\tilde{g}_{\tau}}$ imply
\begin{align}\label{compareDelta}
\int_{\tau-\delta_0}^{\tau}\int_{\Sigma_s}\sum_{1 \leq i_1+i_2+i_3\leq 2}r^{-2}\left|\nabb^{i_1}T^{i_2}(\tilde Z^*)^{i_3}\tilde{\psi}_{\tau}\right|^2\, ds \leq B\int_{\Sigma_{\tau-\delta_0}}\sum_{1 \leq i_1+i_2+i_3\leq 2}\left|\nabb^{i_1}T^{i_2}(\tilde Z^*)^{i_3}\psi\right|^2.
\end{align}
Similarly,
\begin{align}\label{compareDelta2}
\int_{\tau-\delta_0}^{\tau}\int_{\Sigma_s}\sum_{1 \leq i_1+i_2+i_3\leq 2}r^{-2}\left|\nabb^{i_1}T^{i_2}(\tilde Z^*)^{i_3}\psi\right|^2\, ds \leq B\int_{\Sigma_{\tau-\delta_0}}\sum_{1 \leq i_1+i_2+i_3\leq 2}\left|\nabb^{i_1}T^{i_2}(\tilde Z^*)^{i_3}\psi\right|^2.
\end{align}
Finally, Proposition~\ref{h.o.s.suff}, the fact that $\tilde\psi_{\tau}$ is future integrable, and finite in time energy inequalities imply
\begin{align}\label{compareDelta3}
\int_{\tau}^{\infty}\int_{\Sigma_s\cap [r_+,\left(1+\sqrt{2}\right)M]}\left[\left|T\tilde\psi_{\tau}\right|^2 + \left|\tilde\psi_{\tau}\right|^2\right] \leq B\int_{\Sigma_{\tau-\delta_0}}\sum_{1 \leq i_1+i_2+i_3\leq 2}\left|\nabb^{i_1}T^{i_2}(\tilde Z^*)^{i_3}\psi\right|^2.
\end{align}
Combining~(\ref{ILEDinterpolate}),~(\ref{moveDelta}),~(\ref{compareDelta}),~(\ref{compareDelta2}) and~(\ref{compareDelta3}) gives
\begin{align}\label{continuityEst}
&\int_0^{\tau}\int_{\Sigma_s\cap \left\{r \leq \frac{M\left(3+\sqrt{2}\right)}{2}
\right\}}\left(\mathbf{J}^N_{\mu}[\psi]n^{\mu}_{\Sigma_s} + \left|\psi\right|^2\right)\, ds  \\ \nonumber\le & B(m)\left|a-\mathring{a}\right|\left(\int_0^{\tau}\int_{\Sigma_s}\sum_{1 \leq i_1+i_2+i_3\leq 2}r^{-2}\left|\nabb^{i_1}T^{i_2}(\tilde Z^*)^{i_3}\psi\right|^2\, ds + \int_{\Sigma_{\tau-\delta_0}}\sum_{1 \leq i_1+i_2+i_3\leq 2}\left|\nabb^{i_1}T^{i_2}(\tilde Z^*)^{i_3}\tilde{\psi}_{\tau}\right|^2\right)
\\ \nonumber &\qquad+B(m)\int_{\Sigma_0}\left[\mathbf{J}^N_{\mu}[\psi]n^{\mu}_{\Sigma_0} + \left|\psi\right|^2\right].
\end{align}
Now combine Corollary~\ref{combineLemmas} and~(\ref{continuityEst}):
\begin{align*}
&\sup_{\tau' \leq \tau}\int_{\Sigma_{\tau'}}\sum_{1\le i_1+i_2+i_3\le j}
|\nabb^{i_1}T^{i_2}(\tilde Z^*)^{i_3}\psi|^2\\
&+\int_0^{\tau}\int_{\Sigma_{\tau}}\Big(\sum_{1\le i_1+i_2+i_3\le j-1}
|\nabb^{i_1}T^{i_2}(\tilde Z^*)^{i_3}\psi|^2 + \sum_{1\le i_1+i_2+i_3\le j-1}
|\nabb^{i_1}T^{i_2}(\tilde Z^*)^{i_3+1}\psi|^2
\\ &\nonumber\qquad\qquad\qquad + r^{-2-\delta}\left|\psi\right|^2 + 1_{\left[r_+,\left(1+\sqrt{2}\right)M\right]}\sum_{1\le i_1+i_2+i_3\le j}
|\nabb^{i_1}T^{i_2}(\tilde Z^*)^{i_3}\psi|^2\Big)r^{-1-\delta}  \\
\le &B(\de,j,m)\left(\int_0^{\tau}\int_{\Sigma_{s} \cap \left\{r \leq \frac{M\left(3+\sqrt{2}\right)}{2}\right\}}\left|\psi\right|^2\, ds + \int_{\Sigma_0}\sum_{1\le i_1+i_2+i_3\le j}
|\nabb^{i_1}T^{i_2}(\tilde Z^*)^{i_3}\psi|^2\right)  \\
\le & B(\de,j,m)\left|a-\mathring{a}\right|\left(\int_0^{\tau}\int_{\Sigma_s}r^{-2}\sum_{1\le i_1+i_2+i_3\le 2}
|\nabb^{i_1}T^{i_2}(\tilde Z^*)^{i_3}\psi|^2\, ds + \int_{\Sigma_{\tau-\delta_0}}\sum_{1\le i_1+i_2+i_3\le 2}
|\nabb^{i_1}T^{i_2}(\tilde Z^*)^{i_3}\psi|^2\right)
\\ \nonumber &\qquad+B(\de,j,m)\int_{\Sigma_0}\sum_{0\le i_1+i_2+i_3\le j}
|\nabb^{i_1}T^{i_2}(\tilde Z^*)^{i_3}\psi|^2.
\end{align*}
As long as $j \geq 3$, we may take $\epsilon$ in $(\ref{1stcondonep})$ sufficiently small, absorb the $\left|a-\mathring{a}\right|$ term on the left hand side and conclude
\[\sup_{\tau' \leq \tau}\int_{\Sigma_{\tau'}}\sum_{1\le i_1+i_2+i_3\le j}
|\nabb^{i_1}T^{i_2}(\tilde Z^*)^{i_3}\psi|^2 \leq B(j,m)\int_{\Sigma_0}\sum_{0\le i_1+i_2+i_3\le j}
|\nabb^{i_1}T^{i_2}(\tilde Z^*)^{i_3}\psi|^2 < \infty\quad \forall\ j \geq 3.\]
Lastly, we observe that the final inequality clearly remains true if we define $\nabb^{i_1}$ with respect to $g_{a,M}$ instead of $g_{\mathring{a},M}$.
\end{proof}

\subsection{Closedness}
\label{closednesssec}
To finish the proof, it remains to show
\begin{proposition}\label{closed}The set $\mathcal{A}_m$ is closed in $[0,M)$. That is, suppose we have a sequence $\{a_k\}_{k=1}^{\infty}$ with $a_k \in \mathcal{A}_m$ and $a_k \to a \in (-M,M)$. Then $a \in \mathcal{A}_m$.
\end{proposition}
\begin{proof}

Let $\psi$ be a solution the wave equation $\Box_{g_{a,M}}\psi = 0$ arising from smooth, compactly supported initial data which is furthermore supported on a fixed azimuthal frequency $m$.

We may choose $a_0<M$ such that $|a|<a_0$ and without loss of generality, we may
assume that $|a_k|\le a_0$ for all $k$.

We define a sequence of functions $\psi_k$ by solving $\Box_{g_{a_k,M}}\psi_k = 0$ with the same initial data as $\psi$. Using the future-integrability of $\psi_k$, for every $j \geq 1$ we will have
\begin{align}\label{fromFutureInt}
&\int_0^{\infty}\int_{\Sigma_s}
r^{-1-\delta}(1-\eta_{[\left(1+\sqrt{2}\right)M,3M+s^+]})(1-3M/r)^2\sum_{1\le i_1+i_2+i_3\le j}
|\nabb^{i_1}T^{i_2}(\tilde Z^*)^{i_3}\psi_k|^2\\
\nonumber
&\qquad\qquad+r^{-1-\delta}\sum_{1\le i_1+i_2+i_3\le j-1}
\left(|\nabb^{i_1}T^{i_2}(\tilde Z^*)^{i_3+1}\psi_k|^2+|\nabb^{i_1}T^{i_2}(Z^*)^{i_3}\psi_k|^2\right)\\
\nonumber& \, \, \le B(\de,j,m)\int_{\Sigma_0} \sum_{0\le i\le j-1}\mathbf{J}^N_{\mu}[N^i\psi_k]n^{\mu}_{\Sigma_0}.
\end{align}
Now, using the fact that the region $\{r \leq M\left(1+\sqrt{2}\right)\}$ contains
the ergoregion $\mathcal{S}$,
combining~(\ref{fromFutureInt}) and an $N$-based energy estimate yields
\begin{align}\label{futIntEnergyBound}
\sup_{\tau \geq 0}\int_{\Sigma_{\tau}}\mathbf{J}^N_{\mu}[\psi_k]n^{\mu}_{\Sigma_{\tau}} &\leq B(m)\int_{\Sigma_0}\mathbf{J}^N_{\mu}[\psi_k]n^{\mu}_{\Sigma_0} + B(m)\int_0^{\infty}\int_{\Sigma_s\cap\{r \leq M\left(1+\sqrt{2}\right)\}}\mathbf{J}^N_{\mu}[\psi_k]n^{\mu}_{\Sigma_s}
\\ \nonumber &\leq B(m)\int_{\Sigma_0}\mathbf{J}^N_{\mu}[\psi_k]n^{\mu}_{\Sigma_0}.
\end{align}
It remains to upgrade~(\ref{futIntEnergyBound}) to its higher order version in the (by now) standard fashion. First we commute with $T$ and obtain
\begin{align}\label{futIntEnergyBoundT}
\sup_{\tau \geq 0}\int_{\Sigma_{\tau}}\mathbf{J}^N_{\mu}[T\psi_k]n^{\mu}_{\Sigma_{\tau}} &\leq B(m)\int_{\Sigma_0}\mathbf{J}^N_{\mu}[T\psi_k]n^{\mu}_{\Sigma_0}.
\end{align}
Next, we commute with $Y$ and follow the same argument as in the proof of Lemma~\ref{higherPrepLemm}. We obtain
\begin{align}\label{futIntEnergyBoundY}
&\sup_{\tau \geq 0}\int_{\Sigma_{\tau}}\mathbf{J}^N_{\mu}[Y\psi_k]n^{\mu}_{\Sigma_{\tau}} \leq B(m)\int_{\Sigma_0}\left(\mathbf{J}^N_{\mu}[\psi_k]n^{\mu}_{\Sigma_0} + \mathbf{J}^N_{\mu}[T\psi_k]n^{\mu}_{\Sigma_0} + \mathbf{J}^N_{\mu}[Y\psi_k]n^{\mu}_{\Sigma_0}\right).
\end{align}
Just as in the proof of Lemma~\ref{higherPrepLemm}, elliptic estimates imply
\begin{align*}
&\sup_{\tau \geq 0}\int_{\Sigma_{\tau}}\sum_{1\le i_1+i_2+i_3\le 2}
|\nabb^{i_1}T^{i_2}(\tilde Z^*)^{i_3}\psi_k|^2\leq B(m)\int_{\Sigma_0}\left(\mathbf{J}^N_{\mu}[\psi_k]n^{\mu}_{\Sigma_0} + \mathbf{J}^N_{\mu}[T\psi_k]n^{\mu}_{\Sigma_0} + \mathbf{J}^N_{\mu}[Y\psi_k]n^{\mu}_{\Sigma_0}\right).
\end{align*}
Finally, an easy induction argument will imply
\begin{align*}
&\sup_{\tau \geq 0}\int_{\Sigma_{\tau}}\sum_{1\le i_1+i_2+i_3\le j}
|\nabb^{i_1}T^{i_2}(\tilde Z^*)^{i_3}\psi_k|^2 \leq B(j,m)\int_{\Sigma_0}\sum_{1\le i_1+i_2+i_3\le j}
|\nabb^{i_1}T^{i_2}(\tilde Z^*)^{i_3}\psi_k|^2.
\end{align*}
Then, we conclude the proof by observing
\begin{align*}
\int_{\Sigma_{\tau}}\sum_{1\le i_1+i_2+i_3\le j}
|\nabb^{i_1}T^{i_2}(\tilde Z^*)^{i_3}\psi|^2 &= \lim_{k\to\infty}\int_{\Sigma_{\tau}}\sum_{1\le i_1+i_2+i_3\le j}
|\nabb^{i_1}T^{i_2}(\tilde Z^*)^{i_3}\psi_k|^2
\\ \nonumber &\leq B(j,m)\int_{\Sigma_0}\sum_{1\le i_1+i_2+i_3\le j}
|\nabb^{i_1}T^{i_2}(\tilde Z^*)^{i_3}\psi|^2.
\end{align*}
The first equality uses the well-posedness of the wave equation, the smooth dependence of $g_{a,M}$ on $a$ (see Lemma~\ref{lemmaSmoothDependence}) and the fact that $\psi_k$ and $\psi$ have the same initial data along $\Sigma_0$.
\end{proof}

\section{The precise integrated local energy decay statement}\label{precise}
In this section we give will a more precise form of the integrated local energy decay statement. So as to produce a purely physical space estimate, we employed
in the proof of Proposition~\ref{closedILED} a physical space cutoff $\zeta$ (see~(\ref{degenerationfunc})) in the integrated energy decay statement~(\ref{protasn1b}). It is clear from the statement of Theorem~\ref{phaseSpaceILED} in Section~\ref{freqLocEst} that this throws away information (cf.~the discussion in Section~\ref{trappingParam}).

In order to succinctly state the microlocally precise form of integrated local energy decay, we introduce the following notation: For a sufficiently integrable function $\Psi$ on $\mathcal{R}$, we define
\begin{equation}
\label{Ptrapnot}
\mathcal{P}_{\rm trap}\left[\Psi\right] \doteq \frac{1}{\sqrt{2\pi}}\int_{-\infty}^\infty\sum_{m\ell} \left|\zeta-(1-\zeta)r^{-1}r_{\rm trap}\right|e^{-i\omega t} \Psi^{(a\omega)}_{m\ell}(r)S_{m\ell}(a\omega,\cos\theta)e^{im\phi} d\omega,
\end{equation}
where $r_{\rm trap}=r_{\rm trap}(\omega, m, \Lambda_{m\ell})$ is defined in Theorem~\ref{phaseSpaceILED}. Then, we have
\begin{bigprop}\label{preciseILED}
Let $0\le a_0<M$, $0\le a\le a_0$, and let
$\psi$ be a solution of~(\ref{WAVE}) on $\mathcal{R}_0$ as in the reduction of
Section~\ref{closedILED}. Then,
\begin{align}\label{notCrudeILED1}
b\int_0^{\infty}\int_{\Sigma_{\tau}}\left(\left|\tilde Z^*\psi\right|^2r^{-1-\delta} + \left|\psi\right|^2r^{-3-\delta} +
\left|T\mathcal{P}_{\rm trap}[\xi\psi]\right|^2r^{-1-\delta} + \left|\nabb\mathcal{P}_{\rm trap}[\xi\psi]\right|^2r^{-1}\right)\, dr^*\, d\omega &\leq B\int_{\Sigma_0}\mathbf{J}^N_{\mu}[\psi]n^{\mu}_{\Sigma_0}.
\end{align}
where $r_{\rm trap}=r_{\rm trap}(\omega, m, \Lambda_{m\ell})$ is
defined in Theorem~\ref{phaseSpaceILED} and $\xi$ is the cutoff from Section~\ref{summation}.
\end{bigprop}
\begin{proof}One revisits the proof of Proposition~\ref{closedILED} and
simply retains the
nonnegative term on the left hand side of  $(\ref{notCrudeILED1})$
instead of applying the physical space $\zeta$ and the inequality $(\ref{insteadofthis})$ .
\end{proof}

\section{Energy boundedness}\label{boundSuff}

In this section, we establish the uniform boundedness
of the energy flux through $\Sigma_\tau$ for solutions $\psi$ to the wave equation~(\ref{WAVE}):
\begin{bigprop}\label{sufficientIntBound}
Let $0\le a_0<M$, $|a|\le a_0$ and
let $\psi$ be a solution of the wave equation~(\ref{WAVE}) on $\mathcal{R}_0$
as in the reduction of Section~\ref{WPosed}. Then
\[\int_{\Sigma_{\tau}}\mathbf{J}^N_{\mu}[\psi]n^{\mu}_{\Sigma_{\tau}} \leq B\int_{\Sigma_0}\mathbf{J}^N_{\mu}[\psi]n^{\mu}_{\Sigma_0}\quad \forall\ \tau \geq 0.\]
\end{bigprop}

First, recall that the arguments of Sections~\ref{summation} and~\ref{continuityargsec} have shown that $\psi$ is future-integrable and satisfies the integrated decay statements~(\ref{protasn1b}) and~(\ref{protasn1}).

Let $\delta > 0$ be a fixed small parameter, $A_0$ be sufficiently close to $r_+$ and $A_1$ be sufficiently large. The proof proceeds in three steps where the cases $r \in [A_0+\delta,A_1-\delta]$, $r \in [r_+,A_0+\delta]$ and $r \in [A_1-\delta,\infty)$ are each dealt with. As one expects, the first region is the most difficult.
\subsection{Boundedness of $\int_{\Sigma_{\tau}\cap [A_0+\delta,A_1-\delta]}\mathbf{J}^N_{\mu}[\psi]n^{\mu}_{\Sigma_{\tau}}$}
It turns out to be convenient to extend the solution to the entire domain of outer communication $\mathcal{R}$ from $\mathcal{R}_0$.
\subsubsection{Extending the solution}

The trace of $\psi$ and $N\psi$ along the hypersurface $\Sigma_0$ only suffice to determine $\psi$ in the future of $\Sigma_0$. However, an easy domain of dependence argument and finite in time energy estimates allow one to extend $\psi$ to the $\mathcal{R}$ in such a way as to guarantee
\begin{equation}\label{aRemark}
\int_{\hat{\Sigma}_0}\mathbf{J}^N_{\mu}[\psi]n^{\mu}_{\Sigma_{\tau}} \leq B\int_{\Sigma_0}\mathbf{J}^N_{\mu}[\psi]n^{\mu}_{\Sigma_0}.
\end{equation}
Here $\hat{\Sigma}_0$ denotes the image of $\Sigma_0$ under the Boyer-Lindquist coordinate defined map: $t\mapsto -t$.

\subsubsection{Integrated local energy decay for the extended solution}
Since the Boyer-Lindquist defined map $t\mapsto -t$ and $a\mapsto -a$ is an isometry, Proposition~\ref{closedILED} remains true if one goes to the past instead of the future, i.e. if we replace all integrals $\int_0^{\infty}$ with $\int_{-\infty}^0$, and replace $\Sigma_{\tau}$ with $\hat{\Sigma}_{\tau}$. Keeping~(\ref{aRemark}) in mind, we conclude
\begin{align}\label{ILEDextended}
&\int_{-\infty}^{\infty}\int_{\Sigma_{\tau}\cap [A_0,A_1]}\left(\left|\partial_{r^*}\psi\right|^2 + r^{-2}\left|\psi\right|^2 + \zeta\mathbf{J}^N_{\mu}[\psi]n^{\mu}_{\Sigma_{\tau}}\right)r^{-1-\delta}\, d\tau \leq B\left(A_0,A_1\right)\int_{\Sigma_0}\mathbf{J}^N_{\mu}[\psi]n^{\mu}_{\Sigma_0}.
\end{align}

Unfortunately, this version of integrated local energy decay is too crude for our purposes, and we shall need to appeal to the version~(\ref{notCrudeILED1}) of integrated local energy decay.

Let $\chi_{[A_0,A_1]}$ be a bump function which is identically $1$ when $r \in [A_0+\delta,A_1-\delta]$ and $0$ when $r \not\in [A_0,A_1]$. We define
\[\tilde\psi \doteq \chi_{[A_0,A_1]}\psi.\]
We will have
\[
\Box_{g_{a,M}}\tilde\psi = \tilde F \doteq \Box_g\chi_{[A_0,A_1]}\psi + 2\nabla^{\mu}\chi_{[A_0,A_1]}\nabla_{\mu}\psi.
\]
Observe that $\tilde F$ has compact support in $r$, and $\left|\tilde F\right|^2 \leq B\left(\left|\psi\right|^2 + \left|\partial_{r^*}\psi\right|^2\right)$. In particular, $\tilde\psi$ is
sufficiently integrable in the sense of Definition~\ref{sufficient} and outgoing
in the sense of Definition~\ref{sufficient2}. We also have
\begin{align}\label{tildeFEst}
\int_{-\infty}^{\infty}\int_{\Sigma_{\tau}}\left|\tilde F\right|^2 \leq B\int_{-\infty}^{\infty}\int_{\Sigma_{\tau}\cap \{[A_0,A_0+\delta]\cup[A_1-\delta,A_1]\}}\left(\left|\psi\right|^2 + \left|\partial_{r^*}\psi\right|^2\right) \leq B\int_{\Sigma_0}\mathbf{J}^N_{\mu}[\psi]n^{\mu}_{\Sigma_0}.
\end{align}

Now we apply Carter's separation as defined in Section~\ref{cartersupersection} and obtain
\[
\tilde{u}'' + \left(\omega^2 - V\right)\tilde{u} = \tilde{H}.
\]
Note that the compact $r$ support in $[A_0,A_1]$ of $\tilde{\psi}$ is inherited by $\tilde{u}$.
We apply  now Theorem~\ref{phaseSpaceILED}.
In view of the support of $\tilde u$, it follows that
the term $\left|\tilde u\left(-\infty\right)\right|^2$  vanishes.
Furthermore, the right hand sides of all the frequency localised multiplier estimates $(\ref{fromPhaseSpace2})$
are
$O\left(\tilde{H}\right)$, and hence are supported
in $[A_0,A_0+\delta]\cup[A_1-\delta,A_1]$. Consequently,
we can apply the (now trivial) arguments of Section~\ref{summation} to conclude the inhomogeneous version of~(\ref{notCrudeILED1})
\begin{align}\label{notCrudeILED}
b\int_{-\infty}^{\infty}\int_{\Sigma_{\tau}}\left(\left|\partial_{r^*}\tilde\psi\right|^2 + \left|\tilde\psi\right|^2 + \mathbf{J}^N_{\mu}[\mathcal{P}_{\rm trap}\tilde\psi]n^{\mu}_{\Sigma_{\tau}}\right)&\leq  \int_{-\infty}^{\infty}\int_{\Sigma_{\tau}\cap \{[A_0,A_0+\delta]\cup[A_1-\delta,A_1]\}}\left(\left|\psi\right|^2 + \left|\partial\psi\right|^2\right)
\\ \nonumber &\leq B\int_{\Sigma_0}\mathbf{J}^N_{\mu}[\psi]n^{\mu}_{\Sigma_0},
\end{align}
where we recall that $\mathcal{P}_{\rm trap}$ is defined by~(\ref{Ptrapnot}).

\subsubsection{A decomposition}
In order to work around the presence of $\mathcal{P}_{\rm trap}$ in~(\ref{notCrudeILED}), it will be useful to decompose $\tilde\psi$ is pieces, each of which experience trapping near a specific value of $r$. Recalling the definition of $r_{\rm trap}$ from Theorem~\ref{phaseSpaceILED}
we make the following definition.
\begin{definition}\label{decomp}Let $\epsilon > 0$ be a sufficiently small parameter to fixed later. We define
\[\mathcal{C}_0 := \left\{\left(\omega,m,\Lambda\right) : r_{trap}=0
\right\},\]
\[\mathcal{C}_i:= \left\{\left(\omega,m,\Lambda\right) : r_{trap} \in \left[3M - s^- + \left(i-1\right)\epsilon,3M-s^- +i\epsilon\right)\right\}\forall \quad i = 1,\ldots,\lceil \epsilon^{-1}\left(s^++s^-\right)\rceil.\]
\end{definition}

Observe that each value of $\left(\omega,m,\Lambda\right)$ lies in exactly one of the $\mathcal{C}_i$.

\begin{definition}We define $\tilde{\psi}_i$ by a phase space multiplication of $\tilde{\psi}$ by $1_{\mathcal{C}_i}$, the indicator function of $\mathcal{C}_i$:
\[\tilde{\psi}_i \doteq \frac{1}{\sqrt{2\pi}}\int_{-\infty}^\infty\sum_{m\ell} e^{-i\omega t} 1_{\mathcal{C}_i}\tilde{\psi}^{(a\omega)}_{m\ell}(r)S_{m\ell}(a\omega,\cos\theta)e^{im\phi} d\omega.\]
\end{definition}

Note that it immediately follows from Plancherel (see Section~\ref{separationSubsection}) that each $\tilde{\psi}_i$ is sufficiently integrable, and we have $\Box_{g_{a,M}}\tilde{\psi}_i = \tilde{F}_i$
where $\tilde{F}_i$ is defined in the same fashion as $\tilde{\psi}_i$.

It will be useful to observe the following.
\begin{proposition}\label{energyVanish}For each
$i = 0,\ldots,\lceil \epsilon^{-1}\left(s^++s^-\right)\rceil$
there exists a constant $C_i$ and a dyadic sequence $\{\tau^{(i)}_n\}_{n=1}^{\infty}$
such that $\tau^{(i)}_n \to -\infty$ as $n\to\infty$ and
\[\int_{\Sigma_{\tau^{(i)}_n}}\mathbf{J}^N_{\mu}\left[\tilde\psi_i\right]n^{\mu}_{\Sigma_{\tau^{(i)}_n}} \leq \frac{C_i}{\tau^{(i)}_n}.\]
\end{proposition}
\begin{proof}Since each $\tilde{\psi}_i$ is sufficiently integrable and compactly supported in $r$, we have
\[\int_{-\infty}^{\infty}\int_{\Sigma_{\tau}}\mathbf{J}^N_{\mu}\left[\tilde\psi_i\right]n^{\mu}_{\Sigma_{\tau}} < \infty.\]
The proof then concludes with a standard pigeonhole argument.
\end{proof}

\subsubsection{Boundedness}
Finally, we will establish boundedness of the energy of $\tilde\psi$.
\begin{proposition}\label{boundedEnergyInterior}
Under the assumptions of Proposition~\ref{sufficientIntBound}
and with the above notation we have
\[\int_{\Sigma_{\tau}}\mathbf{J}^N_{\mu}\left[\tilde\psi\right]n^{\mu}_{\Sigma_{\tau}} \leq B\int_{\Sigma_0}\mathbf{J}^N_{\mu}\left[\psi\right]n^{\mu}_{\Sigma_0}\quad \forall\ \tau \in (-\infty,\infty).\]
\end{proposition}
\begin{proof}
Since $\tilde\psi = \sum_{i=0}^{\left\lceil \epsilon^{-1}\left(s_++s_-\right)\right\rceil}\tilde{\psi}_i$, it suffices to prove the proposition with $\tilde\psi$ replaced by $\tilde{\psi}_i$.

In Proposition~\ref{timelikeVector} we showed that the vector field $T + \frac{2Mar}{\left(r^2+a^2\right)^2}\Phi$ is timelike in $\mathcal{R}\setminus \mathcal{H}^+$. Given this and taking $\epsilon$ from Definition~\ref{decomp} sufficiently small (and then fixing $\epsilon$), it is easy to construct a $\varphi_\tau$-invariant
timelike vector field $V_i$ on $\mathcal{R}$
which is Killing in the region
\[
r \in\left[3M - s^- + \left(i-1\right)\epsilon,3M-s^- + i\epsilon\right).
\]
Now we apply the energy identity associated to $V_i$ in between the hypersurfaces $\Sigma_{\tau}$ and $\Sigma_{\tau^{(i)}_n}$. Since $\mathcal{P}_{trap}\tilde\psi_i=\tilde\psi_i$ \underline{where $V_i$ is non-Killing}, we obtain
\begin{align}
\int_{\Sigma_{\tau}}\mathbf{J}^{V_i}_{\mu}\left[\tilde\psi_i\right]n^{\mu}_{\Sigma_{\tau}} &\leq B\int_{\tau^{(i)}_n}^{\tau}\int_{\Sigma_{s}\cap [3M-s^- + \left(i-1\right)\epsilon,3M-s^- + i\epsilon]^c}\mathbf{J}^N_{\mu}[\tilde\psi_i]n^{\mu}_{\Sigma_s} + \int_{\Sigma_{\tau^{(i)}_n}}\mathbf{J}^{V_i}_{\mu}\left[\tilde\psi_i\right]n^{\mu}_{\Sigma_{\tau}}
\\ \nonumber &\leq B\int_{-\infty}^{\infty}\int_{\Sigma_{s}\cap [3M-s^- + \left(i-1\right)\epsilon,3M-s^- + i\epsilon]^c}\mathbf{J}^N_{\mu}[\mathcal{P}_{\rm trap}\tilde\psi_i]n^{\mu}_{\Sigma_s} + \frac{BC_i}{\tau_n^{(i)}}
\\ \nonumber &\leq B\int_{-\infty}^{\infty}\int_{\Sigma_{s}}\mathbf{J}^N_{\mu}[\mathcal{P}_{trap}\tilde\psi]n^{\mu}_{\Sigma_s} + \frac{BC_i}{\tau_n^{(i)}}
\\ \nonumber &\leq B\int_{\Sigma_0}\mathbf{J}^N_{\mu}\left[\psi\right]n^{\mu}_{\Sigma_0} + \frac{BC_i}{\tau_n^{(i)}},
\end{align}
where we have used $(\ref{notCrudeILED})$ as well as Plancherel.
It remains to take $n \to \infty$ and to observe (the trivial fact) that,
in view of the support of $\tilde{\psi}_i$ and the $\phi_\tau$-invariance of $V_i$ we have $\mathbf{J}^{V_i}_{\mu}\left[\tilde\psi_i\right]n^{\mu}_{\Sigma_{\tau}} \sim \mathbf{J}^{N}_{\mu}\left[\tilde\psi_i\right]n^{\mu}_{\Sigma_{\tau}}$.
\end{proof}

\subsection{Boundedness of $\int_{\Sigma_{\tau}\cap [r_+,A_0+\delta]}\mathbf{J}^N_{\mu}[\psi]n^{\mu}_{\Sigma_{\tau}}$ and $\int_{\Sigma_{\tau}\cap [A_1-\delta,\infty)}\mathbf{J}^N_{\mu}[\psi]n^{\mu}_{\Sigma_{\tau}}$}
The following is a trivial consequence of the red-shift estimate (Proposition~\ref{ftrs})
and Proposition~\ref{closedILED}:
\begin{align}\label{boundedEnergyHorizon}
\int_{\Sigma_{\tau}\cap [r_+,A_0+\delta)}\mathbf{J}^N_{\mu}\left[\psi\right]n^{\mu}_{\Sigma_{\tau}} &\leq \int_{\Sigma_0}\mathbf{J}^N_{\mu}[\psi]n^{\mu}_{\Sigma_0} + B\int_0^{\tau}\int_{\Sigma_s\cap [A_0+\delta,A_0+2\delta]}\mathbf{J}^N_{\mu}[\psi]n^{\mu}_{\Sigma_s}
\\ \nonumber &\leq B\int_{\Sigma_0}\mathbf{J}^N_{\mu}[\psi]n^{\mu}_{\Sigma_0}.
\end{align}

Similarly, we may consider the energy estimate associated to $\chi_{A_1-\delta}T$ where $\chi_{A_1-\delta}$ is a cutoff which is identically $1$ on $[A_1-\delta,\infty)$ and identically $0$ on $[r_+,A_1-2\delta]$. We obtain
\begin{align}\label{boundedEnergyInfinity}
\int_{\Sigma_{\tau}\cap [A_1-\delta,\infty)}\mathbf{J}^T_{\mu}\left[\psi\right]n^{\mu}_{\Sigma_{\tau}} &\leq \int_{\Sigma_0}\mathbf{J}^N_{\mu}[\psi]n^{\mu}_{\Sigma_0} + B\int_0^{\tau}\int_{\Sigma_s\cap [A_1-2\delta,A_1-\delta]}\mathbf{J}^N_{\mu}[\psi]n^{\mu}_{\Sigma_s}
\\ \nonumber &\leq B\int_{\Sigma_0}\mathbf{J}^N_{\mu}[\psi]n^{\mu}_{\Sigma_0}.
\end{align}

\subsection{Putting everything together and the higher order statement}
Combining Proposition~\ref{boundedEnergyInterior},~(\ref{boundedEnergyHorizon}) and~(\ref{boundedEnergyInfinity}) concludes the proof of Proposition~\ref{sufficientIntBound}.
In view of Section~\ref{theLogic},
this completes the proof of Theorem~\ref{theResult}.

For Theorem~\ref{h.o.s.}, we are left only with proving
the higher order version of Proposition~\ref{sufficientIntBound}:
\begin{proposition}
With the notation of Proposition~\ref{sufficientIntBound},
for every $j \geq 1$
\begin{equation}
\int_{\Sigma_\tau } \sum_{0\le i \le j-1}{\bf J}^N_\mu[N^{i}\psi]n^\mu_{\Sigma_\tau}
\le  B(j)\int_{\Sigma_0} {\sum_{0\le i \le j-1}
{\bf J}^N_\mu[N^{i}\psi]n^\mu_{\Sigma_0}}, \qquad
\forall\tau\ge 0.
\end{equation}
\end{proposition}
\begin{proof}We will be brief, since we have already seen multiple times how to upgrade lower order statements to higher order ones. As usual, we will only consider the case $j=2$ as the general case will follow by an easy induction argument.

First we commute $(\ref{WAVE})$ with $T$ and apply Proposition~\ref{sufficientIntBound}. We obtain
\begin{equation}
\int_{\Sigma_\tau }{\bf J}^N_\mu[T\psi]n^\mu_{\Sigma_\tau}
\le  B(j)\int_{\Sigma_0} {
{\bf J}^N_\mu[T\psi]n^\mu_{\Sigma_0}}, \qquad
\forall\tau\ge 0.
\end{equation}

Next, letting $\chi$ be a cutoff which vanishes for large $r$, we commute with $\chi\Phi$. Using the integrated energy decay to the handle resulting error terms, we obtain
\begin{equation}
\int_{\Sigma_\tau }{\bf J}^N_\mu[\chi\Phi\psi]n^\mu_{\Sigma_\tau}
\le  B(j)\int_{\Sigma_0}\left( {
{\bf J}^N_\mu[N\psi]n^\mu_{\Sigma_0} + {\bf J}^N_{\mu}[\psi]n^{\mu}_{\Sigma_0}}\right), \qquad
\forall\tau\ge 0.
\end{equation}

Finally, we commute with the red-shift commutation vector field $Y$ and apply the argument from the proofs of Lemma~\ref{higherPrepLemm} and Proposition~\ref{h.o.s.suff} to establish
\begin{equation}
\int_{\Sigma_\tau }{\bf J}^N_\mu[Y\psi]n^\mu_{\Sigma_\tau}
\le  B(j)\int_{\Sigma_0}\left( {
{\bf J}^N_\mu[N\psi]n^\mu_{\Sigma_0} + {\bf J}^N_{\mu}[\psi]n^{\mu}_{\Sigma_0}}\right), \qquad
\forall\tau\ge 0.
\end{equation}
The proof concludes via standard elliptic estimates (see the proofs of Lemma~\ref{higherPrepLemm} and Proposition~\ref{h.o.s.suff}).
\end{proof}

In view of Section~\ref{theLogic}, this obtains the remaining
statement $(\ref{bndts1})$  of Theorem~\ref{h.o.s.}.
The proof of both main theorems is thus complete.

\end{document}

%% file: partiii.pstex_t
\begin{picture}(0,0)%
\includegraphics{partiii.pstex}%
\end{picture}%
\setlength{\unitlength}{3158sp}%
\begingroup\makeatletter\ifx\SetFigFont\undefined%
\gdef\SetFigFont#1#2#3#4#5{%
  \reset@font\fontsize{#1}{#2pt}%
  \fontfamily{#3}\fontseries{#4}\fontshape{#5}%
  \selectfont}%
\fi\endgroup%
\begin{picture}(2224,2222)(3839,-5372)
\put(4126,-4786){\rotatebox{315.0}{\makebox(0,0)[lb]{\smash{{\SetFigFont{10}{12.0}{\rmdefault}{\mddefault}{\updefault}{\color[rgb]{0,0,0}$\mathcal{H}^-$}%
}}}}}
\put(5476,-5086){\rotatebox{45.0}{\makebox(0,0)[lb]{\smash{{\SetFigFont{10}{12.0}{\rmdefault}{\mddefault}{\updefault}{\color[rgb]{0,0,0}$\mathcal{I}^-$}%
}}}}}
\put(5476,-3661){\rotatebox{315.0}{\makebox(0,0)[lb]{\smash{{\SetFigFont{10}{12.0}{\rmdefault}{\mddefault}{\updefault}{\color[rgb]{0,0,0}$\mathcal{I}^+$}%
}}}}}
\put(4170,-3837){\rotatebox{45.0}{\makebox(0,0)[lb]{\smash{{\SetFigFont{10}{12.0}{\rmdefault}{\mddefault}{\updefault}{\color[rgb]{0,0,0}$\mathcal{H}^+_0$}%
}}}}}
\put(4576,-4336){\rotatebox{350.0}{\makebox(0,0)[lb]{\smash{{\SetFigFont{10}{12.0}{\rmdefault}{\mddefault}{\updefault}{\color[rgb]{0,0,0}$\Sigma_0$}%
}}}}}
\put(4651,-4786){\makebox(0,0)[lb]{\smash{{\SetFigFont{10}{12.0}{\rmdefault}{\mddefault}{\updefault}{\color[rgb]{0,0,0}$\mathcal{R}$}%
}}}}
\put(4376,-3996){\makebox(0,0)[lb]{\smash{{\SetFigFont{10}{12.0}{\rmdefault}{\mddefault}{\updefault}{\color[rgb]{0,0,0}$\mathcal{R}_0$}%
}}}}
\put(4915,-3870){\rotatebox{340.0}{\makebox(0,0)[lb]{\smash{{\SetFigFont{10}{12.0}{\rmdefault}{\mddefault}{\updefault}{\color[rgb]{0,0,0}$\Sigma_\tau$}%
}}}}}
\end{picture}%

%% file: partiii2.pstex_t
\begin{picture}(0,0)%
\includegraphics{partiii2.pstex}%
\end{picture}%
\setlength{\unitlength}{3158sp}%
\begingroup\makeatletter\ifx\SetFigFont\undefined%
\gdef\SetFigFont#1#2#3#4#5{%
  \reset@font\fontsize{#1}{#2pt}%
  \fontfamily{#3}\fontseries{#4}\fontshape{#5}%
  \selectfont}%
\fi\endgroup%
\begin{picture}(2220,2222)(3839,-5372)
\put(4126,-4786){\rotatebox{315.0}{\makebox(0,0)[lb]{\smash{{\SetFigFont{10}{12.0}{\rmdefault}{\mddefault}{\updefault}{\color[rgb]{0,0,0}$\mathcal{H}^-$}%
}}}}}
\put(5476,-5086){\rotatebox{45.0}{\makebox(0,0)[lb]{\smash{{\SetFigFont{10}{12.0}{\rmdefault}{\mddefault}{\updefault}{\color[rgb]{0,0,0}$\mathcal{I}^-$}%
}}}}}
\put(4205,-3838){\rotatebox{45.0}{\makebox(0,0)[lb]{\smash{{\SetFigFont{10}{12.0}{\rmdefault}{\mddefault}{\updefault}{\color[rgb]{0,0,0}$\mathcal{H}^+_0$}%
}}}}}
\put(5564,-3719){\rotatebox{315.0}{\makebox(0,0)[lb]{\smash{{\SetFigFont{10}{12.0}{\rmdefault}{\mddefault}{\updefault}{\color[rgb]{0,0,0}$\mathcal{I}^+_0$}%
}}}}}
\put(4651,-4786){\makebox(0,0)[lb]{\smash{{\SetFigFont{10}{12.0}{\rmdefault}{\mddefault}{\updefault}{\color[rgb]{0,0,0}$\mathcal{R}$}%
}}}}
\put(4651,-3692){\rotatebox{340.0}{\makebox(0,0)[lb]{\smash{{\SetFigFont{10}{12.0}{\rmdefault}{\mddefault}{\updefault}{\color[rgb]{0,0,0}$\widetilde{\Sigma}_\tau$}%
}}}}}
\put(4426,-3961){\makebox(0,0)[lb]{\smash{{\SetFigFont{10}{12.0}{\rmdefault}{\mddefault}{\updefault}{\color[rgb]{0,0,0}$\mathcal{R}_0$}%
}}}}
\put(4401,-4255){\rotatebox{350.0}{\makebox(0,0)[lb]{\smash{{\SetFigFont{10}{12.0}{\rmdefault}{\mddefault}{\updefault}{\color[rgb]{0,0,0}$\widetilde{\Sigma}_0$}%
}}}}}
\end{picture}%

%% file: PARTIIIsuperimproved.bbl
\begin{thebibliography}{99}
\bibitem{alexakis} S. Alexakis, A. Ionescu and S. Klainerman
\emph{Uniqueness of smooth stationary black holes in vacuum:
small perturbations of the Kerr spaces},  Comm. Math. Phys. {\bf 299} (2010), no. 1, 89--127.

\bibitem{alinhac} S. Alinhac \emph{Energy multipliers for perturbations of Schwarzschild metric}
Comm. Math. Phys. {\bf 288} (2009), no. 1, 199--224.


\bibitem{anblue} L. Andersson and P. Blue
\emph{Hidden symmetries and decay for the wave equation on the Kerr spacetime},
arXiv:0908.2265.

\bibitem{anblue2}
L. Andersson and P. Blue
\emph{Uniform energy bound and asymptotics for the Maxwell field on a slowly rotating Kerr black hole exterior}, arXiv:13102664

\bibitem{AndGlam}
N. Andersson and K. Glampedakis \emph{A superradiance resonance cavity outside rapidly rotating black holes} Phys. Rev.
Lett. {\bf 84} (2000), 4537--4540.

\bibitem{aretakisKerr} S. Aretakis \emph{Decay of axisymmetric solutions of the wave equation
on extreme Kerr backgrounds}, Journal of Functional Analysis, {\bf 263} (2012), no. 9, 2770--2831.

\bibitem{aretakisHor} S. Aretakis \emph{Horizon Instabilities of Extremal Black Holes},
to appear in ATMP, arXiv:1206.6598.

\bibitem{aretakis} S. Aretakis \emph{Stability and instability of extreme Reissner-Nordstr\"om black hole spacetimes for linear scalar perturbations I}, Comm. Math. Phys. {\bf 307} (2011), no. 1, 17--63.

\bibitem{aretakis2} S. Aretakis \emph{Stability and instability of extreme Reissner-Nordstr\"om black hole spacetimes for linear scalar perturbations II}, Ann. Henri Poincar\'e {\bf 12} (2011), no. 8, 1491--1538.


\bibitem{2bachelots}
 A. Bachelot and A. Motet-Bachelot \emph{Les r\'esonances d'un trou noir de Schwarzschild}
 Ann. Inst. H. Poincar\'e Phys. Th\'eor. {\bf 59} (1993), no. 1, 368.

\bibitem{BlueSof0} P. Blue and A. Soffer  \emph{Semilinear wave
equations on the Schwarzschild manifold. I. Local decay
estimates}, Adv. Differential Equations {\bf 8} (2003), no. 5, 595--614.

\bibitem{BlueSof} P. Blue and A. Soffer \emph{Errata for ``Global existence and scattering for the nonlinear Schrodinger equation on Schwarzschild manifolds'', ``Semilinear wave equations on the Schwarzschild manifold I: Local Decay Estimates'', and ``The wave equation on the Schwarzschild metric II: Local Decay for the spin 2 Regge Wheeler equation''}, gr-qc/0608073, 6 pages.


\bibitem{BlueSter} P. Blue and J. Sterbenz  \emph{Uniform decay
of local energy and the semi-linear wave equation
on Schwarzschild space} Comm. Math. Phys. {\bf 268} (2006), no. 2,
481--504.

\bibitem{bh} J. F. Bony and D. H\"afner \emph{Decay and non-decay of the local
energy for the wave equation in the de Sitter-Schwarzschild metric}
Comm. Math. Phys. {\bf 282} (2008), no. 3, 697--719.

\bibitem{carter}
B. Carter \emph{Black hole equilibrium states}, in Black Holes (Les Houches
Lectures), edited B.~S.~DeWitt and C.~DeWitt (Gordon and Breach, New York, 1972).

\bibitem{cartersep2}
B. Carter \emph{Hamilton-Jacobi and Schr\"odinger separable solutions of Einstein's
equations} Comm. Math. Phys. {\bf 10} (1968), 280--310.


\bibitem{chandrasekhar} S. Chandrasekhar \emph{The mathematical theory of
black holes}, Oxford University Press, 1983.

\bibitem{civin}
D. Civin \emph{Stability of subextremal Kerr--Newman exterior spacetimes for linear scalar perturbations}, preprint, 2014

\bibitem{book2} D. Christodoulou \emph{The action principle and partial
differential equations}, Ann. Math. Studies No. 146, 1999

\bibitem{ck}
D. Christodoulou and S. Klainerman \emph{The global nonlinear
stability of the Minkowski space}, Princeton University Press, 1993

\bibitem{cbh} M. Dafermos
\emph{The interior of charged black holes and the problem of uniqueness in
general relativity} Comm. Pure Appl. Math. {\bf 58} (2005), 0445--0504.


\bibitem{dr1} M. Dafermos and I. Rodnianski
\emph{A proof of Price's law for the collapse of a
self-gravitating scalar field}, Invent. Math. {\bf 162}
(2005), 381--457.



\bibitem{dr3} M. Dafermos and I. Rodnianski
\emph{The redshift effect and radiation decay on black hole
spacetimes} Comm. Pure Appl. Math. {\bf 52} (2009), 859--919.

\bibitem{dr4} M. Dafermos and I. Rodnianski
\emph{The wave equation on  Schwarzschild-de Sitter spacetimes},
arXiv:0709.2766v1 [gr-qc].

\bibitem{dr5} M. Dafermos and I. Rodnianski
\emph{A note on energy currents and decay for the wave equation
on a Schwarzschild background}, arXiv:0710.0171v1 [math.AP].


\bibitem{dr6} M. Dafermos and I. Rodnianski
\emph{A proof of the uniform boundedness of solutions to the wave equation on slowly rotating Kerr backgrounds}, Invent. Math. {\bf 185} (2011) , no. 3, 467--559.

\bibitem{jnotes} M.~Dafermos and I.~Rodnianski
\emph{Lectures on black holes and linear waves}, Evolution equations, Clay Mathematics Proceedings, Vol. 17. Amer. Math. Soc., Providence, RI, 2013, pp. 97--205, arXiv:0811.0354 [gr-qc].


\bibitem{icmp} M. Dafermos and I. Rodnianski
\emph{A new physical-space approach to decay for the wave
equation with applications to black hole spacetimes}, in
XVIth International Congress on Mathematical Physics, P. Exner (ed.),
World Scientific, London, 2009, arXiv:0910.4957v1 [math.AP].


\bibitem{stabi} M. Dafermos and I. Rodnianski \emph{The black hole stability
problem for linear scalar perturbations}, in Proceedings of the Twelfth Marcel Grossmann Meeting on General Relativity, T. Damour et al (ed.), World Scientific, Singapore, 2011, pp. 132–189, arXiv:1010.5137


\bibitem{dr7} M. Dafermos and I. Rodnianski
\emph{Decay for solutions of the wave equation
on Kerr exterior spacetimes {I--II}: The cases $|a|\ll M$ or axisymmetry},
arXiv:1010.5132

\bibitem{scattering}
M. Dafermos, G. Holzegel and I. Rodnianski
\emph{A scattering theory construction of dynamical black hole spacetimes},
arXiv:1306.5534

\bibitem{dyatlov1} S. Dyatlov
\emph{Quasi-normal modes and exponential energy decay for the for the Kerr--de Sitter black hole} Comm. Math. Physics {\bf 306} (2011), 119--163.

\bibitem{dyatlov2} S. Dyatlov
\emph{Exponential energy decay for Kerr–de Sitter black holes beyond event horizons}, Mathematical Research Letters {\bf 18} (2011), 1023--1035

\bibitem{dyatlov-last}
S. Dyatlov
\emph{Asymptotics of linear waves and resonances with applications to black holes},
arXiv:1305.1723

\bibitem{fksy} F. Finster, N. Kamran, J. Smoller,
S-T.~Yau \emph{Decay of solutions of the wave equation in Kerr geometry}
Comm. Math. Phys. {\bf 264} (2006), 465--503.

\bibitem{fksy2} F. Finster, N. Kamran, J. Smoller, S-T. Yau
\emph{Erratum: Decay of solutions of the wave equation in Kerr geometry}
Comm. Math. Phys., online first.

\bibitem{franzen} A. Franzen, \emph{The wave equation on black hole interiors},
Ph.D.~Thesis, 2014

\bibitem{gannot} O. Gannot
\emph{Quasinormal modes for Schwarzschild-AdS black holes: exponential convergence to the real axis}, arXiv:1212:1907

\bibitem{he:lssst}
S.~W.~Hawking and G.~F.~R.~Ellis
\emph{The large scale structure of space-time} Cambridge
Monographs on Mathematical Physics, No. 1. Cambridge
University Press, London-New York, 1973.

\bibitem{kostakis2}
G. Holzegel, \emph{Ultimately Schwarzschildean spacetimes
and the black hole stability problem}, arXiv:1010.3216.

\bibitem{holz-smul} G.~Holzegel and J.~Smulevici \emph{Decay properties of Klein-Gordon fields on Kerr-AdS spacetimes} Comm. Pure Appl. Math. {\bf 66} (2013), no. 11, 1751--1802

\bibitem{holz-smul2} G. Holzegel and J. Smulevici
\emph{Quasimodes and a Lower Bound on the Uniform Energy Decay Rate for Kerr-AdS Spacetimes}, arXiv:1303.5944

\bibitem{kw:lss} B. Kay and R. Wald
\emph{Linear stability of Schwarzschild under perturbations
which are nonvanishing on the bifurcation $2$-sphere}
Classical Quantum Gravity {\bf 4} (1987), no. 4, 893--898.

\bibitem{muchT} S. Klainerman \emph{Uniform decay estimates and the Lorentz
invariance of the classical wave equation} Comm. Pure Appl. Math. {\bf 38} (1985),
321--332

\bibitem{labasoffer} I. Laba and A. Soffer \emph{Global existence and scattering
for the nonlinear Schr\"odinger equation on Schwarzschild manifolds}
Helv. Phys. Acta {\bf 72} (1999), no. 4, 272--294.

\bibitem{lauletal} P. Laul, J. Metcalfe, S. Tikare and M. Tohaneanu
\emph{Localized energy estimates on Myers--Perry space-times},
arXiv:1401.0465

\bibitem{luciet} J. Lucietti and H. S. Reall
\emph{Gravitational instability of an extreme Kerr black hole}
Phys. Rev. D {\bf 86} (2012), 104030

\bibitem{luk} J. Luk \emph{Improved decay for solutions to the linear wave equation
on a Schwarzschild black hole}, Ann. Henri Poincar\'e {\bf 11} (2010), no. 5, 805--880.

\bibitem{luk2} J. Luk \emph{A Vector Field Method Approach to Improved Decay for Solutions to the Wave Equation on a Slowly Rotating Kerr Black Hole}, Anal. PDE {\bf 5} (2012), no. 3, 553--625.

\bibitem{luk3} J. Luk \emph{The null condition and global existence for nonlinear wave equations on slowly rotating Kerr spacetimes}, JEMS {\bf 15} (2013), no. 5, 1629--1700.

\bibitem{sbvm} R. Melrose, A. S\'a Barreto, A. Vasy
\emph{Asymptotics of solutions of the wave equation on de Sitter-Schwarzschild space},
arXiv:0811.2229.

\bibitem{mora2}
C. S. Morawetz \emph{Time decay for the nonlinear Klein-Gordon equations}
Proc. Roy. Soc. Ser. A {\bf 206} (1968), 291--296.

\bibitem{moschidis} G. Moschidis, forthcoming.

\bibitem{murataetal}
K. Murata, H. S. Reall and N. Tanahashi
\emph{What happens at the horizon(s) of an extreme black hole?}
Class. Quantum Grav. {\bf 30} (2013) 235007

\bibitem{Ralston} J. Ralston \emph{Solutions of the wave equation with localized energy}
Comm. Pure Appl. Math. {\bf 22} (1969), 807--823.

\bibitem{Sbierski} J. Sbierski
\emph{Characterisation of the Energy of Gaussian Beams on Lorentzian Manifolds - with Applications to Black Hole Spacetimes}, arXiv:1311.2477

\bibitem{schlue} V. Schlue
\emph{Decay of linear waves on higher dimensional Schwarzschild black holes},
Analysis \& PDE {\bf 6}, (2013), 515--600

\bibitem{schlue2} V. Schlue
\emph{Global results for linear waves on expanding Kerr and Schwarzschild de Sitter
cosmologies}, arXiv:1207.6055v2

\bibitem{shlapRot} Y. Shlapentokh-Rothman \emph{Quantitative Mode Stability for the Wave Equation on the Kerr Spacetime}, arXiv:1302.6902, to appear in Ann. Henri Poincar\'e.

\bibitem{shlapRot2} Y. Shlapentokh-Rothman \emph{Exponentially growing finite energy solutions for the Klein-Gordon equation on sub-extremal Kerr spacetimes}, arXiv:1302.3448,
to appear in Commun. Math. Phys.

\bibitem{Starobinsky} A. Starobinsky \emph{Amplification of waves during reflection from a black hole} Soviet Physics JETP {\bf 37} (1973), 28--32.

\bibitem{tatar} D. Tataru \emph{Local decay of waves on asymptotically
flat stationary space-times}, Amer. J. Math. {\bf 135} (2013), no. 2, 361--401.

\bibitem{tattoh} D. Tataru and M. Tohaneanu \emph{Local energy estimate on
Kerr black hole backgrounds}, IMRN (2011), no. 2, 248--292.

\bibitem{vasy} A. Vasy \emph{Microlocal analysis of asymptotically hyperbolic and Kerr-de Sitter spaces, with an appendix by Semyon Dyatlov}
Invent. Math. {\bf 194} (2013), 381--513.

\bibitem{whiting} B. Whiting
\emph{Mode stability of the Kerr black hole} J. Math. Phys. {\bf 30} (1989), 1301.

\bibitem{WunschZworski} J. Wunsch and M. Zworski
\emph{Resolvent estimates for normally hyperbolic trapped sets}
Ann. Inst. Henri Poincar\'e (A), {\bf 12} (2011), 1349--1385.

\bibitem{Yang} S. Yang \emph{Global solutions to nonlinear wave equations in
time dependent inhomogeneous media}, Arch. Ration. Mech. Anal. {\bf 209} (2013), no. 2, 683--728.

\bibitem{shiwu2} S. Yang \emph{Global stability of solutions to nonlinear wave equations}, arXiv:1205.4216.

\bibitem{shiwu}
S. Yang
\emph{On the quasilinear wave equations in time dependent inhomogeneous media},
arXiv:1312.7246.


\end{thebibliography}
